%% file: Granular_IV_v7_3_merged.tex
\documentclass[12pt,letterpaper]{article}
\usepackage[margin=1in]{geometry}
\usepackage{setspace}
\usepackage{amsmath,amssymb,amsfonts,amsthm,bm}
\usepackage{graphicx}
\usepackage{booktabs}
\usepackage{tabularx}
\usepackage{threeparttable}
\usepackage{array}
\usepackage{makeidx}
\usepackage{eurosym}
\usepackage{euscript}
\usepackage{chngcntr}
\usepackage{multirow}
\usepackage{caption}
\usepackage{algorithm}
\usepackage{natbib}
\usepackage{xr}
\usepackage[usenames]{color}

\providecommand{\U}[1]{\protect \rule{.1in}{.1in}}

\allowdisplaybreaks
\newtheorem{theorem}{Theorem}

\newtheorem{assumption}{Assumption}

\newtheorem{lemma}{Lemma}

\newtheorem{remark}{Remark}
\newtheoremstyle{ecta}
{\medskipamount}{\bigskipamount}{\normalfont}{1.4em}{\scshape}{:}{1em}{}
\theoremstyle{ecta}
\newtheorem*{example*}{Example}

\newcolumntype{Y}{>{\centering\arraybackslash}X}
\input{tcilatex}

\begin{document}

\title{Granular Instrumental Variables: Estimation and Inference}
\author{Jinyong Hahn\thanks{%
Department of Economics, UCLA, Los Angeles, CA 90095-1477 USA. Email:\
hahn@econ.ucla.edu} \\
UCLA \and Niu He\thanks{%
Department of Economics, UCLA, Los Angeles, CA 90095-1477 USA. Email:\
niuhe@g.ucla.edu} \\
UCLA \and Zhipeng Liao\thanks{%
Department of Economics, UCLA, Los Angeles, CA 90095-1477 USA. Email:\
zhipeng.liao@econ.ucla.edu} \\
UCLA \and Wenyu Zhou\thanks{%
International Business School, Zhejiang University, Haining, Zhejiang
314400, China. Email: wenyuzhou@intl.zju.edu.cn.} \\
Zhejiang University}
\date{\today }
\maketitle

\begin{abstract}
We develop an estimation and inference framework for granular instrumental variables (GIVs) in models with latent aggregate shocks. Our key insight is that valid GIVs are characterized by the orthogonal complement of the factor-loading space. This characterization yields a feasible procedure for constructing GIVs when factor loadings are unknown and does not require a large cross-sectional dimension. We provide practical procedures for inference and specification testing, and apply the framework to estimate the aggregate equity market multiplier. Our empirical results reveal substantial heterogeneity in equity demand elasticities across investor sectors and may provide nuanced support for the inelastic-markets hypothesis.

\bigskip

\bigskip

\end{abstract}

\section{Introduction\label{sec:intro}}

Understanding how aggregate outcomes respond to shocks is a central
objective in economics and finance. In many applications, researchers
observe a large cross section of entities that are simultaneously exposed to
a small number of common shocks. Examples include firms responding to
aggregate demand conditions, financial institutions adjusting portfolios in
response to market forces, and countries reacting to global macroeconomic
shocks. A common feature of these environments is that the variables of
interest are jointly determined in equilibrium, creating endogeneity
problems that complicate identification and estimation of structural
parameters.

Recently, \cite{gabaix2024granular} proposed a novel identification strategy
based on granular instrumental variables (GIVs). The approach builds on the
insight from the granularity literature that when a small number of firms,
industries, countries, investors, or borrowers account for a non-negligible
share of aggregate activity, idiosyncratic shocks to these units may survive
aggregation and influence aggregate outcomes.\footnote{%
See, among others, \cite{gabaix2011granular,
acemoglu2012network,di2014firms, baqaee2019macroeconomic,gaubert2021granular}%
.} Exploiting this feature, GIV extracts the idiosyncratic component of
observables after controlling for common factors and aggregates these
components using size weights to construct instruments for causal parameters
such as elasticities and multipliers. Under suitable restrictions on the
covariance structure of the idiosyncratic shocks, the resulting instruments
are orthogonal to equilibrium disturbances and can therefore identify
structural parameters. Unlike earlier approaches that use idiosyncratic
shocks to variables excluded from the estimating equation,\footnote{%
See, for instance, \cite{leary2014peer,amiti2018much,amiti2019international}.%
} GIV constructs instruments from the idiosyncratic component of the
variables entering the estimating equation itself. As a result, the
methodology does not rely on traditional excluded instruments, which are
often difficult to justify or unavailable in practice.

Several recent papers extend the baseline GIV framework. %
\citet{banafti2022inferential} study inference in large-$n$, large-$T$
settings with unknown factors and loadings. \citet{baumeister2023uncovering}
develop a likelihood-based approach, while \citet{qian2023heterogeneity}
allows for heterogeneous spillovers. GIV has also become an important
identification device in empirical macroeconomics and finance, including
applications to stock-market demand, exchange rates, bank lending, and asset
pricing.\footnote{%
Recent applications include \citet{galaasen2020granular,
camanho2022global,ma2022mutual,dong2025fast}; many others are discussed in %
\citet{gabaix2024granular}.}

Despite its growing importance, several econometric questions remain
unresolved. The central challenge is that valid GIV construction requires
knowledge of the factor-loading space associated with latent aggregate
shocks, which is rarely observed in practice. Existing implementations
therefore estimate latent factors and construct GIVs in a second step, a
strategy that typically relies on a large cross-sectional dimension and
strong normalization assumptions. Moreover, little is known about the
consequences of estimating the factor-loading space for identification,
estimation, and inference.

This paper develops an estimation and inference framework for structural
models identified by GIVs when the factor-loading space is unknown. Rather
than estimating latent factors and constructing instruments in a second
step, we show that the relevant GIV space can be recovered directly from the
covariance structure of the observables. Specifically, the admissible GIVs
are generated by the orthogonal complement of the factor-loading space,
which can be identified from the eigenspace associated with the smallest
eigenvalues of the covariance matrix. This insight transforms the
construction of GIVs into a covariance-based problem and yields a feasible
procedure for constructing instruments directly from the data. The resulting
estimator remains valid even when the number of entities is fixed and
therefore does not require the cross-sectional dimension to diverge with the
sample size.

The characterization also provides a transparent identification strategy. We
show that all admissible GIVs are generated by the orthogonal complement of
the column space spanned by the factor loadings and the vector of ones. When
the factor-loading space is unknown, the relevant orthogonal complement can
be consistently recovered from the eigenspace associated with the smallest
eigenvalues of the covariance matrix of the observables. This result yields
a feasible GIV estimator and forms the basis for inference with estimated
GIVs.

Building on this characterization, we establish consistency and asymptotic
normality of the feasible GIV estimator and develop practical procedures for
inference and specification testing. In particular, we derive feasible
standard errors, establish the asymptotic validity of an over-identification $%
J$-test when the GIVs are estimated, and propose a BIC-type criterion for
determining the dimension of the factor-loading space. Monte Carlo
simulations show that the feasible estimator performs similarly to an oracle
estimator that knows the true factor-loading space and that the proposed
inference procedures perform well in finite samples.

Beyond estimation and inference, the paper clarifies several identification
issues that arise when factor loadings are unknown. We show that certain
restrictions commonly interpreted as normalizations on the factor loadings
instead impose substantive restrictions on the latent factors. We also
examine the identification strategy in \citet{gabaix2024granular} and show
that the moment conditions in their Proposition~7 may fail to identify
the structural parameters when factor loadings are unknown. These findings
highlight the challenges of identification in the presence of latent
aggregate shocks and motivate the alternative characterization developed in
this paper. The proposed framework avoids these restrictions and extends
naturally to settings with additional exogenous regressors, unbalanced
panels, and heterogeneous demand elasticities.

Finally, we return to the estimation of the aggregate equity market multiplier in demand-based asset pricing \citep{gabaix2021search}. This application is particularly relevant because the factor-loading space is unknown and only twelve investor sectors are available in the data. As a result, the large-cross-section justification underlying existing GIV procedures is difficult to invoke directly, making the setting a natural environment in which to assess the practical importance of estimating GIVs when factor loadings are unobserved. Applying our framework, we obtain estimates and conduct inference for the aggregate multiplier together with a formal specification test of the underlying GIV moment conditions. The empirical results provide evidence consistent with highly inelastic aggregate equity demand.

The remainder of the paper is organized as follows. Section~\ref{sec:
S_model} introduces the main ideas in a simplified framework. Section~\ref%
{sec: G_model} develops the general model, establishes identification, and
presents estimation and inference procedures based on estimated GIVs.
Section~\ref{sec: Extension} studies identification when factor loadings are
unknown, demonstrates a failure of identification in the moment conditions
proposed by \citet{gabaix2024granular}, and develops extensions to models
with exogenous regressors, unbalanced panels, and heterogeneous demand
elasticities. Section~\ref{sec: MC} reports Monte Carlo evidence, Section~%
\ref{sec:emp} presents an empirical application to the aggregate equity
market multiplier, and Section~\ref{sec:conclusion} concludes. Proofs and
additional technical and empirical results are collected in the Online
Appendix.

\textit{Notation.} We use $K$ to denote a generic strictly positive constant that may vary from place to place but does not depend on the sample size $T$. We write $a\equiv b$ to indicate that $a$ is defined as $b$. For any positive integer $k$, let $\mathbf{I}_k$, $\mathbf{1}_k$, and $\mathbf{0}_k$ denote the $k\times k$ identity matrix, the $k\times1$ vector of ones, and the $k\times1$ vector of zeros, respectively. For any vector $x_t\in\mathbb{R}^n$ (possibly indexed by $t$) and any weight vector $\omega\in\mathbb{R}^n$ satisfying $\omega^\top\mathbf{1}_n=1$, define the weighted average $x_{\omega,t}\equiv\omega^\top x_t$. For any matrix $A$, let $\func{col}(A)$ and $\func{rank}(A)$ denote its column space and rank, respectively. We use $\Vert A\Vert$ and $\Vert A\Vert_{\mathrm{o}}$ to denote the Frobenius norm and operator norm of $A$, respectively, and $M_A$ to denote the orthogonal projection matrix onto the orthogonal complement of $\func{col}(A)$. For any square matrix $A$, let $\rho_{\min}(A)$ and $\rho_{\max}(A)$ denote its smallest and largest eigenvalues, respectively. For any two matrices $A$ and $B$, let $\mathrm{diag}(A,B)$ denote the block-diagonal matrix with $A$ and $B$ on its main diagonal, and let $A\otimes B$ denote their Kronecker product. Finally, for any square matrix $A$, let $\func{vech}(A)$ denote the half-vectorization of $A$, obtained by stacking the elements of its lower triangular part (including the diagonal) column by column.

\section{A Simplified Framework \label{sec: S_model}}

We first illustrate the intuition of GIV using the following simplified
model:%
\begin{align}
y_{i,t} & =\phi p_{t}+\eta_{t}+u_{i,t},  \label{S_demand} \\
p_{t} & =\psi y_{S,t}+\varepsilon_{t},   \label{S_supply}
\end{align}
where $y_{i,t}$ denotes the log demand of entity $i\in \{1,\ldots,n\}$ at
time $t$, and $p_{t}$ is the log price common to all entities. The latent
variables $\eta_{t}$ and $u_{i,t}$ with $\mathbb{E}[u_{i,t}]=0$ represent
aggregate and idiosyncratic demand shocks, respectively. The parameter $\phi$
measures the demand elasticity. Equation~(\ref{S_supply}) describes the
supply side, where $y_{S,t}\equiv S^{\top}y_{t}$ denotes aggregate demand, $%
S\equiv(s_{i})_{i\leq n}$, and $y_{t}\equiv(y_{i,t})_{i\leq n}$, with $s_{i}$
denoting the market share of entity $i$. The term $\varepsilon_{t}$ is the
supply shock, and $\psi$ denotes the supply elasticity. Although stylized,
this model is widely used in the macroeconomics and finance literature %
\citep{gabaix2021search,camanho2022global}.

Let $u_{t}\equiv(u_{i,t})_{i\leq n}$. The following assumptions are
maintained throughout this section: 
\begin{equation}
\mathrm{Cov}(\varepsilon_{t},u_{t})=\sigma_{\varepsilon u}\mathbf{1}^{\top}%
_{n},\qquad \mathrm{Cov}(\eta_{t},u_{t})=\sigma_{\eta u}\mathbf{1}^{\top}%
_{n},\qquad \mathrm{Var}(u_{t})=\sigma_{u}^{2}\mathbf{I}_{n}, 
\label{GIV_Validity}
\end{equation}
where $\sigma_{\varepsilon u}$ and $\sigma_{\eta u}$ are constants that need
not be zero. These conditions ensure the validity of the GIVs constructed in
the literature and considered in this section.\footnote{\cite%
{gabaix2024granular} impose the stronger restrictions $\sigma_{\varepsilon
u}=0$ and $\sigma_{\eta u}=0$ for the identification of $\phi$ and $\psi$
(see the first sentence of the paragraph containing their display (3)). As we show
below, these restrictions are not necessary. Moreover, as shown in the next
section, identification and estimation of $\phi$ and $\psi$ do not require
the distributions of the aggregate and idiosyncratic shocks to be
time-invariant. In particular, their variances and covariances, such as $%
\mathrm{Var}(u_{i,t})$ and $\mathrm{Cov}(\eta_{t},u_{t})$, are allowed to
vary over time.}

Following \cite{gabaix2024granular}, we use the equally weighted average $%
y_{e,t}\equiv e^{\top}y_{t}$, where $e\equiv n^{-1}\mathbf{1}_{n}$, to
construct a GIV defined as 
\begin{equation}
z_{t}(e)\equiv y_{S,t}-y_{e,t}.   \label{GIV_1}
\end{equation}
From the demand equation~(\ref{S_demand}), this GIV satisfies $z_{t}(e)=u_{t}^{\top}(S-e)$. Together with the first condition in (\ref{GIV_Validity}), this implies 
\begin{equation}
\mathbb{E}[(p_{t}-\psi y_{S,t})z_{t}(e)]=\mathbb{E}[\varepsilon_{t}u_{t}^{%
\top}](S-e)=\sigma_{\varepsilon u}\mathbf{1}_{n}^{\top}(S-e)=0. 
\label{S_moment_1}
\end{equation}
Similarly,%
\begin{align}
\mathbb{E}[(y_{e,t}-\phi p_{t})z_{t}(e)] & =\mathbb{E}[(%
\eta_{t}+u_{e,t})u_{t}^{\top}](S-e)  \notag \\
& =\mathrm{Cov}(\eta_{t},u_{t})(S-e)+e^{\top}\mathrm{Var}(u_{t})(S-e)
\notag \\
& =\sigma_{\eta u}\mathbf{1}_{n}^{\top}(S-e)+\sigma_{u}^{2}e^{\top}(S-e)=0. 
\label{S_moment_2}
\end{align}
The GIV $z_{t}(e)$ thus provides moment conditions (\ref{S_moment_1}) and (%
\ref{S_moment_2}) for identifying and estimating the elasticities $\phi$ and 
$\psi$.

The above identification strategy can be generalized to construct generic
GIVs 
\begin{equation}
z_{t}(a)\equiv y_{S,t}-y_{a,t}=(S-a)^{\top}u_{t}, 
\label{GIV_Simple_General}
\end{equation}
where $a\in \mathbb{R}^{n}$ satisfies 
\begin{equation}
a\neq S\qquad \text{and}\qquad a^{\top}\mathbf{1}_{n}=1. 
\label{IV_validity_rest}
\end{equation}
Under these conditions, moment restrictions analogous to (\ref{S_moment_1})
and (\ref{S_moment_2}) can be constructed:%
\begin{align}
\mathbb{E}[(p_{t}-\psi y_{S,t})z_{t}(a)] & =0,  \label{S_moment_3} \\
\mathbb{E}[(y_{e,t}-\phi p_{t})z_{t}(a)] & =0.   \label{S_moment_4}
\end{align}
The GIV in (\ref{GIV_1}) corresponds to the special case $a=e$.

When $n>2$, there exist multiple vectors $a$ satisfying (\ref%
{IV_validity_rest}). Therefore multiple GIVs are available and $\psi$ and $%
\phi$ become over-identified. This provides a natural motivation for using
multiple GIVs both to improve efficiency and to conduct specification tests
of instrument validity. We next characterize the resulting set of moment
conditions and clarify its connection to the GIV-based approach.

The restrictions in (\ref{GIV_Validity}) imply the following moment
conditions: 
\begin{align}
\mathbb{E}[(y_{t}-\phi p_{t}\mathbf{1}_{n})(y_{t}-\phi p_{t}\mathbf{1}%
_{n})^{\top}] & =(\mathbb{E}[\eta_{t}^{2}]+2\sigma_{\eta u})\mathbf{1}_{n}%
\mathbf{1}_{n}^{\top}+\sigma_{u}^{2}\mathbf{I}_{n},  \label{eq_moment_1} \\
\mathbb{E}[(p_{t}-\psi y_{S,t})(y_{t}-\phi p_{t}\mathbf{1}_{n})] & =(\mathbb{%
E}[\varepsilon_{t}\eta_{t}]+\sigma_{\varepsilon u})\mathbf{1}_{n}, 
\label{eq_moment_2}
\end{align}
which together provide $n(n+3)/2$ moment conditions.\ To separate the
parameters of interest from the nuisance parameters, let $%
Q\equiv(q_{1},\ldots,q_{n})$ be an $n\times n$ orthonormal matrix with%
\begin{equation}
q_{1}\equiv n^{-1/2}\mathbf{1}_{n}\text{ \ \ and \ }q_{j}%
\equiv(j(j-1))^{-1/2}\left( \sum_{l\leq j-1}\ell_{l}-(j-1)\ell_{j}\right) 
\text{ \ for }j\geq2,   \label{ortho_q_j}
\end{equation}
where $\ell_{j}$ denotes the $j$th canonical basis vector of $\mathbb{R}^{n}$%
. The following lemma provides a non-redundant representation of these
moment conditions.

\begin{lemma}
\label{L1_eq_moment} The non-redundant restrictions in (\ref{eq_moment_1})
can be equivalently written as 
\begin{align}
\mathbb{E}\!\left[ (y_{e,t}-\phi p_{t})Q_{-1}^{\top }y_{t}\right] & =\mathbf{%
0}_{n-1},  \label{L1_eq_moment_d1} \\
\mathrm{vech}\!\left( \mathbb{E}[Q_{-1}^{\top }y_{t}y_{t}^{\top
}Q_{-1}]-\sigma _{u}^{2}\mathbf{I}_{n-1}\right) & =\mathbf{0}_{n(n-1)/2},
\label{L1_eq_moment_d2} \\
\mathbb{E}[(y_{e,t}-\phi p_{t})^{2}]-(\mathbb{E}[\eta _{t}^{2}]+2\sigma
_{\eta u})-n^{-1}\sigma _{u}^{2}& =0,  \label{L1_eq_moment_d3}
\end{align}%
while the restrictions in (\ref{eq_moment_2}) can be equivalently written as 
\begin{align}
\mathbb{E}\!\left[ (p_{t}-\psi y_{S,t})Q_{-1}^{\top }y_{t}\right] & =\mathbf{%
0}_{n-1},  \label{L1_eq_moment_s1} \\
\mathbb{E}\!\left[ (p_{t}-\psi y_{S,t})(y_{e,t}-\phi p_{t})\right] -(\mathbb{%
E}[\varepsilon _{t}\eta _{t}]+\sigma _{\varepsilon u})& =0,
\label{L1_eq_moment_s2}
\end{align}%
where $Q_{-1}\equiv (q_{2},\ldots ,q_{n})$ and $\{{q_{j}\}}_{j=1}^{n}$ is
defined in (\ref{ortho_q_j}).
\end{lemma}

The moment conditions in (\ref{L1_eq_moment_d1}) and (\ref{L1_eq_moment_s1})
are equivalent to those in (\ref{S_moment_3}) and (\ref{S_moment_4}), and
can be directly used within a GMM framework for estimation and inference of
the unknown elasticities $\psi$ and $\phi$.\footnote{%
To establish this equivalence, note first that (\ref{L1_eq_moment_d1}) and (%
\ref{L1_eq_moment_s1}) are constructed using GIVs of the form $q_{j}^{\top
}y_{t}$ for $j\geq2$. For any $a\in \mathbb{R}^{n}$ satisfying (\ref%
{IV_validity_rest}), the corresponding GIV is $(S-a)^{\top}y_{t}$, where $%
(S-a)^{\top}\mathbf{1}_{n}=0$. Since $Q_{-1}$ spans the subspace orthogonal
to $\mathbf{1}_{n}$, it follows that $S-a$ can be written as a linear
combination of the columns of $Q_{-1}$. Hence, the moment conditions in (\ref%
{S_moment_3}) and (\ref{S_moment_4}) are implied by those in (\ref%
{L1_eq_moment_d1}) and (\ref{L1_eq_moment_s1}). Conversely, for each $j\geq2$%
, the vector $S-q_{j}$ satisfies (\ref{IV_validity_rest}), implying that (%
\ref{L1_eq_moment_d1}) and (\ref{L1_eq_moment_s1}) are implied by (\ref%
{S_moment_3}) and (\ref{S_moment_4}).} In contrast, the restrictions in (\ref%
{L1_eq_moment_d2})--(\ref{L1_eq_moment_d3}) and (\ref{L1_eq_moment_s2})
involve only nuisance parameters, namely $\sigma_{u}^{2}$, $\mathbb{E}%
[\eta_{t}^{2}]+2\sigma_{\eta u}$, and $\mathbb{E}[\varepsilon_{t}\eta
_{t}]+\sigma_{\varepsilon u}$.

Specifically, the moment conditions in (\ref{L1_eq_moment_d2}) identify $%
\sigma_{u}^{2}$ and also yield additional restrictions that do not depend on
unknown parameters. Conditional on $\phi$, $\psi$, and $\sigma_{u}^{2}$, the
quantities $\mathbb{E}[\eta_{t}^{2}]+2\sigma_{\eta u}$ and $\mathbb{E}%
[\varepsilon_{t}\eta_{t}]+\sigma_{\varepsilon u}$ are just-identified by (%
\ref{L1_eq_moment_d3}) and (\ref{L1_eq_moment_s2}), respectively. Since $%
\sigma_{u}^{2}$ is over-identified by (\ref{L1_eq_moment_d2}), jointly
estimating $\phi$, $\psi$, and $\sigma_{u}^{2}$ using (\ref{L1_eq_moment_d1}%
), (\ref{L1_eq_moment_d2}), and (\ref{L1_eq_moment_s1}) may yield more
efficient estimators of $\phi$ and $\psi$ than those based only on (\ref%
{L1_eq_moment_d1}) and (\ref{L1_eq_moment_s1}); see, for example, \cite%
{ackerberg2014asymptotic}.

The analysis in this section relies on a simplified demand specification in
which the aggregate shock $\eta_{t}$ enters with a known and homogeneous
loading across entities. In many applications, however, aggregate shocks may
have heterogeneous effects that are not directly observed, giving rise to a
more general factor structure. In the next section, we extend the GIV
framework to this setting, where the demand equation includes unobserved
factors with unknown loadings. This introduces new identification and
estimation challenges, as the moment conditions derived above are no longer
directly applicable when the factor loadings are unknown.

\section{Granular IVs in a General Model \label{sec: G_model}}

In this section, we study estimation and inference using GIVs in a more
general model in which the demand equation (\ref{S_demand}) incorporates a
set of unobserved factors with unknown factor loadings.\footnote{%
The model (\ref{G_demand})-(\ref{G_supply}) can be further extended to
include exogenous regressors in both the demand and supply equations without
affecting the nature of the estimation and inference procedures proposed in
this section; see Subsection \ref{subsec: Ex1} for details.}\ Specifically,
we consider 
\begin{align}
y_{t} & =\phi p_{t}\mathbf{1}_{n}+\lambda \eta_{t}+u_{t},  \label{G_demand}
\\
p_{t} & =\psi y_{S,t}+\varepsilon_{t}.   \label{G_supply}
\end{align}
Here $\eta_{t}$ denotes an $r\times1$ vector of unobserved factors, and $%
\lambda$ is an $n\times r$ matrix of factor loadings. The vector $\eta_{t}$
may also include a constant term, in which case the corresponding loading
captures unobserved entity fixed effects. While the supply equation appears
identical to (\ref{S_supply}),\ we now define\ 
\begin{equation*}
y_{S,t}=S_{t}^{\top}y_{t}, 
\end{equation*}
where $S_{t}\equiv(s_{i,t})_{i\leq n}$, and $s_{i,t}$ is nonnegative and
predetermined at time $t$.\footnote{%
Throughout this section, we assume that the number of entities $n$ is fixed
over $t$. The identification strategy, as well as the estimation and
inference procedures proposed in this section, also apply to settings in
which $n$ varies over time; see Subsection \ref{subsec: Ex2} for details.}

In contrast to the simplified model studied in the previous section, we now
allow the factor loadings $\lambda$ to be unknown, which renders the earlier
identification and estimation results inapplicable and constitutes the main
challenge addressed in this section. One approach to handling the unobserved
factors $\eta_{t}$ is to estimate them from $y_{t}$ (after partialling out $%
p_{t}$ and entity fixed effects) using principal component analysis (see,
e.g., \cite{gabaix2021search} and \cite{banafti2022inferential}). However,
as noted in the literature (see, e.g., \cite{bai2003inferential}), the
consistency of the estimated factors typically requires the number of
entities $n$ to diverge, which stands in sharp contrast to most applications
of GIVs, where the number of entities is relatively small.

Another approach, proposed in \cite{gabaix2024granular}, attempts to
identify $\lambda$ jointly with the other unknown parameters in the model
under normalization restrictions and conditions similar to (but stronger
than) Assumption \ref{ID} below. However, as we show in Subsection \ref%
{subsec: GK}, their identification strategy fails to identify the factor
loadings and may lead to invalid inference.

The method proposed in this section is based on an identification result
that applies for any $n$, whether finite or diverging. Although we follow 
\cite{gabaix2024granular} and construct our inference procedures under an
asymptotic framework with fixed $n$, as discussed in Subsection \ref{subsec:
Ex2}, the method can be straightforwardly extended to settings in which $n$
is large or even exceeds $T$.

\subsection{Identification \label{subsec: G_model_ID}}

In this subsection, we first establish identification of the demand and
supply elasticities $\phi$ and $\psi$ given $\lambda$, thereby extending the
results in Lemma \ref{L1_eq_moment}. We then provide a constructive identification result for the orthogonal complement of $\func{col}((\mathbf{1}_{n}, \lambda))$, which forms the basis for the
estimation and inference procedures developed in the next subsection. We
begin by stating the conditions required for identification.

\begin{assumption}
\label{ID} (i) $\mathbb{E}[u_{t}]=\mathbf{0}_{n}$ and $\mathrm{Var}%
(u_{t})=\sigma_{u,t}^{2}\mathbf{I}_{n}$; (ii) $\mathrm{Cov}%
(\eta_{t},u_{t})=\Gamma_{\eta u,t}\mathbf{1}_{n}^{\top}$, where $%
\Gamma_{\eta u,t}$ is an $r\times1$ vector; (iii) $\mathrm{Cov}%
(\varepsilon_{t},u_{t})=\sigma _{\varepsilon u,t}\mathbf{1}_{n}^{\top}$,
where $\sigma_{\varepsilon u,t}$ is a finite scalar; (iv) $T^{-1}\sum_{t\leq
T}\mathbb{E}[\eta_{t}\eta_{t}^{\top }]$ is nonsingular and $n>\bar{r}$,
where $\bar{r}\equiv \mathrm{rank}((\mathbf{1}_{n},\lambda))$.
\end{assumption}

Assumption \ref{ID}(i)--(iii) generalize the conditions in (\ref%
{GIV_Validity}) by allowing the joint distribution of the demand shocks $%
u_{t}$, the supply shocks $\varepsilon_{t}$, and the factors $\eta_{t}$ to
vary over time. Under these conditions, we obtain%
\begin{align}
\mathbb{E}[(y_{t}-\phi p_{t}\mathbf{1}_{n})(y_{t}-\phi p_{t}\mathbf{1}%
_{n})^{\top}] & =\lambda \mathbb{E}[\eta_{t}\eta_{t}^{\top}]\lambda^{\top
}+\lambda \Gamma_{\eta u,t}\mathbf{1}_{n}^{\top}+\mathbf{1}_{n}\Gamma_{\eta
u,t}^{\top}\lambda^{\top}+\sigma_{u,t}^{2}\mathbf{I}_{n},  \label{G_moment_1}
\\
\mathbb{E}[(p_{t}-\psi y_{S,t})(y_{t}-\phi p_{t}\mathbf{1}_{n})] & =\lambda 
\mathbb{E}[\varepsilon_{t}\eta_{t}]+\sigma_{\varepsilon u,t}\mathbf{1}_{n}, 
\label{G_moment_2}
\end{align}
which provide a total of $n(n+3)/2$ moment conditions for the unknown
parameters $\phi$, $\psi$, $\sigma_{u,t}^{2}$, $\mathbb{E}[\eta_{t}\eta
_{t}^{\top}]$, $\Gamma_{\eta u,t}$, $\mathbb{E}[\varepsilon_{t}\eta_{t}]$,
and $\sigma_{\varepsilon u,t}$. Assumption \ref{ID}(iv) is primarily imposed
to ensure identification of $\func{col}((\mathbf{1}_{n},\lambda))$, whose
orthogonal complement will be used to construct the GIVs.

We now reorganize the moment conditions (\ref{G_moment_1})-(\ref{G_moment_2}%
) according to their roles in identifying the different parameters.

\begin{lemma}
\label{L2_eq_moment}\ Let $\bar{\lambda}\equiv(n^{-1/2}\mathbf{1}_{n},\bar{%
\lambda}_{-1})$ be an $n\times \bar{r}$ orthonormal matrix spanning $\func{col}((\mathbf{1}_{n}, \lambda))$, and let $\bar{\lambda}_{\bot}$ denote its
orthonormal complement. Then the non-redundant restrictions in (\ref%
{G_moment_1}) can be equivalently expressed as 
\begin{align}
\mathbb{E}\! \left[ (y_{e,t}-\phi p_{t})\bar{\lambda}_{\bot}^{\top}y_{t}%
\right] & =\mathbf{0}_{n-\bar{r}},  \label{L2_eq_moment_1} \\
\mathbb{E}\! \left[ \bar{\lambda}_{-1}^{\top}y_{t}y_{t}^{\top}\bar{\lambda }%
_{\bot}\right] & =\mathbf{0}_{(\bar{r}-1)\times(n-\bar{r})},
\label{L2_eq_moment_2} \\
\mathrm{vech}\! \left( \mathbb{E}\! \left[ \bar{\lambda}_{\bot}^{%
\top}y_{t}y_{t}^{\top}\bar{\lambda}_{\bot}\right] -\sigma_{u,t}^{2}\mathbf{I}%
_{n-\bar{r}}\right) & =\mathbf{0}_{(n-\bar{r}+1)(n-\bar{r})/2}, 
\label{L2_eq_moment_3}
\end{align}
and 
\begin{align}
& \mathrm{vech}\! \left( \mathbb{E}\! \left[ \bar{\lambda}^{\top}(y_{t}-\phi
p_{t}\mathbf{1}_{n})(y_{t}-\phi p_{t}\mathbf{1}_{n})^{\top}\bar{\lambda}%
\right] \right)  \notag \\
& \text{ \ \ \ \ \ \ \ \ \ \ \ \ }\overset{}{=}\mathrm{vech}\! \left( \bar{%
\lambda}^{\top}\Big(\lambda \mathbb{E}[\eta_{t}\eta_{t}^{\top}]\lambda^{%
\top}+\lambda \Gamma_{\eta u,t}\mathbf{1}_{n}^{\top}+\mathbf{1}%
_{n}\Gamma_{\eta u,t}^{\top}\lambda^{\top}+\sigma_{u,t}^{2}\mathbf{I}_{n}%
\Big)\bar{\lambda}\right) ,   \label{L2_eq_moment_4}
\end{align}
while the restrictions in (\ref{G_moment_2}) can be equivalently written as 
\begin{align}
\mathbb{E}\! \left[ (p_{t}-\psi y_{S,t})\bar{\lambda}_{\bot}^{\top}y_{t}%
\right] & =\mathbf{0}_{n-\bar{r}},  \label{L2_eq_moment_5} \\
\mathbb{E}\! \left[ \bar{\lambda}^{\top}(p_{t}-\psi y_{S,t})(y_{t}-\phi p_{t}%
\mathbf{1}_{n})\right] -\big(\bar{\lambda}^{\top}\lambda \mathbb{E}%
[\varepsilon_{t}\eta_{t}]+\sigma_{\varepsilon u,t}\bar{\lambda}^{\top }%
\mathbf{1}_{n}\big) & =\mathbf{0}_{\bar{r}}.   \label{L2_eq_moment_6}
\end{align}
\end{lemma}

Lemma \ref{L2_eq_moment} shows that the key moment conditions for
identifying $\phi$ and $\psi$ are given by (\ref{L2_eq_moment_1}) and (\ref%
{L2_eq_moment_5}), which are constructed from the generalized GIVs $\bar{%
\lambda}_{\bot}^{\top}y_{t}$. The moment conditions associated with the
diagonal elements of 
\begin{equation*}
\mathbb{E}\! \left[ \bar{\lambda}_{\bot}^{\top}y_{t}y_{t}^{\top}\bar{\lambda 
}_{\bot}\right] -\sigma_{u,t}^{2}\mathbf{I}_{n-\bar{r}}=\mathbf{0}_{(n-\bar{r%
})\times(n-\bar{r})}
\end{equation*}
provide identifying restrictions for $T^{-1}\sum_{t\leq T}\sigma_{u,t}^{2}$.
In contrast, the moment conditions in (\ref{L2_eq_moment_2}), as well as the
off-diagonal elements of the matrix above, do not involve any unknown
parameters and are therefore redundant from an identification standpoint.
Nevertheless, they may be useful for improving the efficiency of the GMM
estimator and for testing specification assumptions, such as Assumption \ref%
{ID}(i). Finally, (\ref{L2_eq_moment_4}) and (\ref{L2_eq_moment_6}) impose
restrictions on the nuisance parameters $T^{-1}\sum_{t\leq T}\mathbb{E}%
[\eta_{t}\eta_{t}^{\top}]$, $T^{-1}\sum_{t\leq T}\Gamma_{\eta u,t}$, $%
T^{-1}\sum_{t\leq T}\mathbb{E}[\varepsilon_{t}\eta_{t}]$, and $%
T^{-1}\sum_{t\leq T}\sigma_{\varepsilon u,t}$, conditional on the
identification of $\phi$, $\psi$, and $T^{-1}\sum_{t\leq T}\sigma_{u,t}^{2}$.

The moment conditions in (\ref{L2_eq_moment_1}) and (\ref{L2_eq_moment_5})
for the identification of $\phi$ and $\psi$ rely on the generalized GIVs $%
\bar{\lambda}_{\bot}^{\top}y_{t}$. These instruments are, however,
infeasible in practice when the factor loading matrix $\lambda$ is unknown.
We therefore next show how to identify the column space of $\bar{\lambda}$,
which in turn determines the space spanned by $\bar{\lambda}_{\bot}$.

To this end, let $M_{\mathbf{1}_{n}}\equiv \mathbf{I}_{n}-n^{-1}\mathbf{1}%
_{n}\mathbf{1}_{n}^{\top}$ and define the demeaned variables 
\begin{equation}
\tilde{y}_{t}\equiv M_{\mathbf{1}_{n}}y_{t},\qquad \tilde{\lambda}\equiv M_{%
\mathbf{1}_{n}}\lambda,\qquad \tilde{u}_{t}\equiv M_{\mathbf{1}_{n}}u_{t}. 
\label{tilda_vars}
\end{equation}
Applying $M_{\mathbf{1}_{n}}$ to both sides of (\ref{G_demand}) yields 
\begin{equation}
\tilde{y}_{t}=\tilde{\lambda}\eta_{t}+\tilde{u}_{t}. 
\label{Demeaned_Demand}
\end{equation}
Combining (\ref{Demeaned_Demand}) with Assumption \ref{ID}, we obtain 
\begin{equation}
\mathbb{E}[\tilde{y}_{t}\tilde{y}_{t}^{\top}]=\tilde{\lambda}\mathbb{E}%
[\eta_{t}\eta_{t}^{\top}]\tilde{\lambda}^{\top}+\sigma_{u,t}^{2}M_{\mathbf{1}%
_{n}}.   \label{G_demand_factor_1}
\end{equation}
Averaging (\ref{G_demand_factor_1}) over $t$ then yields 
\begin{equation}
\bar{\Sigma}_{\tilde{y}}=\tilde{\lambda}\left( T^{-1}\sum_{t\leq T}\mathbb{E}%
[\eta_{t}\eta_{t}^{\top}]\right) \tilde{\lambda}^{\top}+\bar {\sigma}%
_{u}^{2}M_{\mathbf{1}_{n}}   \label{G_demand_factor}
\end{equation}
where 
\begin{equation*}
\bar{\Sigma}_{\tilde{y}}\equiv T^{-1}\sum_{t\leq T}\mathbb{E}[\tilde{y}_{t}%
\tilde{y}_{t}^{\top}]\text{ \ }\qquad \text{and \ }\qquad \bar{\sigma}%
_{u}^{2}\equiv T^{-1}\sum_{t\leq T}\sigma_{u,t}^{2}. 
\end{equation*}
Here the matrix $\bar{\Sigma}_{\tilde{y}}$ is identified and can be
consistently estimated.\ The following lemma shows that given the
identification of $\bar{\Sigma}_{\tilde{y}}$, (\ref{G_demand_factor}) is
sufficient to identify both $\bar{\sigma}_{u}^{2}$ and the subspace
orthogonal to the column space of $\bar{\lambda}$.

\begin{lemma}
\label{ID_G_GIV_Weight} Under Assumption \ref{ID}(iv), 
\begin{equation}
\bar{\sigma}_{u}^{2}=\min_{a\in \mathcal{B}_{\mathbf{1}_{n}}}a^{\top}\bar{%
\Sigma}_{\tilde{y}}a,   \label{G_demand_sigma_u}
\end{equation}
where 
\begin{equation*}
\mathcal{B}_{\mathbf{1}_{n}}\equiv \left \{ a\in \mathbb{R}^{n}:\mathbf{1}%
_{n}^{\top}a=0,\ \Vert a\Vert=1\right \} . 
\end{equation*}
Moreover, the set of minimizers of (\ref{G_demand_sigma_u}) spans $\func{col}(\bar{%
	\lambda}_{\bot})$.
\end{lemma}

Lemma \ref{ID_G_GIV_Weight} provides a constructive characterization of $%
\bar{\lambda}_{\bot}$, which is central to the construction of generalized
GIVs. In particular, $\bar{\lambda}_{\bot}$ can be recovered as the
eigenspace associated with the smallest eigenvalue of $\bar{\Sigma}_{\tilde{y%
}}$ restricted to $\mathcal{B}_{\mathbf{1}_{n}}$. Intuitively, this
corresponds to extracting directions of cross-sectional variation in $y_{t}$
that are orthogonal to both the common factor structure and the aggregate
component spanned by $\mathbf{1}_{n}$.

The minimization problem in (\ref{G_demand_sigma_u}), however, is defined
over the constrained set $\mathcal{B}_{\mathbf{1}_{n}}$, which is not
directly convenient for implementation. To facilitate computation, we next
provide an equivalent representation that transforms this constrained
problem into an unconstrained eigenvalue problem in $\mathbb{R}^{n-1}$.

\begin{lemma}
\label{G_GIV_Solutions} Suppose Assumption \ref{ID}(iv) holds. Consider the
minimization problem:%
\begin{equation}
\min_{\tilde{a}\in B_{n-1}}\tilde{a}^{\top}Q_{-1}^{\top}\bar{\Sigma}%
_{y}Q_{-1}\tilde{a},   \label{G_GIV_Weights}
\end{equation}
where 
\begin{equation*}
\bar{\Sigma}_{y}\equiv T^{-1}\sum_{t\leq T}\mathbb{E}[y_{t}y_{t}^{\top
}],\qquad B_{n-1}\equiv \{ \tilde{a}\in \mathbb{R}^{n-1}:\Vert \tilde{a}%
\Vert=1\}. 
\end{equation*}
Then $a$ is a minimizer of (\ref{G_demand_sigma_u}) if and only if there
exists a minimizer $\tilde{a}$ of (\ref{G_GIV_Weights}) such that $a=Q_{-1}%
\tilde{a}$.
\end{lemma}

The solutions to (\ref{G_GIV_Weights}) are given by the normalized
eigenvectors associated with the smallest eigenvalue of the symmetric matrix 
$Q_{-1}^{\top}\bar{\Sigma}_{y}Q_{-1}$. By Lemma \ref{ID_G_GIV_Weight} and
Lemma \ref{G_GIV_Solutions}, these eigenvectors, after left multiplication
by $Q_{-1}$, span the same subspace as $\bar{\lambda}_{\bot}$, thereby
providing a feasible representation of the generalized GIVs. In practice, $%
Q_{-1}^{\top }\bar{\Sigma}_{y}Q_{-1}$ can be consistently estimated by its
sample analogue, so the GIVs can be implemented via standard eigenvalue
decomposition without requiring knowledge of the factor loadings $\lambda$.

\subsection{Estimation and inference with GIVs \label{subsec: G_model_EI}}

Building on Lemmas \ref{L2_eq_moment}--\ref{G_GIV_Solutions} in the previous
subsection, the moment conditions (\ref{L2_eq_moment_1}) and (\ref%
{L2_eq_moment_5}) for the identification and estimation of $\theta
\equiv(\phi,\psi)^{\top}$\ can now be written as%
\begin{equation}
\mathbb{E}[\bar{g}_{T}(\theta;A)]=\mathbf{0}_{2(n-\bar{r})}, 
\label{Moment_Cond}
\end{equation}
where%
\begin{equation}
\bar{g}_{T}(\theta;A)\equiv T^{-1}\sum_{t\leq T}%
\begin{pmatrix}
A^{\top}y_{t}(y_{e,t}-\phi p_{t}) \\ 
A^{\top}y_{t}(p_{t}-\psi y_{S,t})%
\end{pmatrix}
.   \label{Moment_Func}
\end{equation}
Here $A\equiv Q_{-1}A_{0}$, where $Q_{-1}$ is an $n\times(n-1)$ matrix
defined in Lemma \ref{L1_eq_moment}, and $A_{0}$ is an $(n-1)\times(n-\bar{r}%
)$ matrix collecting the eigenvectors corresponding to the smallest $n-\bar{r%
}$ eigenvalues of 
\begin{equation*}
S_{y}\equiv Q_{-1}^{\top}\bar{\Sigma}_{y}Q_{-1}. 
\end{equation*}
Since\ $A_{0}$ depends on the unknown population covariance matrix, the
moment function $\bar{g}_{T}(\theta;A)$ is not directly feasible in practice.

To construct feasible moment conditions, we replace $A$ in (\ref{Moment_Func}%
) with $\hat{A}\equiv Q_{-1}\hat{A}_{0}$, where $\hat{A}_{0}$ collects the
eigenvectors corresponding to the smallest $n-\bar{r}$ eigenvalues of 
\begin{equation*}
\hat{S}_{y}\equiv Q_{-1}^{\top}\hat{\Sigma}_{y}Q_{-1},\qquad \text{where}\ 
\hat{\Sigma}_{y}\equiv T^{-1}\sum_{t\leq T}y_{t}y_{t}^{\top}. 
\end{equation*}
The GIV estimator is then defined as 
\begin{equation}
\hat{\theta}(\hat{A})\equiv \arg \min_{\theta \in \Theta}\bar{g}_{T}(\theta ;%
\hat{A})^{\top}W_{0,T}(\hat{A})\bar{g}_{T}(\theta;\hat{A}), 
\label{GMM_Criterion}
\end{equation}
where 
\begin{equation}
W_{0,T}(\hat{A})\equiv((\mathbf{I}_{2}\otimes \hat{A}^{\top})W_{0,T}(\mathbf{%
I}_{2}\otimes \hat{A}))^{-1},   \label{Weight}
\end{equation}
and $W_{0,T}$ is a user-specified symmetric positive definite $2n\times2n$
matrix.

Since $\bar{g}_{T}(\theta;\hat{A})$ is linear in $\theta$, the GIV estimator
admits the closed-form representation 
\begin{equation}
\hat{\theta}(\hat{A})=\big(D_{1,T}(\hat{A})^{\top}W_{0,T}(\hat{A})D_{1,T}(%
\hat{A})\big)^{-1}\big(D_{1,T}(\hat{A})^{\top}W_{0,T}(\hat{A})D_{2,T}(\hat {A%
})\big),   \label{GIV_Form_1}
\end{equation}
where $D_{j,T}(\hat{A})\equiv(\mathbf{I}_{2}\otimes \hat{A}^{\top})D_{j,T}$
for $j=1,2$, and 
\begin{equation}
D_{1,T}\equiv T^{-1}\sum_{t\leq T}\mathrm{diag}\! \left(
y_{t}p_{t},\,y_{t}y_{S,t}\right) ,\text{ }\ D_{2,T}\equiv T^{-1}\sum_{t\leq
T}%
\begin{pmatrix}
y_{t}y_{e,t} \\ 
y_{t}p_{t}%
\end{pmatrix}
.   \label{GIV_Form_2}
\end{equation}

We next present sufficient conditions for establishing the asymptotic
properties of $\hat{\theta}(\hat{A})$ within the same framework as \cite%
{gabaix2024granular}, where the number of entities $n$ is fixed and the
number of observations $T$ (indexed by $t$) tends to infinity.\footnote{%
The asymptotic properties of $\hat{\theta}(\hat{A})$, as well as inference
for the unknown parameter $\theta$, can be extended to the case where both $n
$ and $T$ diverge by applying techniques from the many-moments literature;
see, for example, \cite{han2006gmm} and \cite{newey2009generalized}.} Let $%
\{ \mu _{j}\}_{j\leq n-1}$ denote the eigenvalues of $S_{y}$ arranged in
increasing order, and let $A_{0,\bot}$ denote the matrix collecting the
eigenvectors associated with $\{ \mu_{j}\}_{n-\bar{r}+1\leq j\leq n-1}$.

\begin{assumption}
\label{Asy_Cond_1} (i) $\left \vert \phi \psi-1\right \vert \geq K^{-1}$ and 
$\left \vert \phi \right \vert +\left \vert \psi \right \vert +\left \Vert
\lambda \right \Vert \leq K$;\ (ii)\ for $a,b\in \{u,\eta,\varepsilon \}$, 
\begin{equation*}
T^{-1/2}\sum_{t\leq T}\big(a_{t}b_{t}^{\top}-\mathbb{E}[a_{t}b_{t}^{\top }]%
\big)=O_{p}(1); 
\end{equation*}
(iii) $\mu_{n-\bar{r}+1}-\bar{\sigma}_{u}^{2}>K^{-1}$ and\ $\bar{\sigma}%
_{u}^{2}>K^{-1}$;\ (iv) $\max_{t\leq T}\mathbb{E}[u_{t}^{\top}u_{t}+%
\varepsilon_{t}^{2}+\eta_{t}^{\top}\eta_{t}]\leq K$.
\end{assumption}

Assumption \ref{Asy_Cond_1}(i) ensures that the demand and supply system
admits a well-defined reduced form and, consequently, a unique equilibrium.
Assumption \ref{Asy_Cond_1}(ii) guarantees that the population second
moments of $(u_{t},\varepsilon_{t},\eta_{t})$ are approximated by their
sample counterparts at the rate $T^{-1/2}$. From Lemma \ref{ID_G_GIV_Weight}%
, we have $\mu_{j}=T^{-1}\sum_{t\leq T}\sigma_{u,t}^{2}$ for all $j\leq n-%
\bar{r}$. Therefore, Assumption \ref{Asy_Cond_1}(iii) imposes an eigenvalue
gap condition on $S_{y}$, which is essential for consistent estimation of
the eigenspace $\bar{\lambda}_{\bot}$. This condition can be verified under
a lower bound condition on $\rho_{\min}((\mathbf{1}_{n},\lambda)^{\top }(%
\mathbf{1}_{n},\lambda))$, or on $\rho_{\min}(\lambda^{\top}\lambda)$ when $%
\mathbf{1}_{n}\in \func{col}(\lambda)$; see Lemma \ref{Suff_Cond1_iii} in
Online Appendix \ref{APP_3} for details. Assumption \ref{Asy_Cond_1}(iii)
also requires that the variance of the idiosyncratic demand shock be bounded
away from zero, which is important for maintaining sufficient identification
strength of the GIVs. Finally, Assumption \ref{Asy_Cond_1}(iv), together
with \ref{Asy_Cond_1}(i), ensures that the second moments of $y_{t}$ and $%
p_{t}$ are well defined.

\begin{assumption}
\label{S} The sequence of market shares $\left \{ S_{t}\right \} $
satisfies: (i) 
\begin{equation*}
T^{-1/2}\sum_{t\leq T}(a_{t}b_{t}^{\top}-{\mathbb{E}[}a_{t}b_{t}^{%
	\top}])=O_{p}(1) 
\end{equation*}
for $a_{t},b_{t}\in \{S_{t}^{\top}u_{t},\eta_{t}\otimes S_{t},\varepsilon
_{t}\}$, or $a_{t}\in \{u_{t},\eta_{t}\}$\ and\ $b_{t}\in
\{S_{t}^{\top}u_{t},\eta_{t}\otimes S_{t}\}$; (ii)\ $\mathbf{1}%
_{n}^{\top}S_{t}=1$ for all $t$.
\end{assumption}

Assumption \ref{S}(i) imposes a set of high-level moment conditions ensuring
that sample averages involving the weighted aggregates, such as $S_{t}^{\top
}u_{t}$, satisfy a standard $T^{-1/2}$ law of large numbers. In particular,
it requires that interactions between market-share weights and the
structural shocks $u_{t}$, $\eta_{t}$ and $\varepsilon_{t}$\ exhibit
sufficiently weak temporal dependence and possess finite second moments.
This condition is analogous to Assumption \ref{Asy_Cond_1}(ii), and is
implied by it when $S_{t}$ is time-invariant. More generally, both
conditions can be verified under standard mixing or martingale difference
assumptions. Assumption \ref{S}(ii) is a normalization condition requiring
that the elements of $S_{t}$ sum to one.

Let $v_{t}\equiv y_{e,t}-\phi p_{t}$. For $b\in \{v,\varepsilon \}$, define%
\begin{equation}
\xi_{b,t}\equiv y_{t}b_{t}-\mathbb{E}[y_{t}b_{t}]+(y_{t}y_{t}^{\top }-%
\mathbb{E}[y_{t}y_{t}^{\top}])\, \Upsilon \left( T^{-1}\sum_{t\leq T}\mathbb{%
E}[y_{t}b_{t}]\right) ,   \label{zeta_b}
\end{equation}
where%
\begin{equation*}
\Upsilon \equiv Q_{-1}A_{0,\bot}(\bar{\sigma}_{u}^{2}\mathbf{I}_{\bar{r}%
-1}-\Lambda_{\bot})^{-1}A_{0,\bot}^{\top}Q_{-1}^{\top},\qquad \Lambda_{\bot
}\equiv \mathrm{diag}((\mu_{j})_{n-\bar{r}+1\leq j\leq n-1}). 
\end{equation*}
The random vectors $\xi_{b,t}$, for $b\in \{v,\varepsilon \}$, represent the
estimation errors in the moment conditions used to estimate $\phi$ and $\psi$%
, respectively. The first component, $y_{t}b_{t}-\mathbb{E}[y_{t}b_{t}]$,
captures the sampling variation that would arise even if the\ factor loading
matrix $\lambda$ were known. The second component reflects the additional
estimation error induced by replacing $\bar{\lambda}_{\bot}$ with its
estimator, and hence accounts for the impact of estimating the loading
matrix on the moment conditions.

\begin{assumption}
\label{Asy_Cond_2} (i) $V^{-1/2}T^{-1/2}\sum_{t\leq T}\xi_{t}\rightarrow
_{d}N(0,\mathbf{I}_{2n})$ where $\xi_{t}\equiv(\xi_{v,t}^{\top},\xi
_{\varepsilon,t}^{\top})^{\top}$ with $\rho_{\min}(V)\geq K^{-1}$; (ii) $%
W_{0,T}=W_{0}+o_{p}(1)$, where $W_{0}$ is a nonrandom symmetric matrix with $%
K^{-1}\leq \rho_{\min}(W_{0})\leq \rho_{\max}(W_{0})\leq K$; (iii) $\Vert
T^{-1}\sum_{t\leq T}A^{\top}\mathbb{E}[y_{t}p_{t}]\Vert \geq K^{-1}$ and $%
\Vert T^{-1}\sum_{t\leq T}A^{\top}\mathbb{E}[y_{t}y_{S,t}]\Vert \geq K^{-1}$%
; (iv) there exists a matrix {$\hat{V}$ such that $\hat{V}=V+o_{p}(1)$.}
\end{assumption}

Assumption \ref{Asy_Cond_2}(i) concerns the asymptotic distribution of $%
V^{-1/2}T^{-1/2}\sum_{t\leq T}\xi_{t}$, which can be established via a
central limit theorem. Here $V$ denotes the variance matrix of $%
T^{-1/2}\sum_{t\leq T}\xi_{t}$. Assumption \ref{Asy_Cond_2}(ii) ensures
consistent estimation of the weighting matrix, while Assumption \ref%
{Asy_Cond_2}(iii) guarantees that the GIVs provide sufficient identification
strength for $\phi$ and $\psi$ to be $T^{1/2}$-estimable.\ The latter
condition essentially requires that the weighted mean of the entities'
market shares,\ $S_{u}\equiv T^{-1}\sum_{t\leq T}\sigma_{u,t}^{2}\mathbb{E}%
[S_{t}]$ does not lie in $\func{col}((\mathbf{1}_{n},\lambda))$. It can be
verified under suitable primitive conditions; see Lemma \ref{Suff_Cond2_iii}
in Online Appendix \ref{APP_3} for details. Finally, Assumption \ref%
{Asy_Cond_2}(iv) requires the existence of a consistent estimator of $V$.%
\footnote{%
A consistent estimator of $V$ can be constructed using the estimated shocks $%
\hat{v}_{t}$ and $\hat{\varepsilon}_{t}$ obtained from the GIV estimator $%
\hat{\theta}(\hat{A})$ with identity weighting matrix $W_{0,T}=\mathbf{I}%
_{2n}$; see, for example, (\ref{V_est}) in the implementation algorithm in
Online Appendix \ref{APP_0}. The consistency of this variance estimator is
established in Theorem \ref{V_Est} in the Online Appendix.}

To simplify the notation for the asymptotic variance of the GIV estimator,
define 
\begin{equation*}
\Gamma(D_{1},W_{0},A)\equiv(D_{1}(A)^{\top}W_{0}(A)D_{1}(A))^{-1}D_{1}(A)^{%
\top}W_{0}(A), 
\end{equation*}
where $D_{1}(A)\equiv(\mathbf{I}_{2}\otimes A^{\top})D_{1},$ 
\begin{equation*}
\text{ }D_{1}\equiv T^{-1}\sum_{t\leq T}\mathbb{E}\! \left[ \mathrm{diag}\!
\left( y_{t}p_{t},\,y_{t}y_{S,t}\right) \right] \text{ \ \ \ and \ \ \ }%
W_{0}(A)\equiv((\mathbf{I}_{2}\otimes A^{\top})\,W_{0}\,(\mathbf{I}%
_{2}\otimes A))^{-1}. 
\end{equation*}
The following theorem establishes the asymptotic distribution of the GIV
estimator.

\begin{theorem}
\label{Asy_Dist} Under Assumptions \ref{ID}, \ref{Asy_Cond_1}, \ref{S}, and %
\ref{Asy_Cond_2}(i)--(iii), 
\begin{equation}
(\Gamma(D_{1},W_{0},A)V(A)\Gamma(D_{1},W_{0},A)^{\top})^{-1/2}T^{1/2}(\hat{%
\theta}(\hat{A})-\theta)\; \rightarrow_{d}\;N(0,\mathbf{I}_{2}), 
\label{Asy_Dist_1}
\end{equation}
where 
\begin{equation*}
V(A)\equiv(\mathbf{I}_{2}\otimes A^{\top})\,V\,(\mathbf{I}_{2}\otimes A). 
\end{equation*}
Moreover, if Assumption \ref{Asy_Cond_2}(iv) also holds, then 
\begin{equation}
\Gamma(D_{1,T},W_{0,T},\hat{A})\hat{V}(\hat{A})\Gamma(D_{1,T},W_{0,T},\hat {A%
})^{\top}=\Gamma(D_{1},W_{0},A)V(A)\Gamma(D_{1},W_{0},A)^{\top}+o_{p}(1), 
\label{Asy_Dist_2}
\end{equation}
where 
\begin{equation*}
\hat{V}(\hat{A})\equiv(\mathbf{I}_{2}\otimes \hat{A}^{\top})\, \hat {V}\,(%
\mathbf{I}_{2}\otimes \hat{A}), 
\end{equation*}
and $\Gamma(D_{1,T},W_{0,T},\hat{A})$ is defined analogously to $\Gamma
(D_{1},W_{0},A)$ with $D_{1}$, $W_{0}$, and $A$ replaced by $D_{1,T}$, $%
W_{0,T}$, and $\hat{A}$, respectively.
\end{theorem}

Theorem \ref{Asy_Dist} shows that the asymptotic variance of the GIV
estimator is minimized when $W_{0}=V$. Accordingly, the optimal weighting
matrix can be obtained by setting $W_{0,T}=\hat{V}$ in (\ref{Weight}), which
yields 
\begin{equation*}
W_{\ast,T}(\hat{A})\equiv \hat{V}(\hat{A})^{-1}. 
\end{equation*}
Let $\hat{\theta}^{\ast}(\hat{A})$ denote the corresponding optimally
weighted GIV estimator. It then follows that 
\begin{equation}
(D_{1}(A)^{\top}V(A)^{-1}D_{1}(A))^{1/2}T^{1/2}\bigl(\hat{\theta}^{\ast}(%
\hat{A})-\theta \bigr)\rightarrow_{d}N(0,\mathbf{I}_{2}). 
\label{Asy_Dist_OGMM}
\end{equation}
Standard errors for $\hat{\theta}^{\ast}(\hat{A})$ can be constructed from
the square roots of the diagonal elements of 
\begin{equation}
\bigl(TD_{1,T}(\hat{A})^{\top}W_{\ast,T}(\hat{A})D_{1,T}(\hat{A})\bigr)%
^{-1},   \label{GIV_Est_STD}
\end{equation}
whose validity follows from (\ref{Asy_Dist_2}).

Since there are $2(n-\bar{r})$ moment conditions for the identification and
estimation of $\phi$ and $\psi$, these parameters are over-identified
whenever $n>\bar{r}+1$. In this case, the validity of the GIVs can be
assessed using an over-identification test.

\begin{theorem}
\label{J_Test} Under Assumptions {\ref{ID}, \ref{Asy_Cond_1}, \ref{S}, and %
\ref{Asy_Cond_2}}, 
\begin{equation*}
T\bar{g}_{T}(\hat{\theta}^{\ast}(\hat{A});\hat{A})^{\top}W_{\ast,T}(\hat {A})%
\bar{g}_{T}(\hat{\theta}^{\ast}(\hat{A});\hat{A})\rightarrow_{d}\chi ^{2}%
\bigl(2(n-\bar{r}-1)\bigr). 
\end{equation*}
\end{theorem}

Theorem \ref{J_Test} establishes the asymptotic distribution of the J-test
statistic under the null hypothesis that the moment conditions (\ref%
{Moment_Cond}) are valid. When the GIVs are invalid, the power of the J-test
follows from standard GMM arguments and is therefore omitted for brevity.

\begin{remark}
Theorems \ref{Asy_Dist} and \ref{J_Test} establish estimation and inference
procedures for the demand and supply elasticities based on the full set of
moment conditions. In some applications, however, interest may center on a
single structural parameter, making it natural to estimate the demand and
supply elasticities separately using the corresponding subsets of moment
conditions.

For example, when the demand elasticity is the primary parameter of
interest, estimation may be based on the moment functions 
\begin{equation}
\bar{g}_{\phi,T}(\phi;A)\equiv T^{-1}\sum_{t\leq
T}A^{\top}y_{t}(y_{e,t}-\phi p_{t}).   \label{Moment_phi}
\end{equation}
Using arguments analogous to those in the proof of Theorem \ref{Asy_Dist},
the resulting GIV estimator can be shown to be $T^{1/2}$-consistent and
asymptotically normal. Its asymptotic variance and standard error can be
constructed in the same manner as those of the joint GMM estimator $\hat{%
\theta}^{\ast}(\hat{A})$. In particular, a formula analogous to (\ref%
{GIV_Est_STD}) applies after removing the components associated with the
supply-side moment conditions.
\end{remark}

\begin{remark}
The preceding discussion also extends naturally to specification testing. In
particular, the $J$-test constructed from the moment conditions in (\ref%
{Moment_phi}) and the corresponding GIV estimator provides a test of the
validity of the demand-side moment restrictions. Under correct
specification, the resulting $J$-statistic converges in distribution to $%
\chi^{2}(n-\bar {r}-1)$. Analogous estimation, inference, and
specification-testing results hold when the analysis is based solely on the
supply-side moment conditions.
\end{remark}

\subsection{Estimating the number of GIVs}

Construction of the GIVs requires knowledge of the rank $\bar{r}$ of the
matrix $(\mathbf{1}_{n},\lambda)$, which may not be feasible in practice. In
this subsection, we propose a Bayesian information criterion (BIC) for
consistent estimation of $\bar{r}$. The construction is motivated by Lemma %
\ref{ID_G_GIV_Weight} and Lemma \ref{G_GIV_Solutions}.

Specifically, let $\{ \hat{\mu}_{j}\}_{j\leq n-1}$ denote the eigenvalues of 
$\hat{S}_{y}$ arranged in increasing order. Define the information criterion%
\begin{equation}
\mathrm{BIC}_{T}(j)\equiv \frac{T}{n-j}\sum_{s=1}^{n-j}\frac{(\hat{\mu}_{s}-%
\hat{\mu}_{1})^{2}}{2\hat{\mu}_{s}^{2}}+j\log(T),   \label{BIC}
\end{equation}
for $j\in \mathcal{J}$, where $\mathcal{J}\equiv \{1,\ldots,n-1\}$. The
estimator $\hat{r}$ of $\bar{r}$ is then given by%
\begin{equation}
\hat{r}=\arg \min_{j\in \mathcal{J}}\mathrm{BIC}_{T}(j).   \label{r_hat}
\end{equation}
The consistency of $\hat{r}$ is established in Theorem \ref%
{r_hat_consistency}.

\begin{theorem}
\label{r_hat_consistency}\ Under Assumptions \ref{ID}, \ref{Asy_Cond_1} and \ref{S}, we have $\hat{r}=\bar{r}$ with probability approaching 1.
\end{theorem}

We conclude this subsection by providing intuition for the construction of (%
\ref{BIC}). The criterion $\mathrm{BIC}_{T}(j)$ consists of two components.
The first term, $T(n-j)^{-1}\sum_{s=1}^{n-j}(\hat{\mu}_{s}-\hat{\mu}%
_{1})^{2}/(2\hat{\mu}_{s}^{2})$, is decreasing in $j$ and captures an
over-fitting effect analogous to that in classical regression settings,
where $j$ reflects the dimension or complexity of the model. The second
term, $j\log(T)$, is strictly increasing in $j$ and serves as a penalty on
model complexity. The estimator $\hat{r}$ in (\ref{r_hat}) therefore
balances the trade-off between goodness-of-fit and model complexity.

The eigenvalues $\{ \hat{\mu}_{j}\}_{j\leq n-1}$ are $T^{1/2}$-consistent
estimators of $\{ \mu_{j}\}_{j\leq n-1}$ under Assumptions \ref{Asy_Cond_1}%
(i, ii). Since $\mu_{s}=\bar{\sigma}_{u}^{2}$ for $s\in \{1,\ldots,n-\bar{r}%
\}$, it follows that for $j\geq \bar{r}$, 
\begin{equation*}
\frac{T}{n-j}\sum_{s=1}^{n-j}\frac{(\hat{\mu}_{s}-\hat{\mu}_{1})^{2}}{2\hat{%
\mu}_{s}^{2}}=O_{p}(1). 
\end{equation*}
Consequently, the penalty term $j\log(T)$ dominates the first term in (\ref%
{BIC}). Since $j\log(T)$ is strictly increasing in $j$, $\mathrm{BIC}_{T}(j)$
is asymptotically minimized at $\bar{r}$ over $j\in \{ \bar{r},\ldots,n-1\}$%
. On the other hand, for $j<\bar{r}$, we have $n-j\geq n-\bar {r}+1$, so
Assumption \ref{Asy_Cond_1}(iii) implies that $\mu_{n-j}-\mu_{1}$ is bounded
away from zero. Combined with the $T^{-1/2}$ consistency of $\{ \hat{\mu}%
_{j}\}_{j\leq n-1}$, this yields%
\begin{equation*}
\frac{T}{n-j}\sum_{s=1}^{n-j}\frac{(\hat{\mu}_{s}-\hat{\mu}_{1})^{2}}{2\hat{%
\mu}_{s}^{2}}\geq \frac{T}{n-j}\frac{(\hat{\mu}_{n-j}-\hat{\mu}_{1})^{2}}{2%
\hat{\mu}_{n-j}^{2}}\geq K^{-1}T((\mu_{n-j}-\mu_{1})^{2}-O_{p}(T^{-1/2})), 
\end{equation*}
which diverges at rate $T$. This term therefore dominates the penalty term
of order $\log(T)$, implying that $\mathrm{BIC}_{T}(j)>\mathrm{BIC}_{T}(\bar{%
r})$ wpa1 for all $j<\bar{r}$. Combining the two cases, $\mathrm{BIC}_{T}(j)$
is asymptotically minimized at $j=\bar{r}$ over $j\in \mathcal{J}$, which
ensures the consistency of $\hat{r}$.

\section{Extensions and Discussion\label{sec: Extension}}

This section provides further discussion and extensions of the main results
established in the previous section. First, we show that the identification
strategy in \cite{gabaix2024granular} may fail when the factor loadings are
unknown, potentially leading to inconsistent estimation and invalid
inference under standard GMM procedures. Second, we demonstrate that our
estimation and inference procedures can be straightforwardly extended to
settings with exogenous regressors in the demand and supply equations and to
data with unbalanced features. The latter extension shows that our methods
remain applicable even when the number of entities $n$ is large and may
exceed $T$.

\subsection{Identification failure of factor loadings in \protect\cite%
{gabaix2024granular}\ \label{subsec: GK}}

In this subsection, we show that the identification strategy in \cite%
{gabaix2024granular}, in particular their Proposition 7, may fail when $%
\lambda$ is unknown. \cite{gabaix2024granular} impose a normalization on the
factor loadings $\lambda$ by setting the loadings of the first factor $%
\eta_{1,t}$ to be $\mathbf{1}_{n}$, and the loadings of the remaining
factors $\eta_{2,t}$, denoted by $\lambda_{-1}$, to satisfy\footnote{%
See the second paragraph above Proposition 7 in \cite{gabaix2024granular}.
In the same paragraph, they also impose the restriction that $\mathrm{Var}%
(\lambda _{-1}^{\top}y_{t})$ is diagonal with distinct diagonal entries\ ($%
\lambda _{-1}^{\top}y_{t}$ here is equal to $n\check{\eta}_{t}$ in their
notation). As we show below, this additional restriction should be
interpreted as an assumption on the latent factors rather than a
normalization on the factor loadings. Moreover, imposing this restriction
does not resolve the identification issue in their approach.} 
\begin{equation}
\mathbf{1}_{n}^{\top}\lambda_{-1}=\mathbf{0}_{r-1}^{\top}\text{, \ \ }%
n^{-1}\lambda_{-1}^{\top}\lambda_{-1}=\mathbf{I}_{r-1}.   \label{GK_F_0a}
\end{equation}
Moreover, they assume\footnote{%
See the second paragraph on page 2279, display (3), and Assumption 3 in \cite%
{gabaix2024granular}.} 
\begin{equation}
\mathrm{Var}(u_{t})=\sigma_{u}^{2}\mathbf{I}_{n},\qquad \mathrm{Cov}(\eta
_{t},u_{t})=\mathbf{0}_{r\times n},\qquad \mathbb{E}[\eta_{t}]=\mathbf{0}%
_{r},\qquad \mathrm{Cov}(u_{t},\varepsilon_{t})=\mathbf{0}_{n}. 
\label{GK_F_0b}
\end{equation}
For notational simplicity, we abstract from exogenous regressors in both the
demand and supply equations. The demand equation (\ref{G_demand}) can then
be written as 
\begin{equation}
y_{t}=\phi p_{t}\mathbf{1}_{n}+\eta_{1,t}\mathbf{1}_{n}+\lambda_{-1}\eta
_{2,t}+u_{t},   \label{GK_F_1}
\end{equation}
while the supply equation remains the same as (\ref{G_supply}).

\cite{gabaix2024granular} propose using both $\lambda_{\bot}^{\top}y_{t}$
and $\lambda_{-1}^{\top}y_{t}$ as IVs to construct moment conditions for
identifying and estimating $\phi$, $\psi$, and $\lambda_{-1}$.\footnote{
	They correspond to $z_t(m^y)$ and $\check{\eta}_t(m^y)$ in Proposition~4 of
	\citet{gabaix2024granular}.
} Because $\mathbf{1}%
_{n}^{\top}\lambda_{-1}=\mathbf{0}_{r-1}^{\top}$, one can partial out $%
\eta_{2,t}$ by premultiplying (\ref{GK_F_1}) by $e^{\top}$, yielding 
\begin{equation}
y_{e,t}-\phi p_{t}=\eta_{1,t}+u_{e,t}.   \label{GK_F_2}
\end{equation}
Moreover, because $n^{-1}\lambda_{-1}^{\top}\lambda_{-1}=\mathbf{I}_{r-1}$,
premultiplying (\ref{GK_F_1}) by $\lambda_{-1}^{\top}$, we also obtain 
\begin{equation}
\lambda_{-1}^{\top}y_{t}=n\eta_{2,t}+\lambda_{-1}^{\top}u_{t}. 
\label{GK_F_3}
\end{equation}
From Assumption \ref{ID}(i) and (\ref{GK_F_0b}), it follows that the product
of (\ref{GK_F_3}) and the second term on the right of (\ref{GK_F_2}) is such
that 
\begin{equation}
\mathbb{E}[u_{e,t}\lambda_{-1}^{\top}y_{t}]=n\mathbb{E}[\eta_{2,t}u_{e,t}]+%
\lambda_{-1}^{\top}\mathbb{E}[u_{t}u_{e,t}]=\mathbf{0}_{r-1}. 
\label{GK_F_4}
\end{equation}
However, the product $\mathbb{E}[\eta_{1,t}\lambda_{-1}^{\top}y_{t}]$ of the
first term on the right of (\ref{GK_F_2}) and (\ref{GK_F_3}) may be nonzero
due to possible correlation between $\eta_{1,t}$ and $\eta_{2,t}$.
Therefore, $\lambda_{-1}^{\top}y_{t}$ cannot be directly used together with $%
y_{e,t}-\phi p_{t}$ to form valid moment conditions. To address this issue, 
\cite{gabaix2024granular} introduce the regression coefficient $b_{y}$ of $%
\eta_{1,t}$ on $y_{t}^{\top}\lambda_{-1}$ such that%
\begin{equation}
\mathbb{E}[(\eta_{1,t}-y_{t}^{\top}\lambda_{-1}b_{y})\lambda_{-1}^{%
\top}y_{t}]=\mathbf{0}_{r-1}.   \label{GK_F_4b}
\end{equation}
Combining this with (\ref{GK_F_2}) and (\ref{GK_F_4}) yields 
\begin{equation}
\mathbb{E}[(y_{e,t}-\phi p_{t}-y_{t}^{\top}\lambda_{-1}b_{y})\lambda
_{-1}^{\top}y_{t}]=\mathbf{0}_{r-1}.   \label{GK_F_5a}
\end{equation}
Since $y_{e,t}-\phi p_{t}-y_{t}^{\top}\lambda_{-1}b_{y}=u_{e,t}-u_{t}^{\top
}\lambda_{-1}b_{y}+\eta_{1,t}-n\eta_{2,t}^{\top}b_{y}$, it follows from
Assumption \ref{ID}(i) and (\ref{GK_F_0b}) that 
\begin{equation*}
\mathbb{E}\! \left[ (y_{e,t}-\phi
p_{t}-y_{t}^{\top}\lambda_{-1}b_{y})\lambda_{\bot}^{\top}y_{t}\right] =%
\mathbb{E}\! \left[ (u_{e,t}-u_{t}^{\top}\lambda_{-1}b_{y}+\eta_{1,t}-n%
\eta_{2,t}^{\top}b_{y})\lambda_{\bot }^{\top}u_{t}\right] =\mathbf{0}_{n-r}. 
\end{equation*}
Together with (\ref{GK_F_5a}), this yields the moment conditions from the
demand equation:%
\begin{equation}
\mathbb{E}\left[ (y_{e,t}-\phi p_{t}-y_{t}^{\top}\lambda_{-1}b_{y})\left( 
\begin{array}{c}
\lambda_{\bot}^{\top}y_{t} \\ 
\lambda_{-1}^{\top}y_{t}%
\end{array}
\right) \right] =\mathbf{0}_{n-1}.   \label{GK_F_5}
\end{equation}

Similarly, $\lambda_{-1}^{\top}y_{t}$ cannot be directly used together with $%
p_{t}-\psi y_{S,t}$ to identify $\psi$, because $\lambda_{-1}^{\top}y_{t}$
contains $\eta_{2,t}$, which may be correlated with $\varepsilon_{t}$. \cite%
{gabaix2024granular} therefore propose%
\begin{equation*}
\mathbb{E}\! \left[ (p_{t}-\psi
y_{S,t}-y_{t}^{\top}\lambda_{-1}b_{p})\lambda_{-1}^{\top}y_{t}\right] =%
\mathbf{0}_{r-1}, 
\end{equation*}
where $b_{p}$ is defined by%
\begin{equation*}
\mathbb{E}\! \left[ (\varepsilon_{t}-y_{t}^{\top}\lambda_{-1}b_{p})%
\lambda_{-1}^{\top}y_{t}\right] =\mathbf{0}_{r-1}. 
\end{equation*}
Moreover, by Assumption \ref{ID}(i) and (\ref{GK_F_0b}),%
\begin{equation*}
\mathbb{E}\! \left[ (p_{t}-\psi
y_{S,t}-y_{t}^{\top}\lambda_{-1}b_{p})\lambda_{\bot}^{\top}y_{t}\right] =%
\mathbb{E}\! \left[ (\varepsilon
_{t}-n\eta_{2,t}^{\top}b_{p}-u_{t}^{\top}\lambda_{-1}b_{p})\lambda_{\bot
}^{\top}u_{t}\right] =\mathbf{0}_{n-r}, 
\end{equation*}
which provides additional moment conditions. Therefore, the moment
conditions from the supply equation are%
\begin{equation}
\mathbb{E}\left[ (p_{t}-\psi y_{S,t}-y_{t}^{\top}\lambda_{-1}b_{p})\left( 
\begin{array}{c}
\lambda_{\bot}^{\top}y_{t} \\ 
\lambda_{-1}^{\top}y_{t}%
\end{array}
\right) \right] =\mathbf{0}_{n-1}.   \label{GK_F_6}
\end{equation}

The moment conditions in (\ref{GK_F_5}) and (\ref{GK_F_6}) coincide with
those in Proposition 4 of \cite{gabaix2024granular} when additional
exogenous variables are excluded. \cite{gabaix2024granular} argue in their
Proposition 4 that these conditions identify $\phi$, $\psi$, $b_{y}$, and $%
b_{p}$ when $\lambda_{-1}$ is known. When $\lambda_{-1}$ is unknown, they
propose the additional moment conditions 
\begin{equation}
\mathbb{E}\big[(M_{\mathbf{1}_{n}}-\lambda_{-1}(\lambda_{-1}^{\top}%
\lambda_{-1})^{-1}\lambda_{-1}^{\top})y_{t}y_{t}^{\top}\lambda_{-1}\big]=%
\mathbf{0}_{n\times(r-1)},   \label{GK_F_7}
\end{equation}
which correspond to equation (50) in Proposition 7 of \cite%
{gabaix2024granular}. To verify (\ref{GK_F_7}), note that $M_{\mathbf{1}%
_{n}}\lambda_{-1}=\lambda_{-1}$ by (\ref{GK_F_0a}). Using (\ref{GK_F_1}), 
\begin{equation*}
(M_{\mathbf{1}_{n}}-\lambda_{-1}(\lambda_{-1}^{\top}\lambda_{-1})^{-1}%
\lambda_{-1}^{\top})y_{t}=(M_{\mathbf{1}_{n}}-\lambda_{-1}(\lambda_{-1}^{%
\top }\lambda_{-1})^{-1}\lambda_{-1}^{\top})u_{t}. 
\end{equation*}
Hence, (\ref{GK_F_7}) follows from Assumption \ref{ID}(i).

Since (\ref{GK_F_7}) provides $n(r-1)$ moment conditions for $n(r-1)$
unknown entries in $\lambda_{-1}$, it may seem to deliver exact
identification.\footnote{%
Indeed, \cite{gabaix2024granular} state at the top of page 2294 that
\textquotedblleft The new moment (50) identifies $\check {\lambda}$." In our
notation, their moment (50) corresponds to (\ref{GK_F_7}), while their $%
\check{\lambda}$ is denoted here by $\lambda_{-1}$.} However, as shown in
the lemma below, this is not the case: the restrictions in (\ref{GK_F_7})
fail to uniquely identify $\lambda_{-1}$, even up to rotation.

\begin{lemma}
\label{GK_Non_ID} Suppose that $\mathbb{E}[y_{t}y_{t}^{\top}]$ is finite and
nonsingular, and that the conditions in (\ref{GK_F_0b}) hold. Let $%
\{d_{j}\}_{j=1}^{n-1}$ be an orthonormal basis of eigenvectors associated
with the nonzero eigenvalues of $\mathbb{E}[\tilde{y}_{t}\tilde{y}_{t}^{\top
}]$. For any subset $J\subset \{1,\ldots,n-1\}$ with $|J|=r-1$, let $D_{J}$
collect the columns $d_{j}$, $j\in J$. Then $n^{1/2}D_{J}$ satisfies (\ref%
{GK_F_0a}) and (\ref{GK_F_7}).
\end{lemma}

Lemma \ref{GK_Non_ID} shows that the moment condition (\ref{GK_F_7}),
together with the normalization in (\ref{GK_F_0a}), fails to identify $\func{%
col}(\lambda_{-1})$. Indeed, (\ref{GK_F_0a}) and (\ref{GK_F_7}) admit $%
\binom{n-1}{r-1}$ different choices of $D_{J}$, whose column spaces are
generally distinct. Moreover, 
\begin{equation*}
\func{col}((d_{1},\ldots,d_{n-1}))=\func{col}(M_{\mathbf{1}_{n}})=\func{col}%
((\lambda_{-1},\lambda_{\bot})). 
\end{equation*}
Therefore, if $\func{col}(D_{J})\neq \func{col}(\lambda_{-1})$, then the
orthogonal complement of $(\mathbf{1}_{n},D_{J})$, denoted by $D_{J,\bot}$,
need not be orthogonal to the true factor space $\func{col}(\lambda_{-1})$.
In particular,%
\begin{equation*}
D_{J,\bot}^{\top}\lambda_{-1}\neq \mathbf{0}_{(n-r)\times(r-1)}
\end{equation*}
may hold. Consequently, the candidate GIVs $D_{J,\bot}^{\top}y_{t}$ may
still contain components of the latent factors $\eta_{t}$, since 
\begin{equation*}
D_{J,\bot}^{\top}y_{t}=D_{J,\bot}^{\top}\lambda_{-1}\eta_{2,t}+D_{J,\bot
}^{\top}u_{t}. 
\end{equation*}
As a result, moment conditions constructed from $D_{J,\bot}^{\top}y_{t}$ may
fail to eliminate the latent factor component and therefore need not provide
valid identifying restrictions for the structural parameters.

To make the identification problem more explicit, consider the special case $%
r=n-1$. In this case, $\lambda_{\bot}$ is one-dimensional, and the moment
conditions in (\ref{GK_F_5}) and (\ref{GK_F_6}) provide only exact
identification for the unknown parameters $\phi$, $\psi$, $b_{y}$, and $b_{p}
$, given the true factor loading matrix $\lambda_{-1}$ and its orthogonal
complement $\lambda_{\bot}$. Therefore, identification of these parameters
ultimately relies on identification of $\func{col}(\lambda_{-1})$.

However, as discussed above, (\ref{GK_F_5}) and (\ref{GK_F_6}) admit $n-1$
different choices of $D_{J}$, and hence $n-1$ corresponding choices of $%
D_{J,\bot}$. Since $\lambda_{\bot}$ is one-dimensional, at least $n-2$ of
these choices satisfy 
\begin{equation*}
D_{J,\bot}^{\top}\lambda_{-1}\neq \mathbf{0}_{1\times(r-1)}. 
\end{equation*}
For such choices, the corresponding candidate GIVs $D_{J,\bot}^{\top}y_{t}$
retain latent factor components and therefore generally fail to identify the
true elasticities $\phi$ and $\psi$. Consequently, GMM estimation based on
these invalid instruments would generally converge to pseudo-true values
rather than the true structural parameters. Moreover, since different
choices of $D_{J,\bot}$ generally lead to different pseudo-true values, the
resulting limits need not even be uniquely determined.

Lemma \ref{GK_Non_ID} establishes the non-identification of the factor
loadings under the normalization in (\ref{GK_F_0a}). In addition to (\ref%
{GK_F_0a}), \cite{gabaix2024granular} also assume that $\mathrm{Var}%
(\lambda_{-1}^{\top}y_{t})$ is a diagonal matrix with distinct diagonal
entries.\footnote{%
See the second paragraph above Proposition 7 in \cite{gabaix2024granular}.}
Under (\ref{GK_F_0a}) and (\ref{GK_F_0b}), however, 
\begin{equation*}
\mathrm{Var}(\lambda_{-1}^{\top}y_{t})=\lambda_{-1}^{\top}\mathbb{E}%
[y_{t}y_{t}^{\top}]\lambda_{-1}=n^{2}\mathbb{E}[\eta_{2,t}\eta_{2,t}^{\top
}]+n\sigma_{u}^{2}\mathbf{I}_{r-1}. 
\end{equation*}
Hence, the diagonal structure imposed on $\mathrm{Var}(\lambda_{-1}^{\top
}y_{t})$ amounts to additional restrictions on the latent factors, requiring
that the components of $\eta_{2,t}$ are uncorrelated and have distinct
variances. Since $\mathbb{E}[y_{t}y_{t}^{\top}]$ is unknown, this condition
should be viewed as an assumption on the latent factors rather than a
restriction on the factor loadings.\footnote{%
If one instead imposes the corresponding restriction on the sample second
moment $T^{-1}\sum_{t\leq T}y_{t}y_{t}^{\top}$, then for a broad class of
data-generating processes, the orthonormalized eigenvectors associated with
any collection of $r-1$ eigenvalues may satisfy the empirical counterpart of
this restriction in finite samples due to estimation error in $%
T^{-1}\sum_{t\leq T}y_{t}y_{t}^{\top}$.}

Even after imposing this restriction together with (\ref{GK_F_0a}), the
moment condition (\ref{GK_F_7}) still fails to identify $\func{col}%
(\lambda_{-1})$. Indeed, by (\ref{GK_F_0a}), (\ref{GK_F_0b}), and (\ref%
{GK_F_1}), 
\begin{equation*}
\mathbb{E}[\tilde{y}_{t}\tilde{y}_{t}^{\top}]=M_{\mathbf{1}_{n}}\mathbb{E}%
[y_{t}y_{t}^{\top}]M_{\mathbf{1}_{n}}=\lambda_{-1}\mathbb{E}%
[\eta_{2,t}\eta_{2,t}^{\top}]\lambda_{-1}^{\top}+\sigma_{u}^{2}M_{\mathbf{1}%
_{n}}. 
\end{equation*}
Under the additional assumption that $\mathbb{E}[\eta_{2,t}\eta_{2,t}^{\top}]
$ is diagonal with distinct diagonal entries, the matrix $\mathbb{E}[\tilde {%
y}_{t}\tilde{y}_{t}^{\top}]$ has $r-1$ distinct eigenvalues larger than $%
\sigma_{u}^{2}$, while $\sigma_{u}^{2}$ itself is an eigenvalue with
multiplicity $n-r$.\ The eigenspace associated with $\sigma_{u}^{2}$ is
spanned by $\lambda_{\bot}$. Consequently, there are at least $(r-1)(n-r)+1$ such
choices of $D_{J}$ with distinct column spaces.\footnote{%
Note that one choice is given by $D_{J}=n^{-1/2}\lambda_{-1}$, while the
remaining $(r-1)(n-r)$ choices are obtained by selecting $r-2$ columns from $%
\lambda_{-1}$ columns and one column from $\lambda_{\bot}$.} Since these
columns are orthonormal eigenvectors of $\mathbb{E}[\tilde{y}_{t}\tilde{y}%
_{t}^{\top}]$, it follows that 
\begin{equation*}
\mathrm{Var}(n^{1/2}D_{J}^{\top}y_{t})=nD_{J}^{\top}\mathbb{E}%
[y_{t}y_{t}^{\top}]D_{J}=nD_{J}^{\top}\mathbb{E}[\tilde{y}_{t}\tilde{y}%
_{t}^{\top }]D_{J}
\end{equation*}
is diagonal with distinct diagonal entries. Since there exist $(r-1)(n-r)+1$%
\ such choices of $D_{J}$ with distinct column spaces, this again shows that 
$\func{col}(\lambda_{-1})$ is not identified.

We conclude this subsection with a lemma establishing the rotational
non-uniqueness of the factor loadings identified from (\ref{GK_F_0a}), (\ref%
{GK_F_5})--(\ref{GK_F_7}) in the general case.

\begin{lemma}
\label{GK_Rotation} Suppose that $\{b_{y},b_{p},\lambda_{-1},\lambda_{\perp
}\}$ satisfies (\ref{GK_F_5})--(\ref{GK_F_7}) given $\phi$ and $\psi$. Let $%
C_{1}$ and $C_{2}$ be arbitrary orthogonal matrices of dimensions $%
(r-1)\times(r-1)$ and $(n-r)\times(n-r)$, respectively. Define 
\begin{equation*}
\lambda_{C_{1}}\equiv \lambda_{-1}C_{1},\qquad \lambda_{C_{2}}\equiv
\lambda_{\perp}C_{2}. 
\end{equation*}
Then $\{C_{1}^{\top}b_{y},C_{1}^{\top}b_{p},\lambda_{C_{1}},\lambda_{C_{2}}\}
$ also satisfies (\ref{GK_F_5})--(\ref{GK_F_7}) given $\phi$ and $\psi$.
Moreover, $\lambda_{C_{1}}$ and $\lambda_{C_{2}}$ satisfy the same
normalization and orthogonality restrictions as $\lambda_{-1}$ and $%
\lambda_{\perp}$.
\end{lemma}

Lemma \ref{GK_Rotation} shows that the factor loadings $\lambda_{-1}$ and
the regression coefficients $b_{y}$ and $b_{p}$ are, at best, identified up
to an orthonormal rotation. This raises concerns for estimation and
inference based on these moment conditions, since standard GMM procedures
require uniqueness of the identified parameters and are therefore not
directly applicable in this setting.

The non-identification issue here is more challenging to address than in our
approach, because $\{b_{y},b_{p},\lambda_{-1},\lambda_{\perp}\}$ are jointly
identified together with the demand and supply elasticities. This joint
determination complicates both the computation of the GIV estimator and the
analysis of its statistical properties. In contrast, our approach separates
the identification of the factor loadings from that of the elasticity
parameters. We then exploit the invariance properties of the GIV estimator
and the $J$-test statistic to address the fact that the factor loadings are
only identified up to rotation, thereby allowing standard GMM estimation and
inference to remain valid.

\subsection{Model with exogenous regressors\label{subsec: Ex1}}

This subsection extends the model studied in the previous section by
allowing for additional exogenous regressors in both the demand and supply
equations in (\ref{G_demand})--(\ref{G_supply}). Specifically, we consider 
\begin{align}
y_{t} & =\phi p_{t}\mathbf{1}_{n}+x_{t}\beta+\lambda \eta_{t}+u_{t},
\label{F_demand} \\
p_{t} & =\psi y_{S,t}+w_{t}^{\top}\gamma+\varepsilon_{t},   \label{F_supply}
\end{align}
where $x_{t}\equiv(x_{1,t},\ldots,x_{n,t})^{\top}$ with $x_{i,t}\in \mathbb{R%
}^{d_{x}}$, and $w_{t}\in \mathbb{R}^{d_{w}}$ denote observed exogenous
variables that have direct effects on demand and supply, respectively. The
variables $x_{t}$ include both sector fixed effects and unit-level demand
shifters, while $w_{t}$ captures aggregate supply shifters.\footnote{%
Since $\lambda \eta_{t}$ can be decomposed as $\lambda (\eta_{t}-\mathbb{E}%
[\eta_{t}])+\lambda \mathbb{E}[\eta_{t}]$, and $\lambda \mathbb{E}[\eta_{t}]$
can be absorbed into the sector fixed effects, we assume without loss of
generality that $\mathbb{E}[\eta_{t}]=\mathbf{0}_{r}$ throughout this
subsection.} Under suitable exogeneity conditions, the main identification
and estimation arguments continue to apply after partialling out $x_{t}$
from the demand equation (\ref{F_demand}).\footnote{%
When $x_{t}$ includes variables excluded from the supply equation, these may
serve as IVs for identifying $\psi$ in (\ref{F_supply}) if they are
uncorrelated with $\varepsilon_{t}$. Similarly, variables in $w_{t}$
excluded from the demand equation may identify $\phi$ in (\ref{F_demand}) if
they are uncorrelated with $\eta_{t}$ and $u_{t}$. Although standard in the
classical simultaneous equations literature, this strategy is not widely
used in empirical applications of GIV. We therefore do not assume the
existence or exogeneity of such excluded variables.}

Specifically, multiplying $M_{\mathbf{1}_{n}}$ on both sides of (\ref%
{F_demand}) yields%
\begin{equation*}
\tilde{y}_{t}=\tilde{x}_{t}\beta+\tilde{\lambda}\eta_{t}+\tilde{u}_{t}, 
\end{equation*}
where $\tilde{x}_{t}\equiv M_{\mathbf{1}_{n}}x_{t}$, and $\tilde{y}_{t}$, $%
\tilde{\lambda}$, and $\tilde{u}_{t}$ are defined analogously; see (\ref%
{tilda_vars}). Let $\tilde{y}_{t}^{\ast}\equiv M_{\mathbf{1}_{n}}y_{t}^{\ast}
$, where $y_{t}^{\ast}\equiv y_{t}-x_{t}\beta$. Then the above equation can
be written as 
\begin{equation*}
\tilde{y}_{t}^{\ast}=\tilde{\lambda}\eta_{t}+\tilde{u}_{t}, 
\end{equation*}
which takes a form similar to (\ref{Demeaned_Demand}).\footnote{%
Since $w_{t}$ is invariant across $i$, it is automatically partialled out in 
$\tilde{y}$. Therefore, $\tilde{y}_{t}^{\ast}$ effectively partials out the
exogenous regressors in both the demand and supply equations.} Therefore, we
can apply Lemmas \ref{ID_G_GIV_Weight} and \ref{G_GIV_Solutions} in
Subsection \ref{subsec: G_model_ID}, with the second moment matrix $\bar{%
\Sigma}_{\tilde{y}}$ replaced by%
\begin{equation*}
\bar{\Sigma}_{\tilde{y}^{\ast}}\equiv T^{-1}\sum_{t\leq T}\mathbb{E}[\tilde {%
y}_{t}^{\ast}\tilde{y}_{t}^{\ast \top}], 
\end{equation*}
to identify the subspace $\func{col}(\bar{\lambda}_{\bot})$, which is
orthogonal to $(\mathbf{1}_{n},\lambda)$.

Given $\beta$, the GIVs and moment conditions can be constructed in the same
way as in the previous section, with $y_{t}$ replaced by $y_{t}^{\ast}$. To
proceed, we need to estimate $\beta$, which is required to construct an
estimator for $y_{t}^{\ast}$. The unknown parameter $\beta$ is estimated by 
\begin{equation*}
\hat{\beta}\equiv \Bigl(\sum_{t\leq T}\tilde{x}_{t}^{\top}\tilde{x}_{t}\Bigr)%
^{-1}\Bigl(\sum_{t\leq T}\tilde{x}_{t}^{\top}\tilde{y}_{t}\Bigr). 
\end{equation*}
If $x_{t}$ is exogenous, in the sense that it is uncorrelated with both $%
\eta_{t}$ and $u_{t}$, then standard least squares theory implies that $\hat{%
\beta}$ is a $T^{1/2}$-consistent estimator of $\beta$. Given $\hat{\beta}$%
, define%
\begin{equation*}
\hat{y}_{t}^{\ast}\equiv y_{t}-x_{t}\hat{\beta},\qquad \hat{y}_{e,t}^{\ast
}\equiv e^{\top}\hat{y}_{t}^{\ast}. 
\end{equation*}
We then obtain $\hat{A}\equiv Q_{-1}\hat{A}_{0}$, where $\hat{A}_{0}$
collects the eigenvectors corresponding to the smallest\ $n-\bar{r}$
eigenvalues of 
\begin{equation*}
\hat{S}_{y^{\ast}}\equiv Q_{-1}^{\top}\hat{\Sigma}_{\hat{y}%
^{\ast}}Q_{-1},\qquad \hat{\Sigma}_{\hat{y}^{\ast}}\equiv T^{-1}\sum_{t\leq
T}\hat {y}_{t}^{\ast}\hat{y}_{t}^{\ast \top}. 
\end{equation*}

The moment conditions used to estimate the unknown parameter $\theta
\equiv(\phi,\psi,\gamma^{\top})^{\top}$ are 
\begin{equation}
\bar{g}_{T}(\theta;\hat{A},\hat{\beta})\equiv T^{-1}\sum_{t\leq T} 
\begin{pmatrix}
\hat{A}^{\top}\hat{y}_{t}^{\ast}(\hat{y}_{e,t}^{\ast}-\phi p_{t}) \\ 
\hat{A}^{\top}\hat{y}_{t}^{\ast}(p_{t}-\psi y_{S,t}-w_{t}^{\top}\gamma) \\ 
w_{t}(p_{t}-\psi y_{S,t}-w_{t}^{\top}\gamma)%
\end{pmatrix}
.   \label{F_Moments}
\end{equation}

The GIV estimator is defined as 
\begin{equation}
\hat{\theta}(\hat{A})\equiv \arg \min_{\theta \in \Theta}\bar{g}_{T}(\theta ;%
\hat{A},\hat{\beta})^{\top}W_{0,T}(\hat{A})\bar{g}_{T}(\theta;\hat{A},\hat{%
\beta}),   \label{G_GIV_1}
\end{equation}
where 
\begin{equation}
W_{0,T}(\hat{A})\equiv \Bigl(\mathrm{diag}(\hat{A}^{\top},\hat{A}^{\top },%
\mathbf{I}_{d_{w}})\,W_{0,T}\, \mathrm{diag}(\hat{A},\hat{A},\mathbf{I}%
_{d_{w}})\Bigr)^{-1},   \label{F_Weight}
\end{equation}
and $W_{0,T}$ is a user-specified symmetric positive definite $%
(2n+d_{w})\times(2n+d_{w})$ matrix. Since $\bar{g}_{T}(\theta;\hat{A},\hat{%
\beta})$ is linear in $\theta$, the GIV estimator takes the same form as in (%
\ref{GIV_Form_1}), with $D_{j,T}(\hat{A})$ redefined as%
\begin{equation}
D_{j,T}(\hat{A})\equiv \mathrm{diag}(\hat{A}^{\top},\hat{A}^{\top},\mathbf{I}%
_{d_{w}})D_{j,T},\quad j=1,2,   \label{GIV_Form_3}
\end{equation}
where 
\begin{equation}
D_{1,T}\equiv T^{-1}\sum_{t\leq T}%
\begin{pmatrix}
\hat{y}_{t}^{\ast}p_{t} & \mathbf{0}_{n} & \mathbf{0}_{n\times d_{w}} \\ 
\mathbf{0}_{n} & \hat{y}_{t}^{\ast}y_{S,t} & \hat{y}_{t}^{\ast}w_{t}^{\top}
\\ 
\mathbf{0}_{d_{w}} & w_{t}y_{S,t} & w_{t}w_{t}^{\top}%
\end{pmatrix}
,\quad D_{2,T}\equiv T^{-1}\sum_{t\leq T}%
\begin{pmatrix}
\hat{y}_{t}^{\ast}\hat{y}_{e,t}^{\ast} \\ 
\hat{y}_{t}^{\ast}p_{t} \\ 
w_{t}p_{t}%
\end{pmatrix}
.   \label{GIV_Form_4}
\end{equation}

To conduct inference on $\theta$ and test the validity of the moment
conditions in (\ref{F_Moments}), one must account for the estimation error
in $\hat{\beta}$, since it enters $\bar{g}_{T}(\theta;\hat{A},\hat{\beta})$
through both $\hat{y}_{t}^{\ast}$ and $\hat{A}$. Lemma \ref{Moment_est} in
Online Appendix \ref{APP_5} shows that the randomness introduced by the
estimation error of $\hat{\beta}$ is of higher order. Therefore, the
estimation error of $\hat{\beta}$ is asymptotically negligible and can be
ignored. Consequently, the standard errors of $\hat{\theta}(\hat{A})$ and
the specification tests can be constructed in the same way as in the
previous section. See Algorithm 1 in Online Appendix \ref{APP_0} for details.

\subsection{Model with unbalanced data structure\label{subsec: Ex2}}

The data used to estimate the demand and supply equations have, thus far,
been assumed to follow a balanced structure.\ For instance, $y_{i,t}$
denotes the demand of entity $i$ in period $t$, where $i\in \{1,\ldots,n\}$
and $t\in \{1,\ldots,T\}$. The analysis in the previous section assumes a
balanced data structure, so that each time period is associated with the
same number of entities. In practice, however, entry and exit lead to an
unbalanced data structure. As we show below, the identification and
estimation approach extends naturally to this setting, provided that the
entry and exit decisions of entities are independent of their demand. \ 

Specifically, the demand equation in (\ref{G_demand}) is generalized as 
\begin{equation*}
y_{t}=\phi p_{t}\mathbf{1}_{n_{t}}+\lambda \eta_{t}+u_{t}, 
\end{equation*}
where $y_{t}\equiv(y_{i,t})_{i\leq n_{t}}$ and $n_{t}$ denotes the number of
entities present in the market at time $t$. The factor-loading matrix $%
\lambda$ is of dimension $n_t\times r$, and the idiosyncratic shock vector 
$u_t$ is of dimension $n_t\times1$. The supply equation in (\ref{G_supply})
remains unchanged, with $y_{S,t}\equiv S_{t}^{\top}y_{t}$, where $S_{t}$ is
an $n_{t}\times1$ vector of market shares.

To construct the GIV, consider a subsample of $n_{0}$ entities that are
observed in all periods. Let $y_{t}^{0}$ denote the corresponding subvector
of $y_{t}$. Their demand equation is 
\begin{equation}
y_{t}^{0}=\phi p_{t}\mathbf{1}_{n_{0}}+\lambda_{0}\eta_{t}+u_{t}^{0}, 
\label{n0_Demand}
\end{equation}
where $\lambda_{0}$ and $u_{t}^{0}$ are the associated submatrices of $%
\lambda$ and $u_{t}$. Let $Q_{0,-1}\equiv(q_{0,2},\ldots,q_{0,n_{0}})\in 
\mathbb{R}^{n_{0}\times(n_{0}-1)}$ be defined analogously to (\ref{ortho_q_j}%
) with $n$ replaced by $n_{0}$, and $\lambda_{0,\bot}\in \mathbb{R}%
^{n_{0}\times(n_{0}-\bar{r}_{0})}$ denotes the orthogonal complement of $(%
\mathbf{1}_{n_{0}},\lambda_{0})$ where $\bar{r}_{0}\equiv \mathrm{rank}((%
\mathbf{1}_{n_{0}},\lambda_{0}))$. Since $\lambda
_{0,\bot}^{\top}y_{t}^{0}=\lambda_{0,\bot}^{\top}u_{t}^{0}$, Assumption \ref%
{ID} implies 
\begin{equation*}
\mathbb{E}[\eta_{t}u_{t}^{0\top}\lambda_{0,\bot}]=\mathbf{0}_{r\times (n_{0}-%
\bar{r}_{0})},\quad \mathbf{1}_{n_{t}}^{\top}\mathbb{E}[u_{t}u_{t}^{0\top}]%
\lambda_{0,\bot}=\mathbf{0}_{n_{0}-\bar{r}_{0}},\quad \mathbb{E}%
[\varepsilon_{t}u_{t}^{0\top}\lambda_{0,\bot}]=\mathbf{0}_{n_{0}-\bar{r}%
_{0}}. 
\end{equation*}
These imply $\lambda_{0,\bot}^{\top}y_{t}^{0}$ provides valid moment
conditions%
\begin{equation*}
\mathbb{E}\! \left[ 
\begin{pmatrix}
\lambda_{0,\bot}^{\top}y_{t}^{0}(y_{e,t}-\phi p_{t}) \\ 
\lambda_{0,\bot}^{\top}y_{t}^{0}(p_{t}-\psi y_{S,t})%
\end{pmatrix}
\right] =\mathbf{0}_{2(n_{0}-\bar{r}_{0})}, 
\end{equation*}
which identify $\phi$ and $\psi$.

Lemmas \ref{ID_G_GIV_Weight} and \ref{G_GIV_Solutions} apply to this
subsample. In particular, $\lambda_{0,\bot}$ can be consistently estimated
(up to an orthonormal rotation) by $\hat{A}\equiv Q_{0,-1}\hat{A}_{0}$,
where $\hat{A}_{0}$ collects the eigenvectors corresponding to the smallest $%
n_{0}-\bar{r}_{0}$ eigenvalues of 
\begin{equation*}
\hat{S}_{y^{0}}\equiv Q_{0,-1}^{\top}\hat{\Sigma}_{y^{0}}Q_{0,-1},\qquad 
\hat{\Sigma}_{y^{0}}\equiv T^{-1}\sum_{t\leq T}y_{t}^{0}y_{t}^{0\top}. 
\end{equation*}
The GIV estimator is defined analogously to (\ref{GMM_Criterion}), with
moment function 
\begin{equation*}
\bar{g}_{T}^{0}(\theta;\hat{A})\equiv T^{-1}\sum_{t\leq T}%
\begin{pmatrix}
\hat{A}^{\top}y_{t}^{0}(y_{e,t}-\phi p_{t}) \\ 
\hat{A}^{\top}y_{t}^{0}(p_{t}-\psi y_{S,t})%
\end{pmatrix}
, 
\end{equation*}
and a user-specified symmetric positive definite $2n_{0}\times2n_{0}$ weight
matrix, and it admits an explicit form 
\begin{equation}
\hat{\theta}(\hat{A})=\big(D_{1,T}^{0}(\hat{A})^{\top}W_{0,T}(\hat{A}%
)D_{1,T}^{0}(\hat{A})\big)^{-1}\big(D_{1,T}^{0}(\hat{A})^{\top}W_{0,T}(\hat {%
A})D_{2,T}^{0}(\hat{A})\big),   \label{GIV_Form_n0_1}
\end{equation}
where $D_{j,T}^{0}(\hat{A})\equiv(\mathbf{I}_{2}\otimes \hat{A}%
^{\top})D_{j,T}^{0}$ for $j=1,2$, and 
\begin{equation}
D_{1,T}^{0}\equiv T^{-1}\sum_{t\leq T}\mathrm{diag}\! \left(
y_{t}^{0}p_{t},\,y_{t}^{0}y_{S,t}\right) ,\qquad D_{2,T}^{0}\equiv
T^{-1}\sum_{t\leq T}%
\begin{pmatrix}
y_{t}^{0}y_{e,t} \\ 
y_{t}^{0}p_{t}%
\end{pmatrix}
.   \label{GIV_Form_n0_2}
\end{equation}

The asymptotic normality of the GIV estimator and the asymptotic
distribution of the $J$-test statistic follow from the same arguments as in
the proofs of Theorems \ref{Asy_Dist} and \ref{J_Test}, with the appropriate
modifications to the $D_{1}$ and $V$ matrices. To conserve space, the
implementation details of the GIV estimation and inference procedure are
provided in Algorithm 2 in Online Appendix \ref{APP_0}.

\begin{remark}
The method developed in this subsection constructs GIVs using data from a
subset of $n_{0}$ entities. Once the GIVs are obtained, the full data set
can be used to construct moment conditions for estimating the unknown
parameters in the model. Since $n_{0}$ may be substantially smaller than
both $n$ and $T$, the proposed approach can be applied in settings where the
cross-sectional dimension is large and may even exceed the sample size.
\end{remark}

\begin{remark}
The flexibility of using only a subset of entities to construct the GIVs
also allows the framework to accommodate heterogeneous demand elasticities
across entities, provided that a subset of entities is known to share a
common elasticity.

Specifically, suppose that the first $n_{0}$ entities share a common demand
elasticity $\bar{\phi}$. We may use the demand equations for these entities,
i.e., (\ref{n0_Demand}), to construct the GIVs $\hat{A}^{\top}y_{t}^{0}$,
which can then be employed to form the moment functions 
\begin{equation}
\bar{g}_{T}^{n_{0}}(\bar{\phi};\hat{A})\equiv T^{-1}\sum_{t\leq T}\hat {A}%
^{\top}y_{t}^{0}(y_{n_{0},t}-\bar{\phi}p_{t})\text{ \ \ and \ \ }\bar {g}%
_{T}^{i}(\phi_{i};\hat{A})\equiv T^{-1}\sum_{t\leq T}\hat{A}%
^{\top}y_{t}^{0}(y_{i,t}-\phi_{i}p_{t}),   \label{GIV_Form_n0_3}
\end{equation}
to estimate the common elasticity $\bar{\phi}$ and the entity-specific
elasticities $\phi_{i}$ for $i=n_{0}+1,\ldots,n$, where $%
y_{n_{0},t}=n_{0}^{-1}\mathbf{1}_{n_{0}}^{\top}y_{t}^{0}$. Using arguments
analogous to those in the proof of Theorem \ref{Asy_Dist}, it can be shown
that the resulting GIV estimators of $\bar{\phi}$ and $\phi_{i}$ ($%
i=n_{0}+1,\ldots,n$) are $T^{1/2}$-consistent and asymptotically normal.
Together with consistent estimators of their asymptotic variances, these
results can be used to conduct inference on heterogeneous demand
elasticities and to test hypotheses such as $H_{0}:\bar{\phi}=\phi_{i}$ and $%
H_{0}:\phi_{i}=\phi_{i^{\prime}}$ for $i\neq i^{\prime}$.
\end{remark}

\section{Simulation Studies\label{sec: MC}}

We examine the finite-sample performance of the proposed GIV estimation and
inference procedures through Monte Carlo experiments. Subsection \ref%
{sec-mc1} describes the simulation design, and Subsection \ref{sec-mc2}
reports the results.

\subsection{Simulation Setting\ \label{sec-mc1}}

We consider two simulation designs, a baseline design and an extended
design, to investigate the finite-sample performance of the proposed GIV
estimator, the associated inference procedures, and the specification test.
The baseline design follows the model in (\ref{G_demand})--(\ref{G_supply}),
while the extended design augments the demand equation with three exogenous
regressors.

Solving the demand-supply system in (\ref{F_demand})--(\ref{F_supply}), with
exogenous regressors in the demand equation but no additional exogenous
variables in the supply equation, yields the following reduced-form
expressions under the extended design: 
\begin{align}
y_{t} & =\left( \mathbf{I}_{n}+\frac{\phi \psi}{1-\phi \psi}\mathbf{1}%
_{n}S_{t}^{\top}\right) (x_{t}\beta+\lambda \eta_{t}+u_{t})+\frac{\phi}{%
1-\phi \psi}\mathbf{1}_{n}\varepsilon_{t},  \label{F_Reduced_Demand} \\
p_{t} & =\frac{\psi}{1-\phi \psi}S_{t}^{\top}(x_{t}\beta+\lambda \eta
_{t}+u_{t})+\frac{\varepsilon_{t}}{1-\phi \psi}.   \label{F_Reduced_Supply}
\end{align}
To generate the simulated data, we first draw the demand and supply shocks $%
\eta_{t}$, $u_{t}$, and $\varepsilon_{t}$, the exogenous regressors $x_{t}$,
and the market share vector $S_{t}$ conditional on the parameter values of $%
\phi$, $\psi$, $\beta$, and $\lambda$. These simulated values are then
substituted into (\ref{F_Reduced_Demand})--(\ref{F_Reduced_Supply}) to
obtain $(p_{t},y_{t})$. This procedure generates the simulated observations $%
\{y_{t},p_{t},x_{t},S_{t}\}$ for each period $t$.

The demand and supply shocks are mutually independent and i.i.d.\ across $t$%
, with 
\begin{equation}
\eta_{t}\sim N(\mathbf{0}_{r+1},\Sigma_{\eta}),\qquad u_{t}\sim N(\mathbf{0}%
_{n},\mathbf{I}_{n}),\qquad \varepsilon_{t}\sim
N(0,\sigma_{\varepsilon}^{2}),   \label{MC_Shocks}
\end{equation}
where $\Sigma_{\eta}\equiv((0.1)^{|i-j|})_{i,j\leq r+1}$, and $\sigma
_{\varepsilon}^{2}$ is set to $0.5$. The demand and supply elasticities are
set to $\phi=-0.5$ and $\psi=1.5$, respectively. The exogenous regressors $%
\{x_{i,t}\}$ are generated i.i.d.\ from $N(\mathbf{0}_{3},\mathbf{I}_{3})$
across $i$ and $t$, independently of $(\eta_{t}^{\top},u_{t}^{\top
},\varepsilon_{t})^{\top}$. The market share vector $S_{t}$ is fixed across $%
t$ and follows a Pareto rank-size specification:\ \ $s_{i}\propto \left(
i/n\right) ^{-1/\mu_{S}}$, with tail index $\mu_{S}=0.2$, where the shares
are normalized to satisfy $\sum_{i=1}^{n}s_{i}=1$. Such a power-law profile
is consistent with the size distributions documented for industries, firms,
and financial intermediaries \citep{gabaix2011granular,gabaix2024granular},
and generates the concentrated cross-sectional structure under which
granular variation is informative.

The factor loading matrix is specified as $\lambda=(\mathbf{1}%
_{n},\lambda_{-1})$, where $\lambda_{-1}$ is the $n\times r$ matrix of
loadings on the $r$ non-aggregate latent factors. To construct $\lambda_{-1}$%
, we first draw $nr$ independent $N(0,1)$ random variables to form an $%
n\times r$ matrix $\lambda_{0}$. We then project $\lambda_{0}$ onto the
orthogonal complement of $(\mathbf{1}_{n},S)$ and obtain a preliminary
loading matrix $\tilde{\lambda }_{0}$ through the QR decomposition $M_{(%
\mathbf{1}_{n},S)}\lambda_{0}=\tilde{\lambda}_{0}R_{\lambda_{0}}$,\ where $%
R_{\lambda_{0}}$ is an $r\times r$ upper triangular matrix and the columns
of $\tilde{\lambda}_{0}$ are orthonormal. We then set $\lambda_{-1}=n^{1/2}%
\tilde{\lambda}_{0}$. By construction, the resulting loading matrix $%
\lambda=(\mathbf{1}_{n},\lambda_{-1})$ has rank $\bar{r}=r+1$.

We set $\beta=\mathbf{0}_{d_{x}}$ in (\ref{F_Reduced_Demand})--(\ref%
{F_Reduced_Supply}) for both designs. In the baseline design, $\beta$ is
treated as known and the regressors $x_{t}$ are omitted. In the extended
design, however, $\beta$ is estimated using the procedure described in
Subsection~\ref{subsec: Ex1}. Because $\beta=\mathbf{0}_{d_{x}}$, the
regressors $x_{t}$ play no role in the data-generating process. Thus, any
difference between the two designs reflects the additional estimation error
associated with estimating $\beta$ and partialling out $x_{t}$.

We consider six combinations of the number of entities and the number of
non-aggregate latent factors: 
\begin{equation}
(n,r)\in \{(5,1),(5,2),(8,3),(8,5),(10,5),(10,7)\},   \label{MC_nr_pair}
\end{equation}
together with three sample sizes, $T\in \{150,300,450\}$.\ To evaluate the
finite-sample performance of the GIV estimator, as well as the size
properties of the associated inference and specification tests, we conduct $%
10{,}000$ Monte Carlo replications for each $(n,r,T)$ cell under each
simulation design.

To assess the effect of estimation error arising from the recovery of the
subspace orthogonal to $\func{col}((\mathbf{1}_{n},\lambda))$ on the
performance of the GIV estimator, we consider an oracle GIV estimator
constructed under the assumption that $\func{col}((\mathbf{1}_{n},\lambda))$
is known. Consequently, the matrix $A$, whose columns form a basis for the
orthogonal complement of $\func{col}((\mathbf{1}_{n},\lambda))$ and are used
to construct the GIVs, is treated as known.\footnote{%
In the simulation, $A$ is constructed from the left singular vectors
associated with the zero singular values in the singular value decomposition
of $(\mathbf{1}_{n},\lambda)$.}\ The oracle GIV estimator therefore bypasses
both the BIC step for estimating the number of factors and the estimation of
the orthogonal complement of $\func{col}((\mathbf{1}_{n},\lambda))$. In the
extended design, the oracle estimator additionally treats $\beta$ as known
and partials out $x_{t}\beta$ using the true parameter value. By contrast,
the feasible GIV estimator selects $\bar {r}$ using the BIC criterion in (%
\ref{BIC})--(\ref{r_hat}), and then constructs the GIVs and the
corresponding GIV estimator according to Algorithm 1 of Online Appendix \ref%
{APP_0}.

Both the oracle and feasible GIV estimators are evaluated using their
finite-sample root mean squared errors (RMSEs), with the results reported in
Table \ref{tab:rmse} of the next subsection. We also investigate the
empirical rejection probabilities of the two-sided tests of $H_{0}:\phi=-0.5$
and $H_{0}:\psi=1.5$ at the $5\%$ significance level. The results are
reported in Table \ref{tab:t_size} of the next subsection. Inference is
conducted using $t$-tests based on the standard error estimators described
in Algorithm 1 of Online Appendix \ref{APP_0}, together with the asymptotic
normality of the GIV estimators.

To evaluate the power of the $J$-test, we consider a controlled violation of
the covariance restrictions underlying GIV validity while keeping the
factor-loading matrix $\lambda=(\mathbf{1}_{n},\lambda_{-1})$, the
structural parameters, and the marginal distributions of $u_{t}$, $\eta_{t}$%
, and $\varepsilon_{t}$ unchanged. Let $b_{n}=Ad_{n}$, where $d_{n}\in 
\mathbb{R}^{n-r-1}$ is the unit vector obtained by applying the
Gram--Schmidt procedure to the first standard basis vector against $A^{\top}S
$. For $\rho \in \lbrack0,1)$, we generate $u_{t}$ as 
\begin{equation*}
u_{t}=u_{t}^{\ast}+\frac{\rho}{\sigma_{\varepsilon}}\varepsilon_{t}b_{n}-%
\Bigl(1-(1-\rho^{2})^{1/2}\Bigr)\bigl(b_{n}^{\top}u_{t}^{\ast }\bigr)b_{n}, 
\end{equation*}
where $u_{t}^{\ast}\sim N(\mathbf{0}_{n},\mathbf{I}_{n})$ is generated
together with $\varepsilon_{t}$ and $\eta_{t}$ in the same way as $u_{t}$ in
(\ref{MC_Shocks}).

By construction, $b_{n}\in \func{col}(A)$ implies $b_{n}^{\top }\mathbf{1}%
_{n}=0\ $and$\ b_{n}^{\top}\lambda_{-1}=\mathbf{0}_{r}^{\top}$, so $b_{n}$
is orthogonal to $\func{col}(\lambda)$. It is straightforward to verify that
for any $\rho \in \lbrack0,1)$, the joint distribution of $%
(u_{t},\varepsilon_{t})$ remains Gaussian with 
\begin{equation*}
\func{Var}(u_{t})=\mathbf{I}_{n}\qquad \text{and}\qquad \func{Cov}%
(u_{t},\varepsilon_{t})=\rho \sigma_{\varepsilon}b_{n}. 
\end{equation*}
When $\rho=0$, the design reduces to the correctly specified benchmark
design used above to study the finite-sample properties of the GIV
estimators and inference procedures. When $\rho>0$, however, $u_{t}$ and $%
\varepsilon_{t}$ become correlated, with heterogeneous correlation patterns
determined by $b_{n}$. As a result, the orthogonality conditions underlying
the GIVs are violated, rendering the GIVs invalid. Moreover, the degree of
misspecification increases linearly with $\rho$.

We use the six $(n,r)$ combinations in (\ref{MC_nr_pair}) together with the
three sample sizes $T$ considered above, vary $\rho$ over the grid $\rho
_{j}=0.02j\ $for $j=0,\ldots,20$, and conduct $10{,}000$ simulation
replications for each $(n,r,T,\rho)$ cell under each design. The
significance level of the $J$-test is set at $0.05$, and the resulting
empirical rejection probabilities are reported in Figure~\ref%
{fig:j_power_baseline} of the next subsection.

\begin{table}[!t]
\caption{Root Mean Squared Error of the GIV Estimators}
\label{tab:rmse}\centering
{\small \setlength{\tabcolsep}{4.5pt} \renewcommand{\arraystretch}{1.05} 
\begin{threeparttable}
		\begin{tabular*}{\textwidth}{@{\extracolsep{\fill}}cc cccc cccc@{}}
			\toprule
			& & \multicolumn{4}{c}{Baseline ($d_x = 0$)}
			& \multicolumn{4}{c}{Extended ($d_x = 3$)} \\
			\cmidrule(lr){3-6}\cmidrule(lr){7-10}
			& & \multicolumn{2}{c}{Oracle} & \multicolumn{2}{c}{Feasible}
			& \multicolumn{2}{c}{Oracle} & \multicolumn{2}{c}{Feasible} \\
			\cmidrule(lr){3-4}\cmidrule(lr){5-6}\cmidrule(lr){7-8}\cmidrule(lr){9-10}
			$(n, r)$ & $T$
			& $\phi$ & $\psi$ & $\phi$ & $\psi$
			& $\phi$ & $\psi$ & $\phi$ & $\psi$ \\
			\midrule
			\multirow{3}{*}{$(5, 1)$}
			& 150 & 0.129 & 0.086 & 0.130 & 0.087 & 0.126 & 0.085 & 0.131 & 0.091 \\
			& 300 & 0.088 & 0.060 & 0.092 & 0.065 & 0.088 & 0.060 & 0.089 & 0.065 \\
			& 450 & 0.072 & 0.049 & 0.076 & 0.051 & 0.071 & 0.049 & 0.074 & 0.050 \\
			\addlinespace
			\multirow{3}{*}{$(5, 2)$}
			& 150 & 0.128 & 0.086 & 0.150 & 0.109 & 0.129 & 0.086 & 0.160 & 0.116 \\
			& 300 & 0.089 & 0.060 & 0.106 & 0.078 & 0.088 & 0.061 & 0.103 & 0.083 \\
			& 450 & 0.071 & 0.049 & 0.094 & 0.072 & 0.071 & 0.049 & 0.095 & 0.071 \\
			\addlinespace
			\multirow{3}{*}{$(8, 3)$}
			& 150 & 0.118 & 0.082 & 0.119 & 0.084 & 0.118 & 0.082 & 0.121 & 0.084 \\
			& 300 & 0.081 & 0.057 & 0.081 & 0.058 & 0.080 & 0.058 & 0.081 & 0.059 \\
			& 450 & 0.066 & 0.047 & 0.067 & 0.048 & 0.066 & 0.046 & 0.067 & 0.052 \\
			\addlinespace
			\multirow{3}{*}{$(8, 5)$}
			& 150 & 0.118 & 0.082 & 0.142 & 0.115 & 0.117 & 0.082 & 0.152 & 0.115 \\
			& 300 & 0.080 & 0.057 & 0.100 & 0.070 & 0.081 & 0.057 & 0.111 & 0.081 \\
			& 450 & 0.066 & 0.046 & 0.081 & 0.062 & 0.066 & 0.046 & 0.097 & 0.080 \\
			\addlinespace
			\multirow{3}{*}{$(10, 5)$}
			& 150 & 0.113 & 0.082 & 0.115 & 0.084 & 0.115 & 0.081 & 0.118 & 0.083 \\
			& 300 & 0.079 & 0.057 & 0.080 & 0.057 & 0.080 & 0.057 & 0.081 & 0.058 \\
			& 450 & 0.064 & 0.046 & 0.064 & 0.046 & 0.064 & 0.045 & 0.064 & 0.046 \\
			\addlinespace
			\multirow{3}{*}{$(10, 7)$}
			& 150 & 0.113 & 0.080 & 0.145 & 0.113 & 0.113 & 0.081 & 0.142 & 0.110 \\
			& 300 & 0.079 & 0.056 & 0.102 & 0.083 & 0.080 & 0.056 & 0.107 & 0.093 \\
			& 450 & 0.064 & 0.045 & 0.081 & 0.058 & 0.064 & 0.045 & 0.081 & 0.059 \\
			\bottomrule
		\end{tabular*}
			\begin{tablenotes}[flushleft]
				\footnotesize
				\item \textit{Notes.} This table reports the root mean squared errors of the GIV estimators based on
				$10{,}000$ Monte Carlo replications, separately for the oracle and feasible
				estimators, under the baseline design ($d_x=0$) and the extended design with
				$d_x=3$ exogenous regressors. The oracle estimator treats $\lambda$ (and
				$\beta$ when present) as known. The feasible estimator selects $\bar r$ using
				the BIC criterion in~\eqref{BIC}--\eqref{r_hat} and is computed according to
				Algorithm~1 in Online Appendix~\ref{APP_0}. Both estimators are evaluated using the
				same simulated data within each replication. To reduce the influence of extreme draws, the RMSE is computed using the truncated loss function $\min \{5,(\widehat{\alpha}-\alpha)^2\}$ for a generic estimator $\widehat{\alpha}$ of the parameter $\alpha$.
			\end{tablenotes}
		\end{threeparttable}
}
\end{table}

\subsection{Simulation Results \label{sec-mc2}}

Table \ref{tab:rmse} shows that the RMSEs of both the oracle and feasible
estimators of $\phi$ and $\psi$ decline substantially as the sample size $T$
increases. For example, in the baseline design with $(n,r)=(5,1)$, the RMSE
of the oracle estimator of $\phi$ decreases from $0.129$ at $T=150$, to $%
0.088$ at $T=300$, and further to $0.072$ at $T=450$. Similar improvements
are observed across all configurations and for both structural parameters.

The feasible GIV estimator closely tracks the oracle estimator throughout
the simulation designs. In most cases, the difference in RMSE between the
oracle and feasible estimators is negligible. For instance, in the baseline
design with $(n,r)=(10,5)$, the RMSEs of the feasible estimator $\hat{\phi}%
^{\ast }(\hat{A})$ are $0.115$, $0.080$, and $0.064$ at $T=150$, $300$, and $%
450$, respectively, compared with the corresponding oracle RMSEs of $0.113$, 
$0.079$, and $0.064$. Similar qualitative patterns are observed for the GIV
estimator of $\psi$, as well as for both estimators in the extended design
with exogenous regressors.

We next examine the size properties of the $t$-tests for $H_{0}:\phi=-0.5$
and $H_{0}:\psi=1.5$ at the $5\%$ significance level. The results, reported
in Table \ref{tab:t_size}, show that the empirical rejection probabilities
approach the nominal level as the sample size $T$ increases. When the sample
size is relatively small, i.e., $T=150$, the tests exhibit modest
over-rejection, with the distortion becoming more pronounced in
configurations involving a larger number of GIVs. For example, in the
baseline design with $(n,r)=(8,3)$, where there are four GIVs and eight
moment conditions, the empirical rejection probabilities of the $t$-tests
based on the feasible GIV estimators for $\phi$ and $\psi$ are $0.104$ and $%
0.073$, respectively, at $T=150$. By contrast, when $(n,r)=(8,5)$, where
only two GIVs are available, the corresponding rejection probabilities are $%
0.075$ and $0.061$, respectively. Similar patterns are observed in the
extended design.

The over-rejection observed in these $t$-tests does not appear to be
primarily driven by estimation error in the number of factors or in the null
space of $(\mathbf{1}_{n},\lambda)$, since the tests based on the oracle GIV
estimators, which do not require estimation of these nuisance parameters,
display similar finite-sample behavior. Instead, the over-rejection is
likely related to the well-known many-moment bias in two-step GMM
estimation, of which the GIV estimator is a special case; see, for example, %
\citet{HansenHeatonYaron1996} and \citet{newey2009generalized}. Several
approaches may help mitigate the resulting size distortion in small samples.
For example, instead of estimating $\phi$ and $\psi$ jointly, one may
estimate them using two separate GMM procedures. This reduces the number of
moment conditions used in each estimation problem, although it may sacrifice
some of the efficiency gains from joint GMM estimation. Another possibility
is to employ the continuously updated GMM estimator rather than the two-step
GMM estimator. While this approach may improve finite-sample inference, it
also introduces additional computational burden, since the continuously
updated GMM estimator does not admit a closed-form solution.

\begin{table}[!t]
\caption{Empirical Rejection Probabilities of the Two-sided $t$-tests}
\label{tab:t_size}\centering
{\small \setlength{\tabcolsep}{4.5pt} \renewcommand{\arraystretch}{1.05} 
\begin{threeparttable}
		\begin{tabular*}{\textwidth}{@{\extracolsep{\fill}}cc cccc cccc@{}}
			\toprule
			& & \multicolumn{4}{c}{Baseline ($d_x = 0$)}
			& \multicolumn{4}{c}{Extended ($d_x = 3$)} \\
			\cmidrule(lr){3-6}\cmidrule(lr){7-10}
			& & \multicolumn{2}{c}{Oracle} & \multicolumn{2}{c}{Feasible}
			& \multicolumn{2}{c}{Oracle} & \multicolumn{2}{c}{Feasible} \\
			\cmidrule(lr){3-4}\cmidrule(lr){5-6}\cmidrule(lr){7-8}\cmidrule(lr){9-10}
			$(n, r)$ & $T$
			& $\phi$ & $\psi$ & $\phi$ & $\psi$
			& $\phi$ & $\psi$ & $\phi$ & $\psi$ \\
			\midrule
			\multirow{3}{*}{$(5, 1)$}
			& 150 & 0.091 & 0.063 & 0.094 & 0.062 & 0.089 & 0.058 & 0.090 & 0.061 \\
			& 300 & 0.068 & 0.058 & 0.067 & 0.059 & 0.071 & 0.058 & 0.071 & 0.059 \\
			& 450 & 0.063 & 0.052 & 0.065 & 0.053 & 0.058 & 0.056 & 0.059 & 0.058 \\
			\addlinespace
			\multirow{3}{*}{$(5, 2)$}
			& 150 & 0.071 & 0.058 & 0.072 & 0.058 & 0.073 & 0.057 & 0.076 & 0.059 \\
			& 300 & 0.062 & 0.052 & 0.064 & 0.053 & 0.060 & 0.057 & 0.062 & 0.056 \\
			& 450 & 0.056 & 0.052 & 0.056 & 0.052 & 0.058 & 0.056 & 0.057 & 0.055 \\
			\addlinespace
			\multirow{3}{*}{$(8, 3)$}
			& 150 & 0.101 & 0.069 & 0.104 & 0.073 & 0.104 & 0.070 & 0.113 & 0.073 \\
			& 300 & 0.076 & 0.059 & 0.078 & 0.060 & 0.073 & 0.062 & 0.076 & 0.064 \\
			& 450 & 0.068 & 0.058 & 0.068 & 0.059 & 0.067 & 0.055 & 0.071 & 0.056 \\
			\addlinespace
			\multirow{3}{*}{$(8, 5)$}
			& 150 & 0.070 & 0.058 & 0.075 & 0.061 & 0.069 & 0.059 & 0.075 & 0.061 \\
			& 300 & 0.058 & 0.058 & 0.058 & 0.056 & 0.057 & 0.054 & 0.059 & 0.053 \\
			& 450 & 0.056 & 0.052 & 0.059 & 0.054 & 0.055 & 0.051 & 0.057 & 0.051 \\
			\addlinespace
			\multirow{3}{*}{$(10, 5)$}
			& 150 & 0.096 & 0.072 & 0.102 & 0.077 & 0.101 & 0.067 & 0.107 & 0.073 \\
			& 300 & 0.078 & 0.059 & 0.079 & 0.059 & 0.077 & 0.061 & 0.082 & 0.063 \\
			& 450 & 0.068 & 0.058 & 0.068 & 0.057 & 0.068 & 0.053 & 0.066 & 0.054 \\
			\addlinespace
			\multirow{3}{*}{$(10, 7)$}
			& 150 & 0.069 & 0.057 & 0.073 & 0.064 & 0.068 & 0.057 & 0.075 & 0.062 \\
			& 300 & 0.060 & 0.057 & 0.061 & 0.059 & 0.061 & 0.051 & 0.062 & 0.053 \\
			& 450 & 0.057 & 0.053 & 0.058 & 0.056 & 0.055 & 0.053 & 0.057 & 0.054 \\
			\bottomrule
		\end{tabular*}
			\begin{tablenotes}[flushleft]
	\footnotesize
	\item \textit{Notes.} This table reports the empirical rejection rates of the two-sided $t$-tests of $H_{0}:\phi=-0.5$ and $H_{0}:\psi=1.5$ at the 5\% nominal significance level, using the asymptotic critical value $1.96$, based on $10{,}000$ Monte Carlo replications. Results are reported separately for the oracle and feasible estimators under both the baseline design ($d_x=0$) and the extended design with $d_x=3$ exogenous regressors. The oracle estimator treats $\lambda$ (and $\beta$ when present) as known. The feasible estimator selects $\bar{r}$ using the BIC criterion in \eqref{BIC}--\eqref{r_hat} and is implemented according to Algorithm~1 in Online Appendix~\ref{APP_0}.
\end{tablenotes}
		\end{threeparttable}
}
\end{table}

\begin{table}[!t]
\caption{Empirical Rejection Probabilities of the $J$-test}
\label{tab:j_size}\centering
{\small \setlength{\tabcolsep}{6pt} \renewcommand{\arraystretch}{1.05} 
\begin{threeparttable}
		\begin{tabular*}{\textwidth}{@{\extracolsep{\fill}}cc cc cc@{}}
			\toprule
			& & \multicolumn{2}{c}{Baseline ($d_x = 0$)}
			& \multicolumn{2}{c}{Extended ($d_x = 3$)} \\
			\cmidrule(lr){3-4}\cmidrule(lr){5-6}
			$(n, r)$ & $T$ & Oracle & Feasible & Oracle & Feasible \\
			\midrule
			\multirow{3}{*}{$(5, 1)$}
			& 150 & 0.067 & 0.071 & 0.068 & 0.068 \\
			& 300 & 0.060 & 0.062 & 0.058 & 0.059 \\
			& 450 & 0.053 & 0.052 & 0.054 & 0.055 \\
			\addlinespace
			\multirow{3}{*}{$(5, 2)$}
			& 150 & 0.060 & 0.061 & 0.059 & 0.061 \\
			& 300 & 0.056 & 0.058 & 0.052 & 0.055 \\
			& 450 & 0.056 & 0.057 & 0.058 & 0.057 \\
			\addlinespace
			\multirow{3}{*}{$(8, 3)$}
			& 150 & 0.074 & 0.082 & 0.073 & 0.082 \\
			& 300 & 0.064 & 0.064 & 0.062 & 0.067 \\
			& 450 & 0.058 & 0.058 & 0.060 & 0.061 \\
			\addlinespace
			\multirow{3}{*}{$(8, 5)$}
			& 150 & 0.060 & 0.064 & 0.055 & 0.063 \\
			& 300 & 0.057 & 0.058 & 0.055 & 0.058 \\
			& 450 & 0.050 & 0.050 & 0.053 & 0.057 \\
			\addlinespace
			\multirow{3}{*}{$(10, 5)$}
			& 150 & 0.069 & 0.080 & 0.074 & 0.083 \\
			& 300 & 0.055 & 0.060 & 0.064 & 0.069 \\
			& 450 & 0.056 & 0.061 & 0.053 & 0.058 \\
			\addlinespace
			\multirow{3}{*}{$(10, 7)$}
			& 150 & 0.060 & 0.065 & 0.055 & 0.063 \\
			& 300 & 0.054 & 0.058 & 0.052 & 0.057 \\
			& 450 & 0.054 & 0.056 & 0.053 & 0.054 \\
			\bottomrule
		\end{tabular*}
			\begin{tablenotes}[flushleft]
	\footnotesize
	\item \textit{Notes.} This table reports the empirical rejection rates of the $J$-test at the $5\%$
	nominal significance level based on $10{,}000$ Monte Carlo replications.
	Results are reported separately for the oracle and feasible estimators under
	both the baseline design ($d_x=0$) and the extended design with
	$d_x=3$ exogenous regressors. For the oracle estimator, the reference
	distribution is $\chi^{2}_{2(n-r-2)}$. For the feasible estimator, the
	reference distribution is the replication-specific
	$\chi^{2}_{2(n-\widehat{r}-1)}$. The oracle estimator treats $\lambda$ (and
	$\beta$ when present) as known. The feasible estimator selects $\bar r$ using
	the BIC criterion in~\eqref{BIC}--\eqref{r_hat} and is implemented according
	to Algorithm~1 in Online Appendix~\ref{APP_0}.
\end{tablenotes}
		\end{threeparttable}
}
\end{table}

\begin{figure}[!t]
\caption{Empirical Power of the $J$-test: Baseline Design}
\label{fig:j_power_baseline}\centering
\includegraphics[width=0.85%
\textwidth]{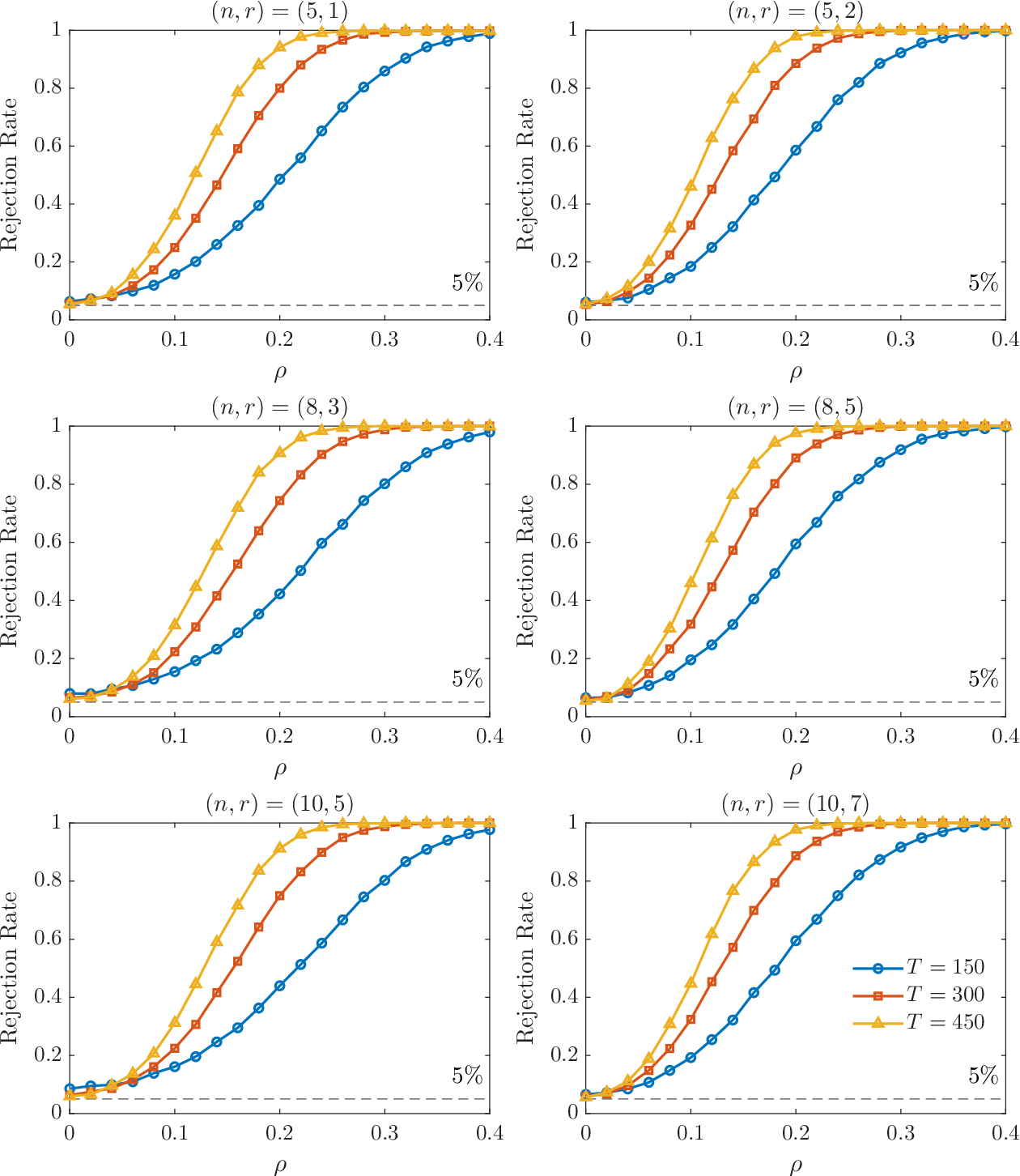} 
\begin{minipage}{0.95\textwidth}
	\footnotesize
	\noindent \textit{Notes.} This figure plots the empirical rejection probabilities of the $J$-test at the 5\% nominal significance level in the baseline design ($d_x=0$), based on $10{,}000$ Monte Carlo replications. Each panel corresponds to one of the six configurations $(n,r)$ listed in~\eqref{MC_nr_pair}, and within each panel three curves correspond to the sample sizes $T\in \{150,300,450\}$. The horizontal axis is $\rho$ and the vertical axis is the rejection rate. The horizontal dashed line marks the nominal $5\%$ level. Power is reported for the feasible GIV estimator, which selects $\bar{r}$ using the BIC criterion in~\eqref{BIC}--\eqref{r_hat} and is implemented according to Algorithm~1 in Online Appendix~\ref{APP_0}.
\end{minipage}
\end{figure}

Finally, we examine the performance of the $J$-test for assessing the
validity of the moment conditions constructed using the GIVs. Table~\ref%
{tab:j_size} reports the empirical rejection probabilities at the $5\%$
nominal significance level under correct specification. The size behavior is
broadly similar to that of the $t$-tests reported in Table~\ref{tab:t_size}:
the empirical rejection probabilities approach the nominal level as the
sample size $T$ increases, while modest over-rejection is observed in small
samples, particularly when the number of moment conditions is relatively
large. For example, in the baseline design with $(n,r)=(10,5)$, the
empirical rejection probabilities of the oracle and feasible $J$-tests are $%
0.069$ and $0.080$, respectively, at $T=150$. In the corresponding extended
design, the rejection probabilities are $0.074$ and $0.083$, respectively.
Similar patterns are observed for $(n,r)=(8,3)$, where the feasible
rejection probabilities are $0.082$ in both the baseline and extended
designs at $T=150$. As the sample size increases, the empirical rejection
probabilities move steadily toward the nominal $5\%$ level across all
configurations.

As discussed earlier for the $t$-tests, the observed small-sample
over-rejection does not appear to be primarily driven by estimation error in
the number of factors or in the null space of $(\mathbf{1}_{n},\lambda)$,
since the oracle and feasible procedures display very similar finite-sample
behavior. Instead, the distortion is likely related to the many-moment
nature of the GMM problem.

We next examine the power of the feasible $J$-test under misspecification.
Figure~\ref{fig:j_power_baseline} plots the empirical rejection
probabilities of the $J$-test as a function of $\rho$, which controls the
severity of the violation of the GIV moment conditions in the baseline
design.\footnote{%
The corresponding results for the extended design are reported in Online
Appendix~\ref{APP_6} and exhibit similar qualitative patterns.} By
construction, each power curve begins near the nominal $5\%$ level when $%
\rho=0$. The rejection probabilities then increase monotonically with $\rho$%
, and the power curves become substantially steeper as the sample size $T$
increases. For example, in the baseline design with $(n,r)=(10,7)$, the test
achieves approximately $80\%$ power at $\rho \approx0.29$ when $T=150$, $%
\rho \approx0.22$ when $T=300$, and $\rho \approx0.16$ when $T=450$. Similar
patterns are observed across all configurations. The feasible $J$%
-test exhibits good power in detecting moderate violations of the moment
restrictions at empirically relevant sample sizes.

\section{Empirical Application: Aggregate Market Multiplier\label{sec:emp}}

A central question in asset pricing is how strongly the aggregate stock
market responds to shifts in investor demand for equities. This response is
summarized by the aggregate market multiplier, denoted by $%
\kappa\equiv-\phi^{-1}$, where $\phi$ is the aggregate demand elasticity.
Economically, $\kappa$ measures the change in aggregate equity value induced
by a one-dollar demand shock. The aggregate demand elasticity has become a
central object of interest in asset pricing and macro-finance because
investor demand, portfolio reallocation, and market segmentation can have
important effects on equilibrium asset prices; see, among others, %
\citet{piazzesi2007asset}, \citet{koijen2019demand}, and %
\citet{gabaix2021search}. In frictionless benchmark models, aggregate demand
is highly elastic, implying a multiplier close to zero. By contrast, the
inelastic-markets hypothesis predicts a substantially larger value of $\kappa
$.

The size of the aggregate multiplier remains controversial. Standard
asset-pricing models imply a macro elasticity of roughly 10 to 20,
corresponding to a multiplier of only 0.05 to 0.1.\footnote{%
See Appendices F and I of \citet{gabaix2021search} for computations of the
macro elasticity implied by the model of \citet{lucas1978asset}, the
rare-disaster models of \citet{barro2006rare} and \citet{gabaix2012variable}%
, and the long-run risks model of \citet{bansal2004risks}.} Yet empirical
estimates of stock-level, factor-level, and aggregate demand elasticities
generally point to substantially less elastic demand. For example, %
\citet{lou2012flow} estimates a stock-level multiplier of about 1.2, %
\citet{pavlova2023benchmarking} report multipliers between 0.3 and 0.5, and \citet{gabaix2021search} report substantially larger multipliers at the aggregate level. These
findings are consistent with the economic intuition that aggregate equity
demand should be less elastic than demand for individual stocks, since
stocks are closer substitutes for one another than for alternative asset
classes such as bonds. At the same time, the estimated multipliers are an
order of magnitude larger than those implied by standard asset-pricing
models, posing a challenge for conventional theories of asset demand.
Resolving this discrepancy requires credible identification of the aggregate
demand elasticity, a key parameter for quantifying the effects of capital
flows, institutional demand shocks, and policy interventions on asset prices.

Obtaining such identification is challenging because prices and quantities
are jointly determined in equilibrium. Demand shocks affect prices, while
prices simultaneously enter investors' demand equations, rendering simple
regressions of demand on prices generally inconsistent. Building on the
demand-based asset-pricing framework of \citet{koijen2019demand}, %
\citet{gabaix2021search} address this endogeneity problem by exploiting the
latent-factor structure in investor demand and estimate aggregate
multipliers ranging from 4.73 to 5.85, with a median estimate close to five.
Their approach, however, relies on consistent estimation of the latent
demand factors and therefore on asymptotic arguments in which the
cross-sectional dimension diverges; see, for example, %
\citet{bai2003inferential}. In this section, we revisit the aggregate
multiplier using the GIV framework developed in Section~\ref{sec: G_model},
which permits valid estimation and inference without requiring the number of
sectors to grow with the sample size.

Following \cite{gabaix2021search}, we study the aggregate multiplier through
a demand system for U.S. equity holdings. Investors are grouped into $n$
equity-holding sectors.\ Let $\Delta q_{i,t}$ denote the fractional
quarterly change in investor $i$'s equity holdings at quarter $t$, with its
empirical counterpart defined in\ \eqref{eq:dq} below; see Online Appendix %
\ref{subsec:variable} for construction details. The demand equation is given
by 
\begin{equation}
\Delta q_{i,t}=\phi \Delta p_{t}+\lambda_{i}^{\top}\eta_{t}+u_{i,t} , 
\label{eq:gk-demand}
\end{equation}
where $\Delta p_{t}$ denotes the quarterly equity market return, $%
\eta_{t}\in \mathbb{R}^{r}$ represents latent aggregate demand factors, and $%
u_{i,t}$ is an idiosyncratic demand shock. The parameter $\phi$ captures the
aggregate demand elasticity and is the primary object of interest.\footnote{
	Unlike \citet{gabaix2021search}, we do not explicitly include aggregate macroeconomic variables, such as GDP growth, in (\ref{eq:gk-demand}). Aggregate variables that enter the demand equation with homogeneous loadings across sectors are eliminated by the orthogonality condition defining the GIVs, since the instruments are constructed to be orthogonal to $\mathbf{1}_n$. Aggregate variables with heterogeneous loadings are absorbed into the latent factor component $\lambda\eta_t$. Under the maintained assumption that the idiosyncratic shocks $u_{i,t}$ are orthogonal to these aggregate components, their omission does not affect identification, estimation, or inference for the demand elasticity $\phi$ and the aggregate multiplier $\kappa$.
}

Because the aggregate supply of equity is approximately fixed in the short
run, market clearing implies that size-weighted net demand equals zero: 
\begin{equation}
q_{S,t}\equiv S_{t}^{\top}\Delta q_{t}=0,   \label{eq:gk-market_clearing}
\end{equation}
where $\Delta q_{t}=(\Delta q_{i,t})_{i\le n}$ and $S_{t}$ denotes the
vector of predetermined market shares. Combining \eqref{eq:gk-demand} and %
\eqref{eq:gk-market_clearing} yields the equilibrium price equation 
\begin{equation}
\Delta p_{t} = \kappa \left( S_{t}^{\top}\lambda \eta_{t}+u_{S,t}\right),   \label{eq:gk-equil}
\end{equation}
where $u_{S,t}\equiv S_{t}^{\top}u_{t}$. Thus, the aggregate multiplier $%
\kappa$ measures the equilibrium price response to aggregate demand shocks,
with a less elastic demand (smaller $|\phi|$) corresponding to a larger multiplier.

Our data are drawn from the Financial Accounts of the United States,
Table~L.224, which reports the equity holdings of major investor sectors.%
\footnote{%
We use the June 2026 vintage of the Financial Accounts, in  which corporate
equity holdings by sector are reported in Table~L.224.} Following \cite%
{gabaix2021search}, our benchmark analysis uses the sample period
1993Q1--2018Q4 and includes twelve sectors that hold U.S. equities
continuously throughout this period. Table~\ref{tab:sectors} in Online
Appendix~\ref{app:sectors} lists these sectors together with their average
market shares. In the data, the fractional change in sector $i$'s equity
holdings is measured as 
\begin{equation}
\Delta q_{i,t}\equiv \frac{w_{i,t}}{w_{i,t-1}R_{t}}-1,   \label{eq:dq}
\end{equation}
where $w_{i,t}$ denotes the value of sector $i$'s equity holdings, $R_{t}
$ is the gross capital-appreciation return on the aggregate stock market,
and $\Delta p_{t}$ is measured by the quarterly simple return on the CRSP
value-weighted index excluding dividends.\footnote{%
Online Appendix~\ref{subsec:variable} details  the construction of $\Delta
q_{i,t}$ and $\Delta p_{t}$, and Online  Appendix~\ref{subsec:sector}
describes the sector classification and market-share weights $S_{t}$.
Following  \cite{gabaix2021search}, pooled sector-level demand growth is
winsorized at  the 5th and 95th percentiles.} Beyond the benchmark sample,
we consider an extended sample spanning 1988Q4--2025Q4, which is the longest
period over which all twelve sectors are continuously observed.

\begin{table}[!t]
\caption{Benchmark Estimates of the Aggregate Equity Market Multiplier}
\label{tab:multiplier}
\centering
{\small \setlength{\tabcolsep}{6pt} \renewcommand{\arraystretch}{1.05}  
\begin{threeparttable}
			\begin{tabular*}{\textwidth}{@{\extracolsep{\fill}}l cc cc@{}}
				\toprule
				& \multicolumn{2}{c}{1993Q1 -- 2018Q4} & \multicolumn{2}{c}{1988Q4 -- 2025Q4} \\
				\cmidrule(lr){2-3}\cmidrule(lr){4-5}
				& $n=12$ & $n=6$ & $n=12$ & $n=6$ \\
				\midrule
				OLS & $5.42^{***}$ & $10.05^{***}$ & $5.66^{***}$  & $11.42^{***}$ \\
				& (0.43)       & (0.78)        & (0.52)        & (1.11)        \\
				FIV & $4.42^{***}$ & $-7.13^{**}$ &  $4.39^{***}$ &  $-2.15$ \\
				&  (0.87)      & (2.83)  & (1.08)       & (2.76) \\
				GIV & $5.05^{***}$ & $8.70^{***}$  & $4.46^{***}$  & $9.42^{***}$  \\
				& (0.30)       & (0.82)        & (0.30)        & (0.93)        \\
				\midrule
				$J$-test ($p$-value)   & $< 0.001$ & 0.817 & $<0.001$ & 0.594 \\
				Estimated $\bar{r}$ & 1     & 1     & 1     & 1     \\
				$T$                  & 104   & 104   & 149   & 149   \\
				\bottomrule
			\end{tabular*}
			\begin{tablenotes}[flushleft]
				\footnotesize
				\item \textit{Notes.} The table reports estimates of the aggregate equity market multiplier, $\widehat{\kappa}=-1/\widehat{\phi}$, for the full twelve-sector panel ($n=12$) and the six-sector granular core ($n=6$). OLS is obtained from a regression of equally weighted demand, $n^{-1}\sum_{i\le n}\Delta q_{i,t}$, on the market return $\Delta p_t$. FIV denotes the factor-residual instrumental-variable estimator of \citet{gabaix2021search}, implemented using Algorithm~3 in Online Appendix~\ref{subsec:gk-replication} with two latent factors and four observed factors (GDP growth, size, value, and momentum). GIV denotes the demand-only granular-IV estimator based on the moment conditions in \eqref{Moment_phi}. The reported $J$-test $p$-value corresponds to the over-identification test, whose asymptotic null distribution is $\chi^2_{n-\bar r-1}$, where $\bar r$ is selected using the BIC criterion in \eqref{BIC}--\eqref{r_hat}. Standard errors are reported in parentheses and are computed using the Newey--West estimator. For the OLS and GIV estimators, $\widehat{\kappa}$ is obtained from $\widehat{\phi}$ through the transformation $\widehat{\kappa}=-1/\widehat{\phi}$ and the corresponding standard errors are computed using the delta method. Significance levels correspond to two-sided tests of $H_0:\kappa=0$. Significance levels: $^{***}\,1\%$, $^{**}\,5\%$, $^{*}\,10\%$.
			\end{tablenotes}
		\end{threeparttable}
}
\end{table}

For comparison, we also report estimates from the factor-residual IV (FIV) estimator of \citet{gabaix2021search}, together with the OLS estimator. The FIV estimator constructs instruments from estimated idiosyncratic demand shocks obtained after removing observed and latent demand factors.\footnote{Online Appendix~\ref{subsec:gk-replication} provides implementation details for the FIV estimator.} The OLS estimator is obtained from a regression of equally weighted demand, $n^{-1}\sum_{i\leq n}\Delta q_{i,t}$, on the market return $\Delta p_t$. Table~\ref{tab:multiplier} reports the resulting estimates and standard errors.

Across all specifications, the OLS estimate of the aggregate multiplier
exceeds its FIV and GIV counterparts. This pattern is consistent with the
endogeneity problem discussed earlier. Positive demand shocks increase both
holdings and prices, causing an uninstrumented regression to attribute part
of the demand shock to the price response. As a result, the demand
elasticity is biased toward zero and the implied multiplier is biased
upward. By exploiting granular demand variation that is orthogonal to common
demand factors, the GIV estimator corrects this source of bias.

When all twelve sectors are used to construct the GIVs, the estimated aggregate multiplier is $5.05$ in the benchmark sample 1993Q1--2018Q4 and $4.46$ in the extended sample 1988Q4--2025Q4, close to the corresponding FIV estimates. However, the over-identification test strongly rejects the associated moment restrictions in both samples, with $p$-values effectively equal to zero. One traditional interpretation of this result is that the instruments are invalid. Alternatively, following a perspective common in the treatment-effects literature, rejection of the over-identification test may reflect heterogeneity in the underlying causal parameters. From this perspective, the evidence suggests that the homogeneous-elasticity specification imposed on all twelve sectors is too restrictive and that demand elasticities may differ substantially across investor sectors.

To investigate this possibility, we restrict attention to the six largest
sectors, which together account for more than 97\% of total equity holdings
in the sample: households, mutual funds and ETFs, the foreign sector,
private pension funds, state and local pension funds, and life insurance
companies. Relative to the twelve-sector specification, both the OLS and GIV
estimates increase substantially. For example, the GIV estimate rises from $%
5.05$ to $8.70$ in the benchmark sample and from $4.46$ to $9.42$ in the
extended sample. More importantly, the over-identification test no longer
rejects, yielding $p$-values of $0.817$ and $0.594$ in the two samples.
These findings suggest that the six largest sectors exhibit more homogeneous
demand behavior and therefore provide a more credible basis for estimating a
common aggregate demand elasticity and the corresponding market multiplier.

\begin{table}[!t]
\caption{Heterogeneous Demand Multipliers by Sector}
\label{tab:hetero}
\centering
{\small \setlength{\tabcolsep}{6pt} \renewcommand{\arraystretch}{1.05}  
\begin{threeparttable}
			\begin{tabular*}{\textwidth}{@{\extracolsep{\fill}}l cc cc@{}}
				\toprule
				& \multicolumn{2}{c}{1993Q1--2018Q4} & \multicolumn{2}{c}{1988Q4--2025Q4} \\
				\cmidrule(lr){2-3}\cmidrule(lr){4-5}
				Sector & $\hat{\kappa}$ & $J$-test ($p$-value) & $\hat{\kappa}$ & $J$-test ($p$-value) \\
				\midrule
				Granular core ($n=6$)
				& $8.70^{***}$ & 0.817 & $9.42^{***}$ & 0.594 \\
				& (0.82)       &       & (0.93)       &       \\
				\addlinespace
				Property and casualty insurers
				& $3.57^{***}$ & 0.094 & $3.96^{***}$ & 0.927 \\
				& (0.31)       &       & (0.38)       &       \\
				\addlinespace
				Federal government retirement funds
				& $-18.15$ & 0.010 & $-17.32$     & $< 0.001$ \\
				& (9.06)       &       & (12.99)     &       \\
				\addlinespace
				State and local governments
				& $3.70^{***}$ & 0.419 & $3.65^{***}$ & 0.202 \\
				& (0.45)       &       & (0.33)       &       \\
				\addlinespace
				Closed-end funds
				& $6.63^{***}$ & 0.410 & $6.52^{***}$ & 0.356 \\
				& (2.08)       &       & (1.71)       &       \\
				\addlinespace
				Banks
				& $6.68^{***}$ & 0.347 & $10.68^{**}$ & 0.786 \\
				& (1.70)       &       & (4.29)       &       \\
				\addlinespace
				Broker-dealers
				& $-40.97$     & 0.064 & $-251.05$    & 0.011 \\
				& (90.85)      &       & (3058.61)    &       \\
				\bottomrule
			\end{tabular*}
			\begin{tablenotes}[flushleft]
				\footnotesize
				\item \textit{Notes.} The table reports sector-specific demand multipliers, $\hat{\kappa}_j=-1/\hat{\phi}_j$, estimated using the GIV moment conditions in \eqref{GIV_Form_n0_3}. The GIVs are constructed from the six largest equity-holding sectors (the granular core), which are assumed to share a common demand elasticity. For each sample period, the first row reproduces the corresponding six-sector estimate from Table~\ref{tab:multiplier}. Standard errors are reported in parentheses and are computed using the Newey--West estimator. The reported $p$-values correspond to the over-identification $J$-test. Rejection of the $J$-test indicates that the moment conditions constructed from the granular core are invalid for the corresponding sector. Standard errors for $\hat{\kappa}_j$ are computed using the delta method, and significance levels correspond to two-sided tests of $H_0:\kappa_j=0$. For sectors whose $J$-test rejects at the 5\% level, point estimates and standard errors are reported for completeness only, and statistical significance is not indicated. Significance levels: $^{***}\,1\%$, $^{**}\,5\%$, $^{*}\,10\%$.
			\end{tablenotes}
		\end{threeparttable}
}
\end{table}

The comparison between the FIV and GIV estimators highlights the importance
of the fixed-$n$ approach. When all twelve sectors are included, the two
estimators deliver qualitatively similar conclusions. The FIV estimates of
the aggregate multiplier are $4.42$ in the benchmark sample and $4.39$ in
the extended sample, close to the corresponding GIV estimates of $5.05$ and $%
4.46$, respectively, and broadly consistent with the estimates reported by %
\citet{gabaix2021search}. Both instrumental-variable estimators also yield
smaller multipliers than OLS, as expected from the endogeneity bias
discussed above.

The contrast becomes much sharper when attention is restricted to the
six-sector granular core. In this case, the FIV estimator produces negative
multiplier estimates of $-7.13$ and $-2.15$, whereas the GIV estimator
yields stable and economically meaningful estimates of $8.70$ and $9.42$.
This divergence reflects the different identification strategies underlying
the two procedures. The FIV estimator constructs instruments from estimated
idiosyncratic demand shocks obtained after removing latent factors through
principal-components analysis and therefore relies on consistent estimation
of those factors. With only six sectors, the cross-sectional dimension is
too small for this approach to be reliable. By contrast, the fixed-$n$ GIV
estimator does not require consistent estimation of latent factors and
remains valid when the number of sectors is small. The resulting GIV
estimates are therefore more credible for the six-sector specification,
which is favored by the over-identification test and appears consistent with
the homogeneous-elasticity restriction underlying the aggregate multiplier.

Table \ref{tab:hetero} sheds light on the source of the rejection of the
twelve-sector specification. Using the GIVs constructed from the six-sector
core, we estimate sector-specific demand elasticities, $\phi_{j}$, and the
corresponding implied market multipliers, $\kappa_{j}\equiv-\phi_{j}^{-1}$,
for the remaining sectors. Several sectors exhibit substantially smaller
multipliers than the six-sector core. Property and casualty insurers and
state and local governments have multipliers between $3.5$ and $4.0$, less
than half the core estimate. Closed-end funds have intermediate multipliers
of roughly $6.5$, while banks exhibit larger and less precisely estimated
multipliers. For all of these sectors, the $J$-test provides little evidence
against the validity of the moment conditions constructed from the
six-sector core.

The two remaining sectors, federal government retirement funds and
broker-dealers, exhibit markedly different behavior. Their estimated multipliers are negative in both samples, implying non-positive estimates of the corresponding demand elasticities. Moreover, these estimates are highly
imprecise and economically difficult to interpret. More importantly, the $J$%
-test rejects the validity of the moment conditions for federal government
retirement funds in both samples. For broker-dealers, the $J$-test is close
to rejection at the 5\% level in the shorter sample ($p$-value $=0.064$) and
is strongly rejected in the longer sample ($p$-value $=0.011$). Taken
together, these findings suggest that the moment conditions constructed from
the six-sector core do not provide valid identifying restrictions for these
two sectors. Consequently, the corresponding multiplier estimates should not
be given an economic interpretation.

Taken together, Tables ~\ref{tab:multiplier} and \ref{tab:hetero} provide
strong evidence that demand elasticities differ substantially across
investor sectors. Among the ten sectors for which the $J$-test does not
reject (treating broker-dealers as invalid given the decisive rejection in
the longer sample and the borderline $p$-value of $0.064$ in the shorter
sample), the estimated multipliers range from approximately $3.6$ to $10.7$,
with the upper end reflecting the imprecisely estimated bank sector. This
heterogeneity explains both the rejection of the twelve-sector specification
and the substantially smaller multiplier obtained when all sectors are
pooled together. By contrast, the six-sector core appears considerably more
homogeneous and yields a stable aggregate multiplier of roughly nine across
both sample periods. This estimate is an order of magnitude larger than the
frictionless benchmark multiplier of $0.05$ to $0.1$, providing evidence consistent with the inelastic-markets hypothesis.

\section{Conclusion\label{sec:conclusion}}

This paper develops an estimation and inference framework for structural
models identified by granular instrumental variables (GIVs). The key insight
is that, under suitable cross-sectional restrictions on the idiosyncratic
shocks, valid GIVs are characterized by the orthogonal complement of the
factor-loading space associated with latent aggregate shocks. This
characterization provides a transparent foundation for GIV-based
identification and allows demand and supply elasticities to be estimated
without conventional excluded instruments or direct observation of the
latent factors.

A central contribution of the paper is to show that the relevant orthogonal
complement can be identified and consistently estimated directly from the
covariance structure of the observables, without first estimating the latent
factors themselves. As a result, the proposed framework remains applicable
even when the number of entities is fixed and does not require the
cross-sectional dimension to diverge with the sample size. Building on this
result, we develop feasible procedures for estimation, inference, and
specification testing based on estimated GIVs. Monte Carlo evidence shows
that the feasible estimator performs similarly to an oracle procedure that
knows the true factor-loading space, while the empirical application yields
evidence consistent with highly inelastic aggregate equity demand.

The analysis also clarifies several identification issues that arise when
factor loadings are unknown. In particular, we show that certain
restrictions commonly imposed in existing implementations of the GIV
methodology should be interpreted as substantive assumptions on the latent
factors rather than innocuous normalizations. We further show that
identification can fail for existing moment-condition approaches when the
factor-loading space is unknown. The framework developed here avoids these
difficulties by focusing directly on the geometry of the factor-loading
space and exploiting the eigenspace structure of the covariance matrix of
the observables.

More broadly, the results demonstrate that cross-sectional heterogeneity can
be used not only to construct granular instruments but also to conduct valid
estimation, inference, and specification testing in models with latent
aggregate shocks. We hope that the framework developed in this paper will
facilitate the use of GIV methods in empirical work and provide a foundation
for future research in settings where latent aggregate forces and granular
heterogeneity interact.

\bigskip

{\small
\input{Granular_IV_v7_3.bbl}

}

\clearpage
\counterwithin{assumption}{section}
\counterwithin{theorem}{section}
\counterwithin{lemma}{section}
\counterwithin{proposition}{section}
\counterwithin{corollary}{section}
\counterwithin{equation}{section}
\counterwithin{table}{section}
\counterwithin{figure}{section}
\renewcommand{\thetable}{\Alph{section}\arabic{table}}
\appendix

\begin{center}
{\huge Online Appendix}

\bigskip
\end{center}

\setcounter{page}{1}
\renewcommand{\thepage}{A-\arabic{page}}
\setcounter{equation}{0}
\setcounter{theorem}{0}
\setcounter{lemma}{0}
\setcounter{proposition}{0}
\setcounter{corollary}{0}

This Online Appendix contains supplementary theoretical, computational, and empirical material for ``Granular Instrumental Variables: Estimation and Inference.'' Section \ref{APP_0} provides implementation details for the proposed GIV procedures. Sections \ref{APP_2} -- \ref{APP_5} contain proofs of the theoretical results and additional technical derivations. Section \ref{APP_6} reports supplementary Monte Carlo evidence, and \ref{app:sectors} presents additional empirical results. References cited only in the Online Appendix are collected at the end of this document.

\section{Implementation Details\label{APP_0}}

\noindent \textsc{Algorithm 1 (GIV Estimation with Balanced Data)}

\noindent \textbf{Step 0.} If exogenous regressors $x_t$ are present, compute the LS estimator
\[
\hat{\beta}\equiv \Bigl(\sum_{t\leq T}\tilde{x}_{t}^{\top}\tilde{x}%
_{t}\Bigr)^{-1}\Bigl(\sum_{t\leq T}\tilde{x}_{t}^{\top}\tilde{y}_{t}\Bigr),
\]
and define $\hat{y}_{t}\equiv y_{t}-x_{t}\hat{\beta}$. Otherwise, set $\hat
{y}_{t}=y_{t}$. If $\bar r$ is unknown, estimate it using the BIC criterion in (\ref{BIC})--(\ref{r_hat}) and replace $\bar r$ by the resulting estimator $\hat r$ throughout the algorithm. 

\noindent \textbf{Step 1.} Compute $y_{S,t}\equiv S_{t}^{\top}y_{t}\ $and
$y_{e,t}\equiv e^{\top}\hat{y}_{t}$, where $\hat{y}_{t}\equiv(\hat{y}%
_{i,t})_{i\leq n}$, $S_{t}=(s_{i,t})_{i\leq n}$, and $e\equiv n^{-1}%
\mathbf{1}_{n}$. Compute
\[
Q_{-1}^{\top}\Bigl(T^{-1}\sum_{t\leq T}\hat{y}_{t}\hat{y}_{t}^{\top
}\Bigr)Q_{-1},
\]
and obtain ordered eigenvalues $\{ \hat{\mu}_{j}\}_{j\leq n-1}$ such that
$\hat{\mu}_{j}\leq \hat{\mu}_{j+1}$. Define
\[
\hat{\Upsilon}\equiv Q_{-1}\hat{A}_{0,\bot}\bigl(\hat{\sigma}_{u}%
^{2}\mathbf{I}_{\bar{r}-1}-\hat{\Lambda}_{\bot}\bigr)^{-1}\hat{A}_{0,\bot
}^{\top}Q_{-1}^{\top},
\]
where $\hat{\sigma}_{u}^{2}\equiv \hat{\mu}_{1}$, $\hat{\Lambda}_{\bot}%
\equiv \mathrm{diag}((\hat{\mu}_{j})_{n-\bar{r}+1\leq j\leq n-1})$, and
$\hat{A}_{0}$ and $\hat{A}_{0,\bot}$ collect the eigenvectors associated with
the smallest $n-\bar{r}$ and largest $\bar{r}-1$ eigenvalues, respectively.

\noindent \textbf{Step 2.} Set $\hat{A}=Q_{-1}\hat{A}_{0}$ and compute the
preliminary GIV estimator $\hat{\theta}_{0}(\hat{A})$ using (\ref{GIV_Form_1})
together with the appropriate moment specification:%
\[%
\begin{array}
[c]{cc}%
\text{Without }w_{t}\text{:} & \text{use (\ref{GIV_Form_2}) and }W_{0,T}%
(\hat{A})=\mathbf{I}_{2(n-\bar{r})}\\
\text{With }w_{t}\text{:} & \text{use (\ref{GIV_Form_4}) and }W_{0,T}(\hat
{A})=\mathbf{I}_{2(n-\bar{r})+d_{w}}%
\end{array}
.
\]

\noindent \textbf{Step 3.} Construct residuals
\[
\hat{v}_{t}\equiv y_{e,t}-\hat{\phi}_{0}(\hat{A})p_{t},
\]
and
\[
\hat{\varepsilon}_{t}\equiv%
\begin{cases}
p_{t}-\hat{\psi}_{0}(\hat{A})y_{S,t}, & \text{without }w_{t},\\
p_{t}-\hat{\psi}_{0}(\hat{A})y_{S,t}-w_{t}^{\top}\hat{\gamma}_{0}(\hat{A}), &
\text{with }w_{t}.
\end{cases}
\]
For $b\in \{v,\varepsilon \}$, define%

\begin{equation}
\hat{\xi}_{b,t}\equiv \hat{y}_{t}\hat{b}_{t}-T^{-1}\sum_{t\leq T}\hat{y}%
_{t}\hat{b}_{t}+\Bigl(\hat{y}_{t}\hat{y}_{t}^{\top}-T^{-1}\sum_{t\leq T}%
\hat{y}_{t}\hat{y}_{t}^{\top}\Bigr)\hat{\Upsilon}\Bigl(T^{-1}\sum_{t\leq
T}\hat{y}_{t}\hat{b}_{t}\Bigr). \label{Zeta_hat}%
\end{equation}

\noindent \textbf{Step 4.} Define
\[
\hat{\xi}_{t}\equiv%
\begin{cases}
(\hat{\xi}_{v,t}^{\top},\hat{\xi}_{\varepsilon,t}^{\top})^{\top}, &
\text{without }w_{t},\\
(\hat{\xi}_{v,t}^{\top},\hat{\xi}_{\varepsilon,t}^{\top},w_{t}^{\top}%
\hat{\varepsilon}_{t})^{\top}, & \text{with }w_{t}.
\end{cases}
\]
Compute
\begin{equation}
\hat{V}\equiv T^{-1}\sum_{t\leq T}\hat{\xi}_{t}\hat{\xi}_{t}^{\top
}-\Bigl(T^{-1}\sum_{t\leq T}\hat{\xi}_{t}\Bigr)\Bigl(T^{-1}\sum_{t\leq T}%
\hat{\xi}_{t}^{\top}\Bigr). \label{V_est}%
\end{equation}
Use a HAC version if $\xi_{t}$ is serially correlated.

\noindent \textbf{Step 5.} Compute the optimal GIV estimator $\hat{\theta
}^{\ast}(\hat{A})$ using (\ref{GIV_Form_1}) with weight matrix
\[
W_{\ast,T}(\hat{A})\equiv%
\begin{cases}
\bigl((\mathbf{I}_{2}\otimes \hat{A}^{\top})\hat{V}(\mathbf{I}_{2}\otimes
\hat{A})\bigr)^{-1}, & \text{without }w_{t},\\
\bigl(\mathrm{diag}(\hat{A}^{\top},\hat{A}^{\top},\mathbf{I}_{d_{w}})\hat
{V}\mathrm{diag}(\hat{A},\hat{A},\mathbf{I}_{d_{w}})\bigr)^{-1}, & \text{with
}w_{t}.
\end{cases}
\]

\noindent \textbf{Step 6.} Report standard errors
\[
\left(  \mathrm{diag}\! \left(  \bigl(TD_{1,T}(\hat{A})^{\top}W_{\ast,T}%
(\hat{A})D_{1,T}(\hat{A})\bigr)^{-1}\right)  \right)  ^{1/2},
\]
and the $p$-value based on the $J$-statistic and the CDF of $\chi^{2}\bigl(2(n-\bar{r}-1)\bigr).$

\bigskip

\noindent \textsc{Algorithm 2 (GIV Estimation with Unbalanced Data)}

\noindent \textbf{Step 1.} If $\bar r$ is unknown, estimate it using the BIC
criterion in (\ref{BIC})--(\ref{r_hat}); otherwise treat $\bar r$ as known. Compute $y_{S,t}\equiv S_{t}^{\top}y_{t}\ $and
$y_{e,t}\equiv e_{t}^{\top}y_{t}$,\ where $y_{t}\equiv(y_{i,t})_{i\leq n_{t}}%
$,\ $S_{t}=(s_{i,t})_{i\leq n_{t}}$ and $e_{t}\equiv n_{t}^{-1}\mathbf{1}%
_{n_{t}}$.

\noindent \textbf{Step 2.} Let $y_{t}^{0}$ denote the $n_{0}\times1$ subvector
of $y_{t}$ observed for all $t$. Let $Q_{0,-1}\equiv(q_{0,2},\ldots
,q_{0,n_{0}})$ be defined analogously to (\ref{ortho_q_j}) with $n$ replaced
by $n_{0}$. Compute
\[
Q_{0,-1}^{\top}\Bigl(T^{-1}\sum_{t\leq T}y_{t}^{0}y_{t}^{0\top}\Bigr)Q_{0,-1}%
,
\]
and obtain its ordered eigenvalues $\{ \hat{\mu}_{0,j}\}_{j\leq n_{0}-1}$.
Obtain
\[
\hat{\Upsilon}_{0}\equiv Q_{0,-1}\hat{A}_{0,\bot}\bigl(\hat{\sigma}_{u}%
^{2}\mathbf{I}_{\bar{r}_{0}-1}-\hat{\Lambda}_{0,\bot}\bigr)^{-1}\hat
{A}_{0,\bot}^{\top}Q_{0,-1}^{\top},
\]
where $\hat{\sigma}_{u}^{2}\equiv \hat{\mu}_{0,1}$, $\hat{\Lambda}_{0,\bot
}\equiv \mathrm{diag}\bigl((\hat{\mu}_{0,j})_{n_{0}-\bar{r}_{0}+1\leq j\leq
n_{0}-1}\bigr)$,\ and $\hat{A}_{0}$ and $\hat{A}_{0,\bot}$ collect the
eigenvectors associated with the smallest $n_{0}-\bar{r}_{0}$ eigenvalues and
the largest $\bar{r}_{0}-1$ eigenvalues, respectively.

\noindent \textbf{Step 3.} Set $\hat{A}=Q_{0,-1}\hat{A}_{0}$ and $W_{0,T}%
(\hat{A})=\mathbf{I}_{2(n_{0}-\bar{r}_{0})}$. Compute the preliminary GIV
estimator $\hat{\theta}_{0}(\hat{A})\equiv \bigl(\hat{\phi}_{0}(\hat{A}%
),\hat{\psi}_{0}(\hat{A})\bigr)^{\top}$ using (\ref{GIV_Form_n0_1}) and
(\ref{GIV_Form_n0_2}). Construct the residuals
\[
\hat{v}_{t}\equiv y_{e,t}-\hat{\phi}_{0}(\hat{A})p_{t},\qquad \hat{\varepsilon
}_{t}\equiv p_{t}-\hat{\psi}_{0}(\hat{A})y_{S,t},
\]
and define, for $b\in \{v,\varepsilon \}$,
\[
\hat{\xi}_{b,t}^{0}\equiv y_{t}^{0}\hat{b}_{t}-T^{-1}\sum_{t\leq T}y_{t}%
^{0}\hat{b}_{t}+\Bigl(y_{t}^{0}y_{t}^{0\top}-T^{-1}\sum_{t\leq T}y_{t}%
^{0}y_{t}^{0\top}\Bigr)\hat{\Upsilon}_{0}\Bigl(T^{-1}\sum_{t\leq T}y_{t}%
^{0}\hat{b}_{t}\Bigr).
\]

\noindent \textbf{Step 4.} Let $\hat{\xi}_{t}^{0}\equiv \bigl(\hat{\xi}%
_{v,t}^{0\top},\hat{\xi}_{\varepsilon,t}^{0\top}\bigr)^{\top}$.\ Define
\[
\hat{V}^{0}\equiv T^{-1}\sum_{t\leq T}\hat{\xi}_{t}^{0}\hat{\xi}_{t}^{0\top
}-\Bigl(T^{-1}\sum_{t\leq T}\hat{\xi}_{t}^{0}\Bigr)\Bigl(T^{-1}\sum_{t\leq
T}\hat{\xi}_{t}^{0}\Bigr)^{\top}.
\]
Use a HAC version of $\hat{V}^{0}$ if $\xi_{t}^{0}$ is serially correlated.

\noindent \textbf{Step 5.} Compute the optimal GIV estimator $\hat{\theta
}^{\ast}(\hat{A})$ using (\ref{GIV_Form_n0_1}) and (\ref{GIV_Form_n0_2}) with
$W_{0,T}(\hat{A})$ replaced by
\[
W_{\ast,T}(\hat{A})\equiv \bigl((\mathbf{I}_{2}\otimes \hat{A}^{\top})\hat
{V}^{0}(\mathbf{I}_{2}\otimes \hat{A})\bigr)^{-1}.
\]

\noindent \textbf{Step 6.} Report the standard errors
\[
\sqrt{\mathrm{Diag}\! \left(  \bigl(TD_{1,T}^{0}(\hat{A})^{\top}W_{\ast
,T}(\hat{A})D_{1,T}^{0}(\hat{A})\bigr)^{-1}\right)  },
\]
and the $p$-value based on $T\bar{g}_{T}^{0}(\hat{\theta}^{\ast}(\hat{A}%
);\hat{A})^{\top}W_{\ast,T}(\hat{A})\bar{g}_{T}^{0}(\hat{\theta}^{\ast}%
(\hat{A});\hat{A})$ and the CDF of $\chi^{2}\bigl(2(n_{0}-\bar{r}%
_{0}-1)\bigr).$

\section{Proof of the Main Results \label{APP_2}}

\noindent \textsc{Proof of Lemma \ref{L1_eq_moment}.} Multiplying both sides of
(\ref{eq_moment_1}) by $Q$ yields
\begin{equation}
\mathbb{E}[Q^{\top}(y_{t}-\phi p_{t}\mathbf{1}_{n})(y_{t}-\phi p_{t}%
\mathbf{1}_{n})^{\top}Q]=(\mathbb{E}[\eta_{t}^{2}]+2\sigma_{\eta u})Q^{\top
}\mathbf{1}_{n}\mathbf{1}_{n}^{\top}Q+\sigma_{u}^{2}\mathbf{I}_{n}.
\label{P_L1_eq_moment_1}%
\end{equation}
Since $Q$ is orthonormal, this can be written as
\begin{equation}
\mathbb{E}\! \left[
\begin{pmatrix}
n(y_{e,t}-\phi p_{t})^{2} & n^{1/2}(y_{e,t}-\phi p_{t})y_{t}^{\top}Q_{-1}\\
n^{1/2}Q_{-1}^{\top}y_{t}(y_{e,t}-\phi p_{t}) & Q_{-1}^{\top}y_{t}y_{t}^{\top
}Q_{-1}%
\end{pmatrix}
\right]  =%
\begin{pmatrix}
n(\mathbb{E}[\eta_{t}^{2}]+2\sigma_{\eta u})+\sigma_{u}^{2} & \mathbf{0}%
_{1\times(n-1)}\\
\mathbf{0}_{(n-1)\times1} & \sigma_{u}^{2}\mathbf{I}_{n-1}%
\end{pmatrix}
. \label{P_L1_eq_moment_2}%
\end{equation}
Matching the upper-right block, the lower-right block, and the $(1,1)$ entry
of the matrices on both sides of (\ref{P_L1_eq_moment_2}) yields the moment
conditions in (\ref{L1_eq_moment_d1})--(\ref{L1_eq_moment_d3}). The moment
conditions in (\ref{L1_eq_moment_s1})--(\ref{L1_eq_moment_s2}) can be obtained
similarly by multiplying $Q^{\top}$ on the left-hand side of both terms in
(\ref{eq_moment_2}).\hfill$Q.E.D.$

\bigskip

\noindent \textsc{Proof of Lemma \ref{L2_eq_moment}.} Multiplying
$(\bar{\lambda},\bar{\lambda}_{\bot})^{\top}$ from the left and $(\bar
{\lambda},\bar{\lambda}_{\bot})$ from the right on both sides of
(\ref{G_moment_1}) yields
\begin{align}
\mathbb{E}\! \left[  \bar{\lambda}^{\top}(y_{t}-\phi p_{t}\mathbf{1}_{n}%
)y_{t}^{\top}\bar{\lambda}_{\bot}\right]   &  =\mathbf{0}_{\bar{r}%
\times(n-\bar{r})},\label{P_L2_eq_moment_1}\\
\mathbb{E}\! \left[  \bar{\lambda}_{\bot}^{\top}y_{t}y_{t}^{\top}\bar{\lambda
}_{\bot}\right]  -\sigma_{u,t}^{2}\mathbf{I}_{n-\bar{r}}  &  =\mathbf{0}%
_{(n-\bar{r})\times(n-\bar{r})},\label{P_L2_eq_moment_2}\\
\mathbb{E}\! \left[  \bar{\lambda}^{\top}(y_{t}-\phi p_{t}\mathbf{1}%
_{n})(y_{t}-\phi p_{t}\mathbf{1}_{n})^{\top}\bar{\lambda}\right]   &
=\bar{\lambda}^{\top}\Big(\lambda \mathbb{E}[\eta_{t}\eta_{t}^{\top}%
]\lambda^{\top}+\lambda \Gamma_{\eta u,t}\mathbf{1}_{n}^{\top}+\mathbf{1}%
_{n}\Gamma_{\eta u,t}^{\top}\lambda^{\top}+\sigma_{u,t}^{2}\mathbf{I}%
_{n}\Big)\bar{\lambda}. \label{P_L2_eq_moment_3}%
\end{align}
By the definition of $\bar{\lambda}$,
\[
\bar{\lambda}^{\top}(y_{t}-\phi p_{t}\mathbf{1}_{n})y_{t}^{\top}\bar{\lambda
}_{\bot}=\left(
\begin{array}
[c]{c}%
n^{-1/2}\mathbf{1}_{n}^{\top}(y_{t}-\phi p_{t}\mathbf{1}_{n})y_{t}^{\top}%
\bar{\lambda}_{\bot}\\
\bar{\lambda}_{-1}^{\top}(y_{t}-\phi p_{t}\mathbf{1}_{n})y_{t}^{\top}%
\bar{\lambda}_{\bot}%
\end{array}
\right)  =\left(
\begin{array}
[c]{c}%
n^{1/2}(y_{e,t}-\phi p_{t})y_{t}^{\top}\bar{\lambda}_{\bot}\\
\bar{\lambda}_{-1}^{\top}y_{t}y_{t}^{\top}\bar{\lambda}_{\bot}%
\end{array}
\right)  .
\]
Therefore, (\ref{L2_eq_moment_1}) and (\ref{L2_eq_moment_2}) follow directly
from (\ref{P_L2_eq_moment_1}). Likewise, (\ref{L2_eq_moment_3}) and
(\ref{L2_eq_moment_4}) follow from (\ref{P_L2_eq_moment_2}) and
(\ref{P_L2_eq_moment_3}), respectively.

Next, multiplying $(\bar{\lambda},\bar{\lambda}_{\bot})^{\top}$ from the left
on both sides of (\ref{G_moment_2}) yields
\begin{align}
\mathbb{E}\! \left[  \bar{\lambda}^{\top}(p_{t}-\psi y_{S,t})(y_{t}-\phi
p_{t}\mathbf{1}_{n})\right]   &  =\bar{\lambda}^{\top}\lambda \mathbb{E}%
[\varepsilon_{t}\eta_{t}]+\sigma_{\varepsilon u,t}\bar{\lambda}^{\top
}\mathbf{1}_{n},\label{P_L2_eq_moment_4}\\
\mathbb{E}\! \left[  \bar{\lambda}_{\bot}^{\top}(p_{t}-\psi y_{S,t}%
)(y_{t}-\phi p_{t}\mathbf{1}_{n})\right]   &  =\mathbf{0}_{n-\bar{r}}.
\label{P_L2_eq_moment_5}%
\end{align}
Since
\[
\bar{\lambda}_{\bot}^{\top}(p_{t}-\psi y_{S,t})(y_{t}-\phi p_{t}\mathbf{1}%
_{n})=(p_{t}-\psi y_{S,t})\bar{\lambda}_{\bot}^{\top}y_{t},
\]
(\ref{L2_eq_moment_5}) follows from (\ref{P_L2_eq_moment_5}), while
(\ref{L2_eq_moment_6}) follows from (\ref{P_L2_eq_moment_4}). \hfill$Q.E.D.$

\bigskip

\noindent \textsc{Proof of Lemma\  \ref{ID_G_GIV_Weight}.} Let $\tilde{r}%
\equiv \mathrm{rank}(\tilde{\lambda})$. Since $\mathrm{rank}((\mathbf{1}%
_{n},M_{\mathbf{1}_{n}}\lambda))=\mathrm{rank}((\mathbf{1}_{n},\lambda))$ and
$\lambda^{\top}M_{\mathbf{1}_{n}}\mathbf{1}_{n}=\mathbf{0}_{r}$, we have
\begin{equation}
\tilde{r}=\mathrm{rank}((\mathbf{1}_{n},\lambda))-1=\bar{r}-1\text{. }
\label{P_ID_G_GIV_Weight_1}%
\end{equation}
This together with the assumption that $T^{-1}\sum_{t\leq T}\mathbb{E}%
[\eta_{t}\eta_{t}^{\top}]$ is nonsingular implies that the matrix
$\tilde{\lambda}\left(  T^{-1}\sum_{t\leq T}\mathbb{E}[\eta_{t}\eta_{t}^{\top
}]\right)  \tilde{\lambda}^{\top}$ has rank $\bar{r}-1$. Hence, for
eigenvectors in $\mathcal{B}_{\mathbf{1}_{n}}$, it has $\bar{r}-1$ strictly
positive eigenvalues, denoted by $\tilde{\mu}_{j}>0$ for $j=1,\ldots,\bar
{r}-1$, and $n-\bar{r}$ zero eigenvalues. For any $x\in \mathcal{B}%
_{\mathbf{1}_{n}}$, we have $x^{\top}M_{\mathbf{1}_{n}}x=1$, so restricting
(\ref{G_demand_factor}) to $\mathcal{B}_{\mathbf{1}_{n}}$ yields%
\[
x^{\top}\bar{\Sigma}_{\tilde{y}}x=x^{\top}\tilde{\lambda}\Big(T^{-1}%
\sum_{t\leq T}\mathbb{E}[\eta_{t}\eta_{t}^{\top}]\Big)\tilde{\lambda}^{\top
}x+\bar{\sigma}_{u}^{2}x^{\top}M_{\mathbf{1}_{n}}x=x^{\top}\tilde{\lambda
}\Big(T^{-1}\sum_{t\leq T}\mathbb{E}[\eta_{t}\eta_{t}^{\top}]\Big)\tilde
{\lambda}^{\top}x+\bar{\sigma}_{u}^{2}.
\]
Therefore, the restricted eigenvalues of $\bar{\Sigma}_{\tilde{y}}$ on
$\mathcal{B}_{\mathbf{1}_{n}}$ are given by $\tilde{\mu}_{j}+T^{-1}\sum_{t\leq
T}\sigma_{u,t}^{2}$, for $j=1,\ldots,\bar{r}-1$, and $\bar{\sigma}_{u}^{2}$
with multiplicity $n-\bar{r}$. Since $n>\bar{r}$, the minimum eigenvalue of
$\bar{\Sigma}_{\tilde{y}}$ on $\mathcal{B}_{\mathbf{1}_{n}}$ is $\bar{\sigma
}_{u}^{2}$ with multiplicity $n-\bar{r}$. This establishes
(\ref{G_demand_sigma_u}).

Finally, the minimizers of (\ref{G_demand_sigma_u}) are the unit eigenvectors
associated with the eigenvalue $\bar{\sigma}_{u}^{2}$, which are precisely the
vectors in $\mathcal{B}_{\mathbf{1}_{n}}$ orthogonal to the column space of
$\tilde{\lambda}$. Equivalently, they are orthogonal to both $\mathbf{1}_{n}$
and $\operatorname{col}(\tilde{\lambda})$. Since $\tilde{\lambda
}=M_{\mathbf{1}_{n}}\lambda$, we have
\[
\operatorname{col}((\mathbf{1}_{n},\tilde{\lambda}))=\operatorname{col}%
((\mathbf{1}_{n},\lambda)).
\]
Hence these eigenvectors are exactly the orthogonal complement of
$\operatorname{col}((\mathbf{1}_{n},\lambda))$, and therefore they span
$\bar{\lambda}_{\bot}$. \hfill$Q.E.D.$

\bigskip

\noindent \textsc{Proof of Lemma \ref{G_GIV_Solutions}.} Recall that
$Q=(n^{-1/2}\mathbf{1}_{n},Q_{-1})$ is an orthonormal matrix, so the columns
of $Q_{-1}$ form an orthonormal basis for the subspace orthogonal to
$\mathbf{1}_{n}$. Therefore, the mapping $\tilde{a}\mapsto a=Q_{-1}\tilde{a}$
is a bijection from $B_{n-1}\ $onto $\mathcal{B}_{\mathbf{1}_{n}}$. In
particular, $a\in \mathcal{B}_{\mathbf{1}_{n}}$ if and only if $a=Q_{-1}%
\tilde{a}$ for some $\tilde{a}\in B_{n-1}$.

Hence, the minimization problem in (\ref{G_demand_sigma_u}) is equivalent to
\[
\min_{a\in \mathcal{B}_{\mathbf{1}_{n}}}a^{\top}\bar{\Sigma}_{\tilde{y}}%
a=\min_{\tilde{a}\in B_{n-1}}\tilde{a}^{\top}Q_{-1}^{\top}\bar{\Sigma}%
_{\tilde{y}}Q_{-1}\tilde{a}.
\]
Next, using $\tilde{y}_{t}=M_{\mathbf{1}_{n}}y_{t}$ and the fact that
$M_{\mathbf{1}_{n}}Q_{-1}=Q_{-1}$ (equivalently, $Q_{-1}^{\top}M_{\mathbf{1}%
_{n}}=Q_{-1}^{\top}$), we obtain
\[
Q_{-1}^{\top}\mathbb{E}[\tilde{y}_{t}\tilde{y}_{t}^{\top}]Q_{-1}=Q_{-1}^{\top
}M_{\mathbf{1}_{n}}\mathbb{E}[y_{t}y_{t}^{\top}]M_{\mathbf{1}_{n}}%
Q_{-1}=Q_{-1}^{\top}\mathbb{E}[y_{t}y_{t}^{\top}]Q_{-1}.
\]
Substituting this into the objective function yields (\ref{G_GIV_Weights}).

Therefore, $a\in \mathcal{B}_{\mathbf{1}_{n}}$ is a minimizer of
(\ref{G_demand_sigma_u}) if and only if $a=Q_{-1}\tilde{a}$ for some minimizer
$\tilde{a}\in B_{n-1}$ of (\ref{G_GIV_Weights}). This completes the
proof.\hfill$Q.E.D.$

\bigskip

\noindent \textsc{Proof of Theorem\  \ref{Asy_Dist}.} We first establish the
claims of the theorem with $\hat{A}$ replaced by $\tilde{A}$ where $\tilde
{A}\equiv \hat{A}\hat{H}_{0}^{\top}$ and $\hat{H}_{0}$ is a rotation matrix
defined in (\ref{Rotation_H}) of Lemma \ref{Asy_Ahat_L1}. We then show that%
\begin{equation}
\hat{\theta}(\tilde{A})=\hat{\theta}(\hat{A}) \label{P_Asy_Dist_0a}%
\end{equation}
and
\begin{equation}
\Gamma(D_{1,T},W_{0,T},\hat{A})\hat{V}(\hat{A})\Gamma(D_{1,T},W_{0,T},\hat
{A})^{\top}=\Gamma(D_{1,T},W_{0,T},\tilde{A})\hat{V}(\tilde{A})\Gamma
(D_{1,T},W_{0,T},\tilde{A})^{\top} \label{P_Asy_Dist_0b}%
\end{equation}
wpa1, where\ $\hat{V}(\tilde{A})\equiv({\mathbf{I}_{2}\otimes \tilde{A}^{\top}%
})\hat{V}({\mathbf{I}_{2}}\otimes \tilde{A})$. The claims of the theorem
therefore follow.

By Assumption \ref{S}(ii), $\max_{t\leq T}\left \Vert S_{t}\right \Vert \leq
K$.\ Therefore, from Assumptions \ref{Asy_Cond_1}(i, iv), (\ref{P_Asy_L0_1})
and (\ref{P_Asy_L0_2}) in the proof of Lemma \ref{Asy_L0}, it follows that%
\begin{equation}
\max_{t\leq T}\left(  ||{\mathbb{E}}[{y_{t}p_{t}]||}+||{\mathbb{E}}[y_{t}%
y_{t}^{\top}{]||}\right)  \leq K. \label{P_Asy_Dist_1a}%
\end{equation}
This along with Lemmas \ref{Asy_L0} and \ref{Asy_L0b} implies that%
\begin{equation}
\Big \lVert T^{-1}\sum_{t\leq T}{y_{t}p_{t}}\Big \rVert+\Big \lVert T^{-1}%
\sum_{t\leq T}{y_{t}y_{S,t}}\Big \rVert \leq K, \label{P_Asy_Dist_1b}%
\end{equation}
wpa1. From Lemmas \ref{Asy_L0} and \ref{Asy_L0b}, {(\ref{P_Asy_Dist_1b})
and\ (\ref{P_Asy_L5_2}) in the proof of Lemma \ref{Asy_L5}, it follows that}%
\begin{equation}
D_{1,T}(\tilde{A})=(\mathbf{I}_{2}\otimes{\tilde{A}^{\top}})T^{-1}\sum_{t\leq
T}\left(
\begin{array}
[c]{cc}%
{y_{t}p_{t}} & \mathbf{0}_{n}\\
\mathbf{0}_{n} & {y_{t}y_{S,t}}%
\end{array}
\right)  =D_{1}(A)+O_{p}(T^{-1/2})=O_{p}(1). \label{P_Asy_Dist_2}%
\end{equation}
Similarly, by {Assumption \ref{Asy_Cond_2}(ii) and (\ref{P_Asy_L5_2}),}%
\begin{equation}
W_{0,T}(\tilde{A})=(({\mathbf{I}_{2}\otimes \tilde{A}^{\top}})W_{0,T}%
({\mathbf{I}_{2}\otimes}\tilde{A}))^{-1}=W_{0}(A)+o_{p}(1)=O_{p}(1),
\label{P_Asy_Dist_3}%
\end{equation}
where\ $W_{0}(A)\equiv(({\mathbf{I}_{2}\otimes A^{\top}})W_{0}({\mathbf{I}%
_{2}}\otimes A))^{-1}$. Combining {(\ref{P_Asy_Dist_2}) and
(\ref{P_Asy_Dist_3}) yields}%
\begin{equation}
D_{1,T}(\tilde{A})^{\top}W_{0,T}(\tilde{A})D_{1,T}(\tilde{A})=D_{1}(A)^{\top
}W_{0}(A)D_{1}(A)+o_{p}(1). \label{P_Asy_Dist_4}%
\end{equation}
Since $A^{\top}A=A_{0}^{\top}Q_{-1}^{\top}Q_{-1}A_{0}={\mathbf{I}}_{n-\bar{r}%
}$, by Assumption {\ref{Asy_Cond_2}(ii), }%
\[
\rho_{\min}(W_{0}(A))\geq(\rho_{\max}(W_{0}))^{-1}(\rho_{\max}({\mathbf{I}%
_{2}\otimes A^{\top}}A))^{-1}\geq K^{-1}.
\]
This together with Assumption {\ref{Asy_Cond_2}(iii) implies }%
\begin{align}
\rho_{\min}\left(  D_{1}(A)^{\top}W_{0}(A)D_{1}(A)\right)   &  \geq \rho_{\min
}(W_{0}(A))\rho_{\min}\left(  D_{1}(A)^{\top}D_{1}(A)\right) \nonumber \\
&  \geq K^{-1}\min \left \{  \Big \lVert T^{-1}\sum_{t\leq T}A^{\top}%
{\mathbb{E}}[{y_{t}p_{t}]}\Big \rVert^{2}{,}\text{ }\Big \lVert T^{-1}%
\sum_{t\leq T}A^{\top}{\mathbb{E}}[{y_{t}y_{S,t}]}\Big \rVert^{2}\right \}
\nonumber \\
&  \geq K^{-1}. \label{P_Asy_Dist_5}%
\end{align}
From\ {(\ref{P_Asy_Dist_4})\ and\ (\ref{P_Asy_Dist_5}), it follows that }%
\begin{equation}
\rho_{\min}\left(  D_{1,T}(\tilde{A})^{\top}W_{0,T}(\tilde{A})D_{1,T}%
(\tilde{A})\right)  \geq(2K)^{-1},\text{ \ wpa1.} \label{P_Asy_Dist_6}%
\end{equation}

By the definition of $D_{2,T}(\tilde{A})$, we can write%
\begin{align}
D_{2,T}(\tilde{A})  &  ={(\mathbf{I}_{2}\otimes \tilde{A}^{\top})}T^{-1}%
\sum_{t\leq T}\left(
\begin{array}
[c]{c}%
{\phi y_{t}p_{t}+y_{t}(y_{e,t}-\phi p_{t})}\\
{\psi y_{t}y_{S,t}+y_{t}(p_{t}-\psi y_{S,t})}%
\end{array}
\right) \nonumber \\
&  ={(\mathbf{I}_{2}\otimes \tilde{A}^{\top})}T^{-1}\sum_{t\leq T}\left(
\begin{array}
[c]{c}%
{\phi y_{t}p_{t}}\\
{\psi y_{t}y_{S,t}}%
\end{array}
\right)  +{(\mathbf{I}_{2}\otimes \tilde{A}^{\top})}T^{-1}\sum_{t\leq T}\left(
\begin{array}
[c]{c}%
{y_{t}(y_{e,t}-\phi p_{t})}\\
{y_{t}(p_{t}-\psi y_{S,t})}%
\end{array}
\right) \nonumber \\
&  =D_{1,T}(\tilde{A})\theta+T^{-1}\sum_{t\leq T}\left(
\begin{array}
[c]{c}%
\tilde{A}^{\top}y_{t}v_{t}\\
\tilde{A}^{\top}y_{t}\varepsilon_{t}%
\end{array}
\right)  . \label{P_Asy_Dist_7}%
\end{align}
By Assumption\ {\ref{Asy_Cond_2}(i)} and Lemma {\ref{Asy_L5}},
\begin{equation}
T^{-1/2}\sum_{t\leq T}\left(
\begin{array}
[c]{c}%
\tilde{A}^{\top}y_{t}v_{t}\\
\tilde{A}^{\top}y_{t}\varepsilon_{t}%
\end{array}
\right)  =T^{-1/2}\sum_{t\leq T}\left(
\begin{array}
[c]{c}%
A^{\top}\xi_{v,t}\\
A^{\top}\xi_{\varepsilon,t}%
\end{array}
\right)  +O_{p}(T^{-1/2})=O_{p}(1). \label{P_Asy_Dist_8}%
\end{equation}
From the definition of $\hat{\theta}(\tilde{A})$,\ {(\ref{P_Asy_Dist_6}) and
(\ref{P_Asy_Dist_7}), it follows that}%
\begin{equation}
\hat{\theta}(\tilde{A})=\theta+(D_{1,T}(\tilde{A})^{\top}W_{0,T}(\tilde
{A})D_{1,T}(\tilde{A}))^{-1}D_{1,T}(\tilde{A})^{\top}W_{0,T}(\tilde{A}%
)T^{-1}\sum_{t\leq T}\left(
\begin{array}
[c]{c}%
\tilde{A}^{\top}y_{t}v_{t}\\
\tilde{A}^{\top}y_{t}\varepsilon_{t}%
\end{array}
\right)  . \label{P_Asy_Dist_9}%
\end{equation}
This along with\ {(\ref{P_Asy_Dist_2}),\ (\ref{P_Asy_Dist_3}),}%
\ {(\ref{P_Asy_Dist_4}) and (\ref{P_Asy_Dist_8}), and (\ref{P_Asy_L5_2})
implies that }
\begin{align}
T^{1/2}(\hat{\theta}(\tilde{A})-\theta)  &  =(D_{1,T}(\tilde{A})^{\top}%
W_{0,T}(\tilde{A})D_{1,T}(\tilde{A}))^{-1}D_{1,T}(\tilde{A})^{\top}%
W_{0,T}(\tilde{A})T^{-1/2}\sum_{t\leq T}\left(
\begin{array}
[c]{c}%
{\tilde{A}^{\top}}y_{t}v_{t}\\
{\tilde{A}^{\top}}y_{t}\varepsilon_{t}%
\end{array}
\right) \nonumber \\
&  =\Gamma(D_{1},W_{0},A)({\mathbf{I}_{2}\otimes A^{\top}})T^{-1/2}\sum_{t\leq
T}\xi_{t}+o_{p}(1). \label{P_Asy_Dist_10}%
\end{align}

For any vector $a\in \mathbb{R}^{2}$ with $a^{\top}a=1$, define
\[
a_{T}\equiv V^{1/2}({\mathbf{I}_{2}\otimes A})\Gamma(D_{1},W_{0},A)^{\top
}(\Gamma(D_{1},W_{0},A)V(A)\Gamma(D_{1},W_{0},A)^{\top})^{-1/2}a.
\]
Then by the definition of $\Gamma(D_{1},W_{0},A)$, we have $a_{T}^{\top}%
a_{T}=1$. From {(\ref{P_Asy_Dist_10}) it follows that }%
\begin{equation}
a^{\top}(\Gamma(D_{1},W_{0},A)V(A)\Gamma(D_{1},W_{0},A)^{\top})^{-1/2}%
T^{1/2}(\hat{\theta}(\tilde{A})-\theta)=a_{T}^{\top}V^{-1/2}T^{-1/2}%
\sum_{t\leq T}\xi_{t}+o_{p}(1). \label{P_Asy_Dist_11}%
\end{equation}
Under Assumption {\ref{Asy_Cond_2}(i), }%
\begin{equation}
V^{-1/2}T^{-1/2}\sum_{t\leq T}\xi_{t}\rightarrow_{d}N(0,{\mathbf{I}}_{2n}).
\label{P_Asy_Dist_12}%
\end{equation}
Since the set $\{a\in \mathbb{R}^{2n}:a^{\top}a=1\}$ is compact, for any
subsequence of $\{T\}$, we can extract a further subsequence $\{T_{k}\}$ such
that
\[
a_{T_{k}}\rightarrow a_{\infty}\in \left \{  a\in \mathbb{R}^{2n}:a^{\top
}a=1\right \}  .
\]
From {(\ref{P_Asy_Dist_12}), it is clear that }$V^{-1/2}T_{k}^{-1/2}%
\sum_{t\leq T_{k}}\xi_{t}=O_{p}(1)$. Therefore, { }%
\[
a_{T_{k}}^{\top}V^{-1/2}T_{k}^{-1/2}\sum_{t\leq T_{k}}\xi_{t}=a_{\infty}%
^{\top}V^{-1/2}T_{k}^{-1/2}\sum_{t\leq T_{k}}\xi_{t}+(a_{T_{k}}-a_{\infty
})^{\top}V^{-1/2}T_{k}^{-1/2}\sum_{t\leq T_{k}}\xi_{t}\rightarrow_{d}N(0,1),
\]
from which, we can conclude that
\begin{equation}
a_{T}^{\top}V^{-1/2}T^{-1/2}\sum_{t\leq T}\xi_{t}\rightarrow_{d}N(0,1).
\label{P_Asy_Dist_13}%
\end{equation}

Combining {(\ref{P_Asy_Dist_11}) and (\ref{P_Asy_Dist_13}), and applying
Slutsky's theorem and the Cramer-Wold device, we deduce that }%
\begin{equation}
(\Gamma(D_{1},W_{0},A)V(A)\Gamma(D_{1},W_{0},A)^{\top})^{-1/2}T^{1/2}%
(\hat{\theta}(\tilde{A})-\theta)\rightarrow_{d}N(0,{\mathbf{I}}_{2}).
\label{P_Asy_Dist_14}%
\end{equation}
This establishes {(\ref{Asy_Dist_1}) with }$\hat{A}$ replaced by $\tilde{A}$.
By {(\ref{P_Asy_Dist_2}),\ (\ref{P_Asy_Dist_3}), (\ref{P_Asy_Dist_4}) and
(\ref{P_Asy_Dist_5})}%
\begin{align}
\Gamma(D_{1,T},W_{0,T},\tilde{A})  &  =(D_{1,T}(\tilde{A})^{\top}%
W_{0,T}(\tilde{A})D_{1,T}(\tilde{A}))^{-1}D_{1,T}(\tilde{A})^{\top}%
W_{0,T}(\tilde{A})\nonumber \\
&  =(D_{1}(A)^{\top}W_{0}(A)D_{1}(A))^{-1}D_{1}(A)^{\top}W_{0}(A)+o_{p}%
(1)\nonumber \\
&  =\Gamma(D_{1},W_{0},A)+o_{p}(1)=O_{p}(1). \label{P_Asy_Dist_15}%
\end{align}
By similar arguments for deriving {(\ref{P_Asy_Dist_3}), we can show that}%
\begin{equation}
\hat{V}(\tilde{A})=({\mathbf{I}_{2}\otimes \tilde{A}^{\top}})\hat
{V}({\mathbf{I}_{2}\otimes}\tilde{A})=V(A)+o_{p}(1)=O_{p}(1),
\label{P_Asy_Dist_15b}%
\end{equation}
which along with {(\ref{P_Asy_Dist_15}) establishes (\ref{Asy_Dist_2}) with
}$\hat{A}$ replaced by $\tilde{A}$.

Using the definition of $J_{T}(\theta;\cdot,\cdot)$ in {(\ref{J_T}), the GMM
problem in (\ref{GMM_Criterion}) }can be written as%
\begin{equation}
\hat{\theta}(\hat{A})\equiv \arg \min_{\theta \in \Theta}J_{T}(\theta;\hat
{A},W_{0,T}). \label{P_Asy_Dist_16}%
\end{equation}
The sufficient conditions of Lemma {\ref{Invariance} hold for }$A_{1}=\hat{A}%
$, $W_{1}=W_{0,T}$ and $C_{1}=\hat{H}_{0}^{\top}$ because: {(i) }$\hat
{A}^{\top}\hat{A}=\hat{A}_{0}^{\top}Q_{-1}^{\top}Q_{-1}\hat{A}_{0}%
=\mathbf{I}_{n-\bar{r}}${ and thus }$\hat{A}$ has full rank by construction;
(ii) $W_{0,T}$ is positive definite wpa1 by Assumption \ref{Asy_Cond_2}(ii);
and (iii) $\hat{H}_{0}^{\top}$ is nonsingular wpa1 by Lemma {\ref{Asy_Ahat_L1}%
. Therefore, applying }Lemma {\ref{Invariance} leads to }%
\begin{equation}
J_{T}(\theta;\hat{A},W_{0,T})=J_{T}(\theta;\hat{A}\hat{H}_{0}^{\top}%
,W_{0,T})=J_{T}(\theta;{\tilde{A}},W_{0,T}) \label{P_Asy_Dist_16b}%
\end{equation}
wpa1, which together with ({\ref{P_Asy_Dist_16}}) shows that $\hat{\theta
}(\hat{A})=\hat{\theta}({\tilde{A}})$ wpa1.

To verify the claim in ({\ref{P_Asy_Dist_0b}}), we first observe that%
\begin{equation}
D_{1,T}(\tilde{A})=D_{1,T}(\hat{A}\hat{H}_{0}^{\top})=({\mathbf{I}_{2}}%
\otimes \hat{H}_{0})D_{1,T}(\hat{A}) \label{P_Asy_Dist_17}%
\end{equation}
and
\begin{equation}
W_{0,T}(\tilde{A})=W_{0,T}(\hat{A}\hat{H}_{0}^{\top})=({\mathbf{I}_{2}}%
\otimes \hat{H}_{0}^{\top})^{-1}W_{0,T}(\hat{A})({\mathbf{I}_{2}}\otimes \hat
{H}_{0})^{-1}, \label{P_Asy_Dist_18}%
\end{equation}
where $\hat{H}_{0}$ is nonsingular wpa1 by Lemma \ref{Asy_Ahat_L1}.
Therefore,
\begin{align}
&  D_{1,T}(\tilde{A})^{\top}W_{0,T}(\tilde{A})D_{1,T}(\tilde{A})\nonumber \\
&  =D_{1,T}(\hat{A})^{\top}({\mathbf{I}_{2}}\otimes \hat{H}_{0}^{\top
})({\mathbf{I}_{2}}\otimes \hat{H}_{0}^{\top})^{-1}W_{0,T}(\hat{A}%
)({\mathbf{I}_{2}}\otimes \hat{H}_{0})^{-1}({\mathbf{I}_{2}}\otimes \hat{H}%
_{0})D_{1,T}(\hat{A})\nonumber \\
&  =D_{1,T}(\hat{A})^{\top}W_{0,T}(\hat{A})D_{1,T}(\hat{A})\text{, \ wpa1.}
\label{P_Asy_Dist_19}%
\end{align}
Next note that%
\[
\hat{V}(\tilde{A})=\hat{V}(\hat{A}\hat{H}_{0}^{\top})=({\mathbf{I}_{2}}%
\otimes \hat{H}_{0})\hat{V}(\hat{A})({\mathbf{I}_{2}}\otimes \hat{H}_{0}^{\top
}),
\]
which along with ({\ref{P_Asy_Dist_17}}) and ({\ref{P_Asy_Dist_18}}) {implies
that}%
\begin{align}
&  D_{1,T}(\tilde{A})^{\top}W_{0,T}(\tilde{A})\hat{V}(\tilde{A})W_{0,T}%
(\tilde{A})D_{1,T}(\tilde{A})\nonumber \\
&  =D_{1,T}(\hat{A})^{\top}({\mathbf{I}_{2}}\otimes \hat{H}_{0}^{\top
})({\mathbf{I}_{2}}\otimes \hat{H}_{0}^{\top})^{-1}W_{0,T}(\hat{A}%
)({\mathbf{I}_{2}}\otimes \hat{H}_{0})^{-1}\nonumber \\
&  \times({\mathbf{I}_{2}}\otimes \hat{H}_{0})\hat{V}(\hat{A})({\mathbf{I}_{2}%
}\otimes \hat{H}_{0}^{\top})\nonumber \\
&  \times({\mathbf{I}_{2}}\otimes \hat{H}_{0}^{\top})^{-1}W_{0,T}(\hat
{A})({\mathbf{I}_{2}}\otimes \hat{H}_{0})^{-1}({\mathbf{I}_{2}}\otimes \hat
{H}_{0})D_{1,T}(\hat{A})\nonumber \\
&  =D_{1,T}(\hat{A})^{\top}W_{0,T}(\hat{A})\hat{V}(\hat{A})W_{0,T}(\hat
{A})D_{1,T}(\hat{A})\text{, \ wpa1.} \label{P_Asy_Dist_20}%
\end{align}
{ }The claim in ({\ref{P_Asy_Dist_0b}}) now follows from\ {(\ref{P_Asy_Dist_6}%
), }({\ref{P_Asy_Dist_19}}), ({\ref{P_Asy_Dist_20}}) and the definitions of
$\Gamma(D_{1,T},W_{0,T},\hat{A})$, $\Gamma(D_{1,T},W_{0,T},\tilde{A})$%
,\ $\hat{V}(\hat{A})$ and $\hat{V}(\tilde{A})$.\hfill$Q.E.D.$

\bigskip

\noindent \textsc{Proof of Theorem \ref{J_Test}.} By similar arguments in
deriving ({\ref{P_Asy_Dist_16b}), we can verify the sufficient conditions of
}Lemma {\ref{Invariance} for }$A_{1}=\hat{A}$, $W_{1}=\hat{V}$ and $C_{1}%
=\hat{H}_{0}^{\top}$. Hence, we can apply Lemma {\ref{Invariance} to show
that\ }%
\[
J_{T}(\theta;\hat{A},\hat{V})=J_{T}(\theta;\hat{A}\hat{H}_{0}^{\top},\hat
{V})=J_{T}(\theta;{\tilde{A}},\hat{V})
\]
wpa1, from which the claim of the theorem follows if
\begin{equation}
J_{T}(\hat{\theta}^{\ast}(\tilde{A});{\tilde{A}},\hat{V})\rightarrow_{d}%
\chi^{2}(2(n-\bar{r}-1)). \label{P_J_Test_0}%
\end{equation}

Note that\ $\hat{\theta}^{\ast}(\tilde{A})^{\top}=(\hat{\phi}^{\ast}(\tilde
{A}),{\hat{\psi}}^{\ast}(\tilde{A}))$. To prove ({\ref{P_J_Test_0}}), we first
observe that%
\begin{align}
T^{1/2}\bar{g}_{T}(\hat{\theta}^{\ast}(\tilde{A});\tilde{A})  &  =T^{-1/2}%
\sum_{t\leq T}\left(
\begin{array}
[c]{c}%
\tilde{A}^{\top}y_{t}(y_{e,t}-\hat{\phi}^{\ast}(\tilde{A})p_{t})\\
\tilde{A}^{\top}y_{t}(p_{t}-{\hat{\psi}}^{\ast}(\tilde{A}){y_{S,t}})
\end{array}
\right) \nonumber \\
&  =T^{-1/2}\sum_{t\leq T}\left(
\begin{array}
[c]{c}%
\tilde{A}^{\top}y_{t}v_{t}\\
\tilde{A}^{\top}y_{t}\varepsilon_{t}%
\end{array}
\right)  -D_{1,T}(\tilde{A})T^{1/2}{(\hat{\theta}^{\ast}(\tilde{A})-\theta)}.
\label{P_J_Test_1}%
\end{align}
By the same arguments for showing\ {(\ref{P_Asy_Dist_10}) but replacing
}$W_{0,T}(\tilde{A})$ by $W_{\ast,T}(\tilde{A})$, we obtain%
\begin{align}
T^{1/2}({\hat{\theta}^{\ast}}(\tilde{A})-\theta)  &  =(D_{1}(A)^{\top
}V(A)^{-1}D_{1}(A))^{-1}D_{1}(A)^{\top}V(A)^{-1}({\mathbf{I}_{2}\otimes
A^{\top}})T^{-1/2}\sum_{t\leq T}\xi_{t}+o_{p}(1)\nonumber \\
&  =\Gamma(D_{1},V,A)V(A)^{1/2}T^{-1/2}\sum_{t\leq T}\xi_{t}^{\ast}+o_{p}(1),
\label{P_J_Test_2}%
\end{align}
where\ $\xi_{t}^{\ast}\equiv V(A)^{-1/2}({\mathbf{I}_{2}\otimes A^{\top}}%
)\xi_{t}$. This along with Assumptions {\ref{Asy_Cond_2}(i, iii) shows that }%
\begin{equation}
T^{1/2}({\hat{\theta}^{\ast}}(\tilde{A})-\theta)=O_{p}(1). \label{P_J_Test_2b}%
\end{equation}
From ({\ref{P_Asy_Dist_2}) and }({\ref{P_J_Test_2b}), it follows that}%
\begin{equation}
D_{1,T}(\tilde{A})T^{1/2}({\hat{\theta}^{\ast}}(\tilde{A})-\theta
)=D_{1}(A)\Gamma(D_{1},V,A)V(A)^{1/2}T^{-1/2}\sum_{t\leq T}\xi_{t}^{\ast
}+o_{p}(1). \label{P_J_Test_3}%
\end{equation}

Combining {(\ref{P_Asy_Dist_8}),} {(\ref{P_J_Test_1}) and (\ref{P_J_Test_3})
yields }%
\begin{align}
T^{1/2}\bar{g}_{T}(\hat{\theta}^{\ast}(\tilde{A});\tilde{A})  &
=({\mathbf{I}_{2}\otimes A^{\top}})T^{-1/2}\sum_{t\leq T}\xi_{t}\nonumber \\
&  -D_{1}(A)\Gamma(D_{1},V,A)V(A)^{1/2}T^{-1/2}\sum_{t\leq T}\xi_{t}^{\ast
}+o_{p}(1)\nonumber \\
&  =V(A)^{1/2}T^{-1/2}\sum_{t\leq T}\xi_{t}^{\ast}-D_{1}(A)\Gamma
(D_{1},V,A)V(A)^{1/2}T^{-1/2}\sum_{t\leq T}\xi_{t}^{\ast}+o_{p}(1)\nonumber \\
&  =V(A)^{1/2}\Pi(A)T^{-1/2}\sum_{t\leq T}\xi_{t}^{\ast}+o_{p}(1)=O_{p}(1),
\label{P_J_Test_4}%
\end{align}
where
\[
\Pi(A)\equiv{\mathbf{I}}_{2(n-\bar{r})}-V(A)^{-1/2}D_{1}(A)(D_{1}(A)^{\top
}V(A)^{-1}D_{1}(A))^{-1}D_{1}(A)^{\top}V(A)^{-1/2}.
\]
By the definition of $W_{\ast,T}(\tilde{A})$ and (\ref{P_Asy_Dist_15b}),
\begin{equation}
(W_{\ast,T}(\tilde{A}))^{-1}=\hat{V}(\tilde{A})=V(A)+o_{p}(1)
\label{P_J_Test_5a}%
\end{equation}
which together with Assumption \ref{Asy_Cond_2}(i) implies that
\begin{equation}
\rho_{\max}(W_{\ast,T}(\tilde{A}))\leq2K, \label{P_J_Test_5b}%
\end{equation}
wpa1.\ Using {(\ref{P_J_Test_4}), (\ref{P_J_Test_5a}) and (\ref{P_J_Test_5b}%
),\ we conclude that}%
\begin{align}
J_{T}(\hat{\theta}^{\ast}(\tilde{A});{\tilde{A}},\hat{V})  &  =T\bar{g}%
_{T}(\hat{\theta}^{\ast}(\tilde{A});\tilde{A})^{\top}W_{\ast,T}(\tilde{A}%
)\bar{g}_{T}(\hat{\theta}^{\ast}(\tilde{A});\tilde{A})\nonumber \\
&  =\Big(T^{-1/2}\sum_{t\leq T}\xi_{t}^{\ast}\Big)^{\top}\Pi(A)\Big(T^{-1/2}%
\sum_{t\leq T}\xi_{t}^{\ast}\Big)+o_{p}(1). \label{P_J_Test_6}%
\end{align}
It is clear that $\Pi(A)$ is an idempotent matrix with%
\begin{equation}
\mathrm{rank}(\Pi(A))=2(n-\bar{r})-2. \label{P_J_Test_7}%
\end{equation}
By Assumption \ref{Asy_Cond_2}(i) and similar arguments for
showing\ {(\ref{P_Asy_Dist_13}), we have}%
\begin{equation}
T^{-1/2}\sum_{t\leq T}\xi_{t}^{\ast}\rightarrow_{d}N(0,\mathbf{I}_{2(n-\bar
{r})}). \label{P_J_Test_8}%
\end{equation}
The claim in {(\ref{P_J_Test_0})} follows from {(\ref{P_J_Test_6}%
)-(\ref{P_J_Test_8}) and Slutsky's Theorem}.\hfill$Q.E.D.$

\bigskip

\noindent \textsc{Proof of Theorem \ref{r_hat_consistency}.} Under Assumptions
\ref{Asy_Cond_1}(i, ii, iii) and \ref{S}, we can apply Lemma \ref{Asy_L0} to
obtain%
\begin{equation}
\left \Vert \hat{S}_{y}-S_{y}\right \Vert =O_{p}(T^{-1/2}),
\label{P_r_hat_consistency_0}%
\end{equation}
which implies that%
\begin{equation}
\max_{j\in \mathcal{J}}|\hat{\mu}_{j}-\mu_{j}|=O_{p}(T^{-1/2}).
\label{P_r_hat_consistency_1}%
\end{equation}

We first consider the case $j>\bar{r}$. For any $j\in \{ \bar{r}+1,\ldots
,n-1\}$,%
\begin{align*}
\frac{\mathrm{BIC}_{T}(j)-\mathrm{BIC}_{T}(\bar{r})}{\log(T)}  &  =j-\bar
{r}+(n-j)^{-1}\sum_{s=1}^{n-j}\frac{T(\hat{\mu}_{s}-\hat{\mu}_{1})^{2}}%
{2\hat{\mu}_{s}^{2}\log(T)}-(n-\bar{r})^{-1}\sum_{s=1}^{n-\bar{r}}\frac
{T(\hat{\mu}_{s}-\hat{\mu}_{1})^{2}}{2\hat{\mu}_{s}^{2}\log(T)}\\
&  \geq j-\bar{r}-(n-\bar{r})^{-1}\sum_{s=n-j+1}^{n-\bar{r}}\frac{T(\hat{\mu
}_{s}-\hat{\mu}_{1})^{2}}{\hat{\mu}_{s}^{2}\log(T)}\\
&  =j-\bar{r}-(n-\bar{r})^{-1}\sum_{s=n-j+1}^{n-\bar{r}}\frac{(T^{1/2}%
(\hat{\mu}_{s}-\mu_{s})-T^{1/2}(\hat{\mu}_{1}-\bar{\sigma}_{u}^{2}))^{2}%
}{(\bar{\sigma}_{u}^{2}+\hat{\mu}_{s}-\bar{\sigma}_{u}^{2})^{2}\log(T)}\\
&  =j-\bar{r}+O_{p}(\log(T)^{-1}),
\end{align*}
where the equality in the third line follows from the fact that $\mu_{j}%
=\bar{\sigma}_{u}^{2}$ for $j\in \{1,\ldots,n-\bar{r}\}$, and the last equality
follows from Assumption \ref{Asy_Cond_1}(iii) together with
(\ref{P_r_hat_consistency_1}). Consequently,
\begin{equation}
\min_{j\in \{ \bar{r}+1,\ldots,n-1\}}\frac{\mathrm{BIC}_{T}(j)-\mathrm{BIC}%
_{T}(\bar{r})}{j-\bar{r}}>\frac{\log(T)}{2},\text{ \ wpa1.}
\label{P_r_hat_consistency_3}%
\end{equation}

Next, consider the case $j<\bar{r}$. For any $j\in \{1,\ldots,\bar{r}-1\}$,
\begin{align}
\frac{\mathrm{BIC}_{T}(j)-\mathrm{BIC}_{T}(\bar{r})}{T}  &  =(n-j)^{-1}%
\sum_{s=1}^{n-j}\frac{(\hat{\mu}_{s}-\hat{\mu}_{1})^{2}}{2\hat{\mu}_{s}^{2}%
}-(n-\bar{r})^{-1}\sum_{s=1}^{n-\bar{r}}\frac{(\hat{\mu}_{s}-\hat{\mu}%
_{1})^{2}}{2\hat{\mu}_{s}^{2}}+\frac{(j-\bar{r})\log(T)}{T}\nonumber \\
&  =(n-j)^{-1}\sum_{s=n-\bar{r}+1}^{n-j}\frac{(\hat{\mu}_{s}-\hat{\mu}%
_{1})^{2}}{2\hat{\mu}_{s}^{2}}\nonumber \\
&  \text{ \  \  \  \  \  \  \  \  \  \  \  \  \  \  \  \  \ }+\frac{(j-\bar{r})}%
{(n-j)(n-\bar{r})}\sum_{s=1}^{n-\bar{r}}\frac{(\hat{\mu}_{s}-\hat{\mu}%
_{1})^{2}}{2\hat{\mu}_{s}^{2}}+\frac{(j-\bar{r})\log(T)}{T}.
\label{P_r_hat_consistency_4a}%
\end{align}
From Assumption \ref{Asy_Cond_1}(iii) and (\ref{P_r_hat_consistency_1}), we
have
\begin{equation}
\frac{(j-\bar{r})}{(n-j)(n-\bar{r})}\sum_{s=1}^{n-\bar{r}}\frac{(\hat{\mu}%
_{s}-\hat{\mu}_{1})^{2}}{2\hat{\mu}_{s}^{2}}=O_{p}(T^{-1}).
\label{P_r_hat_consistency_4b}%
\end{equation}
Similarly,
\[
(n-j)^{-1}\sum_{s=n-\bar{r}+1}^{n-j}\frac{(\hat{\mu}_{s}-\hat{\mu}_{1})^{2}%
}{2\hat{\mu}_{s}^{2}}\geq \frac{\bar{r}-j}{n-j}\frac{(\hat{\mu}_{n-\bar{r}%
+1}-\hat{\mu}_{1})^{2}}{2\hat{\mu}_{n-j}^{2}}=\frac{\bar{r}-j}{n-j}\frac
{(\mu_{n-\bar{r}+1}-\bar{\sigma}_{u}^{2})^{2}}{2\mu_{n-j}^{2}}+O_{p}%
(T^{-1/2}).
\]
Combining this result with (\ref{P_r_hat_consistency_4a}) and
(\ref{P_r_hat_consistency_4b}) yields
\begin{equation}
\frac{\mathrm{BIC}_{T}(j)-\mathrm{BIC}_{T}(\bar{r})}{T}\geq \frac{\bar{r}%
-j}{n-j}\frac{(\mu_{n-\bar{r}+1}-\bar{\sigma}_{u}^{2})^{2}}{2\mu_{n-j}^{2}%
}+O_{p}(T^{-1/2}). \label{P_r_hat_consistency_4c}%
\end{equation}
By (\ref{P_Asy_Dist_1a}), $\mu_{n-j}^{2}\leq K$, while Assumption
\ref{Asy_Cond_1}(iii) implies that $\mu_{n-\bar{r}+1}-\bar{\sigma}_{u}%
^{2}>K^{-1}$. Therefore, (\ref{P_r_hat_consistency_4c}) implies that
\begin{equation}
\min_{j\in \{1,\ldots,\bar{r}-1\}}\frac{\mathrm{BIC}_{T}(j)-\mathrm{BIC}%
_{T}(\bar{r})}{\bar{r}-j}>K^{-1}T,\text{ \ wpa1.}
\label{P_r_hat_consistency_5}%
\end{equation}

The claim of the theorem now follows directly from
(\ref{P_r_hat_consistency_3}) and (\ref{P_r_hat_consistency_5}).\hfill$Q.E.D.$

\bigskip

\noindent \textsc{Proof of Lemma\  \ref{GK_Non_ID}.} Since $M_{\mathbf{1}_{n}%
}\lambda_{-1}=\lambda_{-1}$, the moment condition (\ref{GK_F_7}) is equivalent
to
\begin{equation}
M_{\mathbf{1}_{n}}\mathbb{E}[y_{t}y_{t}^{\top}]M_{\mathbf{1}_{n}}\lambda
_{-1}=\lambda_{-1}(\lambda_{-1}^{\top}\lambda_{-1})^{-1}\lambda_{-1}^{\top
}\mathbb{E}[y_{t}y_{t}^{\top}]\lambda_{-1}. \label{P_GK_Non_ID_1}%
\end{equation}
Since $\mathbb{E}[y_{t}y_{t}^{\top}]$ is nonsingular while $M_{\mathbf{1}_{n}%
}$ has rank $n-1$, $M_{\mathbf{1}_{n}}\mathbb{E}[y_{t}y_{t}^{\top
}]M_{\mathbf{1}_{n}}$ has rank $n-1$. Hence its null space is spanned by
$\mathbf{1}_{n}$, and its range is contained in
\[
\mathcal{H}_{1_{n}}\equiv \{a\in \mathbb{R}^{n}:\mathbf{1}_{n}^{\top}a=0\}.
\]
Therefore, there exists an orthonormal basis $D_{M}\equiv(d_{1},\ldots
,d_{n-1})$ of $\mathcal{H}_{1_{n}}$ such that
\begin{equation}
M_{\mathbf{1}_{n}}\mathbb{E}[y_{t}y_{t}^{\top}]M_{\mathbf{1}_{n}}D_{M}%
=D_{M}\operatorname{diag}(\rho_{1},\ldots,\rho_{n-1}), \label{P_GK_Non_ID_2}%
\end{equation}
where $\rho_{1},\ldots,\rho_{n-1}$ are the nonzero eigenvalues of
$M_{\mathbf{1}_{n}}\mathbb{E}[y_{t}y_{t}^{\top}]M_{\mathbf{1}_{n}}$.

Now fix any subset $J\subset \{1,\ldots,n-1\}$ with $|J|=r-1$, and let $D_{J}$
collect the corresponding columns of $D_{M}$. Then
\begin{equation}
\mathbf{1}_{n}^{\top}D_{J}=0,\qquad D_{J}^{\top}D_{J}=\mathbf{I}_{r-1},
\label{P_GK_Non_ID_2b}%
\end{equation}
and
\begin{equation}
M_{\mathbf{1}_{n}}\mathbb{E}[y_{t}y_{t}^{\top}]M_{\mathbf{1}_{n}}D_{J}%
=D_{J}\operatorname{diag}\{(\rho_{j})_{j\in J}\}. \label{P_GK_Non_ID_3}%
\end{equation}
Because $M_{\mathbf{1}_{n}}D_{J}=D_{J}$, we also have
\begin{equation}
D_{J}^{\top}\mathbb{E}[y_{t}y_{t}^{\top}]D_{J}=D_{J}^{\top}M_{\mathbf{1}_{n}%
}\mathbb{E}[y_{t}y_{t}^{\top}]M_{\mathbf{1}_{n}}D_{J}=\operatorname{diag}%
\{(\rho_{j})_{j\in J}\}. \label{P_GK_Non_ID_4}%
\end{equation}
Combining (\ref{P_GK_Non_ID_3}) and (\ref{P_GK_Non_ID_4}) yields
\[
M_{\mathbf{1}_{n}}\mathbb{E}[y_{t}y_{t}^{\top}]M_{\mathbf{1}_{n}}D_{J}%
=D_{J}(D_{J}^{\top}D_{J})^{-1}D_{J}^{\top}\mathbb{E}[y_{t}y_{t}^{\top}]D_{J}.
\]
This is precisely (\ref{P_GK_Non_ID_1}) with $\lambda_{-1}$ replaced by
$D_{J}$. Hence $D_{J}$ satisfies (\ref{GK_F_7}).

Finally, since $D_{J}^{\top}D_{J}=\mathbf{I}_{r-1}$ and $\mathbf{1}_{n}^{\top
}D_{J}=0$, the scaled matrix $n^{1/2}D_{J}$ satisfies the normalization in
(\ref{GK_F_0a}). The condition (\ref{GK_F_7}) is invariant to nonsingular
right multiplication of $D_{J}$, and therefore it is also satisfied by
$n^{1/2}D_{J}$. This proves the claim.\hfill$Q.E.D.$

\bigskip

\noindent \textsc{Proof of Lemma \ref{GK_Rotation}.} Since $C_{1}$ and $C_{2}$
are orthogonal, we have
\[
\mathbf{1}_{n}^{\top}\lambda_{C_{1}}=\mathbf{1}_{n}^{\top}\lambda_{-1}%
C_{1}=\mathbf{0}_{r-1}^{\top}\text{ \  \ and \  \ }n^{-1}\lambda_{C_{1}}^{\top
}\lambda_{C_{1}}=n^{-1}C_{1}^{\top}\lambda_{-1}^{\top}\lambda_{-1}%
C_{1}=\mathbf{I}_{r-1}.
\]
Moreover,
\[
\lambda_{C_{2}}^{\top}(\mathbf{1}_{n},\lambda_{C_{1}})=C_{2}^{\top}%
\lambda_{\perp}^{\top}(\mathbf{1}_{n},\lambda_{-1}C_{1})=\mathbf{0}%
_{(n-r)\times r}.
\]
This shows that $\lambda_{C_{1}}$ and $\lambda_{C_{2}}$ satisfy the same
restrictions of $\lambda_{-1}$ and $\lambda_{\perp}$.

We next verify the moment conditions. Since
\[
\lambda_{C_{1}}(C_{1}^{\top}b_{y})=\lambda_{-1}C_{1}C_{1}^{\top}b_{y}%
=\lambda_{-1}b_{y},\text{ \  \ }\lambda_{C_{1}}^{\top}y_{t}=C_{1}^{\top}%
\lambda_{-1}^{\top}y_{t},\text{ \ and \  \ }\lambda_{C_{2}}^{\top}y_{t}%
=C_{2}^{\top}\lambda_{\perp}^{\top}y_{t},
\]
we have%
\begin{align*}
&  \mathbb{E}\left[  (y_{e,t}-\phi p_{t}-y_{t}^{\top}\lambda_{C_{1}}%
C_{1}^{\top}b_{y})%
\begin{pmatrix}
\lambda_{C_{2}}^{\top}y_{t}\\
\lambda_{C_{1}}^{\top}y_{t}%
\end{pmatrix}
\right] \\
&  \qquad=\mathrm{diag}(C_{2}^{\top},C_{1}^{\top})\mathbb{E}\left[
(y_{e,t}-\phi p_{t}-y_{t}^{\top}\lambda_{-1}b_{y})%
\begin{pmatrix}
\lambda_{\perp}^{\top}y_{t}\\
\lambda_{-1}^{\top}y_{t}%
\end{pmatrix}
\right]  =\mathbf{0}_{n-1}.
\end{align*}
Thus (\ref{GK_F_5}) continues to hold after the rotation.

Similarly, since $\lambda_{C_{1}}(C_{1}^{\top}b_{p})=\lambda_{-1}b_{p}$, we
obtain
\begin{align*}
&  \mathbb{E}\left[  (p_{t}-\psi y_{S,t}-y_{t}^{\top}\lambda_{C_{1}}%
C_{1}^{\top}b_{p})%
\begin{pmatrix}
\lambda_{C_{2}}^{\top}y_{t}\\
\lambda_{C_{1}}^{\top}y_{t}%
\end{pmatrix}
\right] \\
&  \qquad=\mathrm{diag}(C_{2}^{\top},C_{1}^{\top})\mathbb{E}\left[
(p_{t}-\psi y_{S,t}-y_{t}^{\top}\lambda_{-1}b_{p})%
\begin{pmatrix}
\lambda_{\perp}^{\top}y_{t}\\
\lambda_{-1}^{\top}y_{t}%
\end{pmatrix}
\right]  =\mathbf{0}_{n-1}.
\end{align*}
Therefore, (\ref{GK_F_6}) also continues to hold.

Finally, since
\[
\lambda_{C_{1}}(\lambda_{C_{1}}^{\top}\lambda_{C_{1}})^{-1}\lambda_{C_{1}%
}^{\top}=\lambda_{-1}(\lambda_{-1}^{\top}\lambda_{-1})^{-1}\lambda_{-1}^{\top
},
\]
the projection matrix appearing in (\ref{GK_F_7}) is unchanged when
$\lambda_{-1}$ is replaced by $\lambda_{C_{1}}$. Moreover,
\[
y_{t}y_{t}^{\top}\lambda_{C_{1}}=y_{t}y_{t}^{\top}\lambda_{-1}C_{1}.
\]
Hence the left-hand side of (\ref{GK_F_7}) under $\lambda_{C_{1}}$ equals the
original left-hand side multiplied by $C_{1}$, and is therefore zero. Thus
(\ref{GK_F_7}) holds for $\lambda_{C_{1}}$ as well.\hfill$Q.E.D.$

\section{Auxiliary Lemmas \label{APP_3}}

This subsection presents a set of lemmas that verify Assumptions
\ref{Asy_Cond_1}(iii) and \ref{Asy_Cond_2}(iii) under primitive conditions, as
well as auxiliary results used in the proofs of the main results in Sections
\ref{sec: S_model} and \ref{sec: G_model}. We begin with lemmas that establish
Assumptions \ref{Asy_Cond_1}(iii) and \ref{Asy_Cond_2}(iii).

\begin{lemma}
\label{Suff_Cond1_iii}\ Suppose that $\rho_{\min}(T^{-1}\sum_{t\leq
T}\mathbb{E}[\eta_{t}\eta_{t}^{\top}])\geq K^{-1}$, and either (i)\ $\rho
_{\min}((\mathbf{1}_{n},\lambda)^{\top}(\mathbf{1}_{n},\lambda))\geq K^{-1}$;
or (ii) $\mathbf{1}_{n}\in \operatorname{col}(\lambda)$ and $K^{-1}\leq
\rho_{\min}(\lambda^{\top}\lambda/n)\leq \rho_{\max}(\lambda^{\top}%
\lambda/n)\leq K$. Then Assumption \ref{Asy_Cond_1}(iii) holds.
\end{lemma}

\noindent \textsc{Proof of Lemma\  \ref{Suff_Cond1_iii}.} We begin by noting
that since $Q_{-1}^{\top}Q_{-1}=\mathbf{I}_{n-1}$ and $Q_{-1}^{\top}%
\mathbf{1}_{n}=\mathbf{0}_{n-1}$,
\[
Q_{-1}Q_{-1}^{\top}=Q_{-1}(Q_{-1}^{\top}Q_{-1})^{-1}Q_{-1}^{\top}%
\]
is the projection matrix onto the subspace orthogonal to $\mathbf{1}_{n}$.
Therefore, $\mathbf{I}_{n}-Q_{-1}Q_{-1}^{\top}$ is the projection matrix onto
$\operatorname{col}(\mathbf{1}_{n})$, which implies
\[
\mathbf{I}_{n}-Q_{-1}Q_{-1}^{\top}=n^{-1}\mathbf{1}_{n}\mathbf{1}_{n}^{\top}.
\]
Hence,
\begin{equation}
Q_{-1}Q_{-1}^{\top}=M_{\mathbf{1}_{n}}. \label{P_Suff_Cond1_iii_1}%
\end{equation}

Let $\Sigma_{\eta}\equiv T^{-1}\sum_{t\leq T}\mathbb{E}[\eta_{t}\eta_{t}%
^{\top}]$. The decomposition in (\ref{G_demand_factor}) can be written as
\begin{equation}
\bar{\Sigma}_{\tilde{y}}=\tilde{\lambda}\Sigma_{\eta}\tilde{\lambda}^{\top
}+\bar{\sigma}_{u}^{2}M_{\mathbf{1}_{n}}, \label{P_Suff_Cond1_iii_2}%
\end{equation}
where $\tilde{\lambda}\equiv M_{\mathbf{1}_{n}}\lambda$. Since $\bar{\Sigma
}_{\tilde{y}}=M_{\mathbf{1}_{n}}\bar{\Sigma}_{y}M_{\mathbf{1}_{n}}$ and
$Q_{-1}^{\top}M_{\mathbf{1}_{n}}=Q_{-1}^{\top}$, it follows that
\[
Q_{-1}^{\top}\bar{\Sigma}_{\tilde{y}}Q_{-1}=Q_{-1}^{\top}\bar{\Sigma}%
_{y}Q_{-1},\qquad Q_{-1}^{\top}\tilde{\lambda}\Sigma_{\eta}\tilde{\lambda
}^{\top}Q_{-1}=Q_{-1}^{\top}\lambda \Sigma_{\eta}\lambda^{\top}Q_{-1}.
\]
Therefore, multiplying $Q_{-1}^{\top}$ and $Q_{-1}$ from the left and right of
(\ref{P_Suff_Cond1_iii_2}) yields
\begin{equation}
Q_{-1}^{\top}\bar{\Sigma}_{y}Q_{-1}=Q_{-1}^{\top}\lambda \Sigma_{\eta}%
\lambda^{\top}Q_{-1}+\bar{\sigma}_{u}^{2}\mathbf{I}_{n-1}.
\label{P_Suff_Cond1_iii_3}%
\end{equation}

Since $\mathbf{1}_{n}^{\top}\tilde{\lambda}=0$, we have $\tilde{\lambda}%
\in \operatorname{col}(Q_{-1})$, and hence
\[
\mathrm{rank}(Q_{-1}^{\top}\tilde{\lambda})=\mathrm{rank}(\tilde{\lambda}).
\]
Noting that $Q_{-1}^{\top}\tilde{\lambda}=Q_{-1}^{\top}M_{\mathbf{1}_{n}%
}\lambda=Q_{-1}^{\top}\lambda$, it follows that
\begin{equation}
\mathrm{rank}(Q_{-1}^{\top}\lambda)=\mathrm{rank}(\tilde{\lambda})=\bar{r}-1.
\label{P_Suff_Cond1_iii_4}%
\end{equation}
Because $\rho_{\min}(\Sigma_{\eta})\geq K^{-1}$, (\ref{P_Suff_Cond1_iii_4})
implies that $Q_{-1}^{\top}\lambda \Sigma_{\eta}\lambda^{\top}Q_{-1}$ has
$n-\bar{r}$ zero eigenvalues and $\bar{r}-1$ strictly positive eigenvalues.
From (\ref{P_Suff_Cond1_iii_3}), the $(n-\bar{r}+1)$th eigenvalue of
$Q_{-1}^{\top}\bar{\Sigma}_{y}Q_{-1}$, denoted by $\mu_{n-\bar{r}+1}$,
satisfies
\[
\mu_{n-\bar{r}+1}=\rho_{\min}^{+}\! \left(  Q_{-1}^{\top}\lambda \Sigma_{\eta
}\lambda^{\top}Q_{-1}\right)  +\bar{\sigma}_{u}^{2},
\]
where $\rho_{\min}^{+}(Q_{-1}^{\top}\lambda \Sigma_{\eta}\lambda^{\top}Q_{-1})$
denotes the smallest positive eigenvalue of\ $Q_{-1}^{\top}\lambda \Sigma
_{\eta}\lambda^{\top}Q_{-1}$. Therefore, Assumption \ref{Asy_Cond_1}(iii)
holds provided that
\begin{equation}
\rho_{\min}^{+}\! \left(  Q_{-1}^{\top}\lambda \Sigma_{\eta}\lambda^{\top
}Q_{-1}\right)  \geq K^{-1}. \label{P_Suff_Cond1_iii_5}%
\end{equation}

For any $x\in \mathbb{R}^{n-1}$,
\[
x^{\top}Q_{-1}^{\top}\lambda \Sigma_{\eta}\lambda^{\top}Q_{-1}x\geq \rho_{\min
}(\Sigma_{\eta})\,x^{\top}Q_{-1}^{\top}\lambda \lambda^{\top}Q_{-1}x.
\]
Taking the minimum over all unit vectors $x\in \operatorname{col}(Q_{-1}^{\top
}\lambda)$ yields
\begin{equation}
\rho_{\min}^{+}\! \left(  Q_{-1}^{\top}\lambda \Sigma_{\eta}\lambda^{\top
}Q_{-1}\right)  \geq \rho_{\min}(\Sigma_{\eta})\, \rho_{\min}^{+}\! \left(
Q_{-1}^{\top}\lambda \lambda^{\top}Q_{-1}\right)  , \label{P_Suff_Cond1_iii_6}%
\end{equation}
where $\rho_{\min}^{+}(Q_{-1}^{\top}\lambda \lambda^{\top}Q_{-1})$ denotes the
smallest positive eigenvalue of\ $Q_{-1}^{\top}\lambda \lambda^{\top}Q_{-1}$.
Since $\rho_{\min}(\Sigma_{\eta})\geq K^{-1}$, it suffices to show that
\begin{equation}
\rho_{\min}^{+}\! \left(  Q_{-1}^{\top}\lambda \lambda^{\top}Q_{-1}\right)
\geq K^{-1}. \label{P_Suff_Cond1_iii_7}%
\end{equation}
We next verify (\ref{P_Suff_Cond1_iii_7}) in two separate cases.

\medskip \noindent \textit{Case (i).} Suppose $\rho_{\min}((\mathbf{1}%
_{n},\lambda)^{\top}(\mathbf{1}_{n},\lambda))\geq K^{-1}$. By
(\ref{P_Suff_Cond1_iii_1}),
\[
\lambda^{\top}Q_{-1}Q_{-1}^{\top}\lambda=\lambda^{\top}M_{\mathbf{1}_{n}%
}\lambda.
\]
By the block inverse formula, $(\lambda^{\top}M_{\mathbf{1}_{n}}\lambda)^{-1}$
is the lower-right $r\times r$ block of $((\mathbf{1}_{n},\lambda)^{\top
}(\mathbf{1}_{n},\lambda))^{-1}$. Therefore,
\[
\rho_{\min}(\lambda^{\top}M_{\mathbf{1}_{n}}\lambda)\geq \rho_{\min
}((\mathbf{1}_{n},\lambda)^{\top}(\mathbf{1}_{n},\lambda))\geq K^{-1},
\]
which implies (\ref{P_Suff_Cond1_iii_7}).

\medskip \noindent \textit{Case (ii).} Suppose $\mathbf{1}_{n}\in
\operatorname{col}(\lambda)$ and $\rho_{\min}(\lambda^{\top}\lambda)\geq
K^{-1}$. Define $\alpha_{1}\equiv(\lambda^{\top}\lambda)^{-1}\lambda^{\top
}\mathbf{1}_{n}$ and $\mathcal{A}_{1}\equiv \{x\in \mathbb{R}^{r}:\mathbf{1}%
_{n}^{\top}\lambda x=0\}$. Since $\mathbf{1}_{n}\in \operatorname{col}%
(\lambda)$, we have $\lambda(\lambda^{\top}\lambda)^{-1}\lambda^{\top
}\mathbf{1}_{n}=\mathbf{1}_{n}$ and%
\begin{equation}
(\mathbf{1}_{n}^{\top}\lambda)\alpha_{1}=\mathbf{1}_{n}^{\top}\lambda
(\lambda^{\top}\lambda)^{-1}\lambda^{\top}\mathbf{1}_{n}=\mathbf{1}_{n}^{\top
}\mathbf{1}_{n}=n, \label{P_Suff_Cond1_iii_8a}%
\end{equation}
which implies that $\mathbf{1}_{n}^{\top}\lambda \neq0$ and $\alpha_{1}%
\notin \mathcal{A}_{1}$. Therefore, $\mathcal{A}_{1}$ is a linear subspace of
dimension $r-1$. Let $U\in \mathbb{R}^{r\times(r-1)}$ be a matrix whose columns
form an orthonormal basis of $\mathcal{A}_{1}$. Then the matrix $C_{\lambda
}\equiv(\alpha_{1},U)$ is invertible and it satisfies%
\begin{equation}
\lambda C_{\lambda}=(\mathbf{1}_{n},\lambda_{\perp}),\qquad \mathbf{1}%
_{n}^{\top}\lambda_{\perp}=0\text{,} \label{P_Suff_Cond1_iii_8}%
\end{equation}
\ where $\lambda_{\perp}\equiv \lambda U$. Since the columns of $U$ are
orthonormal,%
\begin{equation}
C_{\lambda}^{\top}C_{\lambda}=\left(
\begin{array}
[c]{cc}%
\Vert \alpha_{1}\Vert^{2} & \alpha_{1}^{\top}U\\
U^{\top}\alpha_{1} & \mathbf{I}_{r-1}%
\end{array}
\right)  \text{, \  \  \  \ where }c_{1}\equiv U^{\top}\alpha_{1}.
\label{P_Suff_Cond1_iii_9}%
\end{equation}
Using the block inverse formula, we obtain%
\[
(C_{\lambda}^{\top}C_{\lambda})^{-1}=\left(
\begin{array}
[c]{cc}%
\delta_{1}^{-2} & -\delta_{1}^{-2}c_{1}^{\top}\\
-\delta_{1}^{-2}c_{1} & \mathbf{I}_{r-1}+\delta_{1}^{-2}c_{1}c_{1}^{\top}%
\end{array}
\right)  ,
\]
where $\delta_{1}\equiv(\Vert \alpha_{1}\Vert^{2}-\Vert c_{1}\Vert^{2})^{1/2}$.
Since $U^{\top}U=\mathbf{I}_{r-1}$,
\[
\Vert c_{1}\Vert^{2}=\Vert U^{\top}\alpha_{1}\Vert^{2}=\Vert UU^{\top}%
\alpha_{1}\Vert^{2}.
\]
Here $UU^{\top}\alpha_{1}=U(U^{\top}U)^{-1}U^{\top}\alpha_{1}$ is the
projection of $\alpha_{1}$ onto $\mathcal{A}_{1}$, which is the subspace
orthogonal to $\lambda^{\top}\mathbf{1}_{n}$. By the Pythagorean theorem,
\[
\delta_{1}^{2}=\Vert \alpha_{1}\Vert^{2}-\Vert UU^{\top}\alpha_{1}\Vert^{2}%
\]
is the squared length of the projection of $\alpha_{1}$ onto the subspace
spanned by $\lambda^{\top}\mathbf{1}_{n}$. Therefore,%
\begin{equation}
\delta_{1}^{2}=\alpha_{1}^{\top}\lambda^{\top}\mathbf{1}_{n}(\mathbf{1}%
_{n}^{\top}\lambda \lambda^{\top}\mathbf{1}_{n})^{-1}\mathbf{1}_{n}^{\top
}\lambda \alpha_{1}=\frac{(\mathbf{1}_{n}^{\top}\lambda(\lambda^{\top}%
\lambda)^{-1}\lambda^{\top}\mathbf{1}_{n})^{2}}{\mathbf{1}_{n}^{\top}%
\lambda \lambda^{\top}\mathbf{1}_{n}}=\frac{n^{2}}{\mathbf{1}_{n}^{\top}%
\lambda \lambda^{\top}\mathbf{1}_{n}}. \label{P_Suff_Cond1_iii_9b}%
\end{equation}
Here the last equality is by (\ref{P_Suff_Cond1_iii_8a}). Using
(\ref{P_Suff_Cond1_iii_8}) and (\ref{P_Suff_Cond1_iii_9}), we have
\begin{align*}
Q_{-1}^{\top}\lambda \lambda^{\top}Q_{-1}  &  =Q_{-1}^{\top}\lambda C_{\lambda
}(C_{\lambda}^{\top}C_{\lambda})^{-1}C_{\lambda}^{\top}\lambda^{\top}Q_{-1}\\
&  =(\mathbf{0}_{n},Q_{-1}^{\top}\lambda_{\perp})\left(
\begin{array}
[c]{cc}%
\delta_{1}^{-2} & -\delta_{1}^{-2}c_{1}^{\top}\\
-\delta_{1}^{-2}c_{1} & \mathbf{I}_{r-1}+\delta_{1}^{-2}c_{1}c_{1}^{\top}%
\end{array}
\right)  (\mathbf{0}_{n},Q_{-1}^{\top}\lambda_{\perp})^{\top}\\
&  =Q_{-1}^{\top}\lambda_{\perp}(\mathbf{I}_{r-1}+\delta_{1}^{-2}c_{1}%
c_{1}^{\top})\lambda_{\perp}^{\top}Q_{-1}.
\end{align*}
It then follows that
\begin{equation}
\rho_{\min}^{+}\! \left(  Q_{-1}^{\top}\lambda \lambda^{\top}Q_{-1}\right)
\geq \rho_{\min}^{+}\! \left(  Q_{-1}^{\top}\lambda_{\perp}\lambda_{\perp
}^{\top}Q_{-1}\right)  \rho_{\min}(\mathbf{I}_{r-1}+\delta_{1}^{-2}c_{1}%
c_{1}^{\top})\geq \rho_{\min}^{+}\! \left(  Q_{-1}^{\top}\lambda_{\perp}%
\lambda_{\perp}^{\top}Q_{-1}\right)  . \label{P_Suff_Cond1_iii_10}%
\end{equation}
Since the positive eigenvalues of $Q_{-1}^{\top}\lambda_{\perp}\lambda_{\perp
}^{\top}Q_{-1}$ are the same as those of $\lambda_{\perp}^{\top}Q_{-1}%
Q_{-1}^{\top}\lambda_{\perp}$, and the latter satisfies%
\[
\lambda_{\perp}^{\top}Q_{-1}Q_{-1}^{\top}\lambda_{\perp}=\lambda_{\perp}%
^{\top}M_{\mathbf{1}_{n}}\lambda_{\perp}=\lambda_{\perp}^{\top}\lambda_{\perp}%
\]
by (\ref{P_Suff_Cond1_iii_1}), it follows that the (\ref{P_Suff_Cond1_iii_7})
holds if
\begin{equation}
\rho_{\min}(\lambda_{\perp}^{\top}\lambda_{\perp})\geq K^{-1}.
\label{P_Suff_Cond1_iii_11}%
\end{equation}
We next verify (\ref{P_Suff_Cond1_iii_11}).

From the first equation in (\ref{P_Suff_Cond1_iii_8}), we have
\[
C_{\lambda}^{\top}\lambda^{\top}\lambda C_{\lambda}=(\lambda C_{\lambda
})^{\top}(\lambda C_{\lambda})=%
\begin{pmatrix}
n & 0\\
0 & \lambda_{\perp}^{\top}\lambda_{\perp}%
\end{pmatrix}
,
\]
which implies that
\begin{equation}
\rho_{\min}(\lambda_{\perp}^{\top}\lambda_{\perp})\geq \rho_{\min}(C_{\lambda
}^{\top}\lambda^{\top}\lambda C_{\lambda})\geq \rho_{\min}(C_{\lambda}^{\top
}C_{\lambda})\rho_{\min}(\lambda^{\top}\lambda)\geq K^{-1}\rho_{\min
}(C_{\lambda}^{\top}C_{\lambda}). \label{P_Suff_Cond1_iii_12}%
\end{equation}
Let $x=(x_{1},x_{2}^{\top})^{\top}\in \mathbb{R}^{r}$, where $x_{1}%
\in \mathbb{R}$ and $x_{2}\in \mathbb{R}^{r-1}$. Then%
\begin{align*}
x^{\top}(C_{\lambda}^{\top}C_{\lambda})^{-1}x  &  =x_{1}^{2}\delta_{1}%
^{-2}-2\delta_{1}^{-2}x_{1}x_{2}^{\top}c_{1}+x_{2}^{\top}x_{2}+\delta_{1}%
^{-2}x_{2}^{\top}c_{1}c_{1}^{\top}x_{2}\\
&  \leq \delta_{1}^{-2}(|x_{1}|+\Vert x_{2}\Vert \Vert c_{1}\Vert)^{2}+\Vert
x_{2}\Vert^{2}\\
&  \leq \left(  1+\delta_{1}^{-2}(1+\Vert c_{1}\Vert)^{2}\right)  \Vert
x\Vert^{2},
\end{align*}
which together with $\Vert c_{1}\Vert \leq \Vert \alpha_{1}\Vert$ implies that
\begin{equation}
\rho_{\min}(C_{\lambda}^{\top}C_{\lambda})=(\rho_{\max}\! \left(  (C_{\lambda
}^{\top}C_{\lambda})^{-1}\right)  )^{-1}\geq \frac{1}{1+\delta_{1}^{-2}%
(1+\Vert \alpha_{1}\Vert)^{2}}. \label{P_Suff_Cond1_iii_13}%
\end{equation}
From (\ref{P_Suff_Cond1_iii_9b}) and $\Vert \alpha_{1}\Vert^{2}=\mathbf{1}%
_{n}^{\top}\lambda(\lambda^{\top}\lambda)^{-2}\lambda^{\top}\mathbf{1}_{n}$,
it follows that%
\begin{align}
\delta_{1}^{-2}(1+\Vert \alpha_{1}\Vert)^{2}  &  \leq2\delta_{1}^{-2}%
(1+\Vert \alpha_{1}\Vert^{2})=\frac{2(\mathbf{1}_{n}^{\top}\lambda \lambda
^{\top}\mathbf{1}_{n})(1+\mathbf{1}_{n}^{\top}\lambda(\lambda^{\top}%
\lambda)^{-2}\lambda^{\top}\mathbf{1}_{n})}{n^{2}}\nonumber \\
&  \leq2\rho_{\max}(\lambda^{\top}\lambda)\frac{\mathbf{1}_{n}^{\top}%
\lambda(\lambda^{\top}\lambda)^{-1}\lambda^{\top}\mathbf{1}_{n}}{n^{2}}\left(
1+\frac{\mathbf{1}_{n}^{\top}\lambda(\lambda^{\top}\lambda)^{-1}\lambda^{\top
}\mathbf{1}_{n}}{\rho_{\min}(\lambda^{\top}\lambda)}\right) \nonumber \\
&  =\frac{2\rho_{\max}(\lambda^{\top}\lambda)}{\rho_{\min}(\lambda^{\top
}\lambda)}\left(  1+\frac{\rho_{\min}(\lambda^{\top}\lambda)}{n}\right)  \leq
K, \label{P_Suff_Cond1_iii_14}%
\end{align}
where the second inequality follows from
\[
\mathbf{1}_{n}^{\top}\lambda \lambda^{\top}\mathbf{1}_{n}=\mathbf{1}_{n}^{\top
}\lambda(\lambda^{\top}\lambda)^{-1/2}(\lambda^{\top}\lambda)(\lambda^{\top
}\lambda)^{-1/2}\lambda^{\top}\mathbf{1}_{n}\leq \rho_{\max}(\lambda^{\top
}\lambda)\mathbf{1}_{n}^{\top}\lambda(\lambda^{\top}\lambda)^{-1}\lambda
^{\top}\mathbf{1}_{n}%
\]
and
\begin{align*}
\mathbf{1}_{n}^{\top}\lambda(\lambda^{\top}\lambda)^{-2}\lambda^{\top
}\mathbf{1}_{n}  &  =\mathbf{1}_{n}^{\top}\lambda(\lambda^{\top}%
\lambda)^{-1/2}(\lambda^{\top}\lambda)^{-1}(\lambda^{\top}\lambda
)^{-1/2}\lambda^{\top}\mathbf{1}_{n}\\
&  \leq \rho_{\max}((\lambda^{\top}\lambda)^{-1})\mathbf{1}_{n}^{\top}%
\lambda(\lambda^{\top}\lambda)^{-1}\lambda^{\top}\mathbf{1}_{n}=\frac
{\mathbf{1}_{n}^{\top}\lambda(\lambda^{\top}\lambda)^{-1}\lambda^{\top
}\mathbf{1}_{n}}{\rho_{\min}(\lambda^{\top}\lambda)},
\end{align*}
and the second equality follows from (\ref{P_Suff_Cond1_iii_8a}). The desired
result in (\ref{P_Suff_Cond1_iii_11}) then follows from
(\ref{P_Suff_Cond1_iii_12}), (\ref{P_Suff_Cond1_iii_13}) and
(\ref{P_Suff_Cond1_iii_14}).\hfill$Q.E.D.$

\bigskip

\begin{lemma}
\label{Suff_Cond2_iii}\ Suppose that $\left \vert \psi \right \vert \geq K^{-1}$,
$0<\left \vert 1-\phi \psi \right \vert \leq K$, $S_{t}$ is independent of $u_{t}$
and $\eta_{t}$, and either (i)\ $\rho_{\min}((\mathbf{1}_{n},\lambda)^{\top
}(\mathbf{1}_{n},\lambda))>0$ and $S_{u}^{\top}M_{(\mathbf{1}_{n},\lambda
)}S_{u}\geq K^{-1}$; or (ii) $\mathbf{1}_{n}\in \operatorname{col}(\lambda)$,
$\rho_{\min}(\lambda^{\top}\lambda)>0$ and $S_{u}^{\top}M_{\lambda}S_{u}\geq
K^{-1}$, where $S_{u}\equiv T^{-1}\sum_{t\leq T}\sigma_{u,t}^{2}%
\mathbb{E}[S_{t}]$. Then Assumption \ref{Asy_Cond_2}(iii) holds.
\end{lemma}

\noindent \textsc{Proof of Lemma\  \ref{Suff_Cond2_iii}.} Since $A\equiv
Q_{-1}A_{0}$ spans the same subspace as $\bar{\lambda}_{\bot}$, which is
orthogonal to $(\mathbf{1}_{n},\lambda)$, we have
\begin{equation}
A^{\top}\mathbf{1}_{n}=\mathbf{0}_{n-\bar{r}}\text{ \  \  \ and \  \  \ }A^{\top
}\lambda=\mathbf{0}_{(n-\bar{r})\times r}. \label{P_Suff_Cond2_iii_1}%
\end{equation}
This, together with the reduced form expression of $y_{t}$ in
(\ref{P_Asy_L0_1}), implies that
\begin{equation}
A^{\top}y_{t}=A^{\top}(u_{t}+\lambda \eta_{t})=A^{\top}u_{t}.
\label{P_Suff_Cond2_iii_2}%
\end{equation}
Using (\ref{P_Suff_Cond2_iii_1}), the reduced-form expression of $p_{t}$ in
(\ref{P_Asy_L0_2}), the expression for $A^{\top}y_{t}$ in
(\ref{P_Suff_Cond2_iii_2}), and Assumption \ref{ID}, we obtain%
\[
\mathbb{E}[A^{\top}y_{t}p_{t}]=\frac{\psi}{1-\phi \psi}A^{\top}\mathbb{E}%
[u_{t}(u_{t}+\lambda \eta_{t})^{\top}S_{t}]+\frac{1}{1-\phi \psi}A^{\top
}\mathbb{E}[u_{t}\varepsilon_{t}]=\frac{\psi \sigma_{u,t}^{2}}{1-\phi \psi
}A^{\top}\mathbb{E}[S_{t}].
\]
Therefore,
\begin{equation}
\Big \lVert T^{-1}\sum_{t\leq T}\mathbb{E}[A^{\top}y_{t}p_{t}]\Big \rVert^{2}%
=\frac{\psi^{2}}{(1-\phi \psi)^{2}}||A^{\top}S_{u}||^{2}.
\label{P_Suff_Cond2_iii_3}%
\end{equation}

Next by the reduced form expression of $y_{t}$ in (\ref{P_Asy_L0_1})%
\[
y_{S,t}\equiv S_{t}^{\top}y_{t}=\frac{1}{1-\phi \psi}S_{t}^{\top}(u_{t}%
+\lambda \eta_{t})+\frac{\phi}{1-\phi \psi}\varepsilon_{t},
\]
which, together with Assumption \ref{ID}, (\ref{P_Suff_Cond2_iii_1}) and
(\ref{P_Suff_Cond2_iii_2}), implies that
\[
\mathbb{E}[A^{\top}y_{t}y_{S,t}]=\frac{1}{1-\phi \psi}A^{\top}\mathbb{E}%
[u_{t}(u_{t}+\lambda \eta_{t})^{\top}S_{t}]+\frac{\phi}{1-\phi \psi}A^{\top
}\mathbb{E}[u_{t}\varepsilon_{t}]=\frac{\sigma_{u,t}^{2}}{1-\phi \psi}A^{\top
}\mathbb{E}[S_{t}].
\]
Therefore,
\begin{equation}
\Big \lVert T^{-1}\sum_{t\leq T}\mathbb{E}[A^{\top}y_{t}y_{S,t}%
]\Big \rVert^{2}=\frac{1}{(1-\phi \psi)^{2}}||A^{\top}S_{u}||^{2}.
\label{P_Suff_Cond2_iii_4}%
\end{equation}

Since $A^{\top}A=\mathbf{I}_{n-\bar{r}}$, we have $\Vert A^{\top}S_{u}%
\Vert=\Vert AA^{\top}S_{u}\Vert$, so $\Vert A^{\top}S_{u}\Vert$ is the length
of the projection of $S_{u}$ onto $\operatorname{col}(A)$. Because
$\operatorname{col}(A)$ is orthogonal to\ $\operatorname{col}(\mathbf{1}%
_{n},\lambda)$, it follows that
\begin{equation}
||A^{\top}S_{u}||^{2}=\left \{
\begin{array}
[c]{cc}%
S_{u}^{\top}M_{(\mathbf{1}_{n},\lambda)}S_{u}, & \text{if }\rho_{\min
}((\mathbf{1}_{n},\lambda)^{\top}(\mathbf{1}_{n},\lambda))>0\\
S_{u}^{\top}M_{\lambda}S_{u}, & \text{if }\mathbf{1}_{n}\in \operatorname{col}%
(\lambda)\text{ and }\rho_{\min}(\lambda^{\top}\lambda)>0
\end{array}
\right.  . \label{P_Suff_Cond2_iii_5}%
\end{equation}
Combining (\ref{P_Suff_Cond2_iii_3}), (\ref{P_Suff_Cond2_iii_4}) and
(\ref{P_Suff_Cond2_iii_5}), as well as the maintained conditions of the lemma
gives%
\[
\Big \lVert T^{-1}\sum_{t\leq T}\mathbb{E}[A^{\top}y_{t}p_{t}]\Big \rVert \geq
K^{-1}\text{ \  \  \ and \  \  \ }\Big \lVert T^{-1}\sum_{t\leq T}\mathbb{E}%
[A^{\top}y_{t}y_{S,t}]\Big \rVert \geq K^{-1}%
\]
This verifies Assumption \ref{Asy_Cond_2}(iii).\hfill$Q.E.D.$

\bigskip

\begin{lemma}
\label{Invariance}\ For any $n\times(n-\bar{r})$ full rank matrix $A_{1}$, any
$2n\times2n$ positive definite matrix $W_{1}$, and any nonsingular $(n-\bar
{r})\times(n-\bar{r})$ matrix $C_{1}$, we have
\[
J_{T}(\theta;A_{1},W_{1})=J_{T}(\theta;A_{1}C_{1},W_{1}),
\]
where
\begin{equation}
J_{T}(\theta;A_{1},W_{1})\equiv \bar{g}_{T}(\theta;A_{1})^{\top}((\mathbf{I}%
_{2}\otimes A_{1}^{\top})W_{1}(\mathbf{I}_{2}\otimes A_{1}))^{-1}\bar{g}%
_{T}(\theta;A_{1}), \label{J_T}%
\end{equation}
and $J_{T}(\theta;A_{1}C_{1},W_{1})$ is defined analogously to $J_{T}%
(\theta;A_{1},W_{1})$ with $A_{1}$ replaced by $A_{1}C_{1}$.
\end{lemma}

\noindent \textsc{Proof of Lemma\  \ref{Invariance}.} By definition of the GMM
criterion,
\begin{equation}
J_{T}(\theta;A_{1}C_{1},W_{1})=\bar{g}_{0,T}(\theta)^{\top}(\mathbf{I}%
_{2}\otimes A_{1}C_{1})((\mathbf{I}_{2}\otimes C_{1}^{\top}A_{1}^{\top}%
)W_{1}(\mathbf{I}_{2}\otimes A_{1}C_{1}))^{-1}(\mathbf{I}_{2}\otimes
C_{1}^{\top}A_{1}^{\top})\bar{g}_{0,T}(\theta) \label{P_Invariance_1}%
\end{equation}
where
\[
\bar{g}_{0,T}(\theta)\equiv%
\begin{pmatrix}
T^{-1}\sum_{t\leq T}y_{t}(y_{e,t}-\phi p_{t})\\
T^{-1}\sum_{t\leq T}y_{t}(p_{t}-\psi y_{S,t})
\end{pmatrix}
.
\]
Using the identity
\[
(\mathbf{I}_{2}\otimes C_{1}^{\top}A_{1}^{\top})=(\mathbf{I}_{2}\otimes
C_{1}^{\top})(\mathbf{I}_{2}\otimes A_{1}^{\top})\  \text{and }(\mathbf{I}%
_{2}\otimes A_{1}C_{1})=(\mathbf{I}_{2}\otimes A_{1})(\mathbf{I}_{2}\otimes
C_{1}),
\]
we obtain
\begin{equation}
((\mathbf{I}_{2}\otimes C_{1}^{\top}A_{1}^{\top})W_{1}(\mathbf{I}_{2}\otimes
A_{1}C_{1}))^{-1}=(\mathbf{I}_{2}\otimes C_{1})^{-1}((\mathbf{I}_{2}\otimes
A_{1}^{\top})W_{1}(\mathbf{I}_{2}\otimes A_{1}))^{-1}(\mathbf{I}_{2}\otimes
C_{1}^{\top})^{-1}. \label{P_Invariance_2}%
\end{equation}
Substituting ({\ref{P_Invariance_2}}) into ({\ref{P_Invariance_1}}), we get%
\begin{align*}
J_{T}(\theta;A_{1}C_{1},W_{1})  &  =\bar{g}_{0,T}(\theta)^{\top}%
(\mathbf{I}_{2}\otimes A_{1})(\mathbf{I}_{2}\otimes C_{1})\\
&  \times(\mathbf{I}_{2}\otimes C_{1})^{-1}((\mathbf{I}_{2}\otimes A_{1}%
^{\top})W_{1}(\mathbf{I}_{2}\otimes A_{1}))^{-1}(\mathbf{I}_{2}\otimes
C_{1}^{\top})^{-1}\\
&  \times(\mathbf{I}_{2}\otimes C_{1}^{\top})(\mathbf{I}_{2}\otimes
A_{1}^{\top})\bar{g}_{0,T}(\theta)\\
&  =\bar{g}_{0,T}(\theta)^{\top}(\mathbf{I}_{2}\otimes A_{1})((\mathbf{I}%
_{2}\otimes A_{1}^{\top})W_{1}(\mathbf{I}_{2}\otimes A_{1}))^{-1}%
(\mathbf{I}_{2}\otimes A_{1}^{\top})\bar{g}_{0,T}(\theta)\\
&  =J_{T}(\theta;A_{1},W_{1}),
\end{align*}
which establishes the claim of the lemma.\hfill$Q.E.D.$

\bigskip

\begin{lemma}
\label{Asy_L0} Under Assumptions {\ref{Asy_Cond_1}(i, ii) and \ref{S}, we
have:}%
\[
T^{-1/2}\sum_{t\leq T}(y_{t}y_{t}^{\top}-{\mathbb{E}[}y_{t}y_{t}^{\top
}])=O_{p}(1)\text{ \ and \ }T^{-1/2}\sum_{t\leq T}(p_{t}y_{t}-{\mathbb{E}%
[}p_{t}y_{t}])=O_{p}(1).
\]

\end{lemma}

\noindent \textsc{Proof of Lemma \ref{Asy_L0}.} From (\ref{S_supply}) and
(\ref{G_demand}), we obtain%
\begin{align}
y_{t}  &  =\left(  \mathbf{I}_{n}+\frac{\phi \psi}{1-\phi \psi}\mathbf{1}%
_{n}S_{t}^{\top}\right)  (u_{t}+\lambda \eta_{t})+\frac{\phi}{1-\phi \psi
}\mathbf{1}_{n}\varepsilon_{t},\label{P_Asy_L0_1}\\
p_{t}  &  =\frac{\psi}{1-\phi \psi}S_{t}^{\top}(u_{t}+\lambda \eta_{t}%
)+\frac{\varepsilon_{t}}{1-\phi \psi}. \label{P_Asy_L0_2}%
\end{align}
Therefore,%
\begin{align}
T^{-1}\sum_{t\leq T}y_{t}y_{t}^{\top}  &  =T^{-1}\sum_{t\leq T}\left(
\mathbf{I}_{n}+\frac{\phi \psi}{1-\phi \psi}\mathbf{1}_{n}S_{t}^{\top}\right)
(u_{t}+\lambda \eta_{t})(u_{t}+\lambda \eta_{t})^{\top}\left(  \mathbf{I}%
_{n}+\frac{\phi \psi}{1-\phi \psi}S_{t}\mathbf{1}_{n}^{\top}\right) \nonumber \\
&  \quad+\frac{\phi}{1-\phi \psi}T^{-1}\sum_{t\leq T}\left(  \mathbf{I}%
_{n}+\frac{\phi \psi}{1-\phi \psi}\mathbf{1}_{n}S_{t}^{\top}\right)
(u_{t}+\lambda \eta_{t})\varepsilon_{t}\mathbf{1}_{n}^{\top}\nonumber \\
&  \quad+\frac{\phi}{1-\phi \psi}\mathbf{1}_{n}T^{-1}\sum_{t\leq T}%
\varepsilon_{t}(u_{t}+\lambda \eta_{t})^{\top}\left(  \mathbf{I}_{n}+\frac
{\phi \psi}{1-\phi \psi}S_{t}\mathbf{1}_{n}^{\top}\right) \nonumber \\
&  \quad+\frac{\phi^{2}}{(1-\phi \psi)^{2}}\mathbf{1}_{n}\mathbf{1}_{n}^{\top
}T^{-1}\sum_{t\leq T}\varepsilon_{t}^{2}, \label{P_Asy_L0_3}%
\end{align}
and%
\begin{align}
T^{-1}\sum_{t\leq T}p_{t}y_{t}  &  =\frac{\psi}{1-\phi \psi}T^{-1}\sum_{t\leq
T}\left(  \mathbf{I}_{n}+\frac{\phi \psi}{1-\phi \psi}\mathbf{1}_{n}S_{t}^{\top
}\right)  (u_{t}+\lambda \eta_{t})(u_{t}+\lambda \eta_{t})^{\top}S_{t}%
\nonumber \\
&  \quad+\frac{1}{1-\phi \psi}T^{-1}\sum_{t\leq T}\left(  \mathbf{I}_{n}%
+\frac{\phi \psi}{1-\phi \psi}\mathbf{1}_{n}S_{t}^{\top}\right)  (u_{t}%
+\lambda \eta_{t})\varepsilon_{t}\nonumber \\
&  \quad+\frac{\phi \psi}{(1-\phi \psi)^{2}}\mathbf{1}_{n}T^{-1}\sum_{t\leq
T}\varepsilon_{t}(u_{t}+\lambda \eta_{t})^{\top}S_{t}+\frac{\phi}{(1-\phi
\psi)^{2}}\mathbf{1}_{n}T^{-1}\sum_{t\leq T}\varepsilon_{t}^{2}.
\label{P_Asy_L0_4}%
\end{align}
From Assumptions {\ref{Asy_Cond_1}(i, ii), we have }%
\begin{align}
T^{-1}\sum_{t\leq T}(u_{t}+\lambda \eta_{t})(u_{t}+\lambda \eta_{t})^{\top}  &
=T^{-1}\sum_{t\leq T}{\mathbb{E}[}(u_{t}+\lambda \eta_{t})(u_{t}+\lambda
\eta_{t})^{\top}]+O_{p}(T^{-1/2}),\label{P_Asy_L0_5}\\
T^{-1}\sum_{t\leq T}(u_{t}+\lambda \eta_{t})\varepsilon_{t}  &  =T^{-1}%
\sum_{t\leq T}{\mathbb{E}[}(u_{t}+\lambda \eta_{t})\varepsilon_{t}%
]+O_{p}(T^{-1/2}),\label{P_Asy_L0_6}\\
T^{-1}\sum_{t\leq T}\varepsilon_{t}^{2}  &  =T^{-1}\sum_{t\leq T}{\mathbb{E}%
[}\varepsilon_{t}^{2}]+O_{p}(T^{-1/2}). \label{P_Asy_L0_7}%
\end{align}
Similarly, from\ Assumptions {\ref{Asy_Cond_1}(i) and \ref{S}(i), it follows
that}%
\begin{align}
T^{-1}\sum_{t\leq T}S_{t}^{\top}(u_{t}+\lambda \eta_{t})(u_{t}+\lambda \eta
_{t})^{\top}S_{t}  &  =T^{-1}\sum_{t\leq T}{\mathbb{E}[}S_{t}^{\top}%
(u_{t}+\lambda \eta_{t})(u_{t}+\lambda \eta_{t})^{\top}S_{t}]+O_{p}%
(T^{-1/2}),\label{P_Asy_L0_8}\\
T^{-1}\sum_{t\leq T}S_{t}^{\top}(u_{t}+\lambda \eta_{t})(u_{t}+\lambda \eta
_{t})^{\top}  &  =T^{-1}\sum_{t\leq T}{\mathbb{E}[}S_{t}^{\top}(u_{t}%
+\lambda \eta_{t})(u_{t}+\lambda \eta_{t})^{\top}]+O_{p}(T^{-1/2}%
),\label{P_Asy_L0_9}\\
T^{-1}\sum_{t\leq T}S_{t}^{\top}(u_{t}+\lambda \eta_{t})\varepsilon_{t}  &
=T^{-1}\sum_{t\leq T}{\mathbb{E}[}S_{t}^{\top}(u_{t}+\lambda \eta
_{t})\varepsilon_{t}]+O_{p}(T^{-1/2}). \label{P_Asy_L0_10}%
\end{align}
The claim of the lemma follows from Assumption {\ref{Asy_Cond_1}(i), and
}({\ref{P_Asy_L0_3})-}({\ref{P_Asy_L0_10}).}\hfill$Q.E.D.$

\bigskip

\begin{lemma}
\label{Asy_L0b}\ Under Assumptions {\ref{Asy_Cond_1}(i, ii) and \ref{S}, we
have:}%
\[
T^{-1/2}\sum_{t\leq T}(y_{S,t}y_{t}-{\mathbb{E}[}y_{S,t}y_{t}])=O_{p}(1).
\]

\end{lemma}

\noindent \textsc{Proof of Lemma \ref{Asy_L0b}.} Since\ $\mathbf{1}_{n}^{\top
}S_{t}=1$, from ({\ref{P_Asy_L0_1}}) we have
\[
y_{S,t}=\frac{1}{1-\phi \psi}S_{t}^{\top}(u_{t}+\lambda \eta_{t})+\frac{\phi
}{1-\phi \psi}\varepsilon_{t}.
\]
Therefore,
\begin{align*}
y_{S,t}y_{t}  &  =\frac{1}{1-\phi \psi}\left(  \mathbf{I}_{n}+\frac{\phi \psi
}{1-\phi \psi}\mathbf{1}_{n}S_{t}^{\top}\right)  (u_{t}+\lambda \eta_{t}%
)(u_{t}+\lambda \eta_{t})^{\top}S_{t}\\
&  +\frac{\phi}{1-\phi \psi}\left(  \mathbf{I}_{n}+\frac{\phi \psi}{1-\phi \psi
}\mathbf{1}_{n}S_{t}^{\top}\right)  (u_{t}+\lambda \eta_{t})\varepsilon_{t}\\
&  +\frac{\phi}{(1-\phi \psi)^{2}}S_{t}^{\top}(u_{t}+\lambda \eta_{t}%
)\varepsilon_{t}\mathbf{1}_{n}+\frac{\phi^{2}}{(1-\phi \psi)^{2}}\mathbf{1}%
_{n}\varepsilon_{t}^{2}.
\end{align*}
The claim of the lemma thus follows from ({\ref{P_Asy_L0_6}}%
)-({\ref{P_Asy_L0_10}}).\hfill$Q.E.D.$

\bigskip

The GIV estimator depends on $\hat{A}\equiv Q_{-1}\hat{A}_{0}$, which serves
as an estimator of $Q_{-1}A_{0}$. Since $A_{0}$ collects the eigenvectors
associated with the smallest $n-\bar{r}$ eigenvalues of $S_{y}$, it is not
uniquely identified when $n-\bar{r}>1$. Therefore, it is generally unrealistic
to expect that $\hat{A}_{0}$ (and hence $\hat{A}$) converges to $A_{0}$ as
$T\rightarrow \infty$. Instead, the lemma below shows that a suitably rotated
version of $\hat{A}_{0}$ is a consistent estimator of $A_{0}$.

\begin{lemma}
\label{Asy_Ahat_L1} Define $K_{0}\equiv(A_{0,\bot}^{\top}\hat{A}_{0}%
)(A_{0}^{\top}\hat{A}_{0})^{-1}$ and
\begin{equation}
\hat{H}_{0}\equiv \left(  \mathbf{I}_{n-\bar{r}}+K_{0}^{\top}K_{0}\right)
^{-1/2}(A_{0}^{\top}\hat{A}_{0}+K_{0}^{\top}A_{0,\bot}^{\top}\hat{A}_{0}).
\label{Rotation_H}%
\end{equation}
Under Assumptions \ref{Asy_Cond_1} and {\ref{S}}, we have
\begin{equation}
\hat{A}_{0}\hat{H}_{0}^{\top}-A_{0}=A_{0,\bot}(\bar{\sigma}_{u}^{2}%
\mathbf{I}_{\bar{r}-1}-\Lambda_{\bot})^{-1}A_{0,\bot}^{\top}\hat{\Delta}%
_{y}A_{0}+O_{p}(T^{-1}), \label{Est_A_tilda}%
\end{equation}
where $\hat{\Delta}_{y}\equiv \hat{S}_{y}-S_{y}$ . Moreover, $\hat{H}_{0}$
satisfies $\hat{H}_{0}^{\top}\hat{H}_{0}=\mathbf{I}_{n-\bar{r}}$ wpa1.
\end{lemma}

\noindent \textsc{Proof of Lemma \ref{Asy_Ahat_L1}.} The proof follows directly
from Lemmas {\ref{Asy_L1} - \ref{Asy_L4}}.\hfill$Q.E.D.$

\bigskip

\begin{lemma}
\label{Asy_L1} Under Assumptions \ref{Asy_Cond_1} and {\ref{S}}, we have
\begin{equation}
||\hat{A}_{0}^{\top}A_{0}||_{o}\geq1-O_{p}(T^{-1/2})\text{ \  \ and
\  \ }\left \Vert K_{0}\right \Vert =O_{p}(T^{-1/2}), \label{Asy_L1_1}%
\end{equation}
where $\left \Vert \cdot \right \Vert _{o}$ denotes the matrix operator norm.
\end{lemma}

\noindent \textsc{Proof of Lemma \ref{Asy_L1}.} By Lemma {\ref{Asy_L0}, }%
\begin{equation}
\hat{\Delta}_{y}\equiv \hat{S}_{y}-S_{y}=O_{p}(T^{-1/2}). \label{S_y_Rate}%
\end{equation}
Recall that $\hat{A}_{0}$ collects the orthonormal eigenvectors associated
with the smallest $n-\bar{r}$ eigenvalues of $\hat{S}_{y}$. We next apply the
Davis--Kahan $\sin \Theta$ theorem by verifying the conditions of Theorem 2 in
\cite{yu2015useful}. Specifically, we aim to establish that
\begin{equation}
\big \lVert A_{0,\bot}^{\top}\hat{A}_{0}\big \rVert \leq \frac{2\min((n-\bar
{r})^{1/2}||\hat{\Delta}_{y}||_{o},||\hat{\Delta}_{y}||)}{\mu_{n-\bar{r}%
+1}-\bar{\sigma}_{u}^{2}}. \label{P_Asy_L1_1}%
\end{equation}
Let $\Theta(A_{0},\hat{A}_{0})$ be the $(n-\bar{r})\times(n-\bar{r})$ diagonal
matrix of principal angles between the true and estimated eigenspaces, with
$j$-th diagonal entry $\arccos(\sigma_{j})$, where $\sigma_{1}\geq \cdots
\geq \sigma_{n-\bar{r}}$ are the singular values of $A_{0}^{\top}\hat{A}_{0}$.
Define $\sin(\Theta(A_{0},\hat{A}_{0}))$ componentwise. Set the parameters of
Theorem 2 in \cite{yu2015useful} as%
\[
p=n-1\text{, \  \ }r=\bar{r}\text{, \ }s=n-1\text{ \ and \  \ }d=s-r+1=n-\bar
{r}.
\]
Let\ $\tilde{\mu}_{j}\equiv \mu_{n-j}$ for $j=1,\ldots,n-1$, and $\tilde{\mu
}_{n}=-\infty$, such that $\tilde{\mu}_{j}\geq \tilde{\mu}_{j+1}$. By
Assumption \ref{Asy_Cond_1}{(iii)},%
\begin{equation}
\min \{ \tilde{\mu}_{r-1}-\tilde{\mu}_{r},\tilde{\mu}_{s}-\tilde{\mu}%
_{s+1}\}=\min \{ \tilde{\mu}_{\bar{r}-1}-\tilde{\mu}_{\bar{r}},\tilde{\mu
}_{n-1}-\tilde{\mu}_{n}\}=\mu_{n-\bar{r}+1}-\mu_{n-\bar{r}}=\mu_{n-\bar{r}%
+1}-\bar{\sigma}_{u}^{2}>K^{-1}. \label{P_Asy_L1_1a}%
\end{equation}
Since\ $A_{0}$ and $\hat{A}_{0}$ consist of the eigenvectors corresponding to
the smallest $n-\bar{r}$ eigenvalues of $S_{y}$ and $\hat{S}_{y}$,
respectively, Theorem 2 of \cite{yu2015useful} yields
\begin{equation}
\big \lVert \sin(\Theta(A_{0},\hat{A}_{0}))\big \rVert \leq \frac{2\min \left \{
d^{1/2}||\hat{\Delta}_{y}||_{o},||\hat{\Delta}_{y}||\right \}  }{\mu_{n-\bar
{r}+1}-\bar{\sigma}_{u}^{2}}. \label{P_Asy_L1_1b}%
\end{equation}
Finally, because $||\sin(\Theta(A_{0},\hat{A}_{0}))||=||A_{0,\bot}^{\top}%
\hat{A}_{0}||$ (see the expression below (A4) in \cite{yu2015useful}), the
desired bound in ({\ref{P_Asy_L1_1}}) follows directly from
({\ref{P_Asy_L1_1b}}).

Using ({\ref{P_Asy_L1_1}}) along with\ ({\ref{S_y_Rate}}) and Assumption
\ref{Asy_Cond_1}{(iii)}, we obtain
\begin{equation}
\big \lVert A_{0,\bot}^{\top}\hat{A}_{0}\big \rVert=O_{p}(T^{-1/2}).
\label{P_Asy_L1_2}%
\end{equation}
Let $P_{0}\equiv(A_{0},A_{0,\bot})$. Since $P_{0}$ and $\hat{A}_{0}$ are
orthonormal by construction,%
\[
(P_{0}^{\top}\hat{A}_{0})^{\top}(P_{0}^{\top}\hat{A}_{0})=\mathbf{I}%
_{n-\bar{r}}.
\]
Expanding the left-hand side,
\begin{equation}
\mathbf{I}_{n-\bar{r}}=((A_{0},A_{0,\bot})^{\top}\hat{A}_{0})^{\top}%
(A_{0},A_{0,\bot})^{\top}\hat{A}_{0}=\hat{A}_{0}^{\top}A_{0}A_{0}^{\top}%
\hat{A}_{0}+\hat{A}_{0}^{\top}A_{0,\bot}A_{0,\bot}^{\top}\hat{A}_{0}.
\label{P_Asy_L1_2b}%
\end{equation}
Taking operator norms and applying ({\ref{P_Asy_L1_2}}),%
\begin{equation}
\big \lVert \hat{A}_{0}^{\top}A_{0}\big \rVert_{o}^{2}\geq1-\big \lVert \hat
{A}_{0}^{\top}A_{0,\bot}\big \rVert_{o}^{2}\geq1-\big \lVert \hat{A}_{0}%
^{\top}A_{0,\bot}\big \rVert^{2}=1-O_{p}(T^{-1}). \label{P_Asy_L1_3}%
\end{equation}
This proves the first claim of ({\ref{Asy_L1_1}}).

From ({\ref{P_Asy_L1_2b}}) and ({\ref{P_Asy_L1_3}}), $\hat{A}_{0}^{\top}A_{0}$ is nonsingular wpa1.
Hence $K_{0}\equiv(A_{0,\bot}^{\top}\hat{A}_{0})(A_{0}^{\top}\hat{A}_{0}%
)^{-1}$ is well defined wpa1. This, together with ({\ref{P_Asy_L1_2}}) and
({\ref{P_Asy_L1_3}}) shows%
\[
\left \Vert K_{0}\right \Vert =O_{p}(T^{-1/2}),
\]
which establishes the second claim.\hfill$Q.E.D.$

\bigskip

\begin{lemma}
\label{Asy_L2} Under Assumptions \ref{Asy_Cond_1} and {\ref{S}}, we have wpa1,%
\begin{equation}
\hat{A}_{0}\hat{H}_{0}^{\top}=P_{0}\left(
\begin{array}
[c]{c}%
\mathbf{I}_{n-\bar{r}}\\
K_{0}%
\end{array}
\right)  (\mathbf{I}_{n-\bar{r}}+K_{0}^{\top}K_{0})^{-1/2}\text{ \ and \ }%
\hat{H}_{0}^{\top}\hat{H}_{0}=\mathbf{I}_{n-\bar{r}}, \label{Asy_L2_1}%
\end{equation}
where $P_{0}\equiv(A_{0},A_{0,\bot})$.
\end{lemma}

\noindent \textsc{Proof of Lemma \ref{Asy_L2}.} Under Assumption
\ref{Asy_Cond_1}, we can use ({\ref{P_Asy_L1_2}}) and ({\ref{P_Asy_L1_2b}}) in
the proof of Lemma {\ref{Asy_L1} to show that }%
\begin{equation}
\left \Vert \mathbf{I}_{n-\bar{r}}-\hat{A}_{0}^{\top}A_{0}A_{0}^{\top}\hat
{A}_{0}\right \Vert =O_{p}(T^{-1}), \label{P_Asy_L2_0}%
\end{equation}
which implies that $A_{0}^{\top}\hat{A}_{0}$ is invertible wpa1. Conditioning
on the event that $A_{0}^{\top}\hat{A}_{0}$ is invertible, we may decompose
$P_{0}^{\top}\hat{A}_{0}=(A_{0},A_{0,\bot})^{\top}\hat{A}_{0}$ as
\begin{equation}
P_{0}^{\top}\hat{A}_{0}=\left(
\begin{array}
[c]{c}%
A_{0}^{\top}\hat{A}_{0}\\
A_{0,\bot}^{\top}\hat{A}_{0}%
\end{array}
\right)  =\left(
\begin{array}
[c]{c}%
\mathbf{I}_{n-\bar{r}}\\
(A_{0,\bot}^{\top}\hat{A}_{0})(A_{0}^{\top}\hat{A}_{0})^{-1}%
\end{array}
\right)  A_{0}^{\top}\hat{A}_{0}=\left(
\begin{array}
[c]{c}%
\mathbf{I}_{n-\bar{r}}\\
K_{0}%
\end{array}
\right)  A_{0}^{\top}\hat{A}_{0}. \label{P_Asy_L2_1}%
\end{equation}
Because both $P_{0}$ and $\hat{A}_{0}$ have orthonormal columns,
\[
(P_{0}^{\top}\hat{A}_{0})^{\top}(P_{0}^{\top}\hat{A}_{0})=\hat{A}_{0}^{\top
}P_{0}P_{0}^{\top}\hat{A}_{0}=\mathbf{I}_{n-\bar{r}},
\]
so $P_{0}^{\top}\hat{A}_{0}$ itself has orthonormal columns. Define
\begin{equation}
U_{0}\equiv \left(
\begin{array}
[c]{c}%
\mathbf{I}_{n-\bar{r}}\\
K_{0}%
\end{array}
\right)  (\mathbf{I}_{n-\bar{r}}+K_{0}^{\top}K_{0})^{-1/2}. \label{P_Asy_L2_2}%
\end{equation}
Then\ $U_{0}^{\top}U_{0}=\mathbf{I}_{n-\bar{r}}$, so $U_{0}$ also has
orthonormal columns.

Since $A_{0}^{\top}\hat{A}_{0}$ and $\mathbf{I}_{n-\bar{r}}+K_{0}^{\top}K_{0}$
are nonsingular,\ ({\ref{P_Asy_L2_1}}) and ({\ref{P_Asy_L2_2}}) imply that
\begin{equation}
\mathrm{\operatorname{col}}(P_{0}^{\top}\hat{A}_{0}%
)=\mathrm{\operatorname{col}}(U_{0})\text{, \ wpa1.} \label{P_Asy_L2_3}%
\end{equation}
Hence there exists a square matrix $R_{0}\ $such that%
\begin{equation}
P_{0}^{\top}\hat{A}_{0}=U_{0}R_{0}. \label{P_Asy_L2_4}%
\end{equation}
Because both $P_{0}^{\top}\hat{A}_{0}$ and $U_{0}$ have\ orthonormal columns,
and $R_{0}$ is a square matrix, we have
\begin{equation}
R_{0}^{\top}R_{0}=R_{0}R_{0}^{\top}=\mathbf{I}_{n-\bar{r}}. \label{P_Asy_L2_5}%
\end{equation}
It follows that:
\begin{equation}
(P_{0}^{\top}\hat{A}_{0})(P_{0}^{\top}\hat{A}_{0})^{\top}=(U_{0}R_{0}%
)(U_{0}R_{0})^{\top}=U_{0}R_{0}R_{0}^{\top}U_{0}^{\top}=U_{0}U_{0}^{\top}.
\label{P_Asy_L2_6}%
\end{equation}

By definition of $\hat{H}_{0}$ (see Lemma \ref{Asy_Ahat_L1}),%
\begin{equation}
\hat{H}_{0}=U_{0}^{\top}(P_{0}^{\top}\hat{A}_{0}), \label{P_Asy_L2_7}%
\end{equation}
based on which we obtain
\begin{equation}
\hat{A}_{0}\hat{H}_{0}^{\top}=\hat{A}_{0}\hat{A}_{0}^{\top}P_{0}U_{0}%
=P_{0}(P_{0}^{\top}\hat{A}_{0})(P_{0}^{\top}\hat{A}_{0})^{\top}U_{0}%
=P_{0}U_{0}U_{0}^{\top}U_{0}=P_{0}U_{0}, \label{P_Asy_L2_8}%
\end{equation}
where we use $P_{0}P_{0}^{\top}=I_{n-1}$, ({\ref{P_Asy_L2_6}}) and
$U_{0}^{\top}U_{0}=\mathbf{I}_{n-\bar{r}}$. This proves the first claim of the lemma.

For the second claim,%
\[
\hat{H}_{0}^{\top}\hat{H}_{0}=(P_{0}^{\top}\hat{A}_{0})^{\top}U_{0}U_{0}%
^{\top}(P_{0}^{\top}\hat{A}_{0})=(U_{0}R_{0})^{\top}U_{0}U_{0}^{\top}%
(U_{0}R_{0})=R_{0}^{\top}(U_{0}^{\top}U_{0})^{2}R_{0}=\mathbf{I}_{n-\bar{r}},
\]
using $U_{0}^{\top}U_{0}=\mathbf{I}_{n-\bar{r}}$, ({\ref{P_Asy_L2_4}}),
({\ref{P_Asy_L2_5}}) and ({\ref{P_Asy_L2_7}}). Both statements in the lemma
now follow.\hfill$Q.E.D.$

\bigskip

\begin{lemma}
\label{Asy_L3} Under Assumptions \ref{Asy_Cond_1} and {\ref{S}}, we have%
\begin{equation}
\tilde{A}_{0}-A_{0}=A_{0,\bot}K_{0}+O_{p}(T^{-1}), \label{Asy_L3_1}%
\end{equation}
where%
\begin{equation}
\tilde{A}_{0}\equiv P_{0}\left(
\begin{array}
[c]{c}%
\mathbf{I}_{n-\bar{r}}\\
K_{0}%
\end{array}
\right)  (\mathbf{I}_{n-\bar{r}}+K_{0}^{\top}K_{0})^{-1/2}. \label{Asy_L3_2}%
\end{equation}

\end{lemma}

\noindent \textsc{Proof of Lemma \ref{Asy_L3}. }The proof uses $K_{0}%
=(A_{0,\bot}^{\top}\hat{A}_{0})(A_{0}^{\top}\hat{A}_{0})^{-1}$, which holds
wpa1 by ({\ref{P_Asy_L2_0}}) in the proof of Lemma {\ref{Asy_L2}. From the
definition of }$\tilde{A}_{0}$,
\begin{align}
\tilde{A}_{0}-A_{0}  &  =(A_{0}+A_{0,\bot}K_{0})(\mathbf{I}_{n-\bar{r}}%
+K_{0}^{\top}K_{0})^{-1/2}-A_{0}\nonumber \\
&  =A_{0,\bot}K_{0}+\left(  A_{0}+A_{0,\bot}K_{0}\right)  \left(  \left(
\mathbf{I}_{n-\bar{r}}+K_{0}^{\top}K_{0}\right)  ^{-1/2}-\mathbf{I}_{n-\bar
{r}}\right)  . \label{P_Asy_L3_1}%
\end{align}
By the perturbation inequality for positive definite matrices (see, e.g.,
Lemma A.2 in \cite{BCCK2015}, applied with $A=\mathbf{I}_{n-\bar{r}}%
+K_{0}^{\top}K_{0}$ and $B=\mathbf{I}_{n-\bar{r}}$ in that lemma),%
\begin{equation}
\big \lVert(\mathbf{I}_{n-\bar{r}}+K_{0}^{\top}K_{0})^{1/2}-\mathbf{I}%
_{n-\bar{r}}\big \rVert_{o}\leq||K_{0}^{\top}K_{0}||_{o}. \label{P_Asy_L3_2}%
\end{equation}
Since $||K_{0}^{\top}K_{0}||_{o}=O_{p}(T^{-1})$ by Lemma {\ref{Asy_L1}}, we
obtain
\begin{equation}
\big \lVert(\mathbf{I}_{n-\bar{r}}+K_{0}^{\top}K_{0})^{1/2}-\mathbf{I}%
_{n-\bar{r}}\big \rVert_{o}\leq||K_{0}^{\top}K_{0}||_{o}=O_{p}(T^{-1}).
\label{P_Asy_L3_3}%
\end{equation}
Next, Lemma {\ref{Asy_L1} along with the triangle inequality and
}({\ref{P_Asy_L3_3}}) implies%
\begin{align}
&  ||(A_{0}+A_{0,\bot}K_{0})(\mathbf{I}_{n-\bar{r}}+K_{0}^{\top}K_{0}%
)^{-1/2}||_{o}\leq(\rho_{\max}((\mathbf{I}_{n-\bar{r}}+K_{0}^{\top}K_{0}%
)^{-1}))^{1/2}||A_{0}+A_{0,\bot}K_{0}||_{o}\nonumber \\
&  \leq \left \Vert A_{0}\right \Vert _{o}+\left \Vert A_{0,\bot}\right \Vert
_{o}\left \Vert K_{0}\right \Vert _{o}=1+O_{p}(T^{-1/2})\leq2, \label{P_Asy_L3_4}%
\end{align}
wpa1. Combining ({\ref{P_Asy_L3_3}}) and ({\ref{P_Asy_L3_4}}) yields
\begin{align}
&  \big \lVert(A_{0}+A_{0,\bot}K_{0})((\mathbf{I}_{n-\bar{r}}+K_{0}^{\top
}K_{0})^{-1/2}-\mathbf{I}_{n-\bar{r}})\big \rVert_{o}\nonumber \\
&  \leq \big \lVert(A_{0}+A_{0,\bot}K_{0})(\mathbf{I}_{n-\bar{r}}+K_{0}^{\top
}K_{0})^{-1/2}\big \rVert_{o}\big \lVert(\mathbf{I}_{n-\bar{r}}+K_{0}^{\top
}K_{0})^{1/2}-\mathbf{I}_{n-\bar{r}}\big \rVert_{o}=O_{p}(T^{-1}).
\label{P_Asy_L3_5}%
\end{align}
The statement of the lemma follows from ({\ref{P_Asy_L3_1}}) and
({\ref{P_Asy_L3_5}}).\hfill$Q.E.D.$

\bigskip

\begin{lemma}
\label{Asy_L4} Under Assumptions \ref{Asy_Cond_1} and {\ref{S}}, we have%
\[
K_{0}=(\bar{\sigma}_{u}^{2}\mathbf{I}_{\bar{r}-1}-\Lambda_{\bot}%
)^{-1}A_{0,\bot}^{\top}\hat{\Delta}_{y}A_{0}+O_{p}(T^{-1}).
\]

\end{lemma}

\noindent \textsc{Proof of Lemma \ref{Asy_L4}. }The proof uses $(A_{0}^{\top
}\hat{A}_{0})^{-}=(A_{0}^{\top}\hat{A}_{0})^{-1}$, which holds wpa1 by
({\ref{P_Asy_L2_0}}) in the proof of Lemma {\ref{Asy_L2}. }Let $\hat{\Lambda}$
be the diagonal matrix of the $n-\bar{r}$ smallest eigenvalues of $\hat{S}%
_{y}$. By definition of $\hat{A}_{0}$, $\hat{S}_{y}\hat{A}_{0}=\hat{A}_{0}%
\hat{\Lambda}$. Premultiplying by $P_{0}^{\top}$ and using $\hat{S}_{y}%
=S_{y}+\hat{\Delta}_{y}$
\begin{equation}
P_{0}^{\top}\hat{A}_{0}\hat{\Lambda}=P_{0}^{\top}\hat{S}_{y}\hat{A}_{0}%
=P_{0}^{\top}\hat{S}_{y}P_{0}P_{0}^{\top}\hat{A}_{0}=(P_{0}^{\top}S_{y}%
P_{0}+P_{0}^{\top}\hat{\Delta}_{y}P_{0})P_{0}^{\top}\hat{A}_{0}.
\label{P_Asy_L4_1}%
\end{equation}
Since
\[
P_{0}^{\top}\hat{A}_{0}=\left(
\begin{array}
[c]{c}%
A_{0}^{\top}\hat{A}_{0}\\
A_{0,\bot}^{\top}\hat{A}_{0}%
\end{array}
\right)  \text{, }P_{0}^{\top}S_{y}P_{0}=\left(
\begin{array}
[c]{cc}%
\bar{\sigma}_{u}^{2}\mathbf{I}_{n-\bar{r}} & 0\\
0 & \Lambda_{\bot}%
\end{array}
\right)  \text{ and }P_{0}^{\top}\hat{\Delta}_{y}P_{0}=\left(
\begin{array}
[c]{cc}%
\hat{\Delta}_{y,11} & \hat{\Delta}_{y,12}\\
\hat{\Delta}_{y,21} & \hat{\Delta}_{y,22}%
\end{array}
\right)  ,
\]
where $\hat{\Delta}_{y,11}\equiv A_{0}^{\top}\hat{\Delta}_{y}A_{0}$,
$\hat{\Delta}_{y,12}\equiv A_{0}^{\top}\hat{\Delta}_{y}A_{0,\bot}$,
$\hat{\Delta}_{y,22}\equiv A_{0,\bot}^{\top}\hat{\Delta}_{y}A_{0,\bot}$ and
$\hat{\Delta}_{y,21}^{\top}=\hat{\Delta}_{y,12}$, equation ({\ref{P_Asy_L4_1}%
}) can be written as
\begin{equation}
\left(
\begin{array}
[c]{c}%
A_{0}^{\top}\hat{A}_{0}\hat{\Lambda}\\
A_{0,\bot}^{\top}\hat{A}_{0}\hat{\Lambda}%
\end{array}
\right)  =\left(
\begin{array}
[c]{c}%
(\bar{\sigma}_{u}^{2}\mathbf{I}_{n-\bar{r}}+\hat{\Delta}_{y,11})A_{0}^{\top
}\hat{A}_{0}+\hat{\Delta}_{y,12}A_{0,\bot}^{\top}\hat{A}_{0}\\
\hat{\Delta}_{y,21}A_{0}^{\top}\hat{A}_{0}+(\Lambda_{\bot}+\hat{\Delta}%
_{y,22})A_{0,\bot}^{\top}\hat{A}_{0}%
\end{array}
\right)  . \label{P_Asy_L4_3}%
\end{equation}
The top block of ({\ref{P_Asy_L4_3}}) implies that%
\[
A_{0}^{\top}\hat{A}_{0}\hat{\Lambda}=(\bar{\sigma}_{u}^{2}\mathbf{I}%
_{n-\bar{r}}+\hat{\Delta}_{y,11})A_{0}^{\top}\hat{A}_{0}+\hat{\Delta}%
_{y,12}A_{0,\bot}^{\top}\hat{A}_{0}.
\]
Premultiplying by $(A_{0}^{\top}\hat{A}_{0})^{-1}$, we obtain
\begin{equation}
\hat{\Lambda}=(A_{0}^{\top}\hat{A}_{0})^{-1}(\bar{\sigma}_{u}^{2}%
\mathbf{I}_{n-\bar{r}}+\hat{\Delta}_{y,11})A_{0}^{\top}\hat{A}_{0}%
+(A_{0}^{\top}\hat{A}_{0})^{-1}\hat{\Delta}_{y,12}A_{0,\bot}^{\top}\hat{A}%
_{0}. \label{P_Asy_L4_4}%
\end{equation}
The second block of ({\ref{P_Asy_L4_3}}) implies that%
\[
A_{0,\bot}^{\top}\hat{A}_{0}\hat{\Lambda}=\hat{\Delta}_{y,21}A_{0}^{\top}%
\hat{A}_{0}+(\Lambda_{\bot}+\hat{\Delta}_{y,22})A_{0,\bot}^{\top}\hat{A}_{0}.
\]
Using ({\ref{P_Asy_L4_4}}) and $K_{0}=(A_{0,\bot}^{\top}\hat{A}_{0}%
)(A_{0}^{\top}\hat{A}_{0})^{-1}$, we obtain%
\[
K_{0}(\bar{\sigma}_{u}^{2}\mathbf{I}_{n-\bar{r}}+\hat{\Delta}_{y,11}%
)+K_{0}\hat{\Delta}_{y,12}K_{0}=\hat{\Delta}_{y,21}+(\Lambda_{\bot}%
+\hat{\Delta}_{y,22})K_{0}.
\]
Rearranging,%
\begin{equation}
(\bar{\sigma}_{u}^{2}\mathbf{I}_{\bar{r}-1}-\Lambda_{\bot})K_{0}=\hat{\Delta
}_{y,21}-K_{0}\hat{\Delta}_{y,11}+\hat{\Delta}_{y,22}K_{0}-K_{0}\hat{\Delta
}_{y,12}K_{0}. \label{P_Asy_L4_5}%
\end{equation}
From ({\ref{S_y_Rate}}) and Lemma {\ref{Asy_L1}, we have}%
\begin{equation}
\big \lVert K_{0}\hat{\Delta}_{y,11}\big \rVert=\big \lVert K_{0}A_{0}^{\top
}\hat{\Delta}_{y}A_{0}\big \rVert \leq \left \Vert K_{0}\right \Vert \left \Vert
A_{0}\right \Vert ^{2}\big \rVert \hat{\Delta}_{y}\big \rVert=O_{p}(T^{-1}).
\label{P_Asy_L4_6}%
\end{equation}
Similarly,
\begin{equation}
\big \lVert \hat{\Delta}_{y,22}K_{0}\big \rVert=\big \lVert A_{0,\bot}^{\top
}\hat{\Delta}_{y}A_{0,\bot}K_{0}\big \rVert \leq \left \Vert K_{0}\right \Vert
\left \Vert A_{0,\bot}\right \Vert ^{2}\big \rVert \hat{\Delta}_{y}%
\big \rVert=O_{p}(T^{-1}), \label{P_Asy_L4_7}%
\end{equation}
and
\begin{equation}
\big \lVert K_{0}\hat{\Delta}_{y,12}K_{0}\big \rVert=\big \lVert K_{0}%
A_{0}^{\top}\hat{\Delta}_{y}A_{0,\bot}K_{0}\big \rVert \leq \left \Vert
K_{0}\right \Vert ^{2}\left \Vert A_{0,\bot}\right \Vert \left \Vert
A_{0}\right \Vert \big \rVert \hat{\Delta}_{y}\big \rVert=O_{p}(T^{-3/2}).
\label{P_Asy_L4_8}%
\end{equation}
Substituting ({\ref{P_Asy_L4_6}})-({\ref{P_Asy_L4_8}}) into ({\ref{P_Asy_L4_5}%
}),%
\[
(\bar{\sigma}_{u}^{2}\mathbf{I}_{\bar{r}-1}-\Lambda_{\bot})K_{0}=\hat{\Delta
}_{y,21}+O_{p}(T^{-1}),
\]
which together with Assumption \ref{Asy_Cond_1}(iii) establishes the
claim.\hfill$Q.E.D.$

\bigskip

\begin{lemma}
\label{Asy_L5}\ Under Assumptions {\ref{ID}, \ref{Asy_Cond_1} and \ref{S}, we
have}%
\begin{equation}
T^{-1/2}\sum_{t\leq T}{\tilde{A}^{\top}}y_{t}b_{t}=A^{\top}T^{-1/2}\sum_{t\leq
T}\xi_{b,t}+O_{p}(T^{-1/2})\text{, \ for }b\in \{v,\varepsilon \},
\label{Asy_L5_1}%
\end{equation}
where $\tilde{A}\equiv Q_{-1}\tilde{A}_{0}$ and $\tilde{A}_{0}=\hat{A}_{0}%
\hat{H}_{0}^{\top}$ by (\ref{Asy_L2_1}) and (\ref{Asy_L3_2}).
\end{lemma}

\noindent \textsc{Proof of Lemma \ref{Asy_L5}.} Under Assumptions
\ref{Asy_Cond_1}(i, ii) and \ref{S}, we can apply Lemmas \ref{Asy_L0} and \ref{Asy_L0b} to
obtain
\begin{equation}
T^{-1}\sum_{t\leq T}(y_{t}v_{t}-\mathbb{E}[y_{t}v_{t}])=O_{p}(T^{-1/2}%
)\quad \text{and}\quad T^{-1}\sum_{t\leq T}(y_{t}\varepsilon_{t}-\mathbb{E}%
[y_{t}\varepsilon_{t}])=O_{p}(T^{-1/2}). \label{P_Asy_L5_1}%
\end{equation}
From Lemma \ref{Asy_Ahat_L1}, we have%
\begin{align}
\tilde{A}-A  &  =Q_{-1}(\hat{A}_{0}\hat{H}_{0}^{\top}-A_{0})\nonumber \\
&  =Q_{-1}A_{0,\bot}(\overline{\sigma}_{u}^{2}\mathbf{I}_{\bar{r}-1}%
-\Lambda_{\bot})^{-1}A_{0,\bot}^{\top}\hat{\Delta}_{y}A_{0}+O_{p}%
(T^{-1})\nonumber \\
&  =Q_{-1}A_{0,\bot}(\overline{\sigma}_{u}^{2}\mathbf{I}_{\bar{r}-1}%
-\Lambda_{\bot})^{-1}A_{0,\bot}^{\top}Q_{-1}^{\top}\left(  T^{-1}\sum_{t\leq
T}(y_{t}y_{t}^{\top}-\mathbb{E}[y_{t}y_{t}^{\top}])\right)  Q_{-1}A_{0}%
+O_{p}(T^{-1})\nonumber \\
&  =\Upsilon \,T^{-1}\sum_{t\leq T}(y_{t}y_{t}^{\top}-\mathbb{E}[y_{t}%
y_{t}^{\top}])A+O_{p}(T^{-1})=O_{p}(T^{-1/2}), \label{P_Asy_L5_2}%
\end{align}
where the last equality follows from Assumption \ref{Asy_Cond_1}(i, iii, iv)
and Lemma \ref{Asy_L0}. Combining (\ref{P_Asy_L5_1}) and (\ref{P_Asy_L5_2}),
we obtain
\begin{align}
T^{-1/2}\sum_{t\leq T}\tilde{A}^{\top}y_{t}v_{t}  &  =T^{-1/2}\sum_{t\leq
T}A^{\top}y_{t}v_{t}+T^{1/2}(\tilde{A}-A)^{\top}T^{-1}\sum_{t\leq T}y_{t}%
v_{t}\nonumber \\
&  =T^{-1/2}\sum_{t\leq T}A^{\top}y_{t}v_{t}+T^{1/2}(\tilde{A}-A)^{\top}%
T^{-1}\sum_{t\leq T}\mathbb{E}[y_{t}v_{t}]+O_{p}(T^{-1/2})\nonumber \\
&  =A^{\top}T^{-1/2}\sum_{t\leq T}y_{t}v_{t}\nonumber \\
&  \quad+A^{\top}T^{-1/2}\sum_{t\leq T}(y_{t}y_{t}^{\top}-\mathbb{E}%
[y_{t}y_{t}^{\top}])\, \Upsilon \,T^{-1}\sum_{t\leq T}\mathbb{E}[y_{t}%
v_{t}]+O_{p}(T^{-1/2}). \label{P_Asy_L5_3}%
\end{align}

Next, note that $A\equiv Q_{-1}A_{0}$ is orthogonal to $(\mathbf{1}%
_{n},\lambda)$. Using (\ref{G_demand}), we can write
\[
A^{\top}y_{t}=A_{0}^{\top}Q_{-1}^{\top}(\phi p_{t}\mathbf{1}_{n}+\lambda
\eta_{t}+u_{t})=A_{0}^{\top}Q_{-1}^{\top}u_{t}.
\]
Together with the definition of $v_{t}$, this implies
\begin{equation}
A^{\top}\mathbb{E}[y_{t}v_{t}]=A_{0}^{\top}Q_{-1}^{\top}\mathbb{E}[u_{t}%
v_{t}]=A_{0}^{\top}Q_{-1}^{\top}\mathbb{E}[u_{t}(y_{e,t}-\phi p_{t})].
\label{P_Asy_L5_4}%
\end{equation}
Since $y_{e,t}=e^{\top}y_{t}$ and $y_{t}=\phi p_{t}\mathbf{1}_{n}+\lambda
\eta_{t}+u_{t}$, we have
\[
\mathbb{E}[u_{t}(y_{e,t}-\phi p_{t})]=\mathbb{E}[u_{t}(\eta_{t}^{\top}%
\lambda^{\top}+u_{t}^{\top})]e.
\]
By Assumption \ref{ID}, it follows that
\begin{equation}
Q_{-1}^{\top}\mathbb{E}[(u_{t}\eta_{t}^{\top}\lambda^{\top}+u_{t}u_{t}^{\top
})]e=n^{-1}Q_{-1}^{\top}(\mathbf{1}_{n}\Gamma_{\eta u,t}^{\top}\lambda^{\top
}+\sigma_{u,t}^{2}\mathbf{I}_{n})\mathbf{1}_{n}=\mathbf{0}_{n-1}.
\label{P_Asy_L5_5}%
\end{equation}
Combining (\ref{P_Asy_L5_4}) and (\ref{P_Asy_L5_5}) yields
\begin{equation}
A^{\top}\mathbb{E}[y_{t}v_{t}]=0. \label{P_Asy_L5_6}%
\end{equation}
Substituting this into (\ref{P_Asy_L5_3}) and using the definition of
$\xi_{v,t}$ establishes (\ref{Asy_L5_1}) for $b=v$. The proof for
$b=\varepsilon$ follows by analogous arguments.\hfill$Q.E.D.$

\section{Proof of Consistency of the Variance Estimator\label{APP_4}}

In this section, we establish the consistency of the variance matrix estimator
$\hat{V}$ proposed in the algorithm under the assumption that $\xi_{t}$ is
uncorrelated across $t$. The following conditions are useful for proving this result.

\begin{assumption}
\label{Asy_Cond_3} (i) For any two distinct eigenvalues $\mu_{j_{1}}$ and
$\mu_{j_{2}}$ of $S_{y}$,\ $\left \vert \mu_{j_{1}}-\mu_{j_{2}}\right \vert
>K^{-1}$; (ii) for any sequence $\{c_{t}\}_{t\leq T}\subset \mathbb{R}$ with
$\max_{t\leq T}|c_{t}|\leq K$, and $a,b\in \{u,\eta,\varepsilon \}$,
\[
T^{-1}\sum_{t\leq T}c_{t}\big(a_{t}b_{t}^{\top}-\mathbb{E}[a_{t}b_{t}^{\top
}]\big)=o_{p}(1);
\]
(iii) for every random variable of the form $b_{t}=\prod_{j=1}^{4}z_{j,t}$,%
\[
T^{-1}\sum_{t\leq T}(b_{t}-\mathbb{E}[b_{t}])=o_{p}(1),
\]
where each $z_{j,t}$ belongs to $\{u_{i,t},\eta_{k,t},\varepsilon_{t}%
,S_{t}^{\top}u_{t},s_{i,t}\eta_{k,t}:i\leq n,k\leq r\}$; (iv) $\max_{t\leq
T}\mathbb{E}[(u_{t}^{\top}u_{t})^{2}+\varepsilon_{t}^{4}+(\eta_{t}^{\top}%
\eta_{t})^{2}]\leq K$.
\end{assumption}

Assumption \ref{Asy_Cond_3}(i) imposes a uniform separation (gap) between
distinct eigenvalues of $S_{y}$, which is useful for establishing the
consistency of $\hat{\Upsilon}$ used in constructing $\hat{\xi}_{b,t}$ for
$b\in \{v,\varepsilon \}$. Since $\mu_{1}=\ldots=\mu_{n-\bar{r}}=\bar{\sigma
}_{u}^{2}<\mu_{n-\bar{r}+1}$, it is clear that Assumption \ref{Asy_Cond_1}%
(iii) is implied by Assumption \ref{Asy_Cond_3}(i). Assumption
\ref{Asy_Cond_3}(ii) is analogous to Assumption \ref{Asy_Cond_1}(ii) and can
be verified using a law of large numbers. Assumption \ref{Asy_Cond_3}(iii)
imposes a law of large numbers for fourth-order products formed from
$\varepsilon_{t}$ and components of $u_{t}$ and $\eta_{t}$. Finally,
Assumption \ref{Asy_Cond_3}(iv) imposes uniform fourth-moment bounds on the
shocks in the demand and supply equations.

\begin{theorem}
\label{V_Est} Suppose that $\xi_{t}$ is uncorrelated across $t$. Under
Assumptions \ref{ID}, \ref{Asy_Cond_1}, \ref{S}, \ref{Asy_Cond_2}(i, ii, iii),
and \ref{Asy_Cond_3}, we have
\[
\hat{V}=V+o_{p}(1),
\]
where $\hat{V}$ is defined in (\ref{V_est}).
\end{theorem}

\noindent \textsc{Proof of Theorem \ref{V_Est}.} Let $\gamma \equiv
\Upsilon \big(T^{-1}\sum_{t\leq T}\mathbb{E}[y_{t}v_{t}]\big)$.\ We then have
\begin{align*}
\xi_{v,t}  &  \equiv y_{t}v_{t}-\mathbb{E}[y_{t}v_{t}]+(y_{t}y_{t}^{\top
}-\mathbb{E}[y_{t}y_{t}^{\top}])\Upsilon \Big(T^{-1}\sum_{t\leq T}%
\mathbb{E}[y_{t}v_{t}]\Big)\\
&  =y_{t}v_{t}-\mathbb{E}[y_{t}v_{t}]+(y_{t}y_{t}^{\top}-\mathbb{E}[y_{t}%
y_{t}^{\top}])\gamma \\
&  =y_{t}v_{t}-\mathbb{E}[y_{t}v_{t}]+(y_{\, \gamma,t}y_{t}-\mathbb{E}[y_{\,
\gamma,t}y_{t}]),
\end{align*}
which implies that
\begin{equation}
\xi_{v,t}^{\top}\xi_{v,t}\leq K\Big(v_{t}^{2}y_{t}^{\top}y_{t}+y_{\gamma
,t}^{2}y_{t}^{\top}y_{t}+\Vert \mathbb{E}[v_{t}y_{t}]\Vert^{2}+\Vert
\mathbb{E}[y_{\gamma,t}y_{t}]\Vert^{2}\Big). \label{P_V_Est_1}%
\end{equation}
Under Assumption \ref{Asy_Cond_1}(i), we can apply (\ref{P_Asy_Dist_1a}) to
show that%
\begin{equation}
\max_{t\leq T}\left \Vert \mathbb{E}[y_{t}v_{t}]\right \Vert \leq \max_{t\leq
T}||\mathbb{E}[y_{t}y_{t}^{\top}]||+\left \vert \phi \right \vert \max_{t\leq
T}\left \Vert \mathbb{E}[y_{t}p_{t}]\right \Vert \leq K. \label{P_V_Est_1b}%
\end{equation}
This together with (\ref{P_L_Upsilon_7}), the triangle
inequality\ and\ Cauchy-Schwarz inequality implies
\begin{equation}
\left \Vert \gamma \right \Vert \leq K \label{P_V_Est_1c}%
\end{equation}
Note that $v_{t}=y_{e,t}-\phi p_{t}$. By (\ref{P_Asy_L0_1}) and
(\ref{P_Asy_L0_2}), both $y_{t}$ and $p_{t}$ are linear functions of $u_{t}$,
$\eta_{t}$, and $\varepsilon_{t}$. It then follows from Assumptions
\ref{Asy_Cond_1}(i), \ref{S}(ii) and \ref{Asy_Cond_3}(iv) that
\begin{equation}
\max_{t\leq T}\Big(\mathbb{E}[v_{t}^{2}y_{t}^{\top}y_{t}]+\mathbb{E}%
[y_{\gamma,t}^{2}y_{t}^{\top}y_{t}]\Big)\leq K. \label{P_V_Est_2}%
\end{equation}
Combining (\ref{P_V_Est_1}) and (\ref{P_V_Est_2}) with Markov's inequality
yields
\begin{equation}
T^{-1}\sum_{t\leq T}\xi_{v,t}^{\top}\xi_{v,t}=O_{p}(1). \label{P_V_Est_3}%
\end{equation}
Similarly, we can show that $T^{-1}\sum_{t\leq T}\xi_{\varepsilon,t}^{\top}%
\xi_{\varepsilon,t}=O_{p}(1)$. Together with (\ref{P_V_Est_3}), this implies
\begin{equation}
T^{-1}\sum_{t\leq T}\xi_{t}^{\top}\xi_{t}=O_{p}(1). \label{P_V_Est_4}%
\end{equation}

By the triangle inequality, the Cauchy--Schwarz inequality, Lemma \ref{L_V_5},
and (\ref{P_V_Est_4}), we have
\[
\big \lVert T^{-1}\sum_{t\leq T}(\hat{\xi}_{t}-\xi_{t})(\hat{\xi}_{t}-\xi
_{t})^{\top}\big \rVert \leq T^{-1}\sum_{t\leq T}\Vert \hat{\xi}_{t}-\xi
_{t}\Vert^{2}=o_{p}(1),
\]
and
\[
\big \lVert T^{-1}\sum_{t\leq T}(\hat{\xi}_{t}-\xi_{t})\xi_{t}^{\top
}\big \rVert \leq T^{-1}\sum_{t\leq T}\Vert \hat{\xi}_{t}-\xi_{t}\Vert \,
\Vert \xi_{t}\Vert=o_{p}(1),
\]
which, along with Lemma \ref{L_V_4}, implies
\begin{align}
T^{-1}\sum_{t\leq T}\hat{\xi}_{t}\hat{\xi}_{t}^{\top}-T^{-1}\sum_{t\leq
T}\mathbb{E}[\xi_{t}\xi_{t}^{\top}]  &  =T^{-1}\sum_{t\leq T}(\xi_{t}\xi
_{t}^{\top}-\mathbb{E}[\xi_{t}\xi_{t}^{\top}])+T^{-1}\sum_{t\leq T}(\hat{\xi
}_{t}-\xi_{t})(\hat{\xi}_{t}-\xi_{t})^{\top}\nonumber \\
&  \quad+T^{-1}\sum_{t\leq T}(\hat{\xi}_{t}-\xi_{t})\xi_{t}^{\top}+T^{-1}%
\sum_{t\leq T}\xi_{t}(\hat{\xi}_{t}-\xi_{t})^{\top}\nonumber \\
&  =o_{p}(1). \label{P_V_Est_5}%
\end{align}

Since $\xi_{t}$ is uncorrelated across $t$, we have $V=T^{-1}\sum_{t\leq
T}\mathbb{E}[\xi_{t}\xi_{t}^{\top}]$. Thus, the result follows from
(\ref{P_V_Est_5}) once we show that
\begin{equation}
T^{-1}\sum_{t\leq T}\hat{\xi}_{t}=o_{p}(1). \label{P_V_Est_6}%
\end{equation}
To verify (\ref{P_V_Est_6}), note that by the triangle inequality,
\begin{align*}
\big \lVert T^{-1}\sum_{t\leq T}\hat{\xi}_{t}\big \rVert  &  \leq
\big \lVert T^{-1}\sum_{t\leq T}(\hat{\xi}_{t}-\xi_{t}%
)\big \rVert+\big \lVert T^{-1}\sum_{t\leq T}\xi_{t}\big \rVert \\
&  \leq \Big(T^{-1}\sum_{t\leq T}\Vert \hat{\xi}_{t}-\xi_{t}\Vert^{2}%
\Big)^{1/2}+O_{p}(T^{-1/2})=o_{p}(1),
\end{align*}
where the second inequality follows from the Cauchy--Schwarz inequality and
Assumption \ref{Asy_Cond_2}(i), and the last equality follows from Lemma
\ref{L_V_5}. This completes the proof.\hfill$Q.E.D.$

\subsection{Useful Lemmas for consistency of variance estimator}

\begin{lemma}
\label{L_A_0_ortho}\ Suppose that the eigenvalues $\{ \mu_{j}\}_{n-\bar
{r}+1\leq j\leq n-1}$ take $m$ distinct values $\{ \mu_{j}^{\ast}\}_{1\leq
j\leq m}$, ordered increasingly, with multiplicities $\{s_{j}\}_{1\leq j\leq
m}$. That is,
\begin{equation}
\mu_{n-\bar{r}+\bar{s}_{j-1}+1}=\cdots=\mu_{n-\bar{r}+\bar{s}_{j}}=\mu
_{j}^{\ast}, \label{L_A_0_ortho_1}%
\end{equation}
for $j\in \{1,\ldots,m\}$, where $\bar{s}_{j-1}=\sum_{l=1}^{j-1}s_{l}$.\ Let
$A_{j}$ collect the orthonormal eigenvectors of $S_{y}$ associated with the
eigenvalues $\{ \mu_{n-\bar{r}+\bar{s}_{j-1}+l}\}_{l=1}^{s_{j}}$, and let
$\hat{A}_{j}$ collect the orthonormal eigenvectors of $\hat{S}_{y}$ associated
with $\{ \hat{\mu}_{n-\bar{r}+\bar{s}_{j-1}+l}\}_{l=1}^{s_{j}}$. Then under
Assumptions\  \ref{Asy_Cond_1}(i, ii),\ {\ref{S}, }\ref{Asy_Cond_2}(i, ii), and
\ref{Asy_Cond_3}(i), there exists an orthogonal matrix $\hat{O}_{j}%
\in \mathbb{R}^{s_{j}\times s_{j}}$ such that for each $j\in \{1,\ldots,m\}$,
\[
\big \lVert \hat{A}_{j}\hat{O}_{j}-A_{j}\big \lVert=O_{p}(T^{-1/2}).
\]

\end{lemma}

\noindent \textsc{Proof of Lemma \ref{L_A_0_ortho}.} Fix $j\in \{1,\ldots,m\}$,
and define
\[
r_{j}\equiv n-\bar{r}+\bar{s}_{j-1}+1,\qquad s_{j}^{\ast}\equiv n-\bar{r}%
+\bar{s}_{j}.
\]
Then $r_{j}\leq s_{j}^{\ast}$, and by (\ref{L_A_0_ortho_1}), the columns of
$A_{j}$ and $\hat{A}_{j}$ are the orthonormal eigenvectors of $S_{y}$ and
$\hat{S}_{y}$, respectively, associated with the block of eigenvalues
\[
\mu_{r_{j}}=\cdots=\mu_{s_{j}^{\ast}}=\mu_{j}^{\ast}.
\]

Since Theorem 2 of \cite{yu2015useful} is stated for eigenvalues ordered in
decreasing order, we apply it to $-S_{y}$ and $-\hat{S}_{y}$.\ By
construction, $\mu_{r_{1}-1}=\mu_{n-\bar{r}}=\bar{\sigma}_{u}^{2}$. Following
\cite{yu2015useful}, we set $\mu_{n}=\infty$, which implies $\mu_{s_{j}^{\ast
}+1}=\infty$. It then follows that there exists an orthogonal matrix $\hat
{O}_{j}\in \mathbb{R}^{s_{j}\times s_{j}}$ such that
\begin{equation}
\big \lVert \hat{A}_{j}\hat{O}_{j}-A_{j}\big \lVert \leq \frac{2^{3/2}%
\min \left \{  s_{j}^{1/2}\big \lVert \hat{S}_{y}-S_{y}\big \lVert_{o},\,
\big \lVert \hat{S}_{y}-S_{y}\big \lVert \right \}  }{\min \left \{  \mu_{r_{j}%
}-\mu_{r_{j}-1},\, \text{\ }\mu_{s_{j}^{\ast}+1}-\mu_{s_{j}^{\ast}}\right \}
}. \label{P_L_A_0_ortho_1}%
\end{equation}

We next show that the denominator in (\ref{P_L_A_0_ortho_1}) is bounded away
from zero uniformly in $j$. For $j\in \{2,\ldots,m\}$, the eigenvalues
$\mu_{r_{j}-1}$ and $\mu_{r_{j}}$ are distinct, and hence Assumption
\ref{Asy_Cond_3}(i) implies
\[
\mu_{r_{j}}-\mu_{r_{j}-1}\geq K^{-1}.
\]
Similarly, for $j\in \{1,\ldots,m-1\}$, the eigenvalues $\mu_{s_{j}^{\ast}}$
and $\mu_{s_{j}^{\ast}+1}$ are distinct, so Assumption \ref{Asy_Cond_3}(i)
implies
\[
\mu_{s_{j}^{\ast}+1}-\mu_{s_{j}^{\ast}}\geq K^{-1}.
\]
For the boundary block $j=1$, the left gap is
\[
\mu_{r_{1}}-\mu_{r_{1}-1}=\mu_{n-\bar{r}+1}-\mu_{n-\bar{r}},
\]
where $\mu_{n-\bar{r}}=\bar{\sigma}_{u}^{2}$ is not contained in the
collection $\{ \mu_{j}\}_{n-\bar{r}+1\leq j\leq n-1}$. By Assumption
\ref{Asy_Cond_3}(i), this gap is bounded away from zero. Therefore,
\begin{equation}
\min \{ \mu_{r_{j}}-\mu_{r_{j}-1},\, \mu_{s_{j}^{\ast}+1}-\mu_{s_{j}^{\ast}}\}
\geq K^{-1} \label{P_L_A_0_ortho_2}%
\end{equation}
for all $j\in \{1,\ldots,m\}$.

Finally, by (\ref{P_r_hat_consistency_0}), $\Vert \hat{S}_{y}-S_{y}\Vert
=O_{p}(T^{-1/2})$. Combining this with (\ref{P_L_A_0_ortho_1}) and
(\ref{P_L_A_0_ortho_2}) establishes the claim of the lemma.\hfill$Q.E.D.$

\bigskip

\begin{lemma}
\label{L_Upsilon} Under Assumptions \ref{Asy_Cond_1}(i, ii, iii, iv),\ {\ref{S},
}\ref{Asy_Cond_2}(i, ii) and \ref{Asy_Cond_3}(i), we have
\begin{equation}
\hat{\Upsilon}=\Upsilon+O_{p}(T^{-1/2})=O_{p}(1), \label{L_Upsilon_1}%
\end{equation}
where $\hat{\Upsilon}\equiv Q_{-1}\hat{A}_{0,\bot}(\hat{\sigma}_{u}%
^{2}\mathbf{I}_{\bar{r}-1}-\hat{\Lambda}_{\bot})^{-1}\hat{A}_{0,\bot}^{\top
}Q_{-1}^{\top}$ and $\Upsilon \equiv Q_{-1}A_{0,\bot}(\bar{\sigma}_{u}%
^{2}\mathbf{I}_{\bar{r}-1}-\Lambda_{\bot})^{-1}A_{0,\bot}^{\top}Q_{-1}^{\top}$.
\end{lemma}

\noindent \textsc{Proof of Lemma \ref{L_Upsilon}.} Suppose that the eigenvalues
$\{ \mu_{j}\}_{n-\bar{r}+1\leq j\leq n-1}$ take $m$ distinct values $\{
\mu_{j}^{\ast}\}_{1\leq j\leq m}$, with multiplicities $\{s_{j}\}_{1\leq j\leq
m}$. Then we can write
\[
A_{0,\bot}=(A_{1},\ldots,A_{m})\qquad \text{and}\qquad \hat{A}_{0,\bot}=(\hat
{A}_{1},\ldots,\hat{A}_{m}),
\]
where $A_{j}$ and $\hat{A}_{j}$ collect the orthonormal eigenvectors of
$S_{y}$ and $\hat{S}_{y}$, respectively, associated with the eigenvalue
$\mu_{j}^{\ast}$ of multiplicity $s_{j}$. By Lemma \ref{L_A_0_ortho}, there
exists a block diagonal matrix $\hat{H}_{0,\bot}^{\top}=\mathrm{diag}(\hat
{O}_{1},\ldots,\hat{O}_{m})$,\ where each $\hat{O}_{j}\in \mathbb{R}%
^{s_{j}\times s_{j}}$ is orthogonal, such that
\begin{equation}
\hat{A}_{0,\bot}\hat{H}_{0,\bot}^{\top}=A_{0,\bot}+O_{p}(T^{-1/2})=O_{p}(1).
\label{P_L_Upsilon_1}%
\end{equation}

From $\Vert \hat{S}_{y}-S_{y}\Vert=O_{p}(T^{-1/2})$, we have
\begin{equation}
\hat{\sigma}_{u}^{2}=\bar{\sigma}_{u}^{2}+O_{p}(T^{-1/2}),\qquad \hat{\Lambda
}_{\bot}=\Lambda_{\bot}+O_{p}(T^{-1/2}), \label{P_L_Upsilon_2}%
\end{equation}
where $\bar{\sigma}_{u}^{2}$ and $\Lambda_{\bot}$ are bounded by
(\ref{P_Asy_Dist_1a}). Since
\[
\Lambda_{\bot}=\mathrm{diag}(\mu_{1}^{\ast}I_{s_{1}},\ldots,\mu_{m}^{\ast
}I_{s_{m}}),
\]
and $\hat{H}_{0,\bot}^{\top}$ is block diagonal with orthogonal blocks, we
have
\begin{equation}
\hat{H}_{0,\bot}\Lambda_{\bot}\hat{H}_{0,\bot}^{\top}=\Lambda_{\bot}.
\label{P_L_Upsilon_3}%
\end{equation}
Therefore,
\begin{equation}
(\hat{\sigma}_{u}^{2}I_{\bar{r}-1}-\hat{H}_{0,\bot}\hat{\Lambda}_{\bot}\hat
{H}_{0,\bot}^{\top})-(\bar{\sigma}_{u}^{2}I_{\bar{r}-1}-\Lambda_{\bot}%
)=(\hat{\sigma}_{u}^{2}-\bar{\sigma}_{u}^{2})I_{\bar{r}-1}-\hat{H}_{0,\bot
}(\hat{\Lambda}_{\bot}-\Lambda_{\bot})\hat{H}_{0,\bot}^{\top}=O_{p}(T^{-1/2}).
\label{P_L_Upsilon_4}%
\end{equation}

By Assumption \ref{Asy_Cond_3}(i), the eigenvalues of $\Lambda_{\bot}%
-\bar{\sigma}_{u}^{2}I_{\bar{r}-1}$ are bounded below by $K^{-1}$. Hence, by
(\ref{P_L_Upsilon_4}),
\begin{equation}
\rho_{\min}\bigl(\hat{H}_{0,\bot}\hat{\Lambda}_{\bot}\hat{H}_{0,\bot}^{\top
}-\hat{\sigma}_{u}^{2}I_{\bar{r}-1}\bigr)\geq(2K)^{-1},\qquad \text{wpa1.}
\label{P_L_Upsilon_5}%
\end{equation}
Since $\hat{H}_{0,\bot}^{\top}\hat{H}_{0,\bot}=I_{\bar{r}-1}$, we have
\[
(\hat{\sigma}_{u}^{2}I_{\bar{r}-1}-\hat{\Lambda}_{\bot})^{-1}=\hat{H}_{0,\bot
}^{\top}(\hat{\sigma}_{u}^{2}I_{\bar{r}-1}-\hat{H}_{0,\bot}\hat{\Lambda}%
_{\bot}\hat{H}_{0,\bot}^{\top})^{-1}\hat{H}_{0,\bot}.
\]
Notice that
\begin{align*}
\hat{\Upsilon}-\Upsilon &  =Q_{-1}\hat{A}_{0,\bot}\hat{H}_{0,\bot}^{\top}%
(\hat{\sigma}_{u}^{2}I_{\bar{r}-1}-\hat{H}_{0,\bot}\hat{\Lambda}_{\bot}\hat
{H}_{0,\bot}^{\top})^{-1}\hat{H}_{0,\bot}\hat{A}_{0,\bot}^{\top}Q_{-1}^{\top
}\\
&  \quad-Q_{-1}A_{0,\bot}(\bar{\sigma}_{u}^{2}I_{\bar{r}-1}-\Lambda_{\bot
})^{-1}A_{0,\bot}^{\top}Q_{-1}^{\top}\\
&  =Q_{-1}(\hat{A}_{0,\bot}\hat{H}_{0,\bot}^{\top}-A_{0,\bot})(\hat{\sigma
}_{u}^{2}I_{\bar{r}-1}-\hat{H}_{0,\bot}\hat{\Lambda}_{\bot}\hat{H}_{0,\bot
}^{\top})^{-1}\hat{H}_{0,\bot}\hat{A}_{0,\bot}^{\top}Q_{-1}^{\top}\\
&  \quad+Q_{-1}A_{0,\bot}\Bigl[(\hat{\sigma}_{u}^{2}I_{\bar{r}-1}-\hat
{H}_{0,\bot}\hat{\Lambda}_{\bot}\hat{H}_{0,\bot}^{\top})^{-1}-(\bar{\sigma
}_{u}^{2}I_{\bar{r}-1}-\Lambda_{\bot})^{-1}\Bigr]\hat{H}_{0,\bot}\hat
{A}_{0,\bot}^{\top}Q_{-1}^{\top}\\
&  \quad+Q_{-1}A_{0,\bot}(\bar{\sigma}_{u}^{2}I_{\bar{r}-1}-\Lambda_{\bot
})^{-1}(\hat{H}_{0,\bot}\hat{A}_{0,\bot}^{\top}-A_{0,\bot}^{\top})Q_{-1}%
^{\top}.
\end{align*}
By (\ref{P_L_Upsilon_1})--(\ref{P_L_Upsilon_5}), each term on the right-hand
side is $O_{p}(T^{-1/2})$. Therefore,
\begin{equation}
\hat{\Upsilon}-\Upsilon=O_{p}(T^{-1/2}), \label{P_L_Upsilon_6}%
\end{equation}
which establishes the first equality in\ (\ref{L_Upsilon_1}).

Since $Q_{-1}^{\top}Q_{-1}=\mathbf{I}_{n-1}$, $A_{0,\bot}^{\top}A_{0,\bot
}=\mathbf{I}_{\bar{r}-1}$ and $\Lambda_{\bot}\equiv \mathrm{diag}((\mu
_{j})_{n-\bar{r}+1\leq j\leq n-1})$,\ from the Cauchy-Schwarz inequality and
Assumption \ref{Asy_Cond_1}(iii), it follows that
\begin{equation}
\left \Vert \Upsilon \right \Vert =\big \lVert Q_{-1}A_{0,\bot}(\bar{\sigma}%
_{u}^{2}\mathbf{I}_{\bar{r}-1}-\Lambda_{\bot})^{-1}A_{0,\bot}^{\top}%
Q_{-1}^{\top}\big \lVert \leq \frac{\left \Vert Q_{-1}\right \Vert ^{2}\left \Vert
A_{0,\bot}\right \Vert ^{2}(\bar{r}-1)}{\mu_{n-\bar{r}+1}-\bar{\sigma}_{u}^{2}%
}\leq K, \label{P_L_Upsilon_7}%
\end{equation}
which together with (\ref{P_L_Upsilon_6}) shows the second equality
in\ (\ref{L_Upsilon_1}).\hfill$Q.E.D.$

\bigskip

\begin{lemma}
\label{L_V_1a} Under Assumptions \ref{ID}, \ref{Asy_Cond_1}, \ref{S} and
\ref{Asy_Cond_2}(i, iii), we have for $b\in \{v,\varepsilon \}$%
\[
T^{-1}\sum_{t\leq T}y_{t}\hat{b}_{t}=T^{-1}\sum_{t\leq T}\mathbb{E}[y_{t}%
b_{t}]+O_{p}(T^{-1/2})=O_{p}(1).
\]

\end{lemma}

\noindent \textsc{Proof of Lemma \ref{L_V_1a}.} Because $v_{t}\equiv
y_{e,t}-\phi p_{t}$, we can write
\begin{align}
T^{-1}\sum_{t\leq T}y_{t}v_{t}  &  =T^{-1}\sum_{t\leq T}\mathbb{E}[y_{t}%
v_{t}]+T^{-1}\sum_{t\leq T}\bigl(y_{t}v_{t}-\mathbb{E}[y_{t}v_{t}%
]\bigr)\nonumber \\
&  =T^{-1}\sum_{t\leq T}\mathbb{E}[y_{t}v_{t}]+T^{-1}\sum_{t\leq T}%
\bigl(y_{t}y_{e,t}-\mathbb{E}[y_{t}y_{e,t}]\bigr)-\phi \,T^{-1}\sum_{t\leq
T}\bigl(y_{t}p_{t}-\mathbb{E}[y_{t}p_{t}]\bigr)\nonumber \\
&  =T^{-1}\sum_{t\leq T}\mathbb{E}[y_{t}v_{t}]+O_{p}(T^{-1/2})=O_{p}(1),
\label{PL_V1a_1}%
\end{align}
where the third equality follows from Lemma \ref{Asy_L0}, and the final
equality uses (\ref{P_V_Est_1b}).

Since $\hat{\theta}_{0}(\hat{A})$ denotes the GIV estimator with weighting
matrix $W_{0,T}=\mathbf{I}_{2n}$ (see Step 2 of the Implementation Algorithm 1
in Appendix \ref{APP_0}), Assumption \ref{Asy_Cond_2}(ii) holds with
$W_{0}=\mathbf{I}_{2n}$. Therefore, under Assumptions \ref{ID},
\ref{Asy_Cond_1}, \ref{S}, and \ref{Asy_Cond_2}(i, iii), we can apply
(\ref{Asy_Dist_1}) of Theorem \ref{Asy_Dist} to obtain
\begin{equation}
\hat{\theta}_{0}(\hat{A})-\theta=O_{p}(T^{-1/2}). \label{PL_V1a_2}%
\end{equation}
Combining this with (\ref{P_Asy_Dist_1b}) yields
\begin{equation}
T^{-1}\sum_{t\leq T}y_{t}\hat{v}_{t}-T^{-1}\sum_{t\leq T}y_{t}v_{t}%
=-\bigl(\hat{\phi}_{0}(\hat{A})-\phi \bigr)T^{-1}\sum_{t\leq T}y_{t}p_{t}%
=O_{p}(T^{-1/2}). \label{PL_V1a_3}%
\end{equation}
The result for $b=v$ follows from (\ref{PL_V1a_1}) and (\ref{PL_V1a_3}). The
case $b=\varepsilon$ can be established analogously and is therefore
omitted.\hfill$Q.E.D.$

\bigskip

\begin{lemma}
\label{L_V_1b} Under Assumptions \ref{Asy_Cond_1}(i), \ref{S} and
\ref{Asy_Cond_3}(ii), we have for any\ $c_{t}\in \mathbb{R}$ with $\max_{t\leq
T}|c_{t}|\leq K$,%
\[
T^{-1}\sum_{t\leq T}c_{t}(y_{t}y_{t}^{\top}-{\mathbb{E}[}y_{t}y_{t}^{\top
}])=o_{p}(1)\text{ \  \  \ and \  \  \ }T^{-1}\sum_{t\leq T}c_{t}(p_{t}%
y_{t}-{\mathbb{E}[}p_{t}y_{t}])=o_{p}(1).
\]

\end{lemma}

\noindent \textsc{Proof of Lemma \ref{L_V_1b}.} Using the reduced-form
representations of $y_{t}$ and $p_{t}$ in (\ref{P_Asy_L0_1}) and
(\ref{P_Asy_L0_2}), the stated results follow by the same argument as in the
proof of Lemma \ref{Asy_L0}, with Assumption \ref{Asy_Cond_1}(ii) replaced by
Assumption \ref{Asy_Cond_3}(ii).\hfill$Q.E.D.$

\bigskip

\begin{lemma}
\label{L_V_2} Under Assumptions \ref{Asy_Cond_1}(i), \ref{S}(ii) and
\ref{Asy_Cond_3}(iii, iv), we have
\[
T^{-1}\sum_{t\leq T}(y_{a_{1},t}y_{a_{2},t}y_{t}y_{t}^{\top}-\mathbb{E}%
[y_{a_{1},t}y_{a_{2},t}y_{t}y_{t}^{\top}])=o_{p}(1),
\]
for any\ $a_{1},a_{2}\in \{a\in \mathbb{R}^{n}:$ $\left \Vert a\right \Vert \leq
K\}$.
\end{lemma}

\noindent \textsc{Proof of Lemma \ref{L_V_2}.}\ From the reduced form
expression of $y_{t}$ in (\ref{P_Asy_L0_1}), the triangle inequality, and the
Cauchy-Schwarz inequality,
\begin{align*}
\left \Vert y_{t}\right \Vert  &  \leq \left(  1+\frac{|\phi \psi|n^{1/2}%
\left \Vert S_{t}\right \Vert }{|1-\phi \psi|}\right)  (\left \Vert u_{t}%
\right \Vert +\left \Vert \lambda \right \Vert \left \Vert \eta_{t}\right \Vert
)+\frac{|\phi|n^{1/2}\left \vert \varepsilon_{t}\right \vert }{|1-\phi \psi|}\\
&  \leq K(\left \Vert u_{t}\right \Vert +\left \Vert \eta_{t}\right \Vert
+\left \vert \varepsilon_{t}\right \vert )
\end{align*}
where the second inequality follows by Assumptions \ref{Asy_Cond_1}(i) and
\ref{S}(ii). Therefore, by Assumption \ref{Asy_Cond_3}(iv),
\[
\mathbb{E}[\left \Vert y_{t}\right \Vert ^{4}]\leq K(\mathbb{E}[\left \Vert
u_{t}\right \Vert ^{4}+\left \Vert \eta_{t}\right \Vert ^{4}+\varepsilon_{t}%
^{4}])\leq K.
\]
This together with the Cauchy-Schwarz inequality implies that
\[
\big \lVert \mathbb{E}[y_{a_{1},t}y_{a_{2},t}y_{t}y_{t}^{\top}%
]\big \lVert \leq K,
\]
for any $a_{1},a_{2}\in \{a\in \mathbb{R}^{n}:$ $\left \Vert a\right \Vert \leq
K\}$. Since $y_{t}$ is a linear function of $u_{t}$, $\eta_{t}$,\ $\varepsilon
_{t}$, $S_{t}^{\top}u_{t}$ and $(\eta_{t}^{\top}\otimes S_{t}^{\top
})\mathrm{vec}(\lambda)$,\ each entry of $y_{t}y_{t}^{\top}$ is a finite
linear combination of random variables of the form
\[
b_{1,t}=\prod_{j=1}^{2}z_{j,t},
\]
where each $z_{j,t}$ belongs to $\{u_{i,t},\eta_{k,t},\varepsilon_{t}%
,S_{t}^{\top}u_{t},(\eta_{t}^{\top}\otimes S_{t}^{\top})\mathrm{vec}%
(\lambda):i\leq n,k\leq r\}$, with coefficients that do not depend on $t$.
Since%
\[
y_{a_{1},t}y_{a_{2},t}=a_{1}^{\top}y_{t}y_{t}^{\top}a_{2},
\]
and $\Vert a_{1}\Vert,\Vert a_{2}\Vert \leq K$, the scalar $y_{a_{1},t}%
y_{a_{2},t}$ is also a finite linear combination of terms of the form
$w_{1,t}$, with coefficients bounded uniformly over admissible $a_{1},a_{2}$.
Therefore, for each $(i_{1},i_{2})$, the $(i_{1},i_{2})$th entry of%
\[
y_{a_{1},t}y_{a_{2},t}y_{t}y_{t}^{\top}%
\]
is a finite linear combination of random variables of the form
\[
b_{t}=\prod_{j=1}^{4}z_{j,t}%
\]
with coefficients uniformly bounded over admissible $a_{1},a_{2}$, where each
$z_{j,t}$ belongs to the same set above.\ Assumption \ref{Asy_Cond_3}(iii)
then yields
\[
T^{-1}\sum_{t\leq T}\bigl(b_{t}-\mathbb{E}[b_{t}]\bigr)=o_{p}(1)
\]
for each such fourth-order polynomial. Since $n$ is fixed, each entry of
\[
T^{-1}\sum_{t\leq T}\bigl(y_{a_{1},t}y_{a_{2},t}y_{t}y_{t}^{\top}%
-\mathbb{E}[y_{a_{1},t}y_{a_{2},t}y_{t}y_{t}^{\top}]\bigr)
\]
is a finite linear combination of $o_{p}(1)$ terms, and is therefore
$o_{p}(1)$. This proves the claim.\hfill$Q.E.D.$

\bigskip

\begin{lemma}
\label{L_V_3} Under Assumptions \ref{Asy_Cond_1}(i), \ref{S}(ii) and
\ref{Asy_Cond_3}(iii, iv), we have for any $a\in \mathbb{R}^{n}$ with $\Vert
a\Vert \leq K$
\[
T^{-1}\sum_{t\leq T}\bigl(y_{a,t}p_{t}y_{t}y_{t}^{\top}-\mathbb{E}%
[y_{a,t}p_{t}y_{t}y_{t}^{\top}]\bigr)=o_{p}(1),
\]
and
\[
T^{-1}\sum_{t\leq T}\bigl(p_{t}^{2}y_{t}y_{t}^{\top}-\mathbb{E}[p_{t}^{2}%
y_{t}y_{t}^{\top}]\bigr)=o_{p}(1).
\]

\end{lemma}

\noindent \textsc{Proof of Lemma \ref{L_V_3}.} Under Assumption
\ref{Asy_Cond_1}(i), it follows from (\ref{P_Asy_L0_2}) that $p_{t}$ is a
linear function of $S_{t}^{\top}u_{t}$, $\eta_{t}\otimes S_{t}$, and
$\varepsilon_{t}$. The stated results then follow by the same argument as in
the proof of Lemma \ref{L_V_2}.\hfill$Q.E.D.$

\bigskip

\begin{lemma}
\label{L_V_4}\ Under Assumptions \ref{Asy_Cond_1}(i), \ref{S}(ii)\ and
\ref{Asy_Cond_3}(ii, iii, iv), we have%
\[
T^{-1}\sum_{t\leq T}\xi_{t}\xi_{t}^{\top}=T^{-1}\sum_{t\leq T}\mathbb{E}%
[\xi_{t}\xi_{t}^{\top}]+o_{p}(1).
\]

\end{lemma}

\noindent \textsc{Proof of Lemma \ref{L_V_4}.} By the definition of $\xi_{t}$,
it suffices to show that
\begin{equation}
T^{-1}\sum_{t\leq T}\xi_{b,t}\xi_{b,t}^{\top}=T^{-1}\sum_{t\leq T}%
\mathbb{E}[\xi_{b,t}\xi_{b,t}^{\top}]+o_{p}(1), \label{P_L_V4_1}%
\end{equation}
for $b\in \{v,\varepsilon \}$, and
\begin{equation}
T^{-1}\sum_{t\leq T}\xi_{v,t}\xi_{\varepsilon,t}^{\top}=T^{-1}\sum_{t\leq
T}\mathbb{E}[\xi_{v,t}\xi_{\varepsilon,t}^{\top}]+o_{p}(1). \label{P_L_V4_2}%
\end{equation}
We establish (\ref{P_L_V4_1}) for $b=v$; the remaining cases follow analogously.

Recall that $\gamma \equiv \Upsilon \big(T^{-1}\sum_{t\leq T}\mathbb{E}%
[y_{t}v_{t}]\big)$, and%
\begin{align*}
\xi_{v,t}  &  \equiv y_{t}v_{t}-\mathbb{E}[y_{t}v_{t}]+(y_{t}y_{t}^{\top
}-\mathbb{E}[y_{t}y_{t}^{\top}])\, \Upsilon \Big(T^{-1}\sum_{t\leq T}%
\mathbb{E}[y_{t}v_{t}]\Big)\\
&  =y_{t}v_{t}-\mathbb{E}[y_{t}v_{t}]+(y_{\, \gamma,t}y_{t}-\mathbb{E}[y_{\,
\gamma,t}y_{t}]).
\end{align*}
This implies that
\begin{align}
T^{-1}\sum_{t\leq T}\xi_{v,t}\xi_{v,t}^{\top}  &  =T^{-1}\sum_{t\leq T}%
(v_{t}y_{t}-\mathbb{E}[v_{t}y_{t}])(v_{t}y_{t}^{\top}-\mathbb{E}[v_{t}%
y_{t}^{\top}])\nonumber \\
&  \quad+T^{-1}\sum_{t\leq T}(v_{t}y_{t}-\mathbb{E}[v_{t}y_{t}])(y_{\gamma
,t}y_{t}^{\top}-\mathbb{E}[y_{\gamma,t}y_{t}^{\top}])\nonumber \\
&  \quad+T^{-1}\sum_{t\leq T}(y_{\gamma,t}y_{t}-\mathbb{E}[y_{\gamma,t}%
y_{t}])(v_{t}y_{t}^{\top}-\mathbb{E}[v_{t}y_{t}^{\top}])\nonumber \\
&  \quad+T^{-1}\sum_{t\leq T}(y_{\gamma,t}y_{t}-\mathbb{E}[y_{\gamma,t}%
y_{t}])(y_{\gamma,t}y_{t}^{\top}-\mathbb{E}[y_{\gamma,t}y_{t}^{\top}]).
\label{P_L_V4_3}%
\end{align}
We next study the four terms on the right-hand side of (\ref{P_L_V4_3}) one by one.

\medskip \noindent \textit{Step 1: The first term.} Expanding the first term
yields
\begin{align}
&  T^{-1}\sum_{t\leq T}(y_{t}v_{t}-\mathbb{E}[y_{t}v_{t}])(y_{t}%
v_{t}-\mathbb{E}[y_{t}v_{t}])^{\top}\nonumber \\
&  =T^{-1}\sum_{t\leq T}v_{t}^{2}y_{t}y_{t}^{\top}+T^{-1}\sum_{t\leq
T}\mathbb{E}[v_{t}y_{t}]\mathbb{E}[v_{t}y_{t}^{\top}]-T^{-1}\sum_{t\leq
T}\bigl(\mathbb{E}[v_{t}y_{t}]v_{t}y_{t}^{\top}+v_{t}y_{t}\mathbb{E}%
[v_{t}y_{t}^{\top}]\bigr). \label{P_L_V4_4}%
\end{align}
Since $v_{t}=y_{e,t}-\phi p_{t}$, Lemmas \ref{L_V_2} and \ref{L_V_3} imply
that
\begin{align}
T^{-1}\sum_{t\leq T}(v_{t}^{2}y_{t}y_{t}^{\top}-\mathbb{E}[v_{t}^{2}y_{t}%
y_{t}^{\top}])  &  =T^{-1}\sum_{t\leq T}(y_{e,t}^{2}y_{t}y_{t}^{\top
}-\mathbb{E}[y_{e,t}^{2}y_{t}y_{t}^{\top}])\nonumber \\
&  \quad+\phi^{2}T^{-1}\sum_{t\leq T}(p_{t}^{2}y_{t}y_{t}^{\top}%
-\mathbb{E}[p_{t}^{2}y_{t}y_{t}^{\top}])\nonumber \\
&  \quad-2\phi T^{-1}\sum_{t\leq T}(y_{e,t}p_{t}y_{t}y_{t}^{\top}%
-\mathbb{E}[y_{e,t}p_{t}y_{t}y_{t}^{\top}])\nonumber \\
&  =o_{p}(1). \label{P_L_V4_5}%
\end{align}
Moreover, since $\max_{t\leq T}\Vert \mathbb{E}[y_{t}v_{t}]\Vert \leq K$ by
(\ref{P_V_Est_1b}), Lemma \ref{L_V_1b} implies that
\begin{align}
T^{-1}\sum_{t\leq T}\mathbb{E}[y_{t}v_{t}](y_{t}^{\top}v_{t}-\mathbb{E}%
[y_{t}^{\top}v_{t}])  &  =T^{-1}\sum_{t\leq T}\mathbb{E}[y_{t}v_{t}%
](y_{t}^{\top}y_{e,t}-\mathbb{E}[y_{t}^{\top}y_{e,t}])\nonumber \\
&  \quad-\phi T^{-1}\sum_{t\leq T}\mathbb{E}[y_{t}v_{t}](y_{t}^{\top}%
p_{t}-\mathbb{E}[y_{t}^{\top}p_{t}])\nonumber \\
&  =o_{p}(1). \label{P_L_V4_7}%
\end{align}
Combining (\ref{P_L_V4_4})--(\ref{P_L_V4_7}) yields
\begin{equation}
T^{-1}\sum_{t\leq T}(y_{t}v_{t}-\mathbb{E}[y_{t}v_{t}])(y_{t}v_{t}%
-\mathbb{E}[y_{t}v_{t}])^{\top}=T^{-1}\sum_{t\leq T}\mathbb{E}[(y_{t}%
v_{t}-\mathbb{E}[y_{t}v_{t}])(y_{t}v_{t}-\mathbb{E}[y_{t}v_{t}])^{\top}%
]+o_{p}(1). \label{P_L_V4_8}%
\end{equation}

\medskip \noindent \textit{Step 2: The cross terms.} For the second term on the
right-hand side of (\ref{P_L_V4_3}), we have
\begin{align}
&  T^{-1}\sum_{t\leq T}(y_{t}v_{t}-\mathbb{E}[y_{t}v_{t}])(y_{\gamma,t}%
y_{t}^{\top}-\mathbb{E}[y_{\gamma,t}y_{t}^{\top}])\nonumber \\
&  =T^{-1}\sum_{t\leq T}y_{\gamma,t}v_{t}y_{t}y_{t}^{\top}-T^{-1}\sum_{t\leq
T}y_{t}v_{t}\mathbb{E}[y_{\gamma,t}y_{t}^{\top}]\nonumber \\
&  \quad-T^{-1}\sum_{t\leq T}\mathbb{E}[y_{t}v_{t}]y_{\gamma,t}y_{t}^{\top
}+T^{-1}\sum_{t\leq T}\mathbb{E}[y_{t}v_{t}]\mathbb{E}[y_{\gamma,t}y_{t}%
^{\top}]. \label{P_L_V4_9}%
\end{align}
Using (\ref{P_V_Est_1c}), Lemmas \ref{L_V_2} and \ref{L_V_3}, we obtain
\begin{align}
T^{-1}\sum_{t\leq T}(y_{\gamma,t}v_{t}y_{t}y_{t}^{\top}-\mathbb{E}%
[y_{\gamma,t}v_{t}y_{t}y_{t}^{\top}])  &  =T^{-1}\sum_{t\leq T}(y_{\gamma
,t}y_{e,t}y_{t}y_{t}^{\top}-\mathbb{E}[y_{\gamma,t}y_{e,t}y_{t}y_{t}^{\top
}])\nonumber \\
&  -\phi T^{-1}\sum_{t\leq T}(y_{\gamma,t}p_{t}y_{t}y_{t}^{\top}%
-\mathbb{E}[y_{\gamma,t}p_{t}y_{t}y_{t}^{\top}])\overset{}{=}o_{p}(1).
\label{P_L_V4_10}%
\end{align}
Moreover, from (\ref{P_Asy_Dist_1a}) and (\ref{P_V_Est_1c}),
\begin{equation}
\max_{t\leq T}\Vert \mathbb{E}[y_{\gamma,t}y_{t}^{\top}]\Vert \leq K,
\label{P_L_V4_10b}%
\end{equation}
which, together with Lemma \ref{L_V_1b}, implies
\begin{align}
T^{-1}\sum_{t\leq T}(y_{t}v_{t}-\mathbb{E}[y_{t}v_{t}])\mathbb{E}[y_{\gamma
,t}y_{t}^{\top}]  &  =T^{-1}\sum_{t\leq T}(y_{t}y_{e,t}-\mathbb{E}%
[y_{t}y_{e,t}])\mathbb{E}[y_{\gamma,t}y_{t}^{\top}]\nonumber \\
&  -\phi T^{-1}\sum_{t\leq T}(y_{t}p_{t}-\mathbb{E}[y_{t}p_{t}])\mathbb{E}%
[y_{\gamma,t}y_{t}^{\top}]\overset{}{=}o_{p}(1). \label{P_L_V4_11}%
\end{align}
Similarly, we can show that
\begin{equation}
T^{-1}\sum_{t\leq T}\mathbb{E}[y_{t}v_{t}](y_{\gamma,t}y_{t}^{\top}%
-\mathbb{E}[y_{\gamma,t}y_{t}^{\top}])=o_{p}(1). \label{P_L_V4_12}%
\end{equation}
From (\ref{P_L_V4_9})-(\ref{P_L_V4_12}), we conclude that
\begin{align}
&  T^{-1}\sum_{t\leq T}(y_{t}v_{t}-\mathbb{E}[y_{t}v_{t}])(y_{\gamma,t}%
y_{t}^{\top}-\mathbb{E}[y_{\gamma,t}y_{t}^{\top}])\, \nonumber \\
&  =T^{-1}\sum_{t\leq T}\mathbb{E}\left[  (y_{t}v_{t}-\mathbb{E}[y_{t}%
v_{t}])(y_{\gamma,t}y_{t}^{\top}-\mathbb{E}[y_{\gamma,t}y_{t}^{\top}])\right]
+o_{p}(1).\, \, \label{P_L_V4_13}%
\end{align}
The third term in (\ref{P_L_V4_3}) can be handled analogously.

\medskip \noindent \textit{Step 3: The last term.} Finally, for the fourth term
in (\ref{P_L_V4_3}), we have%
\begin{align}
&  T^{-1}\sum_{t\leq T}(y_{\gamma,t}y_{t}-\mathbb{E}[y_{\gamma,t}%
y_{t}])(y_{\gamma,t}y_{t}^{\top}-\mathbb{E}[y_{\gamma,t}y_{t}^{\top
}])\nonumber \\
&  =T^{-1}\sum_{t\leq T}y_{\gamma,t}^{2}y_{t}y_{t}^{\top}-T^{-1}\sum_{t\leq
T}\mathbb{E}[y_{\gamma,t}y_{t}]y_{\gamma,t}y_{t}^{\top}\nonumber \\
&  -T^{-1}\sum_{t\leq T}y_{\gamma,t}y_{t}\mathbb{E}[y_{\gamma,t}y_{t}^{\top
}]+T^{-1}\sum_{t\leq T}\mathbb{E}[y_{\gamma,t}y_{t}]\mathbb{E}[y_{\gamma
,t}y_{t}^{\top}]\nonumber \\
&  =T^{-1}\sum_{t\leq T}\mathbb{E}[y_{\gamma,t}^{2}y_{t}y_{t}^{\top}%
]-T^{-1}\sum_{t\leq T}\mathbb{E}[y_{\gamma,t}y_{t}]\mathbb{E}[y_{\gamma
,t}y_{t}^{\top}]+o_{p}(1) \label{P_L_V4_14}%
\end{align}
where the second equality holds by (\ref{P_L_V4_10b}), Lemma \ref{L_V_1b} and
Lemma \ref{L_V_2}.

Combining (\ref{P_L_V4_3}), (\ref{P_L_V4_8}), (\ref{P_L_V4_13}), and
(\ref{P_L_V4_14}) establishes (\ref{P_L_V4_1}) for $b=v$, completing the
proof.\hfill$Q.E.D.$

\bigskip

\begin{lemma}
\label{L_V_5} Under Assumptions \ref{ID}, \ref{Asy_Cond_1}, \ref{S},
\ref{Asy_Cond_2}(i, iii) and \ref{Asy_Cond_3}, we have%
\[
T^{-1}\sum_{t\leq T}(\hat{\xi}_{t}-\xi_{t})^{\top}(\hat{\xi}_{t}-\xi
_{t})=o_{p}(1).
\]

\end{lemma}

\noindent \textsc{Proof of Lemma \ref{L_V_5}.} Let $\hat{\gamma}\equiv
\hat{\Upsilon}T^{-1}\sum_{t\leq T}y_{t}\hat{v}_{t}$. Then, by Lemma
\ref{L_Upsilon} and Lemma \ref{L_V_1a},
\begin{equation}
\hat{\gamma}-\gamma=(\hat{\Upsilon}-\Upsilon)T^{-1}\sum_{t\leq T}y_{t}\hat
{v}_{t}+\Upsilon \Big(T^{-1}\sum_{t\leq T}y_{t}\hat{v}_{t}-T^{-1}\sum_{t\leq
T}\mathbb{E}[y_{t}v_{t}]\Big)=o_{p}(1). \label{P_L_V_5_1}%
\end{equation}
Because%
\begin{align*}
\hat{\xi}_{v,t}  &  \equiv y_{t}\hat{v}_{t}-T^{-1}\sum_{t\leq T}y_{t}\hat
{v}_{t}+\Big(y_{t}y_{t}^{\top}-T^{-1}\sum_{t\leq T}y_{t}y_{t}^{\top}%
\Big)\hat{\Upsilon}\Big(T^{-1}\sum_{t\leq T}y_{t}\hat{v}_{t}\Big)\\
&  =y_{t}\hat{v}_{t}-T^{-1}\sum_{t\leq T}y_{t}\hat{v}_{t}+\Big(y_{t}%
y_{t}^{\top}-T^{-1}\sum_{t\leq T}y_{t}y_{t}^{\top}\Big)\hat{\gamma}%
\end{align*}
and
\[
\xi_{v,t}=y_{t}v_{t}-\mathbb{E}[y_{t}v_{t}]+(y_{\, \gamma,t}y_{t}%
-\mathbb{E}[y_{\, \gamma,t}y_{t}]),
\]
we have
\begin{equation}
\hat{\xi}_{v,t}-\xi_{v,t}=y_{t}(\hat{v}_{t}-v_{t})+\left(  y_{t}y_{t}^{\top
}-\mathbb{E}[y_{t}y_{t}^{\top}]\right)  (\hat{\gamma}-\gamma)-\Delta_{T},
\label{P_L_V_5_2}%
\end{equation}
where
\[
\Delta_{T}=T^{-1}\sum_{t\leq T}(y_{t}\hat{v}_{t}-\mathbb{E}[y_{t}%
v_{t}])+T^{-1}\sum_{t\leq T}(y_{t}y_{t}^{\top}-\mathbb{E}[y_{t}y_{t}^{\top
}])\hat{\gamma}.
\]
Therefore by the Cauchy-Schwarz inequality,
\begin{align}
T^{-1}\sum_{t\leq T}(\hat{\xi}_{v,t}-\xi_{v,t})^{\top}(\hat{\xi}_{v,t}%
-\xi_{v,t})  &  \leq3T^{-1}\sum_{t\leq T}(\hat{v}_{t}-v_{t})^{2}y_{t}^{\top
}y_{t}+3||\Delta_{T}||^{2}\nonumber \\
&  \quad+3(\hat{\gamma}-\gamma)^{\top}T^{-1}\sum_{t\leq T}\left(  y_{t}%
y_{t}^{\top}-\mathbb{E}[y_{t}y_{t}^{\top}]\right)  ^{2}(\hat{\gamma}-\gamma).
\label{P_L_V_5_3}%
\end{align}

Since $\hat{v}_{t}-v_{t}=-(\hat{\phi}_{0}(\hat{A})-\phi)p_{t}$, we have
\begin{equation}
T^{-1}\sum_{t\leq T}(\hat{v}_{t}-v_{t})^{2}y_{t}^{\top}y_{t}=(\hat{\phi}%
_{0}(\hat{A})-\phi)^{2}T^{-1}\sum_{t\leq T}p_{t}^{2}y_{t}^{\top}y_{t}.
\label{P_L_V_5_4}%
\end{equation}
Under Assumptions \ref{Asy_Cond_1}(i) and \ref{Asy_Cond_3}(iv), together with
the reduced-form expressions for $y_{t}$ and $p_{t}$ in (\ref{P_Asy_L0_1}) and
(\ref{P_Asy_L0_2}), we have
\begin{equation}
\max_{t\leq T}\left(  \mathbb{E}[p_{t}^{2}y_{t}^{\top}y_{t}]+\mathbb{E}%
[(y_{t}^{\top}y_{t})^{2}]\right)  \leq K. \label{P_L_V_5_5}%
\end{equation}
Hence, by Markov's inequality,
\begin{equation}
T^{-1}\sum_{t\leq T}p_{t}^{2}y_{t}^{\top}y_{t}=O_{p}(1), \label{P_L_V_5_6}%
\end{equation}
which, together with (\ref{PL_V1a_2}) and (\ref{P_L_V_5_4}), implies
\begin{equation}
T^{-1}\sum_{t\leq T}(\hat{v}_{t}-v_{t})^{2}y_{t}^{\top}y_{t}=o_{p}(1).
\label{P_L_V_5_7}%
\end{equation}

From Lemma \ref{Asy_L0}, Lemma \ref{L_V_1a}, (\ref{P_V_Est_1c}), and
(\ref{P_L_V_5_1}), it follows that
\begin{equation}
||\Delta_{T}||\leq \Big \lVert T^{-1}\sum_{t\leq T}(y_{t}\hat{v}_{t}%
-\mathbb{E}[y_{t}v_{t}])\Big \lVert+\Big \lVert T^{-1}\sum_{t\leq T}%
(y_{t}y_{t}^{\top}-\mathbb{E}[y_{t}y_{t}^{\top}])\hat{\gamma}\Big \lVert=o_{p}%
(1). \label{P_L_V_5_8}%
\end{equation}

Finally, by (\ref{P_L_V_5_5}), the Cauchy-Schwarz inequality and Markov's
inequality,
\[
T^{-1}\sum_{t\leq T}\left(  y_{t}y_{t}^{\top}-\mathbb{E}[y_{t}y_{t}^{\top
}]\right)  ^{2}=O_{p}(1),
\]
which, together with (\ref{P_L_V_5_1}), implies
\begin{equation}
(\hat{\gamma}-\gamma)^{\top}T^{-1}\sum_{t\leq T}\left(  y_{t}y_{t}^{\top
}-\mathbb{E}[y_{t}y_{t}^{\top}]\right)  ^{2}(\hat{\gamma}-\gamma)=o_{p}(1).
\label{P_L_V_5_9}%
\end{equation}
Combining (\ref{P_L_V_5_3}), (\ref{P_L_V_5_7}), (\ref{P_L_V_5_8}), and
(\ref{P_L_V_5_9}) yields
\[
T^{-1}\sum_{t\leq T}(\hat{\xi}_{v,t}-\xi_{v,t})^{\top}(\hat{\xi}_{v,t}%
-\xi_{v,t})=o_{p}(1).
\]

Similarly, we can show that%
\[
T^{-1}\sum_{t\leq T}(\hat{\xi}_{\varepsilon,t}-\xi_{\varepsilon,t})^{\top
}(\hat{\xi}_{\varepsilon,t}-\xi_{\varepsilon,t})=o_{p}(1).
\]
Combining these results establishes the claim of the lemma.\hfill$Q.E.D.$

\section{Asymptotic Variance with Exogenous Regressors\label{APP_5}}

This section studies the asymptotic variance of the GIV estimator proposed in
Subsection \ref{subsec: Ex1}, where both the supply and demand equations
include exogenous regressors. The key challenge is to account for the role of
the estimation error in $\beta$ in the asymptotic variance of the GIV
estimator. As we show below, this estimation error is asymptotically
negligible and can therefore be ignored. Consequently, the standard errors of
the GIV estimator, as well as the $J$-test, can be constructed as if $\beta$
were known. We begin by introducing a set of conditions that characterize the
properties of the additional regressors in the model.

\begin{assumption}
\label{Asy_Cond_EX} (i) For any $t\leq T$ and $i\leq n$,\ we have
$\mathbb{E}[x_{i,t}(u_{t}^{\top},\eta_{t}^{\top})]=\mathbf{0}_{d_{x}%
\times(n+r)}$, $\mathbb{E}[x_{i,t}\varepsilon_{t}]=\mathbf{0}_{d_{x}}%
$\ and\ $\mathbb{E}[w_{t}\varepsilon_{t}]=\mathbf{0}_{d_{w}\times1}%
$;\ (ii)\ let $\tilde{\epsilon}_{t}\equiv \tilde{\lambda}\eta_{t}+\tilde{u}%
_{t}$, then
\[
T^{1/2}(\hat{\beta}-\beta)=\Bigl(T^{-1}\sum_{t\leq T}\mathbb{E}[\tilde{x}%
_{t}^{\top}\tilde{x}_{t}]\Bigr)^{-1}\Bigl(T^{-1/2}\sum_{t\leq T}\tilde{x}%
_{t}^{\top}\tilde{\epsilon}_{t}\Bigr)+o_{p}(1),
\]
where$\  \rho_{\min}(T^{-1}\sum_{t\leq T}\mathbb{E}[\tilde{x}_{t}^{\top}%
\tilde{x}_{t}])\geq K^{-1}$ and $T^{-1/2}\sum_{t\leq T}\tilde{x}_{t}^{\top
}\tilde{\epsilon}_{t}=O_{p}(1)$; (iii) for $a\in \{u,\eta,\varepsilon \}$,%
\[
T^{-1}\sum_{t\leq T}\big(x_{t}\otimes a_{t}-\mathbb{E}[x_{t}\otimes
a_{t}]\big)=O_{p}(T^{-1/2});
\]
(iv) $\max_{t\leq T}\mathbb{E}[||x_{t}||^{2}+||w_{t}||^{2}]\leq K$,
$\left \Vert \beta \right \Vert \leq K$ and $\left \Vert \gamma \right \Vert \leq K$.
\end{assumption}

Assumption \ref{Asy_Cond_EX}(i) ensures that $x_{t}$ and $w_{t}$ are exogenous
in the demand and supply equations, respectively. Assumption \ref{Asy_Cond_EX}%
(ii) provides a linear representation of the estimation error in $\hat{\beta}%
$, implying that $\hat{\beta}$ is $T^{1/2}$-consistent. Assumption
\ref{Asy_Cond_EX}(iii) imposes a $T^{-1/2}$ law of large numbers for products
of $x_{t}$ and the demand and supply shocks. Assumption \ref{Asy_Cond_EX}(iv)
places uniform bounds on the second moments of the exogenous regressors, as
well as on the magnitudes of their coefficients.

Since the matrix $\hat{A}$ in the moment conditions $\bar{g}_{T}(\theta
;\hat{A},\hat{\beta})$ in (\ref{F_Moments}) plays the same role as in $\bar
{g}_{T}(\theta;\hat{A})$ defined in (\ref{Moment_Func}), we can apply the same
arguments as in the proof of Lemma \ref{Invariance} to show that both the GIV
estimator $\hat{\theta}(\hat{A})$ in (\ref{G_GIV_1}) and the $J$-statistic
\[
T\, \bar{g}_{T}(\hat{\theta}(\hat{A});\hat{A},\hat{\beta})^{\top}W_{0,T}%
(\hat{A})\bar{g}_{T}(\hat{\theta}(\hat{A});\hat{A},\hat{\beta})
\]
are invariant to nonsingular rotations of $\hat{A}$ of the form $\hat{A}C_{1}%
$, for any nonsingular $(n-\bar{r})\times(n-\bar{r})$ matrix $C_{1}$.

Moreover, Lemmas \ref{Asy_L0} (with $y_{t}$ in that lemma corresponding to
$y_{t}^{\ast}$ here) and \ref{Sigma_est} imply that
\[
Q_{-1}^{\top}(\hat{\Sigma}_{\hat{y}^{\ast}}-\bar{\Sigma}_{y^{\ast}}%
)Q_{-1}=O_{p}(T^{-1/2}).
\]
Therefore, we can use the same arguments as in the proof of Lemma
\ref{Asy_Ahat_L1} to obtain
\begin{equation}
\hat{A}_{0}\hat{H}_{0}^{\top}-A_{0}=A_{0,\bot}(\bar{\sigma}_{u}^{2}%
\mathbf{I}_{\bar{r}-1}-\Lambda_{\bot})^{-1}A_{0,\bot}^{\top}Q_{-1}^{\top}%
(\hat{\Sigma}_{\hat{y}^{\ast}}-\bar{\Sigma}_{y^{\ast}})Q_{-1}A_{0}%
+O_{p}(T^{-1}), \label{A_hat_Con_Ex}%
\end{equation}
where $\hat{H}_{0}$ is nonsingular with $\hat{H}_{0}^{\top}\hat{H}%
_{0}=\mathbf{I}_{n-\bar{r}}$ wpa1. Here $A_{0}$, $A_{0,\bot}$, $\Lambda_{\bot
}$, and $\hat{H}_{0}$ are defined by replacing $\bar{\Sigma}_{y}$ and
$\hat{\Sigma}_{y}$ in their original definitions with $\bar{\Sigma}_{y^{\ast}%
}$ and $\hat{\Sigma}_{\hat{y}^{\ast}}$, respectively.

Since the contribution of the estimation error in $\beta$ to the asymptotic
variance of the GIV estimator arises through $\bar{g}_{T}(\theta;\hat{A}%
,\hat{\beta})$, and given the invariance of both the GIV estimator
$\hat{\theta}(\hat{A})$ in (\ref{G_GIV_1}) and the $J$-statistic, it suffices
to study the asymptotic variance of $\bar{g}_{T}(\theta;\tilde{A},\hat{\beta
})$ in order to assess the effect of the estimation error in $\beta$. The
asymptotic expansion of $\bar{g}_{T}(\theta;\tilde{A},\hat{\beta})$ is given
in the following lemma.

\begin{lemma}
\label{Moment_est}\ Under Assumptions \ref{ID}, \ref{Asy_Cond_1},
\ref{S} and \ref{Asy_Cond_EX}, we have
\[
\bar{g}_{T}(\theta;\tilde{A},\hat{\beta})=\mathrm{diag}(A^{\top},A^{\top
},\mathbf{I}_{d_{w}})T^{-1}\sum_{t\leq T}\bar{\xi}_{t}+O_{p}(T^{-1}),
\]
where $\bar{\xi}_{t}\equiv(\bar{\xi}_{v,t}^{\top},\bar{\xi}_{\varepsilon
,t}^{\top},w_{t}^{\top}\varepsilon_{t})^{\top}$, and for $b\in \{v,\varepsilon
\}$,
\begin{equation}
\bar{\xi}_{b,t}\equiv y_{t}^{\ast}b_{t}-\mathbb{E}[y_{t}^{\ast}b_{t}%
]+(y_{t}^{\ast}y_{t}^{\ast \top}-\mathbb{E}[y_{t}^{\ast}y_{t}^{\ast \top
}])\Upsilon \Bigl(T^{-1}\sum_{s\leq T}\mathbb{E}[y_{s}^{\ast}b_{s}]\Bigr).
\label{bar_zeta_b}%
\end{equation}

\end{lemma}

\noindent \textsc{Proof of Lemma \ref{Moment_est}.\ }The result follows
immediately from Lemmas \ref{Sigma_est}, \ref{Moment_est_1}, and
\ref{Moment_est_2}, and is therefore omitted.\hfill$Q.E.D.$

Lemma \ref{Moment_est} shows that the asymptotic variance of $\hat{\theta
}(\hat{A})$ is determined by that of $T^{-1/2}\sum_{t\leq T}\bar{\xi}_{t}$.
Compared with $\xi_{t}$ defined in Assumption \ref{Asy_Cond_2}(i), $\bar{\xi
}_{t}$ contains the additional component $w_{t}^{\top}\varepsilon_{t}$, which
arises from the moment conditions used to identify and estimate the
coefficients on $w_{t}$ in the supply equation. More importantly, $\bar{\xi
}_{b,t}$ in (\ref{bar_zeta_b}) has the same form as $\xi_{b,t}$ in
(\ref{zeta_b}), with the only difference being that $y$ there is replaced by
$y^{\ast}$ here. Therefore, the estimation error in $\beta$ does not affect
the asymptotic variance of $\bar{g}_{T}(\theta;\tilde{A},\hat{\beta})$, and
hence does not affect that of the GIV estimator $\hat{\theta}(\hat{A})$ in
(\ref{G_GIV_1}). As a result, the standard errors of $\hat{\theta}(\hat{A})$,
the optimal weight matrix, and the $J$-statistic can be constructed as if
$\beta$ were known.

\bigskip

\begin{lemma}
\label{Sigma_est}\ Under Assumptions \ref{ID}, \ref{Asy_Cond_1},\  \ref{S} and \ref{Asy_Cond_EX}, we have
\[
Q_{-1}^{\top}(\hat{\Sigma}_{\hat{y}^{\ast}}-\bar{\Sigma}_{y^{\ast}}%
)Q_{-1}=Q_{-1}^{\top}T^{-1}\sum_{t\leq T}\big(y_{t}^{\ast}y_{t}^{\ast \top
}-\mathbb{E}[y_{t}^{\ast}y_{t}^{\ast \top}]\big)Q_{-1}+O_{p}(T^{-1}).
\]

\end{lemma}

\noindent \textsc{Proof of Lemma \ref{Sigma_est}.}\ By the definitions of
$\hat{y}_{t}^{\ast}$ and $\hat{y}_{e,t}^{\ast}$,
\begin{equation}
\hat{y}_{t}^{\ast}-y_{t}^{\ast}=-x_{t}(\hat{\beta}-\beta),\qquad \  \hat
{y}_{e,t}^{\ast}-y_{e,t}^{\ast}=-x_{e,t}(\hat{\beta}-\beta),
\label{P_Sigma_est_0}%
\end{equation}
where $x_{e,t}\equiv e^{\top}x_{t}$. It follows that
\begin{align}
\hat{\Sigma}_{\hat{y}^{\ast}}-\bar{\Sigma}_{y^{\ast}}  &  =T^{-1}\sum_{t\leq
T}\hat{y}_{t}^{\ast}\hat{y}_{t}^{\ast \top}-T^{-1}\sum_{t\leq T}\mathbb{E}%
[y_{t}^{\ast}y_{t}^{\ast \top}]\nonumber \\
&  =T^{-1}\sum_{t\leq T}(y_{t}^{\ast}-x_{t}(\hat{\beta}-\beta))(y_{t}^{\ast
}-x_{t}(\hat{\beta}-\beta))^{\top}-T^{-1}\sum_{t\leq T}\mathbb{E}[y_{t}^{\ast
}y_{t}^{\ast \top}]\nonumber \\
&  =T^{-1}\sum_{t\leq T}\big(y_{t}^{\ast}y_{t}^{\ast \top}-\mathbb{E}%
[y_{t}^{\ast}y_{t}^{\ast \top}]\big)-T^{-1}\sum_{t\leq T}y_{t}^{\ast}%
(\hat{\beta}-\beta)^{\top}x_{t}^{\top}\nonumber \\
&  \quad-T^{-1}\sum_{t\leq T}x_{t}(\hat{\beta}-\beta)y_{t}^{\ast \top}%
+T^{-1}\sum_{t\leq T}x_{t}(\hat{\beta}-\beta)(\hat{\beta}-\beta)^{\top}%
x_{t}^{\top}. \label{P_Sigma_est_1}%
\end{align}

From Assumption \ref{Asy_Cond_EX}(ii),
\begin{equation}
\hat{\beta}-\beta=O_{p}(T^{-1/2}). \label{P_Sigma_est_2}%
\end{equation}
By the triangle and Cauchy--Schwarz inequalities,
\begin{equation}
\Big \lVert T^{-1}\sum_{t\leq T}x_{t}(\hat{\beta}-\beta)(\hat{\beta}%
-\beta)^{\top}x_{t}^{\top}\Big \lVert \leq \Vert \hat{\beta}-\beta \Vert
^{2}\,T^{-1}\sum_{t\leq T}\Vert x_{t}\Vert^{2}=O_{p}(T^{-1}),
\label{P_Sigma_est_3}%
\end{equation}
where the last equality follows from Assumption \ref{Asy_Cond_EX}(iv) and
(\ref{P_Sigma_est_2}).

Let $\lambda_{i}$ denote the $i$th row of $\lambda$. By Assumptions
\ref{Asy_Cond_EX}(i, iii)
\begin{align}
T^{-1}\sum_{t\leq T}y_{i,t}^{\ast}x_{t}  &  =T^{-1}\sum_{t\leq T}(\phi
p_{t}+\lambda_{i}\eta_{t}+u_{i,t})x_{t}\nonumber \\
&  =\phi T^{-1}\sum_{t\leq T}p_{t}x_{t}+T^{-1}\sum_{t\leq T}\lambda_{i}%
\eta_{t}x_{t}+T^{-1}\sum_{t\leq T}u_{i,t}x_{t}\nonumber \\
&  =\phi T^{-1}\sum_{t\leq T}p_{t}x_{t}+O_{p}(T^{-1/2}). \label{P_Sigma_est_4}%
\end{align}
Under Assumptions \ref{Asy_Cond_1}(i, iv) and \ref{Asy_Cond_EX}(iv), the
demand and supply equations (\ref{F_demand})-(\ref{F_supply}) admit the
following reduced form expressions%
\begin{align}
y_{t}  &  =\Bigl(\mathbf{I}_{n}+\frac{\phi \psi}{1-\phi \psi}\mathbf{1}_{n}%
S_{t}^{\top}\Bigr)(u_{t}+x_{t}\beta+\lambda \eta_{t})+\frac{\phi \mathbf{1}_{n}%
}{1-\phi \psi}(\varepsilon_{t}+w_{t}^{\top}\gamma),\label{P_Sigma_est_5}\\
p_{t}  &  =\frac{\psi}{1-\phi \psi}S_{t}^{\top}(u_{t}+x_{t}\beta+\lambda
\eta_{t})+\frac{\varepsilon_{t}+w_{t}^{\top}\gamma}{1-\phi \psi}.
\label{P_Sigma_est_6}%
\end{align}
From Assumptions \ref{Asy_Cond_1}(i, iv) and \ref{Asy_Cond_EX}(iv), it follows
that%
\[
\mathbb{E}[p_{t}^{2}]\leq K.
\]
Together with Assumption \ref{Asy_Cond_EX}(iv), the Cauchy--Schwarz
inequality, and Markov's inequality, this implies that%
\begin{equation}
T^{-1}\sum_{t\leq T}p_{t}x_{t}=O_{p}(1). \label{P_Sigma_est_7}%
\end{equation}

Combining (\ref{P_Sigma_est_2}), (\ref{P_Sigma_est_4}) and
(\ref{P_Sigma_est_7}),
\begin{equation}
T^{-1}\sum_{t\leq T}x_{t}(\hat{\beta}-\beta)y_{t}^{\ast \top}=\phi
\Bigl(T^{-1}\sum_{t\leq T}p_{t}x_{t}\Bigr)(\hat{\beta}-\beta)\mathbf{1}%
_{n}^{\top}+O_{p}(T^{-1}). \label{P_Sigma_est_8}%
\end{equation}
Since $\mathbf{1}_{n}^{\top}Q_{-1}=\mathbf{0}_{n-1}$, it follows that
\[
Q_{-1}^{\top}\Bigl(T^{-1}\sum_{t\leq T}x_{t}(\hat{\beta}-\beta)y_{t}^{\ast
\top}\Bigr)Q_{-1}=O_{p}(T^{-1}).
\]
The same bound holds for its transpose. Combining these bounds with the
quadratic term gives the desired result.\hfill$Q.E.D.$

\bigskip

\begin{lemma}
\label{Moment_est_1}\ Under Assumptions \ref{ID}, \ref{Asy_Cond_1},
\ref{S} and \ref{Asy_Cond_EX}, we have%
\begin{align*}
T^{-1/2}\sum_{t\leq T}\tilde{A}^{\top}\hat{y}_{t}^{\ast}(\hat{y}_{e,t}^{\ast
}-\phi p_{t})  &  =A^{\top}T^{-1/2}\sum_{t\leq T}(y_{t}^{\ast}v_{t}%
-\mathbb{E}[y_{t}^{\ast}v_{t}])\\
&  \quad+A^{\top}\Bigl(T^{-1/2}\sum_{t\leq T}(y_{t}^{\ast}y_{t}^{\ast \top
}-\mathbb{E}[y_{t}^{\ast}y_{t}^{\ast \top}])\Bigr)\Upsilon \Bigl(T^{-1}%
\sum_{t\leq T}\mathbb{E}[y_{t}^{\ast}v_{t}]\Bigr)\\
&  \quad+O_{p}(T^{-1/2}).
\end{align*}

\end{lemma}

\noindent \textsc{Proof of Lemma \ref{Moment_est_1}.} We begin with the
decomposition
\begin{align}
T^{-1/2}\sum_{t\leq T}\tilde{A}^{\top}\hat{y}_{t}^{\ast}(\hat{y}_{e,t}^{\ast
}-\phi p_{t})  &  =T^{-1/2}\sum_{t\leq T}A^{\top}y_{t}^{\ast}(y_{e,t}^{\ast
}-\phi p_{t})+T^{-1/2}\sum_{t\leq T}A^{\top}y_{t}^{\ast}(\hat{y}_{e,t}^{\ast
}-y_{e,t}^{\ast})\nonumber \\
&  \quad+A^{\top}T^{-1/2}\sum_{t\leq T}(\hat{y}_{t}^{\ast}-y_{t}^{\ast}%
)(\hat{y}_{e,t}^{\ast}-\phi p_{t})+(\tilde{A}-A)^{\top}T^{-1/2}\sum_{t\leq
T}\hat{y}_{t}^{\ast}(\hat{y}_{e,t}^{\ast}-\phi p_{t}). \label{P_Moment_est1_1}%
\end{align}
Since $v_{t}\equiv y_{e,t}^{\ast}-\phi p_{t}$, by (\ref{P_Asy_L5_6}) we have
\begin{equation}
T^{-1/2}\sum_{t\leq T}A^{\top}y_{t}^{\ast}(y_{e,t}^{\ast}-\phi p_{t}%
)=T^{-1/2}\sum_{t\leq T}A^{\top}y_{t}^{\ast}v_{t}=A^{\top}T^{-1/2}\sum_{t\leq
T}(y_{t}^{\ast}v_{t}-\mathbb{E}[y_{t}^{\ast}v_{t}]). \label{P_Moment_est1_1b}%
\end{equation}
We next control the last three terms on the right-hand side of
(\ref{P_Moment_est1_1}).

\medskip \noindent \textit{Step 1: Second term.} Using $\hat{\beta}-\beta
=O_{p}(T^{-1/2})$ in (\ref{P_Sigma_est_2}) and the same arguments as in
(\ref{P_Sigma_est_4}), we have%
\begin{align}
T^{-1/2}\sum_{t\leq T}y_{t}^{\ast}(\hat{y}_{e,t}^{\ast}-y_{e,t}^{\ast})  &
=-T^{-1/2}\sum_{t\leq T}y_{t}^{\ast}x_{e,t}(\hat{\beta}-\beta)\nonumber \\
&  =-\mathbf{1}_{n}\Bigl(\phi T^{-1}\sum_{t\leq T}p_{t}x_{e,t}\Bigr)T^{1/2}%
(\hat{\beta}-\beta)+O_{p}(T^{-1/2}). \label{P_Moment_est1_2a}%
\end{align}
Since $A^{\top}\mathbf{1}_{n}=A_{0}^{\top}Q_{-1}^{\top}\mathbf{1}%
_{n}=\mathbf{0}_{n-1}$, it follows that
\begin{equation}
T^{-1/2}\sum_{t\leq T}A^{\top}y_{t}^{\ast}(\hat{y}_{e,t}^{\ast}-y_{e,t}^{\ast
})=O_{p}(T^{-1/2}). \label{P_Moment_est1_2}%
\end{equation}

\medskip \noindent \textit{Step 2: Third term.} By Assumptions \ref{Asy_Cond_1}%
(i) and \ref{Asy_Cond_EX}(i, iii),%
\[
T^{-1/2}\sum_{t\leq T}x_{t}(y_{e,t}^{\ast}-\phi p_{t})=O_{p}(1).
\]
Combining this with $\hat{\beta}-\beta=O_{p}(T^{-1/2})$ yields
\begin{equation}
T^{-1/2}\sum_{t\leq T}(\hat{y}_{t}^{\ast}-y_{t}^{\ast})(y_{e,t}^{\ast}-\phi
p_{t})=-\,T^{-1/2}\sum_{t\leq T}x_{t}(\hat{\beta}-\beta)(y_{e,t}^{\ast}-\phi
p_{t})=O_{p}(T^{-1/2}). \label{P_Moment_est1_3}%
\end{equation}
Moreover, by the same arguments as in (\ref{P_Sigma_est_3}),
\[
T^{-1/2}\sum_{t\leq T}x_{t}(\hat{\beta}-\beta)x_{e,t}(\hat{\beta}-\beta
)=O_{p}(T^{-1/2}).
\]
Therefore,
\begin{align}
T^{-1/2}\sum_{t\leq T}(\hat{y}_{t}^{\ast}-y_{t}^{\ast})(\hat{y}_{e,t}^{\ast
}-\phi p_{t})  &  =-T^{-1/2}\sum_{t\leq T}x_{t}(\hat{\beta}-\beta
)(y_{e,t}^{\ast}-\phi p_{t})\nonumber \\
&  \quad+T^{-1/2}\sum_{t\leq T}x_{t}(\hat{\beta}-\beta)x_{e,t}(\hat{\beta
}-\beta)=O_{p}(T^{-1/2}). \label{P_Moment_est1_4}%
\end{align}
Hence,%
\begin{equation}
A^{\top}T^{-1/2}\sum_{t\leq T}(\hat{y}_{t}^{\ast}-y_{t}^{\ast})(\hat{y}%
_{e,t}^{\ast}-\phi p_{t})=O_{p}(T^{-1/2}). \label{P_Moment_est1_4b}%
\end{equation}

\medskip \noindent \textit{Step 3: Fourth term.} To study the last term, we
first use (\ref{P_Moment_est1_4}) to obtain
\begin{align}
T^{-1/2}\sum_{t\leq T}\hat{y}_{t}^{\ast}(\hat{y}_{e,t}^{\ast}-\phi p_{t})  &
=T^{-1/2}\sum_{t\leq T}y_{t}^{\ast}(\hat{y}_{e,t}^{\ast}-\phi p_{t}%
)+T^{-1/2}\sum_{t\leq T}(\hat{y}_{t}^{\ast}-y_{t}^{\ast})(\hat{y}_{e,t}^{\ast
}-\phi p_{t})\nonumber \\
&  =T^{-1/2}\sum_{t\leq T}y_{t}^{\ast}v_{t}+T^{-1/2}\sum_{t\leq T}y_{t}^{\ast
}(\hat{y}_{e,t}^{\ast}-y_{e,t}^{\ast})+O_{p}(T^{-1/2}),
\label{P_Moment_est1_5}%
\end{align}
where $v_{t}\equiv y_{e,t}^{\ast}-\phi p_{t}$. From (\ref{P_Sigma_est_7}), we
have
\begin{equation}
T^{-1}\sum_{t\leq T}p_{t}x_{e,t}=O_{p}(1), \label{P_Moment_est1_8}%
\end{equation}
which, together with (\ref{P_Sigma_est_2}), implies
\begin{equation}
\mathbf{1}_{n}\Bigl(\phi T^{-1}\sum_{t\leq T}p_{t}x_{e,t}\Bigr)T^{1/2}%
(\hat{\beta}-\beta)=O_{p}(1). \label{P_Moment_est1_9}%
\end{equation}
From Lemmas \ref{Asy_L0} and \ref{Asy_Ahat_L1} (with $y_{t}$ in that lemma
corresponding to $y_{t}^{\ast}$ here), and \ref{Sigma_est}, we have
\begin{equation}
T^{1/2}(\tilde{A}-A)=\Upsilon \Bigl(T^{-1/2}\sum_{t\leq T}(y_{t}^{\ast}%
y_{t}^{\ast \top}-\mathbb{E}[y_{t}^{\ast}y_{t}^{\ast \top}])\Bigr)A+O_{p}%
(T^{-1/2})=O_{p}(1), \label{P_Moment_est1_10}%
\end{equation}
where the second equality follows from (\ref{P_L_Upsilon_7}) and Lemma
\ref{Asy_L0}.

Using $\Upsilon \mathbf{1}_{n}=0$, together with (\ref{P_Moment_est1_2a}) and
(\ref{P_Moment_est1_9}), we obtain
\begin{equation}
(\tilde{A}-A)^{\top}T^{-1/2}\sum_{t\leq T}y_{t}^{\ast}(\hat{y}_{e,t}^{\ast
}-y_{e,t}^{\ast})=O_{p}(T^{-1/2}). \label{P_Moment_est1_11}%
\end{equation}
By (\ref{P_V_Est_1b}) and (\ref{P_Asy_L5_1}) (with $y_{t}$ there corresponding
to $y_{t}^{\ast}$ here), we have
\begin{equation}
T^{-1}\sum_{t\leq T}\mathbb{E}[y_{t}^{\ast}v_{t}]=O(1),\quad T^{-1}\sum_{t\leq
T}(y_{t}^{\ast}v_{t}-\mathbb{E}[y_{t}^{\ast}v_{t}])=O_{p}(1).
\label{P_Moment_est1_12}%
\end{equation}
From\ (\ref{P_Moment_est1_10}) and (\ref{P_Moment_est1_12}), it follows that
\begin{align}
(\tilde{A}-A)^{\top}T^{-1/2}\sum_{t\leq T}y_{t}^{\ast}v_{t}  &  =A^{\top
}\Bigl(T^{-1/2}\sum_{t\leq T}(y_{t}^{\ast}y_{t}^{\ast \top}-\mathbb{E}%
[y_{t}^{\ast}y_{t}^{\ast \top}])\Bigr)\Upsilon \Bigl(T^{-1}\sum_{t\leq
T}\mathbb{E}[y_{t}^{\ast}v_{t}]\Bigr)\nonumber \\
&  \quad+O_{p}(T^{-1/2}). \label{P_Moment_est1_13}%
\end{align}
Combining (\ref{P_Moment_est1_5}), (\ref{P_Moment_est1_11}) and
(\ref{P_Moment_est1_13}),\ we obtain%
\begin{align}
(\tilde{A}-A)^{\top}T^{-1/2}\sum_{t\leq T}\hat{y}_{t}^{\ast}(\hat{y}%
_{e,t}^{\ast}-\phi p_{t})  &  =A^{\top}\Bigl(T^{-1/2}\sum_{t\leq T}%
(y_{t}^{\ast}y_{t}^{\ast \top}-\mathbb{E}[y_{t}^{\ast}y_{t}^{\ast \top
}])\Bigr)\Upsilon \Bigl(T^{-1}\sum_{t\leq T}\mathbb{E}[y_{t}^{\ast}%
v_{t}]\Bigr)\nonumber \\
&  \quad+O_{p}(T^{-1/2}). \label{P_Moment_est1_14}%
\end{align}

The result then follows from (\ref{P_Moment_est1_1}), (\ref{P_Moment_est1_1b}%
), (\ref{P_Moment_est1_2}), (\ref{P_Moment_est1_4b}), and
(\ref{P_Moment_est1_14}). \hfill$Q.E.D.$

\bigskip \ 

\begin{lemma}
\label{Moment_est_2}\ Under Assumptions \ref{ID}, \ref{Asy_Cond_1},
\ref{S} and \ref{Asy_Cond_EX}, we have
\begin{align*}
T^{-1/2}\sum_{t\leq T}\tilde{A}^{\top}\hat{y}_{t}^{\ast}\varepsilon_{t}  &
=A^{\top}T^{-1/2}\sum_{t\leq T}(y_{t}^{\ast}\varepsilon_{t}-\mathbb{E}%
[y_{t}^{\ast}\varepsilon_{t}])\\
&  \quad+A^{\top}\Bigl(T^{-1/2}\sum_{t\leq T}(y_{t}^{\ast}y_{t}^{\ast \top
}-\mathbb{E}[y_{t}^{\ast}y_{t}^{\ast \top}])\Bigr)\Upsilon \Bigl(T^{-1}%
\sum_{t\leq T}\mathbb{E}[y_{t}^{\ast}\varepsilon_{t}]\Bigr)\\
&  +O_{p}(T^{-1/2}).
\end{align*}

\end{lemma}

\noindent \textsc{Proof of Lemma \ref{Moment_est_2}.} We begin with the
decomposition
\begin{equation}
T^{-1/2}\sum_{t\leq T}\tilde{A}^{\top}\hat{y}_{t}^{\ast}\varepsilon
_{t}=T^{-1/2}\sum_{t\leq T}A^{\top}y_{t}^{\ast}\varepsilon_{t}+(\tilde
{A}-A)^{\top}T^{-1/2}\sum_{t\leq T}y_{t}^{\ast}\varepsilon_{t}+T^{-1/2}%
\sum_{t\leq T}\tilde{A}^{\top}(\hat{y}_{t}^{\ast}-y_{t}^{\ast})\varepsilon
_{t}. \label{P_Moment_est2_1}%
\end{equation}
By arguments similar to those used to establish (\ref{P_Asy_L5_6}), we have%
\[
T^{-1/2}\sum_{t\leq T}A^{\top}\mathbb{E}[y_{t}^{\ast}\varepsilon
_{t}]=\mathbf{0}_{n-\bar{r}}.
\]
Therefore,
\begin{equation}
T^{-1/2}\sum_{t\leq T}A^{\top}y_{t}^{\ast}\varepsilon_{t}=A^{\top}T^{-1/2}%
\sum_{t\leq T}(y_{t}^{\ast}\varepsilon_{t}-\mathbb{E}[y_{t}^{\ast}%
\varepsilon_{t}]). \label{P_Moment_est2_1b}%
\end{equation}

We next study the last term on the right-hand side of (\ref{P_Moment_est2_1}).
By definition,
\begin{equation}
T^{-1/2}\sum_{t\leq T}\tilde{A}^{\top}(\hat{y}_{t}^{\ast}-y_{t}^{\ast
})\varepsilon_{t}=-A^{\top}T^{-1/2}\sum_{t\leq T}x_{t}\varepsilon_{t}%
(\hat{\beta}-\beta)-(\tilde{A}-A)^{\top}T^{-1/2}\sum_{t\leq T}x_{t}%
\varepsilon_{t}(\hat{\beta}-\beta), \label{P_Moment_est2_2}%
\end{equation}
where we use $\hat{y}_{t}^{\ast}-y_{t}^{\ast}=-x_{t}(\hat{\beta}-\beta)$. From
Assumption \ref{Asy_Cond_EX}(iii), together with the condition $T^{-1}%
\sum_{t\leq T}\mathbb{E}[x_{t}\varepsilon_{t}]=0$, it follows that
\begin{equation}
T^{-1/2}\sum_{t\leq T}x_{t}\varepsilon_{t}=O_{p}(1). \label{P_Moment_est2_3}%
\end{equation}
This, combined with (\ref{P_Sigma_est_2}), (\ref{P_Moment_est1_10}), and
(\ref{P_Moment_est2_2}), implies that
\begin{equation}
T^{-1/2}\sum_{t\leq T}\tilde{A}^{\top}(\hat{y}_{t}^{\ast}-y_{t}^{\ast
})\varepsilon_{t}=O_{p}(T^{-1/2}). \label{P_Moment_est2_4}%
\end{equation}

Next, by arguments similar to those used in (\ref{P_Moment_est1_12}), we have
\[
T^{-1}\sum_{t\leq T}\mathbb{E}[y_{t}^{\ast}\varepsilon_{t}]=O(1),\qquad
T^{-1}\sum_{t\leq T}(y_{t}^{\ast}\varepsilon_{t}-\mathbb{E}[y_{t}^{\ast
}\varepsilon_{t}])=O_{p}(1).
\]
Combining this with (\ref{P_Moment_est1_10}) yields
\begin{equation}
(\tilde{A}-A)^{\top}T^{-1/2}\sum_{t\leq T}y_{t}^{\ast}\varepsilon_{t}=A^{\top
}\Bigl(T^{-1/2}\sum_{t\leq T}(y_{t}^{\ast}y_{t}^{\ast \top}-\mathbb{E}%
[y_{t}^{\ast}y_{t}^{\ast \top}])\Bigr)\Upsilon \Bigl(T^{-1}\sum_{t\leq
T}\mathbb{E}[y_{t}^{\ast}\varepsilon_{t}]\Bigr)+O_{p}(T^{-1/2}).
\label{P_Moment_est2_5}%
\end{equation}

Substituting (\ref{P_Moment_est2_1b}), (\ref{P_Moment_est2_4}) and
(\ref{P_Moment_est2_5}) into (\ref{P_Moment_est2_1}) completes the
proof.\hfill$Q.E.D.$

\section{Additional Simulation Results\label{APP_6}}

This appendix reports two additional sets of simulation results related to
Section~\ref{sec: MC}: the power of the feasible $J$-test in the extended
design and the performance of the BIC criterion in estimating the number of
latent factors.

\begin{figure}[!t]
	\caption{Empirical Power of the $J$-test: Extended Design ($d_{x}=3$)}%
	\label{fig:j_power_extended}
	\centering
	\includegraphics[width=0.9\textwidth]{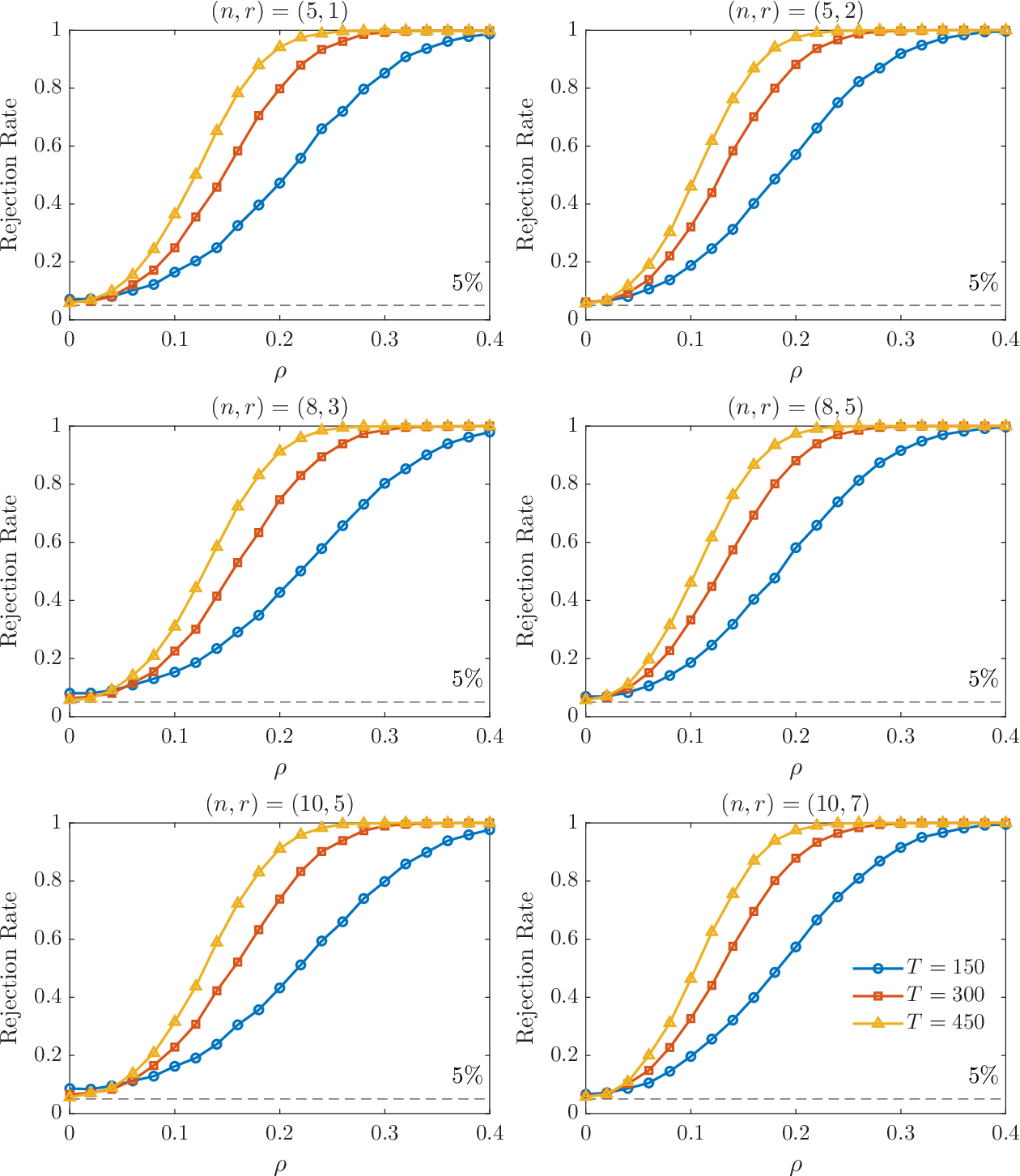}
	\begin{minipage}{0.95\textwidth}
		\footnotesize
		\noindent \textit{Notes.} This figure plots the empirical rejection probabilities of the $J$-test at the $5\%$ nominal significance level in the
		extended design with $d_x=3$ exogenous regressors, based on $10{,}000$ Monte
		Carlo replications. Each panel corresponds to one of the six configurations
		$(n,r)$ listed in~\eqref{MC_nr_pair}, and within each panel the three curves
		correspond to the sample sizes $T\in \{150,300,450\}$. The horizontal axis
		is $\rho$, while the vertical axis reports the
		rejection rate. The horizontal dashed line indicates the nominal $5\%$ level.
		Power is reported for the feasible GIV estimator, which selects $\bar r$ using
		the BIC criterion in~\eqref{BIC}--\eqref{r_hat} and is implemented according
		to Algorithm~1 in Appendix~\ref{APP_0}. In the extended design, the estimator additionally partials out $x_t\beta$ using OLS.
	\end{minipage}
\end{figure}

\subsection{Power of the $J$-test in the Extended Design}

In the extended design with $d_{x}=3$ exogenous regressors, the feasible
$J$-test exhibits behavior similar to that in the baseline design shown in
Figure~\ref{fig:j_power_baseline}: the power increases monotonically with
$\rho$, rises substantially as the sample size $T$ increases, and approaches
one by $\rho=0.4$ across all configurations (see
Figure~\ref{fig:j_power_extended} for details). Partialling out the exogenous
regressors $x_{t}$ has only a minor effect on power. Relative to the baseline
design, the power curves in the extended design lie only slightly below their
baseline counterparts at intermediate values of $\rho$, reflecting the
additional estimation error from estimating $\beta$.

\begin{table}[!t]
	\caption{BIC Selection Diagnostics for $\widehat{r}$.}%
	\label{tab:bic}%
	\centering
	{\small \setlength{\tabcolsep}{4.5pt} \renewcommand{\arraystretch}{1.05}
		\begin{threeparttable}
			\begin{tabular*}{\textwidth}{@{\extracolsep{\fill}}cc ccc ccc@{}}
				\toprule
				& & \multicolumn{3}{c}{Baseline ($d_x = 0$)}
				& \multicolumn{3}{c}{Extended ($d_x = 3$)} \\
				\cmidrule(lr){3-5}\cmidrule(lr){6-8}
				$(n, r)$ & $T$
				& $\textnormal{mean}(\widehat r)$
				& $\textnormal{mode}(\widehat r)$
				& $\mathbb{P}(\widehat r = \bar r)$
				& $\textnormal{mean}(\widehat r)$
				& $\textnormal{mode}(\widehat r)$
				& $\mathbb{P}(\widehat r = \bar r)$ \\
				\midrule
				\multirow{3}{*}{$(5, 1)$}
				& 150 & 2.008  & 2 & 0.994  & 2.008  & 2 & 0.994 \\
				& 300 & 2.010  & 2 & 0.992  & 2.010  & 2 & 0.992 \\
				& 450 & 2.009  & 2 & 0.993  & 2.010  & 2 & 0.991 \\
				\addlinespace
				\multirow{3}{*}{$(5, 2)$}
				& 150 & 3.020  & 3 & 0.980  & 3.022  & 3 & 0.978 \\
				& 300 & 3.021  & 3 & 0.979  & 3.021  & 3 & 0.979 \\
				& 450 & 3.021  & 3 & 0.979  & 3.017  & 3 & 0.983 \\
				\addlinespace
				\multirow{3}{*}{$(8, 3)$}
				& 150 & 4.000  & 4 & 1.000  & 4.000  & 4 & 1.000 \\
				& 300 & 4.002  & 4 & 0.999  & 4.001  & 4 & 0.999 \\
				& 450 & 4.002  & 4 & 0.998  & 4.002  & 4 & 0.998 \\
				\addlinespace
				\multirow{3}{*}{$(8, 5)$}
				& 150 & 6.025  & 6 & 0.975  & 6.025  & 6 & 0.975 \\
				& 300 & 6.022  & 6 & 0.978  & 6.022  & 6 & 0.978 \\
				& 450 & 6.020  & 6 & 0.980  & 6.020  & 6 & 0.980 \\
				\addlinespace
				\multirow{3}{*}{$(10, 5)$}
				& 150 & 6.001  & 6 & 0.999  & 6.001  & 6 & 1.000 \\
				& 300 & 6.002  & 6 & 0.999  & 6.001  & 6 & 0.999 \\
				& 450 & 6.003  & 6 & 0.998  & 6.002  & 6 & 0.998 \\
				\addlinespace
				\multirow{3}{*}{$(10, 7)$}
				& 150 & 8.026  & 8 & 0.974  & 8.027  & 8 & 0.974 \\
				& 300 & 8.024  & 8 & 0.976  & 8.022  & 8 & 0.978 \\
				& 450 & 8.022  & 8 & 0.978  & 8.023  & 8 & 0.977 \\
				\bottomrule
			\end{tabular*}
			\begin{tablenotes}[flushleft]
				\footnotesize
				\item \textit{Notes.} This table reports BIC selection diagnostics for the rank estimator
				$\widehat r$, based on $10{,}000$ Monte Carlo replications for each design
				configuration. Results are presented separately for the baseline design
				($d_x=0$) and the extended design with $d_x=3$ exogenous regressors. The
				columns $\mathrm{mean}(\widehat r)$, $\mathrm{mode}(\widehat r)$, and
				$\mathbb{P}(\widehat r=\bar r)$ report the average value, modal value, and
				probability of correct selection of the BIC estimator $\widehat r$ defined
				in~\eqref{BIC}--\eqref{r_hat}, where $\bar r=r+1$ denotes the true rank of
				$(\mathbf{1}_{n},\lambda)$.
			\end{tablenotes}
		\end{threeparttable}
}\end{table}

\subsection{BIC Rank Selection}

The feasible GIV estimation and inference procedures determine the rank $\bar
r$ using the BIC estimator proposed in \eqref{BIC}--\eqref{r_hat}.
Table~\ref{tab:bic} shows that the BIC criterion performs remarkably well in
estimating the true rank across all configurations of $(n,r)$ in both the
baseline and extended designs. In particular, the modal value of $\widehat r$
always coincides with the true rank $\bar r=r+1$, and the average value of
$\widehat r$ is extremely close to $\bar r$ in every case.

The probability of correct selection is uniformly high across all
configurations, typically exceeding $97\%$ and often above $99\%$. For
example, in the baseline design with $(n,r)=(5,1)$, the correct selection
probabilities are $0.994$, $0.992$, and $0.993$ at $T=150$, $300$, and $450$,
respectively. Even in configurations with a larger number of latent factors,
such as $(n,r)=(10,7)$, the correct selection probabilities remain close to
$98\%$ in both designs. Moreover, the performance of the BIC criterion is
nearly identical in the baseline and extended designs, indicating that the
additional estimation of $\beta$ in the extended model has little effect on
rank determination.

Overall, these results suggest that the proposed BIC procedure provides a
highly reliable estimator of the rank $\bar{r}$, supporting the practical
implementation of the feasible GIV estimation and inference procedures.

\section{Additional Empirical Results\label{app:sectors}}

\subsection{Variable Construction \label{subsec:variable}}

We construct the sector-level holding growth $\Delta q_{i,t}$ and the
aggregate price change $\Delta p_{t}$ following Appendix~C of \citet{GIVOA-gabaix2021search}.

The Financial Accounts report the dollar value of each sector's U.S. equity
holdings, denoted by $w_{i,t}$. Changes in $w_{i,t}$ reflect both capital
appreciation of existing holdings and net purchases or sales of equities. To
isolate the latter component, let $q_{i,t}$ denote the number of shares held
by sector $i$ and let $P_{t}$ denote the aggregate equity price. Since
$w_{i,t}=q_{i,t}P_{t}$, a sector that does not trade between $t-1$ and $t$
would end the quarter holding $w_{i,t-1}R_{t}$ dollars of equity, where
$R_{t}=P_{t}/P_{t-1}$ is the gross capital-appreciation return on the
aggregate market. It follows that
\[
\Delta q_{i,t}\equiv \frac{q_{i,t}}{q_{i,t-1}}-1=\frac{w_{i,t}/P_{t}}%
{w_{i,t-1}/P_{t-1}}-1=\frac{w_{i,t}}{w_{i,t-1}R_{t}}-1,
\]
which is the empirical measure used in \eqref{eq:dq}. Define $\Delta
f_{i,t}\equiv(w_{i,t}-w_{i,t-1}R_{t})/w_{i,t-1}$. We can equivalently write
$\Delta q_{i,t}=\Delta f_{i,t}/R_{t}$, which corresponds to the flow-based
measure reported in Appendix~C.1.2 of \citet{GIVOA-gabaix2021search}. The gross
capital-appreciation return $R_{t}$ is constructed from the CRSP
value-weighted index excluding dividends and compounded from monthly to
quarterly frequency.

The aggregate price change $\Delta p_{t}$ is measured by the quarterly simple
return on the CRSP value-weighted index including dividends, compounded from
monthly returns.

\subsection{Sector Construction and Weights \label{subsec:sector}}

We construct the sector classification from the corporate-equity holdings
reported in Financial Accounts Table~L.224, following the data construction in
Appendix~C of \citet{GIVOA-gabaix2021search}. The resulting cross-section consists
of twelve equity-holding sectors, listed in Table~\ref{tab:sectors}. These
sectors correspond closely to the investor categories used in \citet{GIVOA-gabaix2021search}.

For each sector, we compute its time-$t$ market share as
\begin{equation}
S_{i,t}\equiv \frac{w_{i,t-1}}{\sum_{j\leq12}w_{j,t-1}}. \label{GK_Sit}%
\end{equation}
Thus,\ $S_{i,t}$ measures sector $i$'s share of the total equity holdings of
the twelve modeled sectors and satisfies $\sum_{i=1}^{12}S_{i,t}=1$.
Table~\ref{tab:sectors} reports the average market share, $\bar{S}_{i}%
=T^{-1}\sum_{t\leq T}S_{i,t}$,\ computed over the benchmark sample 1993Q1--2018Q4.

\begin{table}[ptbh]
\caption{Institutional Investor Sectors Holding U.S. Equities}%
\label{tab:sectors}
\centering
\begin{threeparttable}
		\small
		\begin{tabular*}{0.95\textwidth}{@{\extracolsep{\fill}} c l r r c @{}}
			\toprule
			Rank & Sector & Avg.\ share (\%) & Cum.\ share (\%) & Cross-section \\
			\midrule
			1  & Households                          & $41.63$ & $41.63$  & $6,\,12$ \\
			2  & Mutual funds and ETFs               & $23.42$ & $65.05$  & $6,\,12$ \\
			3  & Foreign sector                      & $13.00$ & $78.05$  & $6,\,12$ \\
			4  & Private pension funds               & $9.92$  & $87.97$  & $6,\,12$ \\
			5  & State and local pension funds       & $7.54$  & $95.51$  & $6,\,12$ \\
			6  & Life insurance companies            & $1.89$  & $97.40$  & $6,\,12$ \\
			\addlinespace
			7  & Property and casualty insurers      & $1.17$  & $98.57$  & $12$ \\
			8  & Federal government retirement funds & $0.44$  & $99.01$  & $12$ \\
			9  & State and local governments         & $0.43$  & $99.44$  & $12$ \\
			10 & Closed-end funds                    & $0.38$  & $99.82$  & $12$ \\
			11 & Banks                               & $0.12$  & $99.94$  & $12$ \\
			12 & Broker-dealers                      & $0.06$  & $100.00$ & $12$ \\
			\bottomrule
		\end{tabular*}
		\begin{tablenotes}[flushleft]
			\footnotesize
			\item \textit{Notes.} The table lists the twelve institutional investor sectors holding U.S. equities in the Financial Accounts of the United States (Table~L.224), ranked by average market share over 1993Q1--2018Q4. ``Avg.\ share'' denotes the sector's average share of aggregate equity holdings (in percent), and the shares sum to 100. ``Cum.\ share'' denotes the cumulative share accounted for by the sector and all sectors ranked above it. ``Cross-section'' identifies the estimation samples that include the sector: $n=6$ corresponds to the granular core formed by the six largest sectors, while $n=12$ corresponds to the full twelve-sector panel.
		\end{tablenotes}
	\end{threeparttable}
\end{table}

We rank sectors according to their average market shares $\bar{S}_{i}$. The
six largest sectors jointly account for 97.4\% of total equity holdings and
constitute the granular core ($n=6$): households, mutual funds and ETFs, the
foreign sector, private pension funds, state and local pension funds, and life
insurance companies. The remaining six sectors each account for less than
1.2\% of total equity holdings and, together with the granular core, form the
full twelve-sector panel ($n=12$).

\subsection{Implementation of the Factor-residual IV Estimator
\label{subsec:gk-replication}}

This subsection describes our implementation of the FIV estimator of
\citet{GIVOA-gabaix2021search}. To facilitate comparison with their results, we
follow the estimation procedure in Section 4.2 of their paper as closely as possible.

\bigskip

\noindent \textsc{Algorithm 3 (FIV Estimation)}

\noindent \textbf{Step 0.} For each sector $i$, let $\hat{\sigma}_{q,i}^{2}$
denote the time-series sample variance of $\Delta q_{i,t}$. Construct the
precision weights
\[
E_{i}=\frac{1/\bar{\sigma}_{q,i}^{2}}{\sum_{j\leq n}1/\bar{\sigma}_{q,j}^{2}%
},\qquad \bar{\sigma}_{q,i}^{2}=\max \bigl \{ \hat{\sigma}_{q,i}^{2}%
,\  \operatorname{median}_{j\leq n}\hat{\sigma}_{q,j}^{2}\bigr \}.
\]

\noindent \textbf{Step 1.} Estimate the two-way fixed-effects regression by
weighted least squares:
\[
(\hat{\alpha},\hat{\gamma},\hat{\beta})=\arg \min_{{\alpha_{i}},{\gamma_{t}%
},{\beta_{i}}}\sum_{t\leq T}\sum_{i\leq n}E_{i}\left(  \Delta q_{i,t}%
-\alpha_{i}-\gamma_{t}-\beta_{i}^{\top}\eta_{o,t}\right)  ^{2},
\]
where $\alpha_{i}$ and $\gamma_{t}$ are sector and time fixed effects,
respectively, and $\eta_{o,t}$ denotes the vector of observed factors with
sector-specific loadings $\beta_{i}$. The residual is
\[
\Delta \check{q}_{i,t}\equiv \Delta q_{i,t}-\hat{\alpha}_{i}-\hat{\gamma}%
_{t}-\hat{\beta}_{i}^{\top}\eta_{o,t}.
\]

\noindent \textbf{Step 2.} Extract $r$ latent factors, denoted by $\hat{\eta
}_{l,t}$, as the principal components of the weighted residuals $E_{i}%
^{1/2}\Delta \check{q}_{i,t}$. Given the estimated factors $\hat{\eta}_{l,t}$,
estimate the corresponding factor loadings $\hat{\lambda}_{i}$ and recover the
idiosyncratic shock $\check{u}_{i,t}$ from the regression
\[
\Delta \check{q}_{i,t}=\hat{\lambda}_{i}^{\top}\hat{\eta}_{l,t}+\check{u}%
_{i,t}.
\]

\noindent \textbf{Step 3.} Construct the FIV as the size-weighted sum of the
estimated idiosyncratic shocks
\[
z_{t}=\sum_{i\leq n}S_{i,t}\, \check{u}_{i,t},
\]
where $S_{i,t}$ is defined in (\ref{GK_Sit}), and estimate the market
multiplier from the time-series regression
\begin{equation}
\Delta p_{t}=\alpha+\kappa z_{t}+\Gamma^{\top}\eta_{t}+e_{t},\text{
\  \  \ }\eta_{t}=(\eta_{o,t}^{\top},\hat{\eta}_{l,t}^{\top})^{\top}
\label{reduced_reg}%
\end{equation}
where $\Delta p_{t}$ denotes the aggregate market return and $\kappa$ is the
aggregate market multiplier.

\noindent \textbf{Step 4.} Recompute the precision weights using the
time-series sample variance of the estimated idiosyncratic shocks $\check
u_{i,t}$ and repeat Steps~1--3 once. Report the coefficient estimates and
their Newey--West standard errors from the time-series regression in
\eqref{reduced_reg}, without adjusting for the estimation error in $z_{t}$ and
$\eta_{t}$.

\bigskip

In our empirical implementation of the FIV estimator, the observed factors
$\eta_{o,t}$ consist of GDP growth together with a set of asset-pricing
factors. Specifically, following \citet{GIVOA-gabaix2021search}, we run monthly
cross-sectional regressions of stock returns on log market capitalization, the
log book-to-market ratio, and a momentum signal. Individual stock returns,
prices, and shares outstanding are obtained from CRSP, while book equity used
to construct the book-to-market ratio is obtained from Compustat, matching the
data inputs used by \citet{GIVOA-gabaix2021search}. GDP growth is
	computed from real gross domestic product (series GDPC1), produced by the U.S.
	Bureau of Economic Analysis and retrieved from the Federal Reserve Economic
	Data (FRED) database maintained by the Federal Reserve Bank of St.\ Louis. The
resulting Fama--MacBeth slope estimates are aggregated to the quarterly
frequency and used as the observed size, value, and momentum factors.

\bigskip

{\small
\input{GIV_OA_merged.bbl}

}

\end{document}

%% file: Granular_IV_v7_3.bbl
\ifx\undefined\BySame
\newcommand{\BySame}{\leavevmode\rule[.5ex]{3em}{.5pt}\ }
\fi
\ifx\undefined\textsc
\newcommand{\textsc}[1]{{\sc #1}}
\newcommand{\emph}[1]{{\em #1\/}}
\let\tmpsmall\small
\renewcommand{\small}{\tmpsmall\sc}
\fi

%% file: GIV_OA_merged.bbl
\ifx\undefined\BySame
\newcommand{\BySame}{\leavevmode\rule[.5ex]{3em}{.5pt}\ }
\fi
\ifx\undefined\textsc
\newcommand{\textsc}[1]{{\sc #1}}
\newcommand{\emph}[1]{{\em #1\/}}
\let\tmpsmall\small
\renewcommand{\small}{\tmpsmall\sc}
\fi

%% file: GIV.bib
@article{koijen2019demand,
	title={A demand system approach to asset pricing},
	author={Koijen, Ralph SJ and Yogo, Motohiro},
	journal={Journal of Political Economy},
	volume={127},
	number={4},
	pages={1475--1515},
	year={2019},
	publisher={The University of Chicago Press Chicago, IL}
}

@Article{HansenHeatonYaron1996,
  author  = {Hansen, Lars Peter and Heaton, John and Yaron, Amir},
  journal = {Journal of Business \& Economic Statistics},
  title   = {Finite-sample properties of some alternative {GMM} estimators},
  year    = {1996},
  number  = {3},
  pages   = {262--280},
  volume  = {14},
  doi     = {10.1080/07350015.1996.10524656},
}

@article{gabaix2011granular,
	title={The granular origins of aggregate fluctuations},
	author={Gabaix, Xavier},
	journal={Econometrica},
	volume={79},
	number={3},
	pages={733--772},
	year={2011},
	publisher={Wiley Online Library}
}

@Article{gabaix2024granular,
  author    = {Gabaix, Xavier and Koijen, Ralph SJ},
  journal   = {Journal of Political Economy},
  title     = {Granular instrumental variables},
  year      = {2024},
  number    = {7},
  pages     = {2274--2303},
  volume    = {132},
  publisher = {The University of Chicago Press Chicago, IL},
}

@article{qian2023heterogeneity,
  title={Heterogeneity-robust granular instruments},
  author={Qian, Eric},
  journal={arXiv preprint arXiv:2304.01273},
  year={2023}
}

@Unpublished{gabaix2021search,
  author = {Gabaix, Xavier and Koijen, Ralph SJ},
  note   = {SSRN Working Paper No. 3686935},
  title  = {In search of the origins of financial fluctuations: The inelastic markets hypothesis},
  year   = {2023},
  url    = {https://papers.ssrn.com/sol3/papers.cfm?abstract_id=3686935},
}

@article{camanho2022global,
  title={Global portfolio rebalancing and exchange rates},
  author={Camanho, Nelson and Hau, Harald and Rey, H{\'e}lene},
  journal={The Review of Financial Studies},
  volume={35},
  number={11},
  pages={5228--5274},
  year={2022},
  publisher={Oxford University Press}
}

@article{ma2022mutual,
  title={Mutual fund liquidity transformation and reverse flight to liquidity},
  author={Ma, Yiming and Xiao, Kairong and Zeng, Yao},
  journal={The Review of Financial Studies},
  volume={35},
  number={10},
  pages={4674--4711},
  year={2022},
  publisher={Oxford University Press}
}

@techreport{galaasen2020granular,
  title={Granular credit risk},
  author={Galaasen, Sigurd and Jamilov, Rustam and Juelsrud, Ragnar and Rey, H{\'e}lene},
  year={2020},
  institution={National Bureau of Economic Research}
}

@Article{dong2025fast,
  author  = {Dong, Xi and Kang, Namho and Peress, Joel},
  journal = {The Review of Financial Studies},
  title   = {Fast and slow arbitrage: The predictive power of persistent capital flows for factor returns},
  year    = {2025},
  number  = {9},
  pages   = {2936--2987},
  volume  = {38},
}

@article{bai2003inferential,
  title={Inferential theory for factor models of large dimensions},
  author={Bai, Jushan},
  journal={Econometrica},
  volume={71},
  number={1},
  pages={135--171},
  year={2003},
  publisher={Wiley Online Library}
}

@Article{yu2015useful,
  author    = {Yu, Yi and Wang, Tengyao and Samworth, Richard J},
  journal   = {Biometrika},
  title     = {A useful variant of the Davis--Kahan theorem for statisticians},
  year      = {2015},
  number    = {2},
  pages     = {315--323},
  volume    = {102},
  publisher = {Oxford University Press},
}

@Article{BCCK2015,
  author   = {Alexandre Belloni and Victor Chernozhukov and Denis Chetverikov and Kengo Kato},
  journal  = {Journal of Econometrics},
  title    = {Some New Asymptotic Theory for Least Squares Series: Pointwise and Uniform Results},
  year     = {2015},
  issn     = {0304-4076},
  note     = {High Dimensional Problems in Econometrics},
  number   = {2},
  pages    = {345 - 366},
  volume   = {186},
  doi      = {https://doi.org/10.1016/j.jeconom.2015.02.014},
  url      = {http://www.sciencedirect.com/science/article/pii/S030440761500038X},
}

@Article{ackerberg2014asymptotic,
  author    = {Ackerberg, Daniel and Chen, Xiaohong and Hahn, Jinyong and Liao, Zhipeng},
  journal   = {Review of Economic Studies},
  title     = {Asymptotic efficiency of semiparametric two-step GMM},
  year      = {2014},
  number    = {3},
  pages     = {919--943},
  volume    = {81},
  publisher = {Oxford University Press},
}

@Article{banafti2022inferential,
  author  = {Banafti, Saman and Lee, Tae-Hwy},
  journal = {arXiv preprint arXiv:2201.06605},
  title   = {Inferential theory for granular instrumental variables in high dimensions},
  year    = {2022},
}

@Article{newey2009generalized,
  author    = {Newey, Whitney K and Windmeijer, Frank},
  journal   = {Econometrica},
  title     = {Generalized method of moments with many weak moment conditions},
  year      = {2009},
  number    = {3},
  pages     = {687--719},
  volume    = {77},
  publisher = {Wiley Online Library},
}

@Article{han2006gmm,
  author    = {Han, Chirok and Phillips, Peter CB},
  journal   = {Econometrica},
  title     = {GMM with many moment conditions},
  year      = {2006},
  number    = {1},
  pages     = {147--192},
  volume    = {74},
  publisher = {Wiley Online Library},
}

@article{acemoglu2012network,
  title={The network origins of aggregate fluctuations},
  author={Acemoglu, Daron and Carvalho, Vasco M and Ozdaglar, Asuman and Tahbaz-Salehi, Alireza},
  journal={Econometrica},
  volume={80},
  number={5},
  pages={1977--2016},
  year={2012},
  publisher={Wiley Online Library}
}

@article{di2014firms,
  title={Firms, destinations, and aggregate fluctuations},
  author={di Giovanni, Julian and Levchenko, Andrei A and Mejean, Isabelle},
  journal={Econometrica},
  volume={82},
  number={4},
  pages={1303--1340},
  year={2014},
  publisher={Wiley Online Library}
}

@article{baqaee2019macroeconomic,
  title={The macroeconomic impact of microeconomic shocks: Beyond Hulten's theorem},
  author={Baqaee, David Rezza and Farhi, Emmanuel},
  journal={Econometrica},
  volume={87},
  number={4},
  pages={1155--1203},
  year={2019},
  publisher={Wiley Online Library}
}

@article{gaubert2021granular,
  title={Granular comparative advantage},
  author={Gaubert, Cecile and Itskhoki, Oleg},
  journal={Journal of Political Economy},
  volume={129},
  number={3},
  pages={871--939},
  year={2021},
  publisher={The University of Chicago Press Chicago, IL}
}

@article{leary2014peer,
  title={Do peer firms affect corporate financial policy?},
  author={Leary, Mark T and Roberts, Michael R},
  journal={The Journal of Finance},
  volume={69},
  number={1},
  pages={139--178},
  year={2014},
  publisher={Wiley Online Library}
}

@article{amiti2018much,
  title={How much do idiosyncratic bank shocks affect investment? Evidence from matched bank-firm loan data},
  author={Amiti, Mary and Weinstein, David E},
  journal={Journal of Political Economy},
  volume={126},
  number={2},
  pages={525--587},
  year={2018},
  publisher={University of Chicago Press Chicago, IL}
}

@article{amiti2019international,
  title={International shocks, variable markups, and domestic prices},
  author={Amiti, Mary and Itskhoki, Oleg and Konings, Jozef},
  journal={The Review of Economic Studies},
  volume={86},
  number={6},
  pages={2356--2402},
  year={2019},
  publisher={Oxford University Press}
}

@Article{baumeister2023uncovering,
  author      = {Baumeister, Christiane and Hamilton, James D},
  journal     = {Working paper, UCSD},
  title       = {Uncovering disaggregated oil market dynamics: A full-information approach to granular instrumental variables},
  year        = {2023},
  institution = {Working paper, UCSD},
}

@Article{lucas1978asset,
  author  = {Lucas, Robert E. Jr.},
  journal = {Econometrica},
  title   = {Asset prices in an exchange economy},
  year    = {1978},
  number  = {6},
  pages   = {1429--1445},
  volume  = {46},
}

@article{barro2006rare,
  title={Rare disasters and asset markets in the twentieth century},
  author={Barro, Robert J},
  journal={The Quarterly Journal of Economics},
  volume={121},
  number={3},
  pages={823--866},
  year={2006},
  publisher={MIT Press}
}

@article{gabaix2012variable,
  title={Variable rare disasters: An exactly solved framework for ten puzzles in macro-finance},
  author={Gabaix, Xavier},
  journal={The Quarterly Journal of Economics},
  volume={127},
  number={2},
  pages={645--700},
  year={2012},
  publisher={MIT Press}
}

@article{bansal2004risks,
  title={Risks for the long run: A potential resolution of asset pricing puzzles},
  author={Bansal, Ravi and Yaron, Amir},
  journal={The Journal of Finance},
  volume={59},
  number={4},
  pages={1481--1509},
  year={2004},
  publisher={Wiley Online Library}
}

@article{lou2012flow,
  title={A flow-based explanation for return predictability},
  author={Lou, Dong},
  journal={The Review of Financial Studies},
  volume={25},
  number={12},
  pages={3457--3489},
  year={2012},
  publisher={Society for Financial Studies}
}

@article{pavlova2023benchmarking,
  title={Benchmarking intensity},
  author={Pavlova, Anna and Sikorskaya, Taisiya},
  journal={The Review of Financial Studies},
  volume={36},
  number={3},
  pages={859--903},
  year={2023},
  publisher={Oxford University Press}
}

@Article{piazzesi2007asset,
  author    = {Piazzesi, Monika and Schneider, Martin},
  journal   = {Journal of the European Economic Association},
  title     = {Asset prices and asset quantities},
  year      = {2007},
  number    = {2-3},
  pages     = {380--389},
  volume    = {5},
  publisher = {Oxford University Press},
}
